\documentclass{article}

\def\noheaderplainsetup{

\topmargin=0pt \headheight=0pt \headsep=0pt  \oddsidemargin=0pt \evensidemargin=0pt  \textheight=9.1truein \textwidth=6.5truein}   

\noheaderplainsetup

\usepackage{amsfonts}

\begin{document}


\newcommand{\clthree}{\mbox{\bf CL3}}
\newcommand{\clfour}{\mbox{\bf CL4}}
\newcommand{\arfour}{\mbox{\bf PTA}} 
\newcommand{\pa}{\mbox{\bf PA}} 
\newcommand{\bound}{\mathfrak{b}} 


\newcommand{\zero}{\mbox{\small {\bf 0}}}
\newcommand{\adi}{\hspace{2pt}\raisebox{0.02cm}{\mbox{\small $\sqsupset$}}\hspace{2pt}} 
\newcommand{\plus}{\mbox{\hspace{1pt}\raisebox{0.05cm}{\tiny\boldmath $+$}\hspace{1pt}}}
\newcommand{\mult}{\mbox{\hspace{1pt}\raisebox{0.05cm}{\tiny\boldmath $\times$}\hspace{1pt}}}
\newcommand{\mminus}{\mbox{\hspace{1pt}\raisebox{0.05cm}{\tiny\boldmath $-$}\hspace{1pt}}}
\newcommand{\equals}{\mbox{\hspace{1pt}\raisebox{0.05cm}{\tiny\boldmath $=$}\hspace{1pt}}}
\newcommand{\notequals}{\mbox{\hspace{1pt}\raisebox{0.05cm}{\tiny\boldmath $\not=$}\hspace{1pt}}}
\newcommand{\successor}{\mbox{\hspace{1pt}\boldmath $'$}}

\newcommand{\mless}{\mbox{\hspace{1pt}\raisebox{0.05cm}{\tiny\boldmath $<$}\hspace{1pt}}}
\newcommand{\mgreater}{\mbox{\hspace{1pt}\raisebox{0.05cm}{\tiny\boldmath $>$}\hspace{1pt}}}
\newcommand{\mleq}{\mbox{\hspace{1pt}\raisebox{0.05cm}{\tiny\boldmath $\leq$}\hspace{1pt}}}
\newcommand{\mgeq}{\mbox{\hspace{1pt}\raisebox{0.05cm}{\tiny\boldmath $\geq$}\hspace{1pt}}}

\newcommand{\elz}[1]{\mbox{$\parallel\hspace{-3pt} #1 \hspace{-3pt}\parallel$}} 
\newcommand{\elzi}[1]{\mbox{\scriptsize $\parallel\hspace{-3pt} #1 \hspace{-3pt}\parallel$}}
\newcommand{\emptyrun}{\langle\rangle} 
\newcommand{\oo}{\bot}            
\newcommand{\pp}{\top}            
\newcommand{\xx}{\wp}               
\newcommand{\legal}[2]{\mbox{\bf Lr}^{#1}_{#2}} 
\newcommand{\win}[2]{\mbox{\bf Wn}^{#1}_{#2}} 
\newcommand{\seq}[1]{\langle #1 \rangle}           


\newcommand{\pst}{\mbox{\raisebox{-0.01cm}{\scriptsize $\wedge$}\hspace{-4pt}\raisebox{0.16cm}{\tiny $\mid$}\hspace{2pt}}}
\newcommand{\pcost}{\mbox{\raisebox{0.12cm}{\scriptsize $\vee$}\hspace{-4pt}\raisebox{0.02cm}{\tiny $\mid$}\hspace{2pt}}}

\newcommand{\gneg}{\mbox{\small $\neg$}}                  
\newcommand{\mli}{\hspace{2pt}\mbox{\small $\rightarrow$}\hspace{2pt}}                      
\newcommand{\cla}{\mbox{$\forall$}}      
\newcommand{\cle}{\mbox{$\exists$}}        
\newcommand{\mld}{\hspace{2pt}\mbox{\small $\vee$}\hspace{2pt}}     
\newcommand{\mlc}{\hspace{2pt}\mbox{\small $\wedge$}\hspace{2pt}}   
\newcommand{\mlci}{\hspace{2pt}\mbox{\footnotesize $\wedge$}\hspace{2pt}}   
\newcommand{\ade}{\mbox{\large $\sqcup$}}      
\newcommand{\ada}{\mbox{\large $\sqcap$}}      
\newcommand{\add}{\hspace{2pt}\mbox{\small $\sqcup$}\hspace{2pt}}                     
\newcommand{\adc}{\hspace{2pt}\mbox{\small $\sqcap$}\hspace{2pt}} 
\newcommand{\adci}{\hspace{2pt}\mbox{\footnotesize $\sqcap$}\hspace{2pt}}              
\newcommand{\clai}{\forall}     
\newcommand{\clei}{\exists}        
\newcommand{\tlg}{\bot}               
\newcommand{\twg}{\top}               
\newcommand{\fintimpl}{\mbox{\hspace{2pt}$\bullet$\hspace{-0.14cm} \raisebox{-0.058cm}{\Large --}\hspace{-6pt}\raisebox{0.008cm}{\scriptsize $\wr$}\hspace{-1pt}\raisebox{0.008cm}{\scriptsize $\wr$}\hspace{4pt}}}
\newcommand{\col}[1]{\mbox{$#1$:}}


\newtheorem{theoremm}{Theorem}[section]
\newtheorem{factt}[theoremm]{Fact}
\newtheorem{definitionn}[theoremm]{Definition}
\newtheorem{lemmaa}[theoremm]{Lemma}
\newtheorem{propositionn}[theoremm]{Proposition}
\newtheorem{conventionn}[theoremm]{Convention}
\newtheorem{examplee}[theoremm]{Example}
\newtheorem{exercisee}[theoremm]{Exercise}
\newenvironment{definition}{\begin{definitionn} \em}{ \end{definitionn}}
\newenvironment{theorem}{\begin{theoremm}}{\end{theoremm}}
\newenvironment{lemma}{\begin{lemmaa}}{\end{lemmaa}}
\newenvironment{fact}{\begin{factt}}{\end{factt}}
\newenvironment{proposition}{\begin{propositionn} }{\end{propositionn}}
\newenvironment{convention}{\begin{conventionn} \em}{\end{conventionn}}
\newenvironment{example}{\begin{examplee} \em}{\end{examplee}}
\newenvironment{exercise}{\begin{exercisee} \em}{\end{exercisee}}
\newenvironment{proof}{ {\bf Proof.} }{\  \rule{2.5mm}{2.5mm} \vspace{.2in} }
\newenvironment{idea}{ {\bf Idea.} }{\  \rule{1.5mm}{1.5mm} \vspace{.15in} }
\newenvironment{subproof}{ {\em Proof.} }{\  \rule{2mm}{2mm} \vspace{.1in} }

\title{Ptarithmetic}
\author{Giorgi Japaridze\thanks{This material is based upon work supported by the National Science Foundation under Grant No. 0208816}   \\  
\\ {\footnotesize Department of Computing Sciences, Villanova University, 800 Lancaster Avenue, Villanova, PA 19085, USA.}\\
{\footnotesize Email: giorgi.japaridze@villanova.edu
 \ URL: http://www.csc.villanova.edu/$^\sim$japaridz/}}
\date{}
\maketitle

\begin{abstract} The present article introduces  {\em ptarithmetic} (short for ``polynomial time arithmetic'') --- a formal number theory similar to  the well known Peano arithmetic, but based  on the recently born {\em computability logic}   instead of classical logic. The formulas of ptarithmetic represent interactive computational problems rather than just true/false statements, and their ``truth'' is understood as existence of a polynomial time solution. The  system  of ptarithmetic elaborated in this article is shown to be sound and complete. Sound     in the sense that every theorem $T$ of the system represents an interactive  number-theoretic computational problem with a polynomial time solution and, furthermore, such a solution can be effectively extracted from a proof of $T$. And complete in the sense that every interactive number-theoretic problem with a polynomial time solution is represented by some theorem $T$ of the system. 

The paper is self-contained, and can be read without any prior familiarity with computability logic.  
\end{abstract}


\


\tableofcontents

\section{Introduction}\label{intr}

{\em Computability logic} (CL), introduced in \cite{Jap03,Japic,Japfin}, is a semantical, mathematical and philosophical platform, and an ambitious program, for redeveloping logic as a formal theory of computability, as opposed to the formal theory of truth which logic has more traditionally been. 
 Under the approach of CL, formulas represent computational problems, 
and their ``truth'' is seen as algorithmic solvability. In turn, computational problems --- understood in their  most general, {\em interactive} sense --- are defined as games played by a machine against its environment, with ``algorithmic solvability'' meaning existence of a machine that wins the game against any possible behavior of the environment. And an open-ended collection of the most basic and natural operations 
on computational problems forms the logical vocabulary of the theory.  With this semantics, CL provides a systematic answer to the fundamental question ``{\em what can be computed?}\hspace{1pt}'', just as classical logic is a systematic tool for telling what is true. Furthermore, as it turns out, in positive cases ``{\em what} can be computed'' always allows itself to be replaced by ``{\em how} can be computed'', which makes CL of potential interest in not only theoretical computer science, but many more applied areas as well, including interactive knowledge base systems, resource oriented systems for planning and action, or declarative programming languages. 

While potential applications have been repeatedly pointed out and outlined in introductory papers on CL, so far all technical efforts had been mainly focused on finding axiomatizations for various fragments of this semantically conceived and inordinately expressive logic. Considerable advances have already been made in this direction (\cite{Japtocl1}-\cite{Cirq}, \cite{Japtcs}-\cite{Japfour}, \cite{Japtowards}, \cite{Ver}), and more results in the same style are probably still to come. It should be however remembered that the main value of CL, or anything else claiming to be a ``Logic'' with a capital ``L'', will  eventually be determined by whether and how it relates to the outside, extra-logical world. In this respect, unlike many  other systems officially qualified as ``logics'', the merits of classical logic  are obvious, most eloquently demonstrated by the fact that applied formal theories, a model example of which is {\em Peano arithmetic} $\pa$,\label{iPA} can be and have been successfully based on it. Unlike pure logics with their meaningless symbols, such theories are direct tools for studying and navigating the real world with its non-man-made, meaningful objects, such as natural numbers in the case of arithmetic. To make this point more clear to a computer scientist, one could compare a pure  logic  with a programming language, and applied theories based on it with application programs written in that language. A programming language created for its own sake, mathematically or esthetically appealing  but otherwise unusable as a general-purpose, comprehensive basis for application  programs, would hardly be of much interest.  

So, in parallel with studying possible axiomatizations and various metaproperties of pure computability logic, it would certainly be worthwhile to devote some efforts to justifying its right on existence through revealing its power and appeal as a basis for applied theories. First and so far the only concrete steps in this direction have been made only very recently in \cite{Japtowards}, where a CL-based system {\bf CLA1} of (Peano) arithmetic was constructed.\footnote{The paper \cite{Xu} (in Chinese) is apparently another exception, focused on applications of CL in AI.} Unlike its classical-logic-based counterpart {\bf PA}, {\bf CLA1} is not merely about what arithmetical facts are {\em true}, but about what arithmetical problems can be actually {\em computed} or effectively {\em solved}. More precisely, every formula of the language of {\bf CLA1} expresses a number-theoretic computational {\em problem} (rather than just a true/false {\em fact}), every theorem expresses a problem that has an algorithmic solution,  and every proof encodes such a solution. Does not this sound  exactly like what the constructivists have been calling for? 

Unlike the mathematical or philosophical constructivism, however, and even unlike the early-day theory of computation, modern computer science has long
understood that, what really matters, is not just {\em computability}, but rather {\em efficient computability}. So, the next natural step on the road of revealing the importance of CL for computer science would be showing that it can be used for  studying efficient computability just as successfully as for studying computability-in-principle. Anyone familiar with the earlier work on CL could have found reasons for optimistic expectations here. Namely, every provable formula of any of the known sound  axiomatizations of CL happens to be a scheme of  not only ``always computable'' problems, but ``always efficiently computable'' problems just as well, whatever efficiency exactly mens in the context of interactive computation that CL operates in. That is, at the level of pure logic, computability and efficient computability yield the same classes of valid principles. The study of logic abounds with   phenomena in this style. One example would be the well known fact about classical logic, according to which validity with respect to all possible models is equivalent to validity with respect to just models with countable domains.       

At the level of reasonably expressive applied theories, however, one should certainly expect significant differences depending on whether the underlying concept of interest is  efficient computability or computability-in-principle. For instance, the earlier-mentioned system {\bf CLA1} proves formulas expressing computable but not always efficiently computable arithmetical problems. The purpose of the present paper is to construct a CL-based system for arithmetic which, unlike {\bf CLA1}, proves only efficiently --- specifically, polynomial time --- computable problems. The new applied formal theory $\arfour$\label{iPTA} (``{\em ptarithmetic}'', short for ``polynomial time arithmetic'') presented in Section \ref{ss11} achieves this purpose. 

Just like {\bf CLA1}, our present system $\arfour$ is not only a cognitive, but also a problem-solving tool: in order to find a solution for a given problem, it would be sufficient to write the problem in the language of the system, and find a proof for it. An algorithmic solution for the problem then would automatically come together with such a proof. However, unlike the solutions extracted from {\bf CLA1}-proofs, which might be intractable, the solutions extracted from $\arfour$-proofs would always be efficient. 

Furthermore, $\arfour$ turns out to be not only sound, but also complete in a certain reasonable sense that we call {\em extensional completeness}.\label{iextcom} According to the latter, every number-theoretic computational problem that has a polynomial time solution is represented by some theorem of  $\arfour$. Taking into account that there are many ways to represent the same problem, extensional completeness is weaker than what can be called {\em intensional completeness},\label{iintcom} according to which any formula representing an (efficiently) computable problem is provable. In these terms, G\"{o}del's celebrated theorem,\label{igincom} here  with ``truth''=``computability'', is about intensional rather than extensional incompleteness. In fact, extensional completeness is not at all interesting in the context of classical-logic-based theories such as $\pa$. In such theories, unlike computability-logic-based theories, it is trivially achieved, as the provable formula $\twg$ represents every true sentence.  

Syntactically, our $\arfour$ is an extension of {\bf PA}, and the semantics of the former is a conservative generalization of the semantics of the latter. Namely, the formulas of $\pa$, which form only a proper subclass of the formulas of $\arfour$, are seen as special, ``moveless'' sorts of problems/games, automatically solved/won when true and failed/lost when false. This makes the classical concept of truth just a special case of computability in our sense --- it is nothing but computability restricted to (the problems represented) by the traditional sorts of formulas. And this means that G\"{o}del's incompleteness theorems automatically extend from $\pa$ to $\arfour$, so that, unlike extensional completeness,  intensional completeness in $\arfour$  or any other sufficiently expressive CL-based applied theory is impossible to achieve in principle. As for {\bf CLA1}, it turns out to be incomplete in both senses. Section \ref{sincom} shows that any sufficiently expressive sound  system would be (not only intensionally but also) extensionally incomplete, as long as the semantics of the system is based on unrestricted (as opposed to, say, efficient) computability.

Among the main moral merits of the present investigation and its contributions to the overall CL project is an illustration of the fact that,  in constructing CL-based applied theories, successfully switching from computability to efficient computability is possible and even more than just possible. As noted, efficient computability, in fact, turns out to be much better behaved than computability-in-principle: the former allows us to achieve completeness in a sense in which 
the latter yields inherent incompleteness.    

An advanced reader will easily understand that the present paper, while focused on the system $\arfour$ of (pt)arithmetic, in fact is not only about arithmetic, but also just as much about CL-based applied theories or knowledge base systems in general, with $\arfour$ only serving as a model example of such systems and a starting point for what may be a separate (sub)line of research within the CL enterprise. 
Generally, the nonlogical axioms or the knowledge base of a CL-based applied  system would be any collection of (formulas expressing) problems whose algorithmic or efficient solutions are known. Sometimes, together with nonlogical axioms, we may also have nonlogical rules of inference, preserving the property of computability or efficient computability. An example of such a rule is the {\em polynomial time induction} (PTI) rule of  $\arfour$.  Then, the soundness of the corresponding underlying axiomatization of CL (in our present case, it is system  \clthree\ studied in \cite{Japtcs}) --- which usually comes in the strong form called {\em uniform-constructive soundness} --- guarantees that every theorem $T$ of the theory also has an effective or efficient solution and that, furthermore, such
a solution can be effectively extracted from a proof of $T$. 
It is this fact that, as mentioned,  makes CL-based systems problem-solving tools.

Having said the above, motivationally (re)introducing and (re)justifying computability logic is not among the goals of the present paper. This job has been done in \cite{Jap03,Japic,Japfin},  and the reader would benefit from getting familiar with any of those pieces of literature first, of which most recommended is the first 10 tutorial-style sections of \cite{Japfin}. While helpful in fully appreciating the import of the present results, however,
from the purely technical point of view, such  familiarity  is not necessary, as this paper provides all relevant definitions. 

\section{An informal overview of the main operations on games}\label{ss2}

As noted, formulas in CL represent computational problems. Such problems are understood as games between two players: $\pp$,\label{ipp} called   {\bf machine},\label{imachine} and $\oo$,\label{ioo} called  {\bf environment}.\label{ienvironment} $\pp$ is a mechanical device with a fully determined, algorithmic behavior.  On the other hand, there are no restrictions on the behavior of $\oo$. A given machine is considered to be  {\em solving} a given problem iff it wins the corresponding game no matter 
how the environment acts. 

Standard atomic sentences, such as ``$0\equals 0$'' or ``Peggy is John's mother'', are understood as special sorts of games, called {\bf elementary}.\label{ielem1} There are no moves in elementary games, and they are automatically won or lost. Specifically, the elementary game represented by a true sentence is won (without making any moves) by the machine, and the elementary game represented by a false sentence is won by the environment.  
 
Logical operators are understood as operations on games/problems. One of the important groups of such operations, called {\bf choice operations},\label{ichoiceop} comprises  $\adc,\add,\ada,\ade$. These are called {\bf choice conjunction},\label{ichoicecon} {\bf choice disjunction},\label{ichoicedis} {\bf choice universal quantifier}\label{ichoiceuq} and {\bf choice existential quantifier},\label{ichoiceeq} respectively. $A_0\adc A_1$ is a game where the first legal move (``choice"), which should be either $0$ or $1$, is by $\oo$. After such a move/choice $i$ is made, the play continues and the winner is determined according to the rules of $A_i$; if a choice is never made, $\oo$ loses. 
 $A_0\add A_1$ is defined in a symmetric way with the roles of $\oo$ and $\pp$ interchanged: here it is $\pp$ who makes an initial choice and who loses if such a choice is not made. With the universe of discourse being $\{0,1,10,11,100,\ldots\}$ (natural numbers identified with their binary representations), the meanings of the quantifiers $\ada$ and $\ade$  can now be explained by 
\[\ada x A(x)= A(0)\adc A(1)\adc A(10)\adc A(11)\adc A(100)\adc \ldots\] and \[\ade x A(x)= A(0)\add A(1)\add A(10)\add A(11)\add A(100)\add \ldots.\] 

So, for example, 
\[\ada x\bigl(\mbox{\em Prime}(x)\add \mbox{\em Composite}(x)\bigr)\]
is a game where the first move is by the environment. Such a move should consist in selecting a particular number $n$ for $x$, intuitively amounting to asking whether $n$ is prime or composite. This move brings the game down to (in the sense that the game continues as) 
\[\mbox{\em Prime}(n)\add \mbox{\em Composite}(n).\]
Now the machine has to move, or else it loses. The move should consist in choosing one of the two disjuncts. Let us say the left disjunct is chosen, which further brings the game down to $\mbox{\em Prime}(n)$. The latter is an elementary game, and here the interaction ends. The machine wins iff it has chosen a true disjunct. The choice of the left disjunct by the machine thus amounts to claiming/answering that $n$ is prime. Overall, as we see,  $\ada x\bigl(\mbox{\em Prime}(x)\add \mbox{\em Composite}(x)\bigr)$ represents the problem of deciding the primality question.\footnote{For simplicity, here we treat ``Composite'' as the complement of ``Prime'', even though, strictly speaking, this is not quite so: the numbers $0$ and $1$ are neither prime nor composite. Writing ``Nonprime'' instead of ``Composite'' would easily correct this minor inaccuracy.} 

Similarly, 
\[\ada x\ada y\ade z(z\equals x\mult  y)\]
is the problem of computing the product of any two numbers. Here the first two moves are by the environment, which selects some particular $m= x$ and $n= y$, thus asking the machine to tell what the product of $m$ and $n$ is. The machine wins if and only if, in response, it selects a (the) number $k$ for $z$ such that $k\equals m\mult  n$. 

The present paper replaces the above-described choice quantifiers $\ada$ and $\ade$ with their {\em bounded} counterparts $\ada^{\bound}$\label{iadab} and $\ade^{\bound}$,\label{iadeb} where $\bound$ is a variable. These are the same as $\ada$ and $\ade$, with the difference that the choice here is limited only to the objects of the universe of discourse whose sizes do not exceed a certain  bound,\label{ibound} which is represented by the variable $\bound$.  So,  $\ada^{\bound} x A(x)$ is essentially the same as $\ada x \bigl(|x|\mleq \bound\mli A(x)\bigr)$  and $\ade^{\bound} x A(x)$ is essentially  the same as  $\ade x \bigl(|x|\mleq \bound\mlc A(x)\bigr)$, where (the meanings of $\mli,\mlc$ will be explained shortly and) $|x|\mleq \bound$ means ``the size of $x$ does not exceed $\bound$''. As we are going to see later, it is exactly the value of  $\bound$ with respect to which the computational complexity of games will be measured.

Another group of game operations dealt with in this paper, two of which have already been used in the previous paragraph, comprises $\gneg,\mlc,\mld,\mli$. Employing the classical symbols for these operations is no accident, as they are conservative generalizations of the corresponding Boolean operations from elementary games to all games. 

{\bf Negation} $\gneg$\label{igneg} is a role-switch operation: it turns $\pp$'s moves and wins into $\oo$'s moves and wins, and vice versa. Since elementary games have no moves, only the winners are switched there, so that, as noted, $\gneg$ acts just as the ordinary classical negation. For instance, as $\pp$ is the winner in $0\plus 1\equals 1$, the winner in $\gneg 0\plus 1\equals 1$ will be $\oo$. That is, $\pp$ wins the negation $\gneg A$ of an elementary game $A$ iff it loses $A$, i.e., if $A$ is false. As for the meaning of negation when applied to nonelementary games, at this point it may be useful to observe that $\gneg$ interacts with choice operations in the kind old 
DeMorgan fashion. For example, it would not be hard to see that \[\gneg \ada x\ada y\ade z(z\equals x\mult  y)\ = \ 
\ade x\ade y\ada z(z\notequals x\mult  y).\]

The operations $\mlc$\label{imlc} and $\mld$\label{imld} are called {\bf parallel conjunction} and {\bf parallel disjunction}, respectively.  Playing $A_0\mlc A_1$ (resp. $A_0\mld A_1$) means playing the two games in parallel where, in order to win, $\pp$ needs to win in both (resp. at least one) of the components $A_i$. It is obvious that, just as in the case of negation, $\mlc$ and $\mld$ act as classical conjunction and disjunction when applied to elementary games. For instance, $0\plus 1\equals 1\mld 0\mult  1\equals 1$ is a game automatically won by the machine. There are no moves in it as there are no moves in either disjunct, and the machine is an automatic winner because it is so in the left disjunct. To appreciate the difference between the two --- choice and parallel --- groups of connectives, compare \[\ada x\bigl(\mbox{\em Prime}(x)\add \gneg \mbox{\em Prime}(x)\bigr)\] and \[\ada x\bigl(\mbox{\em Prime}(x)\mld \gneg \mbox{\em Prime}(x)\bigr).\] The former is a computationally nontrivial problem, existence of an easy (polynomial time) solution for which had remained an open question until a few years ago. As for the latter, it is trivial, as the machine has nothing to do in it: the first (and only) move is by the environment, consisting in choosing a number $n$ for $x$. Whatever $n$ is chosen, the machine wins, as $\mbox{\em Prime}(n)\mld \gneg \mbox{\em Prime}(n)$ is a true sentence and hence an automatically $\pp$-won elementary game.

The operation $\mli$,\label{imli} called {\bf reduction}, is defined by $A\mli B= (\gneg A)\mld B$. Intuitively, this is indeed the problem of {\em reducing} 
$B$ to $A$: solving $A\mli B$ means solving $B$ while having $A$ as an external {\em computational resource}. Resources are symmetric to problems: what is a problem to solve for one player is a resource that the other player can use, and vice versa. Since 
$A$ is negated in  $(\gneg A)\mld B$ and negation means switching the roles, $A$ appears as a resource rather than problem for 
$\pp$ in $A\mli B$. 

Consider $\ada x\ade  y(y\equals x^2)$. Anyone who knows the definition of $x^2$ in terms of $\mult $ (but perhaps does not know the meaning of multiplication, or is unable to compute this  function for whatever reason) would be able to solve the  problem 
\begin{equation}\label{april15}
\ada z\ada u\ade v(v\equals  z\mult  u)\ \mli \ \ada x\ade y(y\equals x^2),
\end{equation}
i.e., the problem 
\[\ade z\ade u\ada v(v\notequals z\mult  u)\ \mld \ \ada x\ade y(y\equals x^2),\]
 as it is about reducing the consequent to the antecedent. 
A solution here goes like this. Wait till the environment specifies a value $n$ for $x$, i.e. asks ``what is the square of $n$?''. Do not try to immediately answer this question, but rather specify the same value $n$ for both $z$ and $u$, thus asking the counterquestion: ``what is $n$ times $n$?''. The environment will have to provide a correct answer $m$  to this counterquestion (i.e., specify $v$ as $m$ where $m= n\mult  n$), or else it loses. Then, specify $y$ as $m$, and rest your case. Note that, in this solution, the machine did not have to compute multiplication, doing which had become the environment's responsibility. The machine only correctly reduced the problem of computing square to the problem of computing product, which made it the winner.

Another group of operations that play an important role in CL comprises  $\cla$\label{icla} and its dual $\cle$\label{icle} (with $\cle xA(x)= \gneg\cla x\gneg A(x)$), called  {\bf blind universal quantifier} and {\bf blind existential quantifier}, respectively.  $\cla xA(x)$ 
can be thought of as a ``version" of $\ada xA(x)$ where the particular value of $x$ that the environment selects is invisible to the machine, so that it has to play blindly in a way that guarantees success no matter what that value is. 

Compare the problems
\[\ada x\bigl(\mbox{\em Even$(x)$}\add \mbox{\em Odd$(x)$}\bigr)\] and  \[\cla x\bigl(\mbox{\em Even$(x)$}\add \mbox{\em Odd$(x)$}\bigr).\] 
Both of them are about telling whether a given number is even or odd; the difference is only in whether that ``given number" is known to the machine or not. The first problem is an easy-to-win, two-move-deep game of a structure that we have already seen.  The second game, on the other hand, is one-move deep with only  the machine to make a move --- select the ``true"  disjunct, which is hardly possible to do as the value of $x$ remains unspecified.

Just like all other operations for which we use classical symbols, the meanings of $\cla$ and $\cle$ are exactly classical  when applied to elementary games. Having this full collection of classical operations makes computability logic a generalization and conservative extension of classical logic. 

Going back to an earlier example, even though  (\ref{april15}) expresses a ``very easily solvable'' problem, that formula is still not logically valid. Note that the successfulness of the reduction strategy of the consequent to the antecedent that we provided for it relies on the nonlogical fact that $x^2\equals x\mult  x$. That strategy would fail in a general case where the meanings of $x^2$ and $x\mult  x$ may not necessarily be the same. On the other hand, the goal of CL as a general-purpose problem-solving tool should be to allow us find purely logical solutions, i.e., solutions that do not require any special, domain-specific knowledge and (thus) would be good no matter what the particular predicate or function symbols of the formulas mean. Any knowledge that might be relevant should be explicitly stated and included either in the antecedent of a given formula or in the set of axioms (``implicit antecedents'' for every potential formula) of a CL-based theory. 
In our present case,  formula (\ref{april15}) easily turns into a logically valid one by adding, to its antecedent,  the definition of square in terms of multiplication:
\begin{equation}\label{april16}
\cla w (w^2\equals w\mult  w) \mlc \ada z\ada u\ade v(v\equals  z\mult  u)\ \mli \ \ada x\ade y(y\equals x^2).
\end{equation}
The strategy that we provided earlier for (\ref{april15}) is just as good for (\ref{april16}), with the difference that it is successful for (\ref{april16}) no matter what $x^2$ and $z\mult  u$ mean, whereas, in the case of (\ref{april15}), it was guaranteed to be successful only under the standard arithmetic interpretations of the square and product functions. Thus, our strategy for (\ref{april16}) is, in fact, a ``purely logical'' solution. Again, among the purposes of computability logic is to serve as a tool for finding such ``purely logical'' solutions, so that it can be applied to any domain of study rather than specific domains such as that of arithmetic, and to arbitrary meanings of nonlogical symbols rather than particular meanings such as that of the multiplication function for the symbol $\mult $. 

The above examples should not suggest that blind quantifiers are meaningful or useful  only when applied to elementary problems. The following is an example of an effectively winnable nonelementary $\cla$-game:

\begin{equation}\label{lkj}\cla y\Bigl(\mbox{\em Even$(y)$}\add \mbox{\em Odd$(y)$}\ \mli\ \ada x\bigl(\mbox{\em Even$(x\plus  y)$}\add
\mbox{\em Odd$(x\plus  y)$}\bigr)\Bigr).\vspace{-3pt}\end{equation}
Solving this problem, which means reducing the consequent to the antecedent without knowing the value of $y$, is easy: 
$\pp$ waits till $\oo$ selects  a value $n$ for $x$, and also tells --- by selecting a disjunct in the antecedent --- whether $y$ is even or odd. Then, 
if $n$ and $y$ are both even or both odd, $\pp$ chooses the first $\add$-disjunct in the consequent, otherwise it chooses the second $\add$-disjunct. Replacing the $\cla y$ prefix by $\ada y$ would significantly weaken the problem, obligating the environment to specify a value for $y$. Our strategy does not really need to know the exact value of $y$, as it only exploits the information about $y$'s being even or odd, provided by the antecedent of the formula.

Many more --- natural, meaningful and useful --- operations beyond the ones discussed in this section have been introduced and studied in computability logic. Here we have only surveyed those that are relevant to our present investigation.

\section{Constant games}\label{cg}

Now we are getting down to formal definitions of the concepts informally explained in the previous section. 

To define games formally, we need some technical terms and conventions. Let us agree that by a {\bf move}\label{imove} we mean any finite string over the standard keyboard alphabet. 
A {\bf labeled move} ({\bf labmove})\label{ilabmove} is a move prefixed with $\pp$ or $\oo$, with such a prefix ({\bf label})\label{ilabel} indicating which player has made the move. 
A {\bf run}\label{irun} is a (finite or infinite) sequence of labmoves, and a {\bf position}\label{iposition} is a finite run. 

\begin{convention}\label{conv1}
We will be exclusively using the letters $\Gamma,\Delta,\Phi$  for runs, and  $\alpha,\beta$ for moves. The letter $\xx$\label{ixx} will always be a variable for players, and \[\overline{\xx}\label{ixxneg}\]  will mean ``$\xx$'s adversary'' (``the other player'').
Runs will be often delimited by ``$\langle$" and ``$\rangle$", with $\emptyrun$ thus denoting the {\bf empty run}.\label{iempty} The meaning of an expression such as $\seq{\Phi,\xx\alpha,\Gamma}$ must be clear: this is the result of appending to the position $\seq{\Phi}$ 
the labmove $\seq{\xx\alpha}$ and then the run $\seq{\Gamma}$.  
\end{convention}

The following is a formal definition of what we call constant games, combined with some less formal conventions regarding the usage of certain terminology.

\begin{definition}\label{game}
 A {\bf constant game}\label{iconstantgame} is a pair $A= (\legal{A}{},\win{A}{})$, where:

1. $\legal{A}{}$\label{ilr} is a set of runs  satisfying the condition that a (finite or infinite) run is in $\legal{A}{}$ iff all of its nonempty finite  initial
segments are in $\legal{A}{}$ (notice that this implies $\emptyrun\in\legal{A}{}$). The elements of $\legal{A}{}$ are
said to be {\bf legal runs}\label{ilegrun} of $A$, and all other runs are said to be {\bf illegal}.\label{iillegrun} We say that $\alpha$ is a {\bf legal move}\label{ilegmove} for $\xx$ in a position $\Phi$ of $A$ iff $\seq{\Phi,\xx\alpha}\in\legal{A}{}$; otherwise 
$\alpha$ is {\bf illegal}.\label{iillegmove} When the last move of the shortest illegal initial segment of $\Gamma$  is $\xx$-labeled, we say that $\Gamma$ is a {\bf $\xx$-illegal}\label{ipillegal} run of $A$. 

2. $\win{A}{}$\label{iwn}  is a function that sends every run $\Gamma$ to one of the players $\pp$ or $\oo$, satisfying the condition that if $\Gamma$ is a $\xx$-illegal run of $A$, then $\win{A}{}\seq{\Gamma}= \overline{\xx}$. When $\win{A}{}\seq{\Gamma}= \xx$, we say that $\Gamma$ is a {\bf $\xx$-won}\label{iwon} (or {\bf won by $\xx$}) run of $A$; otherwise $\Gamma$ is {\bf lost}\label{ilost} by $\xx$. Thus, an illegal run is always lost by the player who has made the first illegal move in it.  
\end{definition}

An important operation not explicitly mentioned in Section \ref{ss2} is what is called {\em prefixation}.\label{iprefixation}
This operation takes two arguments: a constant game $A$ and a position $\Phi$ 
 that must 
be a legal position of $A$ (otherwise the operation is undefined), and returns the game $\seq{\Phi}A$.
Intuitively, $\seq{\Phi}A$ is the game playing which means playing $A$ starting (continuing) from position $\Phi$. 
That is, $\seq{\Phi}A$ is the game to which $A$ {\bf evolves} (will be ``{\bf brought down}") after the moves of $\Phi$ have been made. We have already used this intuition when explaining the meaning of choice operations in Section \ref{ss2}: we said that after $\oo$ makes an initial move $i\in\{0,1\}$,
 the game 
$A_0\adc A_1$ continues as $A_i$. What this meant was nothing but that 
$\seq{\oo i}(A_0\adc A_1)= A_i$.
Similarly, $\seq{\pp i}(A_0\add A_1)= A_i$. Here is a definition of prefixation:

\begin{definition}\label{prfx}
Let $A$ be a constant game and $\Phi$ a legal position of $A$. The game 
$\seq{\Phi}A$\label{ipr} is defined by: 
\begin{itemize}
\item $\legal{\seq{\Phi}A}{}= \{\Gamma\ |\ \seq{\Phi,\Gamma}\in\legal{A}{}\}$;
\item $\win{\seq{\Phi}A}{}\seq{\Gamma}= \win{A}{}\seq{\Phi,\Gamma}$.
\end{itemize}
\end{definition}

\begin{convention}\label{poscon}
A terminological convention important to remember is that we often identify a legal position $\Phi$ of a game $A$ with the game $\seq{\Phi}A$. So, for instance, we may say that the move $1$ by $\oo$ brings the game $B_0\adc B_1$ down to the position $B_1$. Strictly speaking, $B_1$ is not a position but a game, and what {\em is} a position is $\seq{\oo 1}$, which we here identified with the game $B_1=\seq{\oo 1}(B_0\adc B_1)$.
\end{convention}

We say that a constant game $A$ is {\bf finite-depth} iff there is an integer $d$ such that no legal run of $A$ contains more than $d$ labmoves. The smallest of such integers $d$ is called the {\bf depth}\label{idepth} of $A$. An {\bf elementary game}\label{ielgame2} is a game of depth $0$.

In this paper will exclusively deal with finite-depth games. This restriction of focus makes many definitions and proofs simpler. Namely, in order to define a finite-depth-preserving game operation $O(A_1,\ldots,A_n)$ applied to such games, it suffices  to specify the following:

\begin{description}
\item[(i)] Who wins $O(A_1,\ldots,A_n)$ if no moves are made, i.e., the value of $\win{O(A_1,\ldots,A_n)}{}\emptyrun$.
\item[(ii)] What are the {\bf initial legal (lab)moves},\label{iilm} i.e., the elements of  $\{\xx\alpha\ |\ \seq{\xx\alpha}\in\legal{O(A_1,\ldots,A_n)}{}\}$, and to 
what game  is $O(A_1,\ldots,A_n)$ brought down after such an initial legal labmove $\xx\alpha$ is made. Recall that, by saying that a given labmove $\xx\alpha$ brings a given game $A$ down to $B$, we mean that $\seq{\xx\alpha}A= B$.  
\end{description}
Then, the set of legal runs of $O(A_1,\ldots,A_n)$ will be uniquely defined, and so will be the winner in every legal (and hence finite) run of the game. 

Below we define a number of operations for finite-depth games only. Each of these operations can be easily seen to preserve the finite-depth property. Of course, more general definitions of these operations --- not restricted to finite-depth games --- do exist (see, e.g., \cite{Japfin}), but in this paper we are trying to keep things as simple as possible, and reintroduce only as much of computability logic as necessary.

\begin{definition}\label{op} Let $A$, $B$, $A_0,A_1,\ldots$ be finite-depth constant games, and $n$ be a positive integer.\vspace{9pt}

\noindent 1. $\gneg A$\label{igneg2} is defined by: 
\begin{quote}\begin{description}
\item[(i)] $\win{\gneg A}{}\emptyrun = \xx$ iff $\win{A}{}\emptyrun =\overline{\xx}$. 
\item[(ii)] $\seq{\xx\alpha}\in\legal{\gneg A}{}$ iff $\seq{\overline{\xx}\alpha}\in\legal{A}{}$. Such an initial legal labmove $\xx\alpha$ brings the game down to 
$\gneg \seq{\overline{\xx}\alpha}A$.\vspace{5pt}
\end{description}\end{quote}

\noindent 2. $A_0\adc\ldots\adc  A_n$\label{iadc2} is defined by: 
\begin{quote}\begin{description}
\item[(i)] $\win{A_0\adci\ldots\adci  A_n}{}\emptyrun = \pp$. 
\item[(ii)] $\seq{\xx\alpha}\in\legal{A_0\adci\ldots\adci  A_n}{}$ iff $\xx= \oo$ and $\alpha= i  \in\{0,\ldots,n\}$.\footnote{According to our conventions, such a natural number $i$ is identified with its binary representation. The same applies to the other clauses of this definition.}  Such an initial legal labmove $\oo i$ brings the game down to 
$A_i$.\vspace{5pt} 
\end{description}\end{quote}

\noindent 3. $A_0\mlc\ldots\mlc A_n$\label{imlc2} is defined by: 
\begin{quote}\begin{description}
\item[(i)] $\win{A_0\mlci\ldots\mlci  A_n}{}\emptyrun= \pp$ iff, for each $i\in\{0,\ldots,n\}$,  $\win{A_i}{}\emptyrun= \pp$. 
\item[(ii)] $\seq{\xx\alpha}\in\legal{A_0\mlci\ldots\mlci  A_n}{}$ iff $\alpha= i.\beta$, where $i\in\{0,\ldots,n\}$ and $\seq{\xx\beta}\in\legal{A_i}{}$. Such an initial legal labmove $\xx i.\beta$ brings the game down to  
\[ A_0\mlc\ldots\mlc A_{i-1}\mlc \seq{\xx\beta}A_i\mlc A_{i\plus 1}\mlc\ldots\mlc A_n.\vspace{3pt}\] 
\end{description}\end{quote}

\noindent 4. $A_0\add\ldots\add A_n$\label{iadd2} and $A_0\mld\ldots\mld  A_n$
are defined exactly as $A_0\adc\ldots\adc A_n$ and $A_0\mlc\ldots\mlc  A_n$, respectively, only with ``$\pp$" and ``$\oo$" interchanged.\vspace{7pt}


\noindent 5. In addition to the earlier-established meanings, the symbols $\twg$\label{itwg2} and $\tlg$ also denote two special --- simplest --- constant games, defined by  $\win{\twg}{}\emptyrun=\pp$, $\win{\tlg}{}\emptyrun= \oo$ and $\legal{\twg}{}= \legal{\tlg}{}= \{\emptyrun\}$.\vspace{7pt} 

\noindent 6. $A\mli B$\label{imli2} is treated as an abbreviation of $(\gneg A)\mld B$.
\end{definition}

\begin{example}
 The game $(0\equals 0\adc 0\equals 1)\mli(10\equals 11\adc 10\equals 10)$, i.e. \[\gneg (0\equals 0\adc 0\equals 1)\mld(10\equals 11\adc 10\equals 10),\]
 has thirteen legal runs, which are: 
\begin{description}
\item[1] $\seq{}$. It is won by $\pp$, because $\pp$ is the winner in the right $\mld$-disjunct (consequent).
\item[2] $\seq{\pp 0.0}$. (The labmove of) this run brings the game down to $\gneg 0\equals 0\mld(10\equals 11\adc 10\equals 10)$, and $\pp$ is the winner for the same reason as in the previous case.
\item[3] $\seq{\pp 0.1}$. It brings the game down to $\gneg 0\equals 1\mld(10\equals 11\adc 10\equals 10)$, and $\pp$ is the winner because it wins in both $\mld$-disjuncts. 
\item[4] $\seq{\oo 1.0}$. It brings the game down to $\gneg(0\equals 0\adc 0\equals 1)\mld 10\equals 11$.  $\pp$ loses as it loses in both $\mld$-disjuncts. 
\item[5] $\seq{\oo 1.1}$. It brings the game down to $\gneg (0\equals 0\adc 0\equals 1)\mld 10\equals 10$.  $\pp$ wins as it wins in the right $\mld$-disjunct. 
\item[6-7] $\seq{\pp 0.0,\oo 1.0}$ and $\seq{\oo 1.0, \pp 0.0}$. Both bring the game down to the false $\gneg 0\equals 0 \mld 10\equals 11$, and both are lost by  $\pp$. 
\item[8-9] $\seq{\pp 0.1,\oo 1.0}$ and $\seq{\oo 1.0, \pp 0.1}$. Both bring the game down to the true $\gneg 0\equals 1 \mld 10\equals 11$, which makes  $\pp$ the winner.
\item[10-11] $\seq{\pp 0.0,\oo 1.1}$ and $\seq{\oo 1.1, \pp 0.0}$. Both bring the game down to the true $\gneg 0\equals 0 \mld 10\equals 10$, so $\pp$ wins.
\item[12-13] $\seq{\pp 0.1,\oo 1.1}$ and $\seq{\oo 1.1, \pp 0.1}$. Both bring the game down to the true $\gneg 0\equals 1 \mld 10\equals 10$, so $\pp$ wins.
\end{description}
\end{example}

\section{Games as generalized predicates}\label{nncg}

Constant games can be seen as generalized propositions: while propositions in classical logic are just elements 
of $\{\twg,\tlg\}$, constant games are functions from runs to $\{\twg,\tlg\}$.
As we know, however, propositions only offer a very limited expressive power, 
and classical logic needs 
to consider the more general concept of predicates, with propositions being nothing but special --- constant --- cases of predicates. The situation in computability logic is similar. Our concept of a (simply) game generalizes that of a constant game in the same sense as the classical concept of a predicate generalizes that of a proposition.

We fix an infinite set of expressions called {\bf variables}:\label{ivariable} 
\[\{\mathfrak{w}_0,\mathfrak{w}_1,\mathfrak{w}_2,\mathfrak{w}_3,\ldots\}.\] The letters \[x,y,z,s,r,t,u,v,w\] will be used as metavariables for these variables. The Gothic letter \[\bound\label{ipi}\] will be exclusively used as a metaname for the variable $\mathfrak{w}_0$, which
is going to have a special status throughout our entire treatment. 

We also fix another infinite set of expressions called {\bf constants}:\label{iconstant} 
\[\{0,1,10,11,100,101,110,111,1000,\ldots\}.\] These are thus  {\bf binary numerals}\label{ibinnum} --- the strings matching the regular expression $0\cup 1(0\cup 1)^*$.  We will be typically identifying such strings --- by some rather innocent abuse of concepts --- with the natural numbers represented by them in the standard binary notation, and vice versa. The above collection of constants is going to be exactly the {\em universe of discourse} 
--- i.e., the set over which the variables range --- in all cases that we consider. We will be mostly using $a,b,c,d$ as metavariables for constants. 

By a {\bf valuation}\label{ivaluation} we mean 
a function $e$ that sends each variable $x$ to a constant $e(x)$. In these terms, a classical predicate $p$ can be understood as 
a function that sends each valuation $e$ to a proposition, i.e., to a constant predicate.   Similarly, what we call a game sends valuations to constant games: 

\begin{definition}\label{ngame}
A {\bf game} is a function $A$ from valuations to constant games. We write $e[A]$\label{iea} (rather than $A(e)$) to denote the constant game returned by $A$ for valuation $e$. Such a constant game $e[A]$ is said to be an {\bf instance}\label{iinstance} of $A$. 
For readability, we usually write $\legal{A}{e}$\label{ilre} and $\win{A}{e}$ instead of $\legal{e[A]}{}$ and $\win{e[A]}{}$.
\end{definition}

Just as this is the case with propositions versus predicates, constant games in the sense of Definition \ref{game} will
be thought of as special, constant cases of games in the sense of Definition \ref{ngame}. In particular, each constant game $A'$ is the game $A$ such that, for every valuation $e$,
$e[A]= A' $. From now on we will no longer distinguish between such $A$ and $A' $, so that, if $A$ is a constant game,
it is its own instance, with $A= e[A]$ for every $e$.

Where $n$ is a natural number, we say that a game $A$ is {\bf $n$-ary}\label{igarity} iff there is are $n$ variables such that, for any two valuations $e_1$ and $e_2$ that agree on all those variables, we have $e_1[A]= e_2[A]$. Generally, a game that is $n$-ary for some $n$, is said to be {\em finitary}.\label{ifinitary} Our paper is going to exclusively deal with finitary games and, for this reason, we agree that, from now on, when we say ``game'', we usually mean ``finitary game''.  

We say that a game $A$ {\bf depends} on a variable $x$ iff there are two valuations  $e_1,e_2$ that agree on all variables except $x$ such that $e_1[A]\not= e_2[A]$. An $n$-ary game thus depends on at most $n$ variables. And constant games are nothing but $0$-ary games, i.e., games that do not depend on any variables. 

We say that a (not necessarily constant) game $A$ is {\bf elementary}\label{ielem2} iff so are all of its instances $e[A]$. And we say that $A$ is {\bf finite-depth}\label{fdpth} iff there is a (smallest) integer $d$, called the {\bf depth} of $A$, such that the depth of no instance of $A$ exceeds $d$.

Just as constant games are generalized propositions, games can be treated as generalized predicates. Namely, we will see each predicate $p$ of whatever arity as  the same-arity elementary game such that, for every valuation $e$,
$\win{p}{e}\emptyrun=\pp$ iff $p$ is true at $e$.  
And vice versa: every elementary game $p$ will be seen as the same-arity predicate which is true at a given valuation $e$ iff  $\win{p}{e}\emptyrun=\pp$.   
Thus, for us, ``predicate'' and ``elementary game'' are going to be synonyms. Accordingly,  any standard terminological or notational conventions familiar from the literature for predicates also apply to them seen as elementary games. 

Just as the Boolean operations straightforwardly extend from propositions to all predicates, our operations 
$\gneg,\mlc,\mld,\mli,\adc,\add$ extend from constant games to all games. This is done by simply stipulating that $e[\ldots]$ commutes with all of those operations: $\gneg A$ is 
the game such that, for every valuation $e$, $e[\gneg A]=\gneg e[A]$; $A\adc B$ is the game such that,
for every $e$, $e[A\adc B]= e[A]\adc e[B]$; etc. 

The operation of prefixation also extends to nonconstant games:  $\seq{\Phi}A$ should be understood as the unique game such that, for every $e$, $e[\seq{\Phi}A]= \seq{\Phi}e[A]$. However, unlike the cases with all other operations,  $\seq{\Phi}A$,  as a function from valuations to constant games, may be partial even if $A$ is total. Namely, it will be defined only for those valuations $e$ for which we have $\Phi\in\legal{A}{e}$. Let us call not-always-defined ``games'' {\bf partial} (as opposed to the {\bf total} games of Definition \ref{ngame}). In the rare cases when we write $\seq{\Phi}A$ for a non-constant game $A$ (which always happens in just intermediate steps), it should be remembered that possibly we are dealing with a partial rather than a total game.  Otherwise, the default meaning of the word ``game'' is always a total game.

\begin{definition}\label{sov}
Let $A$ be a game, $x_1,\ldots,x_n$ be pairwise distinct variables, and $c_1,\ldots,c_n$ be  constants. 
The result of {\bf substituting $x_1,\ldots,x_n$ by $c_1,\ldots,c_n$ in $A$}, denoted $A(x_1/c_1,\ldots,x_n/c_n)$, is defined by stipulating that, for every valuation $e$, $e[A(x_1/c_1,\ldots,x_n/c_n)]= e'[A]$, where $e'$ is the valuation that sends each $x_i$ to $c_i$ and agrees with $e$ on all other variables. 
\end{definition}

Following the standard readability-improving practice established in the literature for predicates, we will often fix pairwise distinct  variables $x_1,\ldots,x_n$ for a game $A$ and write $A$ as $A(x_1,\ldots,x_n)$. 
Representing $A$ in this form  sets a context in which we can write $A(c_1,\ldots,c_n)$ to mean the same as the more clumsy expression $A(x_1/c_1,\ldots,x_n/c_n)$. 

\begin{definition}\label{bq}
Below $x$ is an arbitrary variable other than $\bound$, and $A(x)$ is an arbitrary finite-depth game.    

1. We define $\ada^0 xA(x)=\ade^0xA(x)=A(0)$ and, for any positive integer $b$, with $1^b$ standing for the binary numeral consisting of $b$ ``$1$''s,  we define the games $\ada^bxA(x)$ and $\ade^bxA(x)$ as follows: 

   \[\ada^b xA(x)\ = \ A(0)\adc A(1)\adc A(10)\adc A(11)\adc A(100)\adc A(101)\adc \ldots\adc A(1^b);\] \[\ade^b xA(x)\ =\ A(0)\add A(1)\add A(10)\add A(11)\add A(100)\add A(101)\add \ldots \add A(1^b).\] 

2. Using the above notation, we define \[\ada^\bound x A(x)\label{ibcuq2}\] as the unique game such that, for any valuation $e$, $e[\ada^\bound x A(x)]= e[\ada^b x A(x)]$, where $b= e(\bound)$. Similarly, \[\ade^\bound x A(x)\label{ibceq2}\] is the unique game such that, for any valuation $e$, $e[\ade^\bound x A(x)]=e[\ade^b x A(x)]$, where $b=e(\bound)$. 
$\ada^\bound$ and $\ade^\bound$ are said to be {\bf bounded choice universal  quantifier}\label{ibcq} and {\bf bounded choice existential quantifier}, respectively. 
\end{definition}

As we see, $\ada^\bound$ and $\ade^\bound$ are like the ordinary choice quantifiers $\ada,\ade$ of computability logic explained in Section \ref{ss2}, with the only difference that the size of a constant chosen for $x$ in $\ada^\bound x$ or $\ade^\bound x$ should not exceed the value of $\bound$. (The case of that value being $0$ is a minor technical exception which can be safely forgotten.)

\begin{convention}\label{con1}
Because throughout the rest of this paper we exclusively deal with the bounded choice quantifiers $\ada^\bound,\ade^\bound$ (and never with the ordinary $\ada,\ade$ discussed  in Section \ref{ss2}), and because the variable $\bound$ is fixed and is the same everywhere, we agree that, {\em from now on}, when we write $\ada$ or $\ade$, we {\em always} mean $\ada^\bound$ or $\ade^\bound$, respectively.

This is not a change of interpretation of $\ada,\ade$ but rather some, rather innocent,  abuse of notation.  
\end{convention}
 
We will say that a game $A$ is  {\bf unistructural}\label{iunistructural} iff, for any two valuations $e_1$ and $e_2$ that agree on $\bound$,   we have $\legal{A}{e_1}= \legal{A}{e_2}$. Of course, all constant or elementary games are unistructural. It can also be easily seen that all our game operations preserve the unistructural property of games. For the purposes of the present paper, considering only unistructural games would be sufficient. 

We define the remaining operations $\cla$ and $\cle$ only for unistructural games:

\begin{definition}\label{op5} Let $x$ be a variable other than $\bound$, and $A(x)$ be a finite-depth unistructural game.\vspace{9pt}

\noindent 1. $\cla x A(x)$ is defined by stipulating that, for every valuation $e$, player $\xx$ and move $\alpha$, we have: 
\begin{quote}\begin{description}
\item[(i)] $\win{\clai x A(x)}{e}\emptyrun= \pp$ iff, for every constant\footnote{It is important to note that, unlike the case with the choice quantifiers, here we are not imposing any restrictions on the size of such a constant.}  $c$, $\win{A(c)}{e}\emptyrun= \pp$. 
\item[(ii)] $\seq{\xx\alpha}\in\legal{\clai x A(x)}{e}$ iff $\seq{\xx\alpha}\in\legal{A(x)}{e}$. Such an initial legal labmove $\xx\alpha$ brings the game $e[\cla x A(x)]$ down to 
$e[\cla x\seq{\xx\alpha}A(x)]$.\vspace{5pt}
\end{description}\end{quote}
\noindent 2. $\cle x A(x)$ is defined in exactly the same way, only with $\pp$ and $\oo$ interchanged.  
\end{definition}

It is worth noting that $\cla x A(x)$ and $\cle x A(x)$ are total even if the game $\seq{\xx\alpha}A(x)$ used in their definition is only partial. 
 
\begin{example}\label{may14}
Let $G$ be the game (\ref{lkj})
discussed earlier in Section \ref{ss2} (only, now $\ada$ seen as $\ada^{\bound}$), and let $e$ be a  valuation with $e(\bound)= 10$. The sequence 
$\seq{\oo 1.11,\ \oo 0.0,\ \pp 1.1}$ 
is a legal run of $e[G]$, the effects of the moves of which are shown below:
\[\begin{array}{ll}
e[G]:  & \cla y\Bigl(\mbox{\em Even}(y)\add \mbox{\em Odd}(y) \mli \ada^{10}  x\bigl(\mbox{\em Even}(x\plus y)\add \mbox{\em Odd}(x\plus y)\bigr)\Bigr)\\
\seq{\oo 1.11}e[G]:  & \cla y\bigl(\mbox{\em Even}(y)\add \mbox{\em Odd}(y) \mli \mbox{\em Even}(11\plus y)\add \mbox{\em Odd}(11\plus y)\bigr)\\
\seq{\oo 1.11, \oo 0.0}e[G]: &  \cla y\bigl(\mbox{\em Even}(y) \mli \mbox{\em Even}(11\plus y)\add \mbox{\em Odd}(11\plus y)\bigr)\\
\seq{\oo 1.11, \oo 0.0,\pp 1.1}e[G]: & \cla y\bigl(\mbox{\em Even}(y) \mli \mbox{\em Odd}(11\plus y)\bigr)
\end{array}\]
The play hits (ends as) the true proposition $\cla y\bigl(\mbox{\em Even}(y) \mli \mbox{\em Odd}(11\plus y)\bigr)$ and hence is won by $\pp$. 

\end{example}

Before closing this section, we want to make the rather straightforward observation that the DeMorgan dualities hold for all of our sorts of conjunctions, disjunctions and quantifiers, and so does the double negation principle. That is,  we always have:
\[\gneg\gneg A= A;\vspace{-3pt}\]
\[\gneg(A\mlc B)= \gneg A\mld\gneg B; \ \ \ \ \gneg(A\mld B)= \gneg A\mlc\gneg B;\vspace{-3pt}\]
\[\gneg(A\adc B)= \gneg A\add\gneg B;  \ \ \ \ \gneg(A\add B)= \gneg A\adc\gneg B;\vspace{-3pt}\]
\[\gneg \cla xA(x)= \cle x\gneg A(x); \ \ \ \ \gneg \cle xA(x)= \cla x\gneg A(x);\vspace{-3pt}\]
\[\gneg \ada xA(x)= \ade x\gneg A(x); \ \ \ \ \gneg \ade xA(x)= \ada x\gneg A(x).\]

\section{Algorithmic strategies through interactive machines}\label{icp}

In traditional game-semantical approaches, including Blass's \cite{Bla72,Bla92} approach which is the closest precursor of ours, player's strategies are understood as {\em functions} --- typically as functions from interaction histories (positions) to moves, or sometimes (\cite{Abr94}) as functions that only look at the latest move of the history. This {\em strategies-as-functions} approach, however, is inapplicable in the context of computability logic, whose relaxed semantics, in striving to get rid of any ``bureaucratic pollutants'' and only deal with the remaining true essence of games,  does not impose any regulations on which player can or should move in a given situation. Here, in many cases, either player may have (legal) moves, and then it is unclear whether the next move should be the one prescribed by $\pp$'s strategy function or the one prescribed by the strategy function of $\oo$. In fact, for a game semantics whose ambition is to provide a comprehensive, natural and direct tool for modeling interaction, the strategies-as-functions approach would be simply less than adequate, even if technically possible. This is so for the simple reason that  the strategies that real computers follow are not functions. If the strategy of your personal computer was a function from the history of interaction with you, then its performance would keep noticeably worsening due to the need to read the continuously lengthening --- and, in fact, practically infinite --- interaction history every time before responding. Fully ignoring that history and looking only at your latest keystroke in the spirit of \cite{Abr94} is also certainly not what your computer does, either.  

In computability logic, ($\pp$'s effective) strategies are defined in terms of interactive machines, where computation is one continuous process interspersed with --- and influenced by --- multiple ``input'' (environment's moves) and ``output'' (machine's moves) events. Of several, seemingly rather different yet equivalent,  machine models of interactive computation studied in CL, here we will employ the most basic, {\bf HPM}\label{ihpm} (``Hard-Play Machine'') model.

An HPM is nothing but a Turing machine with the additional capability of making moves. The adversary can also move at any time, with such moves being the only nondeterministic events from the machine's perspective. Along with the ordinary  work tape, the machine has two additional  tapes called the valuation tape and the  run tape. The valuation tape, serving as a static input, spells some (arbitrary but fixed) valuation applied to the game. And the run tape, serving as a dynamic input, at any time  spells the ``current position'' of the play. The role of these two tapes is to make both the valuation and the run fully visible to the machine.  

In these terms,  an  algorithmic solution ($\pp$'s winning strategy) for a given  game $A$ is understood as an HPM $\cal M$ such that,  no matter how the environment acts during its interaction with $\cal M$ (what moves it makes and when), and no matter what valuation $e$ is spelled on the valuation tape, the run incrementally spelled on the run tape is a $\pp$-won run of $e[A]$.  

As for $\oo$'s strategies, there is no need to define them: all possible behaviors by $\oo$ are accounted for by the different possible nondeterministic updates  of the run tape of an HPM. 

In the above outline, we described HPMs in a relaxed fashion, without being specific about technical details such as, say, how, exactly, moves are made by the machine, how many moves either player can make at once, what happens if both players attempt to move ``simultaneously'', etc. As it turns out, all reasonable design choices yield the same class of winnable games as long as we consider a certain natural subclass of games called {\bf static}.\label{istatic} Such games are obtained by imposing a certain simple formal condition on games (see, e.g., Section 5 of \cite{Japfin}), which we do not reproduce here as nothing in this paper relies on it. We shall only point out that, intuitively, static games are interactive tasks where the relative speeds of the players are irrelevant, as it never hurts a player to postpone making moves. In other words, static games are games that are contests of intellect rather than contests of speed. And one of the theses that computability logic philosophically relies on is that static games present an adequate formal counterpart of our intuitive concept of ``pure'', speed-independent interactive computational problems. Correspondingly, computability logic restricts its attention (more specifically, possible interpretations of the atoms of its formal language) to static games. All elementary games turn out to be trivially static, and the class of static games turns out to be closed under all game operations studied in computability logic. More specifically, all games expressible in the language of the later-defined logic $\clthree$, or theory $\arfour$, are static.   
And, in this paper, we use the 
 term ``{\bf computational problem}", or simply ``{\bf problem}", is a synonym of ``static game''.

\section{The HPM model in greater detail}

As noted, computability of static games is rather robust with respect to the technical details of the underlying model of interaction. And the  loose description of HPMs that we gave in the previous section would be sufficient for most purposes, just as mankind had been rather comfortably studying and using algorithms long before the Church-Turing thesis in its precise form came around. Namely, relying on just the intuitive concept of algorithmic strategies (believed in CL to be adequately captured by the HPM model) would be sufficient if we only needed to show existence of such strategies for various games. As it happens, however, later sections of this paper need  to arithmetize such strategies in order to prove the promised extensional completeness of ptarithmetic. The complexity-theoretic concepts defined in the next section also require certain more specific details about HPMs, and in this section we provide such details. It should be pointed out again that most --- if not all --- of such details are ``negotiable'', as different reasonable arrangements would   yield equivalent models. 

Just like an ordinary Turing machine, an HPM has a finite set of {\bf states},\label{istate} one of which has the special status of being the {\bf start state}. There are no accept, reject, or halt states, but there are specially designated states called {\bf move states}.\label{imovestate} It is assumed that the start state is not among the move states. As noted earlier, this is a three-tape machine, with a read-only {\bf valuation tape},\label{ivaluationtape} read-write {\bf work tape},\label{iworktape}  and read-only {\bf run tape}.\label{iruntape}  Each tape has a beginning but no end, and is divided into infinitely many {\bf cells},\label{icell} arranged in the left-to-right order. At any time, each cell will contain one symbol from a certain fixed finite set of {\bf tape symbols}.\label{itapesymbol} The {\bf blank} symbol, as well as $\pp$ and $\oo$, are among the tape symbols. 
We also assume that these three symbols  are not among the symbols that any (legal or illegal) move can ever contain.  
Each tape has its own {\bf scanning head},\label{ihead} at any given time looking (located) at one of the cells of the tape.  A transition from one {\bf computation step}  (``{\bf clock cycle}'')\label{icc}   to another happens according to the fixed {\bf transition function}\label{itf} of the machine. The latter, depending on the current state, and the symbols seen by the three heads on the corresponding tapes, deterministically prescribes the next state, the tape symbol by which the old symbol should be overwritten in the current cell   (the cell currently scanned by the  head) of the work tape, and, for each head, the direction --- one cell left or one cell right --- in which the head should move. A constraint here is that the blank symbol, $\pp$ or $\oo$ can never be written by the machine on the work tape. An attempt to move left when the head of a  given   tape is looking at the first (leftmost) cell  results in staying put. So does an attempt 
to move right when the head is looking at the blank symbol. 

When the machine starts working, it is in its start state, all three scanning heads are looking at the first cells of the corresponding tapes, the valuation tape spells some valuation $e$ by listing the values of the variables $\mathfrak{w}_0,\mathfrak{w}_1,\mathfrak{w}_2,\ldots$ (in this precise order) separated by commas, and (all cells of) the work and run tapes are blank (i.e., contain the blank symbol). Whenever the machine enters a move state, the string $\alpha$ spelled by (the contents of) its work tape cells, starting from the first cell and ending with the cell immediately left to the work-tape scanning head,   will be automatically appended --- at the beginning of the next clock cycle --- to the contents of the run tape in the $\pp$-prefixed form  $\pp\alpha$. And, on every transition, whether the machine is in a move state or not, any finite sequence $\oo\beta_1,\ldots,\oo\beta_m$ of $\oo$-labeled moves may be nondeterministically appended to the content of the run tape. If the above two events happen on the same clock cycle, then the moves will be appended to the contents of the run tape in the following order: $\pp\alpha\oo\beta_1\ldots\oo\beta_m$ (note the technicality that labmoves are listed on the run tape without blanks or commas between them). 

With each labmove  that emerges on the run tape we associate its {\bf timestamp},\label{itimestamp} which is the number of the clock cycle immediately preceding the cycle on which the move first emerged on the run tape. Intuitively, the timestamp indicates on which cycle the move was {\em made} rather than {\em appeared} on the run tape; a move made during cycle $\#i$ appears on the run tape on cycle $\#i\plus 1$ rather than $\#i$. Also, we agree that the count of clock cycles starts from $0$, meaning that the very first clock cycle is cycle $\#0$ rather than $\#1$. 

A {\bf configuration}\label{iconfiguration} is a full description of (the ``current'') contents of the work and run tapes, the locations of the three scanning heads, and the state of the machine. 
An {\bf $e$-computation branch}\label{icb} is an infinite sequence $C_0,C_1,C_2,\ldots$ of configurations, where $C_0$ is the initial configuration (as explained earlier), and every $C_{i\plus 1}$ is a configuration that could have legally followed (again,  in the sense explained earlier) $C_i$ when the valuation $e$ is spelled  on the valuation tape. For an $e$-computation branch $B$, the {\bf run spelled by $B$}\label{irsb} is the run $\Gamma$ incrementally spelled on the run tape in the corresponding scenario of interaction. We say that such a $\Gamma$ is {\bf a run generated by}\label{irgb} the machine on valuation $e$. 

We say that a given HPM $\cal M$ {\bf wins} ({\bf computes}, {\bf solves}) a given  game $A$ on valuation $e$ --- and write ${\cal M}\models_e A$\label{imodels} --- iff every run $\Gamma$ generated by $\cal M$ on valuation $e$ is a $\pp$-won run of $e[A]$. We say that $A$ is {\bf computable}\label{icomputable} iff there is an HPM $\cal M$ such that, for every valuation $e$, ${\cal M}\models_e A$; such an HPM is said to be an (algorithmic) {\bf solution},\label{isol} or {\bf winning strategy}, for $A$.

\section{Towards interactive complexity}\label{s7}

At present, the theory of interactive computation is far from being well developed, and even less so is the theory of interactive complexity. The studies of interactive computation in the context of complexity, while having going on since long ago, have been relatively scattered, and interaction 
has often been used for better understanding certain traditional, non-interactive complexity issues (examples would be alternating computation \cite{Chandra}, or   interactive proof systems and Arthur-Merlin games \cite{Goldwasser,Babai}) rather than being treated as an object of systematic studies 
in its own rights. 
 As if complexity theory was not ``complex'' enough already, taking it to the interactive level would most certainly generate a by an order of magnitude greater diversity of species from the complexity zoo. 

The present paper is the first modest attempt to bring complexity issues into computability logic and the corresponding part of the under-construction theory of interactive computation. Here we introduce one, perhaps the simplest, way of measuring interactive complexity out of the huge and interesting potential variety of complexity measures meaningful and useful in the interactive context. 
 
Games happen to be so expressive that most, if not all, ways of measuring complexity will be meaningful and interesting only for certain (sub)classes of games and not quite so, or not so at all, for other classes. Our present approach is no exception. The time complexity concept that we are going to introduce is meaningfully applicable only to games that, in positive (winnable) cases, can be brought by $\pp$ to a successful end within a  finite number of moves. In addition, every instance of a game under consideration should be such that the length of any move in any legal run of it never exceeds a certain bound which only depends on the value of our special-status variable $\bound$. As mentioned earlier, it is exactly the value of this variable relative to which the computational complexity of games will be measured.

The  above class of games includes all games obtained by closing elementary games (predicates) under the operations of Sections \ref{cg} and \ref{nncg}, which also happens to be the class of games expressible in the language of the later-defined logic $\clthree$. Indeed, consider any such game $A$. Obviously the number of moves in any legal run --- and hence any $\pp$-won run --- of any instance of $A$ cannot exceed its $(\adc,\add,\ada,\ade)$-depth;  the sizes of   moves associated with $\adc,\add$ are constant; and the sizes of moves associated with $\ada,\ade$, in any given instance of the game,  never exceed a certain constant plus the value of the variable $\bound$. 

Games for which our present complexity concepts are meaningful also include the much wider class of games expressible in the language of logic {\bf CL12} introduced in \cite{Japtowards}, if the quantifiers $\ada,\ade$ of the latter are understood (as they are in this paper) as their bounded counterparts $\ada^{\bound},\ade^{\bound}$. While those games may have arbitrarily long or even infinite legal runs, all runs won by $\pp$ are still finite.

Bringing computability logic to a complexity-sensitive level  also naturally calls for  considering only {\bf bounded valuations}.\label{ibv} By a bounded valuation we mean  a valuation $e$ such that, for any variable $x$, the size of the binary numeral $e(x)$ does not exceed the value   $e(\bound)$ of $\bound$ (note: the {\em value} of $\bound$ rather than the {\em size} of that value). This condition makes it possible to treat free variables in the same way as if they were $\ada$-bounded.  

The starting philosophical-motivational point of our present approach to time complexity is that it should be an   indicator of ``how soon the game(s) can be won'' in the worst case, with ``how soon'' referring to the number of computation steps (clock cycles) a given HPM $\cal M$ takes to reach a final and winning position. There is a little correction to be made in this characterization though. The point is that part of its time $\cal M$ may spend just waiting for its adversary to move, and it would be unfair to bill $\cal M$ for the time for which probably it is not responsible.  Our solution is to subtract from the overall time the moveless intervals preceding the adversary's moves, i.e. the intervals that intuitively correspond to the adversary's ``thinking periods''. These intuitions are accounted for by the following definitions.

Let $\cal M$ be an HPM, $e$ a bounded valuation, $B$ any $e$-computation branch of $\cal M$, and $\Gamma$ the run spelled by $B$. 
For any  labmove $\lambda$ of  $\Gamma$, we define the {\bf thinking period}\label{itp} for $\lambda$ as $m\mminus n$, where $m$ is the timestamp of $\lambda$ and $n$ is the timestamp of the labmove immediately preceding $\lambda$ in $\Gamma$, or is $0$ if there are no such labmoves. 
Next, we define {\bf $\pp$'s time}\label{itm} in $B$ (or in $\Gamma$) as the sum of the thinking periods for all $\pp$-labeled moves of $\Gamma$. {\bf $\oo$'s time} is defined similarly. Note that, for either player $\xx$, $\xx$'s time will be finite iff there are only finitely many moves made by $\xx$; otherwise it will be infinite.

\begin{definition}\label{deftc}
Let $A$ be a game, $h$ a function from natural numbers to natural numbers, and $\cal M$ an HPM. 

1. We say that {\bf $\cal M$ runs in time $h$}, or that $\cal M$ is an {\bf $h$ time machine}, iff, for any bounded valuation $e$ and any $e$-computation branch $B$ of $\cal M$, $\pp$'s time in $B$ is less than $h\bigl(e(\bound)\bigr)$. 

2. We say that {\bf $\cal M$ wins} ({\bf computes}, {\bf solves}) {\bf $A$ in time $h$}, or that {\bf $\cal M$ is an $h$ time solution for $A$},  iff $\cal M$ is an $h$ time machine and, for any bounded valuation $e$, ${\cal M}\models_e A$. 

3. We say that $A$ is {\bf computable} ({\bf winnable}, {\bf solvable}) {\bf in time $h$} iff it has an $h$ time solution. 

4. We say that {\bf $\cal M$ runs in polynomial time}, or that $\cal M$ is a {\bf polynomial time machine},\label{iptm} iff it runs in time $h$ for some polynomial function $h$.   

5. We say that {\bf $\cal M$ wins} ({\bf computes}, {\bf solves}) {\bf $A$ in polynomial time}, or that {\bf $\cal M$ is a polynomial time solution for $A$},\label{ipts}  iff $\cal M$ is an $h$ time solution for $A$ for some polynomial function $h$. Symbolically, this will be written as  
\[{\cal M}\models^{P} A.\label{imodelsp}\]

6. We say that $A$ is {\bf computable} ({\bf winnable}, {\bf solvable}) {\bf in polynomial time},\label{iptccc} or {\bf polynomial time computable}  ({\bf winnable}, {\bf solvable}),\label{iptc} iff it has a polynomial time solution. 

\end{definition}

Many concepts introduced within the framework of computability are generalizations ---  for  the interactive context --- of ordinary and well-studied concepts of the traditional theory of computation. The above-defined  time complexity or polynomial time computability are among such concepts. Let us look at the traditional  notion of polynomial time decidability of a predicate $p(x)$ for instance. With a moment's thought, it can be seen to be equivalent to polynomial time computability (in the sense of Definition \ref{deftc}) of the game $p(x)\add \gneg p(x)$, or --- if you prefer --- the game $\ada x\bigl(p(x)\add\gneg p(x)\bigr)$ (these two games are essentially the same, with the only difference that, in one case, the value of $x$ will have to be read from the valuation tape, while in the other case from the run tape).

\section{The language  of logic $\clthree$ and its semantics}\label{ss6}

Logic $\clthree$ will be axiomatically constructed in Section \ref{ss8}. The present section is merely devoted to its {\em language}. The building blocks of this formal language are:

\begin{itemize} 
\item {\bf Nonlogical predicate letters},\label{ipl} for which we use $p,q$ (possibly indexed) as metavariables. With each predicate letter is associated a nonnegative integer called its {\bf arity}.\label{iar2} We assume that, for any $n$, there are infinitely many $n$-ary predicate letters.   
\item {\bf Function letters},\label{ifl} for which we use $f,g$ as metavariables. Again, each function letter comes with a fixed {\bf arity},\label{iar3} and  we assume that, for any $n$, there are infinitely many $n$-ary function letters.  
\item The binary {\bf logical predicate letter} $\equals $.
\item Infinitely many {\bf variables}.  These are the same as the ones fixed in Section \ref{nncg}.
\item Technical symbols: the left parenthesis, the right parenthesis, and the comma. 
\end{itemize}

{\bf Terms},\label{ipterm} for which we use $\tau,\theta,\omega,\psi,\xi$ (possibly indexed) as metavariables,  are defined as the elements of the smallest set of expressions such that:
\begin{itemize}
\item Variables are terms.
\item If $f$ is an $n$-ary function letter and $\tau_1,\ldots,\tau_n$ are terms, then $f(\tau_1,\ldots,\tau_n)$ is a term. When $f$ is $0$-ary, we write $f$ instead of $f()$. 
\end{itemize}

{\bf $\clthree$-formulas},\label{icl3f} or, in most contexts simply {\bf formulas}, are defined as the elements of the smallest set of expressions such that:
\begin{itemize}
\item If $p$ is an $n$-ary predicate letter and $\tau_1,\ldots,\tau_n$ are terms, then $p(\tau_1,\ldots,\tau_n)$ is an ({\bf atomic}) formula. We write $\tau_1\equals \tau_2$ instead of $\equals (\tau_1,\tau_2)$. Also, when $p$ is $0$-ary, we write $p$ instead of $p()$.  
\item If $E$ is an atomic formula, $\gneg (E)$ is a formula. We can write $\tau_1\notequals \tau_2$ instead of $\gneg (\tau_1\equals \tau_2)$.
\item $\tlg$ and $\twg$ are formulas. 
\item If $E_1,\ldots,E_n$ ($n\mgeq 2$) are formulas, then so are $(E_1)\mlc\ldots\mlc (E_n)$, $(E_1)\mld\ldots\mld (E_n)$, $(E_1)\adc\ldots\adc (E_n)$, $(E_1)\add\ldots\add (E_n)$. 
\item If $E$ is a formula and $x$ is a variable other than $\bound$, then $\cla x (E)$, $\cle x (E)$, $\ada x(E)$, $\ade x(E)$ are formulas. 
\end{itemize}

Note that, terminologically, $\twg$ and $\tlg$ do not count as atoms. For us, atoms are formulas containing no logical operators. The formulas $\twg$ and $\tlg$ do not qualify because they {\em are} ($0$-ary) logical operators themselves.  

 Sometimes we can write $E_1\mlc\ldots\mlc E_n$ for an unspecified $n\mgeq 1$ (rather than $n\mgeq 2$). Such a formula, in the case $n= 1$, should be understood as simply $E_1$. Similarly for $\mld,\adc,\add$.  

Also, where $S$ is a set of formulas, we may write 
\[\mlc S\]
for the $\mlc$-conjunction of the elements of $S$. Again, if $S$ only has one element $F$, then $\mlc S$ is simply $F$. Similarly for $\mld,\adc,\add$. Furthermore, we do not rule out the possibility of $S$ being empty when using this notation. It is our convention that, when $S$ is empty, both $\mlc S$ and $\adc S$ mean $\twg$, and both $\mld S$ and $\add S$ mean $\tlg$.

$\gneg E$, where $E$ is not atomic, will be understood as a standard abbreviation: 
$\gneg\twg=\tlg$, $\gneg\gneg E= E$, $\gneg(A\mlc B)= \gneg A\mld \gneg B$, $\gneg \ada xE= \ade x\gneg E$, etc. And $E\mli F$ will be understood as an abbreviation of $\gneg E\mld F$. Also, if we write \[E_1\mli E_2\mli E_3\mli\ldots\mli E_n,\] this is to be understood as an abbreviation of $E_1\mli(E_2\mli(E_3\mli(\ldots (E_{n-1}\mli E_n)\ldots)))$.

Parentheses will often be omitted --- as we just did --- if there is no danger of ambiguity. When omitting parentheses, we assume that $\gneg$ and the quantifiers have the highest precedence, and $\mli$ has the lowest precedence. So, for instance, $\gneg \ada x E \mli F\mld G$ means $(\gneg(\ada x(E)))\mli((F)\mld (G))$.   

The expressions $\vec{x},\vec{y},\ldots$ will usually stand for tuples of variables. Similarly for $\vec{\tau},\vec{\theta},\ldots$ (for tuples of terms) or $\vec{a},\vec{b},\ldots$ (for tuples of constants). 
  
The definitions of {\em free} and {\em bound} occurrences of variables are standard, with $\ada,\ade$ acting as quantifiers along with $\cla,\cle$. We will try to use $x,y,z$ for bound  variables only, while use $s,r,t,u,v,w$ for free variables only. There may be some occasional violations of this commitment though.

\begin{convention}\label{jan26}
The present conventions apply not only to the language of $\clthree$ but also to the other formal languages that we deal with  later, such as those of $\clfour$ and $\arfour$. 

1. For safety and simplicity, throughout the rest of this paper we assume that no formula that we ever consider --- unless strictly implied otherwise by the context ---  may have both bound and free occurrences of the same variable. This restriction, of course, does not yield any loss of expressive power as variables can always be renamed so as to satisfy this condition.   

2. Sometimes a formula $F$ will be represented as $F(s_1,\ldots,s_n)$, where the $s_i$ are variables. 
When doing so, we do not necessarily mean that each  $s_i$ has a free occurrence in $F$, or that every variable occurring free in $F$ is among $s_1,\ldots,s_n$. However, it {\em will}  always be assumed (usually only implicitly) that the $s_i$ are pairwise distinct, and have no bound occurrences in $F$.  In the context set by the above representation, $F(\tau_1,\ldots,\tau_n)$ will mean the result of replacing, in $F$, each  occurrence of each $s_i$   by term $\tau_i$. When writing $F(\tau_1,\ldots,\tau_n)$, it will always be assumed (again, usually only implicitly) that the terms $\tau_1,\ldots,\tau_n$ contain no variables that have bound occurences in $F$, so that there are no unpleasant collisions of variables when doing replacements.  

3. Similar --- well established in the literature --- notational conventions apply to terms.
\end{convention}

An {\bf interpretation}\label{iint}\footnote{The concept of an interpretation in CL is usually more general than the present one. Interpretations in our present sense are called  {\bf perfect}. But here we omit the word ``perfect'' as we do not consider any nonperfect interpretations, anyway.} is a function $^*$ that sends each $n$-ary  predicate letter $p$ to an $n$-ary predicate (elementary game) $p^*(s_1,\ldots,s_n)$ which does not depend on any variables other than $s_1,\ldots,s_n$; it also sends each $n$-ary function letter $f$ to a function \[f^*:\ \{0,1,10,11,100,\ldots\}^n\rightarrow \{0,1,10,11,100,\ldots\};\]
the additional condition required to be satisfied by $^*$ is that $\equals^*$ is an equivalence relation on $\{0,1,10,\ldots\}$ preserved by $f^*$ for each function symbol $f$, and respected by $p^*$ for each nonlogical predicate symbol $p$.\footnote{That is, $\equals^*$ is a congruence relation. More commonly classical logic simply treats  $\equals$ as the identity predicate. That treatment of $\equals$, however, is known to be equivalent --- in every respect relevant for us --- to our present one. Namely, the latter turns into the former by seeing 
any two $\equals^*$-equivalent constants  as two different {\em names} of the same {\em object} of the universe, as ``Evening Star'' and ``Morning Star'' are.} 

The above uniquely extends to  a  mapping that sends each term $\tau$ to a function $\tau^*$, and each formula $F$ to a game $F^*$, by stipulating that: 
\begin{enumerate}
\item $s^* =  s$ (any variable $s$). 
\item Where $f$ is an $n$-ary function letter and $\tau_1,\ldots,\tau_n$ are terms, $\bigl(f(\tau_1,\ldots,\tau_n)\bigr)^* =  f^*(\tau_{1}^{*},\ldots,\tau_{n}^{*})$. 
\item Where $p$ is an $n$-ary  predicate letter  and $\tau_1,\ldots,\tau_n$ are terms, $\bigl(p(\tau_1,\ldots,\tau_n)\bigr)^* =  p^*(\tau_{1}^{*},\ldots,\tau_{n}^{*})$. 
\item $^{*}$ commutes with all logical operators, seeing them as the corresponding game operations:  
\begin{itemize}
\item $\twg^* =  \twg$; 
\item $\tlg^* =  \tlg$; 
\item $(\gneg F)^{*} =  \gneg F^{*}$; 
\item $(E_1\mlc\ldots\mlc E_n)^{*} =  E^{*}_{1}\mlc \ldots\mlc E^{*}_n$; 
\item $(E_1\mld\ldots\mld E_n)^{*} =  E^{*}_{1}\mld \ldots\mld E^{*}_{n}$; 
\item $(E_1\adc\ldots\adc E_n)^{*} =  E^{*}_{1}\adc \ldots\adc E^{*}_n$; 
\item $(E_1\add\ldots\add E_n)^{*} =  E^{*}_{1}\add \ldots\add E^{*}_{n}$;
\item $(\cla x E)^{*} =  \cla x(E^{*})$; 
\item $(\cle x E)^{*} =  \cle x(E^{*})$; 
\item $(\ada x E)^{*} =  \ada x(E^{*})$; 
\item $(\ade x E)^{*} =  \ade x(E^{*})$.\footnote{Remember Convention \ref{con1}, according to which $\ada$ means $\ada^{\bound}$ and $\ade$ means $\ade^{\bound}$.}
\end{itemize}
\end{enumerate}

When $O$ is a function symbol, a predicate symbol, or a formula, and $O^* =  W$, we say that $^*$ {\bf interprets} $O$ as $W$. We can also refer to such a $W$ as 
``{\bf $O$ under interpretation $^*$}''.

When a given formula is represented as $F(x_1,\ldots,x_n)$, we will typically write $F^*(x_1,\ldots,x_n)$ instead of 
$\bigl(F(x_1,\ldots,x_n)\bigr)^*$. A similar practice will be used for terms as well.

\begin{definition}\label{ddd2}
We say that an HPM $\cal M$ is a {\bf uniform polynomial time solution}\label{iupts}  for a formula $F$ iff, for any interpretation $^*$, $\cal M$ is a polynomial time solution for  $F^*$.
\end{definition}

Intuitively, a uniform polynomial time solution is a ``purely logical'' efficient solution. ``Logical'' in the sense that it does not depend on the meanings of the nonlogical symbols (predicate and function letters) ---  does not depend on a (the) interpretation $^*$, that is. It is exactly these kinds of   solutions that we are interested in when seeing CL as a logical basis for applied theories or knowledge base systems. As a universal-utility tool, CL (or a CL-based compiler) would have no knowledge of the meanings of those nonlogical symbols
(the meanings that will be changing from application to application and from theory to theory), other than what is explicitly  given by the target formula and the axioms or the knowledge base   
of the system. 

\section{Some closure properties of polynomial time computability}\label{ss7}

In this section we establish certain important  closure properties for polynomial time computability of games. For simplicity we restrict them to games expressible in the language of $\clthree$, even though it should be pointed out that these results can be stated and proven in much more general forms  than presented here.

By an (inference) {\bf rule}  we mean a binary relation $\cal R$ between sequences of formulas and formulas,  instances of which are schematically written as  
\begin{equation}\label{erule}
\frac{X_1\hspace{20pt}\ldots\hspace{20pt}X_n}{X},\end{equation}
where $X_1,\ldots,X_n$ are (metavariables for) formulas called the {\bf premises}, and $X$ is (a metavariable for) a formula called the {\bf conclusion}. 
Whenever ${\cal R}(\seq{X_1,\ldots,X_n},X)$ holds, we say that $X$ {\bf follows from} $X_1,\ldots,X_n$ by $\cal R$.

We say that such a rule ${\cal R}$ is {\bf uniform-constructively sound}\label{iucs} iff there is an effective procedure that takes any instance $(\seq{X_1,\ldots,X_n},X)$ of the rule, any HPMs  ${\cal M}_1,\ldots,{\cal M}_n$ and returns an HPM ${\cal M}$ such that, for any interpretation $^*$, whenever  ${\cal M}_1\models^{P} X_{1}^{*},\ldots,{\cal M}_n\models^{P} X_{n}^{*} $, we have ${\cal M}\models^{P}X^{*}$.

Our formulations of rules, as well as our later treatment, rely on the following notational and terminological conventions.

\begin{enumerate}
\item A {\bf positive occurrence} of a subformula is an occurrence that is not in the scope of $\gneg$. Since officially only atoms may come with a $\gneg$,  occurrences of  non-atomic subformulas  will always be positive. 
\item A {\bf surface occurrence}\label{isoc} of a subformula is an occurrence that is 
not in the scope of any choice operators ($\adc,\add,\ada,\ade$). 
\item A formula not containing choice operators --- i.e., a formula of the classical language --- is said to be {\bf elementary}.\label{ief}  
\item The {\bf elementarization} \[\elz{F}\label{ielz}\] of a formula $F$ is the result of replacing
in $F$ all surface occurrences of $\add$- and $\ade$-subformulas by $\tlg$, and all surface occurrences of $\adc$- and $\ada$-subformulas by $\twg$. Note that $\elz{F}$ is (indeed) an elementary formula.
\item We will be using the notation \[F[E_1,\ldots,E_n]\] to mean a formula $F$ together with some (single) fixed  positive surface occurrences of each subformula $E_i$. Here the {\em formulas} $E_i$ are not required to be pairwise distinct, but their {\em occurrences} are. Using this notation sets a context in which $F[H_1,\ldots,H_n]$ will mean the result of replacing in $F[E_1,\ldots,E_n]$ the (fixed) occurrence of each $E_i$ by $H_i$.  Note again that here we are talking about some {\em occurrences} of $E_1,\ldots,E_n$. Only those occurrences get replaced when moving from $F[E_1,\ldots,E_n]$ to $F[H_1,\ldots,H_n]$, even if the formula also had some other occurrences of $E_1,\ldots,E_n$.
\item In any context where the notation of the previous clause is used (specifically, in the formulations of the rules of $\add$-Choose, $\ade$-Choose and Wait below), all formulas are assumed to be in negation normal form, meaning that they contain no $\mli$, and no $\gneg$ applied to non-atomic subformulas.  
\end{enumerate}

Below we prove the uniform-constructive soundness of several rules.   Our proofs will be limited to showing how to construct an  HPM $\cal M$ from an arbitrary instance --- in the form (\ref{erule}) --- of the rule and arbitrary HPMs ${\cal M}_1,\ldots,{\cal M}_n$ (purported solutions for the premises).  In each case it will be immediately clear from our description of $\cal M$  that it can be constructed effectively, that it runs  in polynomial time as long as so do  ${\cal M}_1,\ldots,{\cal M}_n$,  and that its work in no way does depend on an interpretation $^*$ applied to the games involved. Since an interpretation $^*$ is typically irrelevant in such proofs, we will often omit it and write simply $F$ where, strictly speaking, $F^*$ is meant. That is, we identify formulas  with the games into which they turn once an interpretation is applied to them. Likewise, we may omit a valuation $e$ and write $F$ instead of $e[F]$ or $e[F^*]$. 

\subsection{$\add$-Choose}\label{sep10aa} $\add$-Choose is the following rule:
\[\frac{F[H_i]}{F[H_0\add\ldots\add H_n]},\]
where $n\geq 1$ and $i\in\{0,\ldots,n\}$.

Whenever a formula $F$ follows from a formula $E$ by $\add$-Choose, we say that $E$ is a {\bf $\add$-Choose-premise}\label{iaddprem} of $F$.  

\begin{theorem}\label{sep10a}
$\add$-Choose is uniform-constructively sound.  
\end{theorem}

\begin{idea} This rule most directly encodes an action that $\cal M$ should perform in order to successfully solve the conclusion.  Namely, $\cal M$ should choose $H_i$ and then continue playing as the machine that (presumably) solves the premise. 
\end{idea}

\begin{proof} Let ${\cal M}_1$ be an arbitrary HPM (a purported polynomial time solution for the premise). We let ${\cal M}$ (the will-be polynomial time solution for the conclusion) be the machine that  works as follows. 

At the beginning, without looking at its run tape or valuation tape, ${\cal M}$ makes the move $\alpha$ that brings $ F[H_0\add\ldots\add H_n]$ down to $F[H_i]$. For instance, if $F[H_0\add\ldots\add H_n]$ is $X\mlc(Y\mld (Z\add T))$ and $F[H_i]$ is $ X\mlc(Y\mld Z)$, then $1.1.0$ is such a move. 

What  ${\cal M}$ does after that can be characterized as ``turning itself into ${\cal M}_1$'' and playing the rest of the game as ${\cal M}_1$ would. In more detail, $\cal M$ starts simulating and mimicking ${\cal M}_1$. During this simulation, ${\cal M}$ ``imagines'' that ${\cal M}_1$ has the same valuation $e$ on its valuation tape as ${\cal M}$ itself has, and that the run tape of ${\cal M}_1$ spells the same run as its own run tape does, with the difference that the move $\alpha$ made by ${\cal M}$ at the beginning is ignored (as if it was not there). To achieve the effect of consistency between the real and imaginary valuation and run tapes, what $\cal M$ does is that, every time the simulated ${\cal M}_1$ reads a cell of its valuation or run tape, $\cal M$ reads the content of the corresponding cell of its own valuation or run tape, and feeds that content to the simulation as the content of the cell that ${\cal M}_1$ was reading.   And whenever, during the simulation, ${\cal M}_1$ makes a move, ${\cal M}$ makes the same move in the real play. 

The run generated by $\cal M$ in the real play will look like $\seq{\pp\alpha,\Gamma}$. It is not hard to see that then $\Gamma$ will be a run generated by ${\cal M}_1$.  So, if ${\cal M}_1$ wins $F[H_i]$, implying that $\win{F[H_i]}{e}\seq{\Gamma}=\pp$, then  $\cal M$ wins $F[H_0\add\ldots\add H_n]$, because  $\win{F[H_0\add\ldots\add H_n]}{e}\seq{\pp\alpha,\Gamma}=\win{F[H_i]}{e}\seq{\Gamma}$.

Simulation does impose a certain overhead, which makes $\cal M$ slower than ${\cal M}_1$. But, with some analysis, details of which are left to the reader, it can be seen that the slowdown would be at most polynomial, meaning that, if ${\cal M}_1$ runs in polynomial time, then so does $\cal M$. 
\end{proof}

\subsection{$\ade$-Choose}\label{sep10bb} $\ade$-Choose is the following rule:
\[\frac{F[H(s)]}{F[\ade x H(x)]},\]
where $x$ is any non-$\bound$ variable, $s$ is a variable with no bound occurrences in the premise, and $H(s)$ is the result of replacing by $s$ all free occurrences of $x$ in $H(x)$ (rather than vice versa).

Whenever a formula $F$ follows from a formula $E$ by $\ade$-Choose, we say that $E$ is a {\bf $\ade$-Choose-premise}\label{iadeprem} of $F$. 

\begin{theorem}\label{sep10b}
$\ade$-Choose is uniform-constructively sound.  
\end{theorem}

\begin{idea} Very similar to the previous case. $\cal M$ should specify $x$ as (the value of) $s$, and then continue playing as the machine that solves the premise.  
\end{idea}

\begin{proof} Let ${\cal M}_1$ be an arbitrary HPM (a purported polynomial time solution for the premise). We let ${\cal M}$ (the will-be polynomial time solution for the conclusion) be the machine that, with a valuation $e$ spelled on its valuation tape, works as follows. At the beginning,  ${\cal M}$ makes the move that brings $F[\ade x H(x)]$ down to $F[H(s)]$. For instance, if $F[\ade x H(x)]$ is $X\mlc(Y\mld \ade x Z(x))$ and $F[H(s)]$ is $X\mlc (Y\mld Z(s))$, then $1.1.c$ is such a move, where  $c= e(s)$ (the machine will have to read $c$ from its valuation tape).  After this move, ${\cal M}$ starts simulating and mimicking ${\cal M}_1$ in the same fashion as in the proof of Theorem \ref{sep10a}. And, again, as long as ${\cal M}_1$ wins $F[H(s)]$ in polynomial time, ${\cal M}$ wins $F[\ade x H(x)]$ in polynomial time. 
\end{proof}

\subsection{Wait}\label{iwait} Wait is the following rule:
\[\frac{\elz{F}\hspace{25pt} F_1\hspace{25pt}\ldots\hspace{25pt}F_n}{F}\]
(remember that $\elz{F}$ means the elementarization of $F$), where $n\geq 0$  and the following two conditions are satisfied:
\begin{enumerate}
\item Whenever $F$ has the form $X[Y_0\adc\ldots\adc Y_m]$, each formula $ X[Y_i]$ ($0\leq i\leq m$) is among $F_1,\ldots,F_n$.
\item  Whenever $F$ has the form $ X[\ada xY(x)]$, for some variable $s$ different from $\bound$ and not occurring in $F$, the formula  $ X[Y(s)]$ is among  $F_1,\ldots,F_n$. Here  $Y(s)$ is the result of replacing by $s$ all free occurrences of $x$ in $Y(x)$ (rather than vice versa).
\end{enumerate}

Whenever the above relation holds, we say that $\elz{F}$ is the {\bf special Wait-premise} of $F$, and that $F_1,\ldots,F_n$ are {\bf ordinary Wait-premises}\label{iwaitpremise} of $F$.

The following lemma, on which we are going to rely in this subsection, can be verified by a straightforward induction on the complexity of $F$, which we omit. Remember that $\seq{}$ stands for the empty run.

\begin{lemma}\label{new1}
For any formula $F$, interpretation $^*$ and valuation $e$, $\win{F^*}{e}\emptyrun= \win{\elzi{F}^*}{e}\emptyrun$.
\end{lemma}

\begin{theorem}\label{sep10c}
Wait is uniform-constructively sound.  
\end{theorem}

\begin{idea} $\cal M$ should wait (hence the name ``Wait'' for the rule) until the adversary makes a move. If this never happens,  in view of the presence of the premise $\elz{F}$, a win for $\cal M$ is guaranteed by Lemma \ref{new1}. Otherwise, any (legal) move by the adversary essentially brings the conclusion down to one of  the premises 
$F_1,\ldots,F_n$; then $\cal M$ continues playing as the machine that wins that premise.
\end{idea}

\begin{proof} Assume ${\cal M}_0,{\cal M}_1,\ldots,{\cal M}_n$ are polynomial time solutions for $\elz{F},F_1,\ldots,F_n$, respectively. We  
 let ${\cal M}$, the will-be solution for $F$, whose construction does not depend on the just-made assumption, be a machine that, with a bounded valuation $e$ spelled on its valuation tape,  works as follows. 

 At the beginning,  ${\cal M}$ keeps waiting until the environment makes a move. If such a move is never made, then the run that is generated is empty.  Since  $\elz{F}$ is elementary and ${\cal M}_0$ wins it,   it is classically true (a false elementary game would be automatically lost by any machine).   But then, in view of Lemma \ref{new1}, $\cal M$  wins (the empty run of) $F$. And, note, $\pp$'s time in this case is $0$. 

 Suppose now the environment makes a move. Note that the time during which the machine was waiting does not contribute anything to $\pp$'s time. We may assume that the  move made by the environment is legal, or else the machine immediately wins. With a little thought, one can see that any legal move $\alpha$ by the environment brings the game $e[F]$ down to $g[F_i]$ for a certain bounded valuation $g$ --- with $g(\bound)= e(\bound)$ --- and one of the premises $F_i$ of the rule. For example, if $F$ is $(X\adc Y)\mld  \ada x Z(x)$, then a legal move $\alpha$ by the environment should be either $0.0$ or $0.1$ or $1.c$ for some constant $c$ (of size $\mleq e(\bound)$). In the case $\alpha=0.0$, the above-mentioned premise $F_i$ will be $X \mld  \ada x Z(x)$, and $g$ will be the same as $e$.  In the case $\alpha=0.1$,   $F_i$ will be $Y \mld  \ada x Z(x)$, and $g$, again, will be the same as $e$. Finally, in the case $\alpha=1.c$,   $F_i$ will be $(X\adc Y) \mld Z(s)$  for a variable $s$ different from $\bound$ and not occurring in $F$, and   $g$ will be the valuation that sends $s$ to $c$ and agrees with $e$ on all other variables, so that $g[(X\adc Y) \mld Z(s)]$ is $e[(X\adc Y) \mld Z(c)]$, with the latter being the game to which $e[F]$ is brought down by the labmove $\oo 1.c$. 

After the above event, ${\cal M}$ starts simulating and mimicking the machine ${\cal M}_i$  in the same fashion as in the proofs of Theorems \ref{sep10a} and \ref{sep10b}, with the only difference that,  if $g\not=e$, the imaginary valuation tape of the simulated machine now spells $g$ rather than $e$. 

As in the earlier proofs, it can be seen that $\cal M$, constructed as above, is a polynomial time solution for $F$. 
\end{proof}

\subsection{Modus Ponens (MP)}\label{imp} Modus Ponens is the following rule:
\[\frac{F_0\hspace{25pt}\ldots\hspace{25pt}F_n\hspace{25pt}F_0\mlc\ldots\mlc F_n \mli F}{F},\]
where $n\geq 0$.

\begin{theorem}\label{april22}
Modus Ponens is uniform-constructively sound.  
\end{theorem}

\begin{idea} Together with  the real play of $F$, $\cal M$ plays an imaginary game for each of the premises, in which it mimics the machines that win those premises. In addition, it applies copycat between each premise $F_i$ and the corresponding conjunct of the antecedent of the rightmost premise, as well as between (the real) $F$ and the consequent of that  premise.  
\end{idea}

\begin{proof} 
Assume ${\cal M}_0,\ldots,{\cal M}_n$ and ${\cal N}$ are HPMs that win $F_0,\ldots,F_n$ and  $F_0\mlc\ldots\mlc F_n\mli F$ in polynomial time, respectively (as in the previous proofs, our construction of ${\cal M}$ does not depend on this assumption; only the to-be-made conclusion  ${\cal M}\models^{P} F$  does). For simplicity, below we reason under the assumption that $n\geq 1$. Extending our reasoning so as to also include the case $n= 0$ does not present a problem. 

 We let ${\cal M}$ be the following HPM. Its work on a valuation $e$ consists in simulating, in parallel, the machines ${\cal M}_0,\ldots,{\cal M}_{n},{\cal N}$ with  the same $e$ on their valuation tapes, and also continuously polling (in parallel with simulation) its own valuation tape to see if the environment has made a new move. These simulations proceed in the same fashion as in the proofs of the earlier theorems, with the only difference that now $\cal M$ actually maintains records of the contents of the imaginary run tapes of the simulated machines (in the proof of Theorem \ref{sep10a}, $\cal M$ was simply using its own run tape in the role of such a ``record'').

As before, we may assume that the environment does not make illegal moves, for then $\cal M$ immediately wins. We can also safely assume that the simulated machines do not make illegal moves, or else our assumptions about their winning the corresponding games would be wrong.\footnote{Since we  need to construct $\cal M$ no matter whether those assumptions are true or not, we can let $\cal M$ simply stop making any moves as soon as it detects some illegal behavior.}  If so, in the process of the above simulation-polling  routine, now and then, one of the following four types of events will be happening (or rather detected):

{\em Event 1}. ${\cal M}_i$ ($0\leq i\leq n$) makes a move $\alpha$. Then  ${\cal M}$ appends the labmove $\oo 0.i.\alpha$ at the end of the position spelled on the imaginary run tape of ${\cal N}$ in its simulation.\footnote{Here and later, of course, an implicit stipulation is that the position spelled on the imaginary run tape of the machine (${\cal M}_i$ in the present case) that made the move is also correspondingly updated (in the present case, by appending $\pp\alpha$ to it).} 

{\em Event 2}. ${\cal N}$ makes a move $0.i.\alpha$ ($0\leq i\leq n$). Then  ${\cal M}$ appends the labmove $\oo\alpha$ at the end of the imaginary run tape of ${\cal M}_i$ in its simulation.

{\em Event 3}. ${\cal N}$ makes a move $1.\alpha$. Then ${\cal M}$ makes the move $\alpha$ in the real play. 

{\em Event 4}. The environment makes a move $\alpha$ in the real play. Then ${\cal M}$ appends the labmove $\oo 1.\alpha$ at the end of the imaginary run tape of ${\cal N}$ in its simulation. 

What is going on here is that $\cal M$ applies copycat between $n+2$ pairs of (sub)games, real or imaginary. Namely, it mimics, in (the real play of) $F$,  $\cal N$'s moves made
in the consequent of  (the imaginary play of)  $F_0\mlc\ldots\mlc F_n \mli F$, and vice versa: uses (the real) environment's moves made (in the real play of) $F$ as (an imaginary) environment's moves in the consequent of $F_0\mlc\ldots\mlc F_n \mli F$.  Also, for each $i\in\{0,\ldots,n\}$, $\cal M$ uses the moves made by ${\cal M}_i$ in $F_i$ as environment's 
moves in the $F_i$ component of $F_0\mlc\ldots\mlc F_n \mli F$, and vice versa: uses the moves made by $\cal N$ in that component as environment's moves in $F_i$.  
Therefore, the final positions\footnote{Remember Convention \ref{poscon}.}  hit by the above imaginary and real plays will be
\[\mbox{$F'_0,\ \ldots,\ F'_n,\ F'_1\mlc\ldots\mlc F'_n\mli F'$ and $F'$}\]
for some $F'_0,\ldots,F'_n,F'$. Our assumption that the machines ${\cal M}_0,\ldots,{\cal M}_n$ and ${\cal N}$ win the games $F_0, \ldots, F_n$ and $F_1\mlc\ldots\mlc F_n\mli F$ implies that each $G\in\{F'_0,\ \ldots,\ F'_n,\ F'_1\mlc\ldots\mlc F'_n\mli F'\}$ is $\pp$-won, in the sense that $\win{G}{e}\seq{}=\pp$. It is then obvious that so should be  $F'$. Thus, the (real) play of $F$ brings it down to the $\pp$-won  $F'$, meaning that $\cal M$ wins $F$. 
 
With some thought, one can also see that ${\cal M}$ runs in polynomial time. The only reason why $\cal M$ may spend ``too much'' time thinking before making a move could be that it waited ``too long'' to see what move was made by one (or several) of the simulated machines. But this would not happen because, by our assumption, those machines run in polynomial time, so, whenever they make a move,  it never takes them ``too long'' to do so.     
\end{proof}

\section{Logic $\clthree$}\label{ss8}

Before we get to our version of formal arithmetic, it would not hurt to identify the (pure) logic on which it is based --- based in the same sense as the traditional Peano arithmetic is based on classical logic. This logic is $\clthree$. With minor technical differences not worth our attention and not warranting a new name for the logic, our present version of $\clthree$ is the same as the same-name logic introduced and studied in \cite{Japtcs}.\footnote{In fact, an essentially the same logic, under the name {\bf L}, was already known as early as in   \cite{Jap02}, where it emerged in the related yet quite different context of the {\em Logic of Tasks}.}   
 
The {\bf language} of $\clthree$ has already been described in Section \ref{ss6}. 

The {\bf axioms} of this system are all classically valid elementary formulas. 
Here by classical validity, in view of G\"{o}del's completeness theorem,  we mean provability in classical first-order calculus. Specifically, in  
classical first-order calculus with  function letters and $\equals $, where $\equals $ is treated as the logical {\em identity} predicate (so that, say, $x\equals x$, $x\equals y\mli (E(x)\mli E(y))$, etc. are valid/provable).
 
As for the {\bf rules of inference} of $\clthree$, they are the $\add$-Choose, $\ade$-Choose and Wait rules of Section \ref{ss7}. As will be easily seen from the forthcoming soundness and completeness theorem for $\clthree$ (in conjunction with Theorem \ref{april22}), $\clthree$ is closed under Modus Ponens.  So, there is no need for officially including it among the rules of inference, doing which would  destroy the otherwise analytic property of the system.
  
A {\bf $\clthree$-proof} of a formula $F$ is a sequence $E_1,\ldots,E_n$ of formulas, with $E_n = F$, such that each $E_i$ is either an axiom or follows from some earlier formulas of the sequence by one of the rules of $\clthree$.
When a $\clthree$-proof of $F$ exists, we say that $F$ is {\bf provable} in $\clthree$, and write $\clthree\vdash F$. Otherwise we write $\clthree\not\vdash F$.  Similarly for any other formal systems as well.

\begin{example}\label{j28a}
The formula $\cla xp(x)\mli\ada xp(x)$ is provable in $\clthree$. It follows by Wait from the axioms $\cla xp(x)\mli\twg$ (special Wait-premise) and $\cla xp(x)\mli p(s)$ (ordinary Wait-premise). 

On the other hand, the formula $\ada xp(x)\mli\cla xp(x)$, i.e. $\ade x\gneg p(x)\mld \cla xp(x)$, in not provable. Indeed, this formula  has no $\add$-Choose-premises because it does not contain $\add$. Its elementarization (special Wait-premise) is $\tlg\mld \cla xp(x)$ which is not an axiom nor the conclusion of any rules. Hence $\ade x\gneg p(x)\mld \cla xp(x)$ cannot be derived by Wait, either. This leaves us with $\ade$-Choose. But if $\ade x\gneg p(x)\mld \cla xp(x)$ is derived  by $\ade$-Choose, then the premise should be $\gneg p(s)\mld \cla xp(x)$ for some variable $s$. The latter, however, is neither an axiom nor the conclusion of any of the three rules of $\clthree$. 
\end{example}

\begin{example}\label{j28b}
The formula $\ada x\ade y\bigl(p(x)\mli p(y)\bigr)$, whose elementarization is $\twg$, is provable in $\clthree$ as follows:\vspace{7pt}

\noindent 1. $\begin{array}{l}
\twg
\end{array}$  \ \ Axiom\vspace{3pt}

\noindent 2. $\begin{array}{l}
p(s)\mli p(s)
\end{array}$  \ \ Axiom\vspace{3pt}

\noindent 3. $\begin{array}{l}
\ade y\bigl(p(s)\mli p(y)\bigr)
\end{array}$  \ \ $\ade$-Choose: 2\vspace{3pt}

\noindent 4. $\begin{array}{l}
\ada x\ade y\bigl(p(x)\mli p(y)\bigr)
\end{array}$  \ \ Wait: 1,3\vspace{7pt}

On the other hand, the formula $\ade y\ada x \bigl(p(x)\mli p(y)\bigr)$ can be seen to be unprovable, even though its classical counterpart $\cle y\cla x \bigl(p(x)\mli p(y)\bigr)$ is an axiom and hence provable. 
\end{example}

\begin{example}\label{j28c}
While the formula $\cle x \bigl(x\equals f(s)\bigr)$  is classically valid and hence provable in $\clthree$, its constructive counterpart 
$\ade x \bigl(x\equals f(s)\bigr)$ can be easily seen to  be unprovable. This is no surprise. In view of the expected soundness of $\clthree$,  provability  of $\ade x \bigl(x\equals f(s)\bigr)$ would imply that every function $f$ is computable (worse yet: efficiently computable), which, of course, is not the case.     
\end{example}

\begin{exercise}\label{feb1a}
To see the resource-consciousness of $\clthree$, show that it does not prove $p\adc q\mli (p\adc q)\mlc (p\adc q)$, even though this formula has the form $F\mli F\mlc F$ of a classical tautology.
\end{exercise}

\begin{theorem}\label{main}
$\clthree\vdash X$ iff $X$ has a uniform polynomial time solution (any formula $X$).
Furthermore:

{\bf Uniform-constructive soundness:}\label{iucsl} There is an effective procedure that takes any $\clthree$-proof of any formula $X$ and 
constructs a uniform polynomial time solution for $X$.  

{\bf Completeness:} If $\clthree\not\vdash X$, then, for any  HPM $\cal M$, there is an interpretation $^*$ such that $\cal M$ does not win $X^*$ (let alone winning in polynomial time).  
\end{theorem}

\begin{idea} The soundness of $\clthree$ was, in fact, already established in the preceding section. For completeness, assume $\clthree\not\vdash X$ and consider any HPM $\cal M$. If $\elz{X}$ is an axiom,  a smart environment can always make a move that brings $X$ down to an unprovable ordinary Wait-premise of 
$X$,  or else $X$ would be derivable by Wait; such a Wait-premise is less complex than $X$ and, by the induction hypothesis, $\cal M$ loses. If $\elz{X}$ is not an axiom, then it is false under a certain interpretation, and therefore $\cal M$ will have to make a move to avoid an automatic loss.  But any (legal) move by $\cal M$ brings $X$ down to an unprovable Choose-premise of it (or else $X$ would be derivable  by a Choose rule) and, by the induction hypothesis, $\cal M$ again loses.   
\end{idea}

\begin{proof} Modulo the results of Section \ref{ss7}, the soundness (``only if'') part of this theorem, in the strong ``uniform-constructive'' form, is straightforward. Formally this fact can be proven by induction on the lengths of   $\clthree$-proofs. All axioms of $\clthree$ are obviously ``solved'' by a machine that does nothing at all. Next, as an induction hypothesis, assume $X_1,\ldots,X_n$ are $\clthree$-provable formulas, ${\cal M}_1,\ldots,{\cal M}_n$ are uniform polynomial time solutions for them, and $X$ follows from those formulas by one of the rules of $\clthree$. Then, as immediately implied by the results of  Section \ref{ss7}, we can effectively construct a uniform polynomial time solution ${\cal M}$ for $X$. 

The rest of this proof will be devoted to the completeness (``if'') part of the theorem.   

 Consider an arbitrary formula $X$ with $\clthree\not\vdash X$, and an arbitrary HPM $\cal M$.  Here we describe a scenario of the environment's behavior in interaction with $\cal M$ --- call this ``behavior'' the {\em counterstrategy} --- that makes $\cal M$ lose $F^*$ on $e$ for a certain appropriately selected interpretation $^*$ and a certain appropriately selected bounded valuation $e$    even if the time of $\cal M$ is not limited at all.   

For a formula $Y$ and valuation $g$, we say that $g$ is {\bf $Y$-distinctive}\label{idistinctive} iff $g$ assigns different values to different free variables of $Y$. 
We select $e$ to be an $X$-distinctive bounded valuation. Here we let $e(\bound)$ be ``sufficiently large'' to allow certain flexibilities needed below. 
 
How our counterstrategy acts depends on the current game (formula, ``position'')    $Y$ to which the original game   $X$ has been brought down by the present time in the play. Initially,  $Y$ is $X$.  

As an induction hypothesis, we assume that $\clthree\not\vdash Y$ and  $e$ is $Y$-distinctive.
We separately consider the following two cases.

{\em Case 1:} $\elz{Y}$ is classically valid. Then there should be a $\clthree$-unprovable formula $Z$ --- an ordinary Wait-premise of $Y$ --- satisfying the  conditions of one of the following two subcases, for otherwise $Y$ would be derivable 
by Wait. Our counterstrategy selects one such $Z$ (say, lexicographically the smallest one), and acts according to the corresponding prescription as given below. 

{\em Subcase 1.1:} $Y$ has the form $F[G_0\adc\ldots\adc G_m]$, and $Z$ is $ F[G_i]$ ($i\in\{0,\ldots,m\}$). In this case, the counterstrategy makes the move that brings $Y$ down to $Z$, and calls itself on $Z$  in the role of the ``current'' formula $Y$. 
 
{\em Subcase 1.2:} $Y$ has the form $ F[\ada xG(x)]$, and $Z$ is $F[G(s)]$, where $s$ is a variable different from $\bound$ and not occurring in $Y$. We may assume here that $e(s)$ is different from any $e(r)$ where $r$ is any other ($\not=s$) free variable of $Y$. Thus, $e$ remains a $Z$-distinctive valuation. In this case, our counterstrategy makes the move  that brings $Y$ down to $Z$ (such a move is the one that specifies the value of $x$ as $e(s)$ in the indicated occurrence of $\ada xG(x)$), and calls itself on $Z$ in the role of $Y$.  

{\em Case 2:} $\elz{Y}$ is not classically valid.   Then our counterstrategy inactively waits until $\cal M$ makes a move. 

{\em Subcase 2.1.} If such a move is never made, then the run that is generated is empty. Since $e$ is a $Y$-distinctive valuation, of course, it is also $\elz{Y}$-distinctive.  It is a common knowledge from classical logic that, whenever a formula $F$ is invalid (as is $\elz{Y}$ in our present case) and $e$ is an $F$-distinctive valuation, $e[F]$ is false in some model. So, $e[\elz{Y}]$ is false in/under some model/interpretation $^*$. This, in view of Lemma \ref{new1}, implies that $\win{Y^*}{e}\seq{}= \oo$ and hence $\cal M$ is the loser in the overall play of $X^*$ on $e$.

{\em Subcase 2.2.} Now suppose $\cal M$ makes a move. We may assume that such a move is legal, or else $\cal M$ immediately loses. With a little thought, one can see that any legal move $\alpha$ by $\cal M$ will bring the game down to $Z$ for a certain  formula $Z$ such that $Y$ follows from $Z$ by $\add$-Choose or $\ade$-Choose, and $e$ remains --- or, at least, can be safely assumed to remain --- $Z$-distinctive. But then, since $\clthree\not\vdash Y$,  we also have $\clthree\not\vdash Z$. In this case, our counterstrategy calls itself on $Z$ in the role of $Y$. 
 
It is clear that, sooner or later, the interaction will end according to the scenario of Subcase 2.1, in which case, as we observed, $\cal M$ will be the loser in the overall play of $X^*$ on $e$ for a certain interpretation $^*$. 
\end{proof}
   
\section{ {\bf CL4}, the metalogic  of  $\clthree$}\label{ss66}
In this section we present an auxiliary deductive system $\clfour$. Syntactically, it is a conservative extension of $\clthree$. Semantically, we treat $\clfour$ as a {\em metalogic} of $\clthree$, in the sense that the formulas of $\clfour$ are seen as schemata of $\clthree$-formulas, and the theorems of $\clfour$ as schemata of theorems of $\clthree$. System $\clfour$ --- in an unsubstantially different form --- was initially introduced and studied in \cite{Japtcs2} where, unlike our present treatment, it was seen as a logic (rather than metalogic) in its own rights, soundly and completely axiomatizing a more expressive fragment of computability logic than $\clthree$ does. Simplicity is the only reason why here we prefer to see $\clfour$  as just a metalogic. 

The language of $\clfour$ is obtained from that of $\clthree$ by adding to it nonlogical {\bf general letters}\label{igl2}, on top of the predicate letters of the language of $\clthree$ that in this new context, following the terminological tradition of computability logic,  we rename into {\bf elementary letters}.\label{iel2} We continue using the lowercase $p,q$ (possibly indexed) as metavariables for elementary letters, and will be using the uppercase $P,Q$ (possibly indexed) as metavariables for general letters.    Just as this is the case with the elementary letters, we have infinitely many $n$-ary general letters for each arity (natural number) $n$. In our present approach, the nonlogical elementary letters of the language of $\clfour$ will be seen as metavariables for elementary formulas of the language of  $\clthree$, the general letters of the language of $\clfour$ will be seen as metavariables for any, not-necessarily-elementary, formulas of the language of $\clthree$, and the function   letters of the language of $ \clfour$ will be seen as metavariables for terms of the language of $\clthree$. 

{\em Formulas} of the language of $\clfour$, to which we refer as {\bf $\clfour$-formulas},\label{icl4f} are built from atoms, terms, variables and operators in exactly the same way as $\clthree$-formulas are, with the only difference that now, along with the old {\bf elementary atoms}\label{ielat} --- atoms of the form $p(\tau_1,\ldots,\tau_n)$ where $p$ is an $n$-ary elementary letter and the $\tau_i$ are terms --- we also have {\bf general atoms},\label{igenat} which are of the form $P(\tau_1,\ldots,\tau_n)$, where $P$ is an $n$-ary general letter and the $\tau_i$ are terms. An  {\bf elementary literal}\label{iell} is $\twg$, $\tlg$, or an elementary atom with or without negation $\gneg$. And a {\bf general literal}\label{igenl} is a general atom with or without negation.  As before, we always assume that negation can only occur in literals; $\gneg$ applied to a non-atomic formula, as well as $\mli$, are treated as abbreviations.  The concepts of a surface occurrence, positive occurrence etc. straightforwardly extend from the language of $\clthree$ to the language of $\clfour$. 

We say that a $\clfour$-formula is {\bf elementary} iff it does not contain  general atoms and choice operators. Thus, ``elementary $\clfour$-formula'', ``elementary $\clthree$-formula'' and ``formula of classical logic'' mean the same. Note that we see the predicate letters of classical logic as elementary rather than general letters. 

The {\bf elementarization}\label{ielz2} $\elz{F}$ of a $\clfour$-formula $F$ is the result of replacing in it all surface occurrences of $\adc$- and $\ada$-subformulas by $\twg$, all   surface occurrences of $\add$- and $\ade$-subformulas by $\tlg$, and all positive surface occurrences of general literals by $\tlg$. 

$\clfour$ has exactly the same axioms as $\clthree$ does (all classically valid elementary formulas), and has four rules of inference. The first three rules are  nothing but the rules of $\add$-Choose, $\ade$-Choose and Wait of $\clthree$, only now applied to any $\clfour$-formulas rather than just $\clthree$-formulas. The additional, fourth rule, which we call {\bf Match},\label{imatch} is the following:
\[\frac{F[p(\vec{\tau}),\gneg p(\vec{\theta})]}{F[P(\vec{\tau}),\gneg P(\vec{\theta})]},\]
where $P$ is any $n$-ary general letter, $p$ is any $n$-ary nonlogical elementary letter not occurring in the conclusion, and $\vec{\tau}$,$\vec{\theta}$ are any $n$-tuples of terms; also, according to our earlier notational conventions,   $F[P(\vec{\tau}),\gneg P(\vec{\theta})]$ is a formula with two fixed positive occurrences of the literals $P(\vec{\tau})$ and $\gneg P(\vec{\theta})$, and $F[p(\vec{\tau}),\gneg p(\vec{\theta})]$ is the result of replacing in $F[P(\vec{\tau}),\gneg P(\vec{\theta})]$ the above two occurrences by  $p(\vec{\tau})$ and $\gneg p(\vec{\theta})$, respectively.\vspace{7pt}

It may help some readers to know that $\clfour$ is an  extension of additive-multiplicative affine logic (classical linear logic with weakening), with the letters of the latter understood as our general letters.  This fact is an immediate consequence of the earlier-known soundness of affine logic 
(proven in \cite{Japfin}) and completeness of  $\clfour$ (proven in \cite{Japtcs2})  with respect to the semantics of computability logic. As seen from the following example, the extension, however, is not conservative.

\begin{example}\label{j28d}
Below is a $\clfour$-proof of the formula $(P\mlc P)\mld (P\mlc P)\mli (P\mld P)\mlc (P\mld P)$. The latter was used by Blass \cite{Bla92} as an example of a game-semantically valid principle not provable in affine logic.\vspace{7pt}

\noindent 1. $\begin{array}{l}
(p_1\mlc p_2)\mld (p_3\mlc p_4)\mli (p_1\mld p_3)\mlc (p_2\mld p_4)
\end{array}$  \ \ Axiom\vspace{3pt}

\noindent 2. $\begin{array}{l}
(p_1\mlc p_2)\mld (p_3\mlc P)\mli (p_1\mld p_3)\mlc (p_2\mld P)
\end{array}$  \ \ Match: 1\vspace{3pt}

\noindent 3. $\begin{array}{l}
(p_1\mlc p_2)\mld (P\mlc P)\mli (p_1\mld P)\mlc (p_2\mld P)
\end{array}$  \ \ Match: 2\vspace{3pt}

\noindent 4. $\begin{array}{l}
(p_1\mlc P)\mld (P\mlc P)\mli (p_1\mld P)\mlc (P\mld P)
\end{array}$  \ \ Match: 3\vspace{3pt}

\noindent 5. $\begin{array}{l}
(P\mlc P)\mld (P\mlc P)\mli (P\mld P)\mlc (P\mld P)
\end{array}$  \ \ Match: 4\vspace{7pt}
\end{example}

\begin{example}\label{j28f}
In Example \ref{j28b} we saw a $\clthree$-proof of  $\ada x\ade y\bigl(p(x)\mli p(y)\bigr)$. The same proof, of course, is also a $\clfour$-proof. Below is a $\clfour$-proof of the stronger version of this formula where we have an uppercase rather than lowercase $P$:\vspace{7pt}

\noindent 1. $\begin{array}{l}
\twg
\end{array}$  Axiom \vspace{3pt}

\noindent 2. $\begin{array}{l}
p(s)\mli p(s)
\end{array}$  \ \ Axiom\vspace{3pt}

\noindent 3. $\begin{array}{l}
P(s)\mli P(s)
\end{array}$  \ \ Match: 2\vspace{3pt}

\noindent 4. $\begin{array}{l}
\ade y\bigl(P(s)\mli P(y)\bigr)
\end{array}$  \ \ $\ade$-Choose: 3\vspace{3pt}

\noindent 5. $\begin{array}{l}
\ada x\ade y\bigl(P(x)\mli P(y)\bigr)
\end{array}$  \ \ Wait: 1,4\vspace{7pt}
\end{example}

\begin{example}\label{j28g}
While $\clfour$ proves the elementary formula $p\mli p\mlc p$, it does not prove its general counterpart $P\mli P\mlc P$. Indeed,  $\elz{P\mli P\mlc P}=\twg\mli \tlg\mlc\tlg$ and hence, obviously, $P\mli P\mlc P$ cannot be derived by Wait. This formula cannot be derived by Choose rules either, because it contains no choice operators. Finally, if it is derived by Match, the premise should be $p\mli P\mlc p$ or $p\mli p\mlc P$. In either case, such a premise  cannot be proven, as it contains no choice operators and its elementarization is $p\mli\tlg\mlc p$ or $p\mli p\mlc \tlg$. 
\end{example}

Let $F$ be a $\clfour$-formula. A {\bf substitution}\label{isubstitution} for $F$  is a function $^\heartsuit$ that sends:
\begin{itemize}
\item each nonlogical $n$-ary elementary letter  $p$ of $F$ to an elementary $\clthree$-formula $p^\heartsuit(x_1,\ldots,x_n)$ --- with here and below $x_1,\ldots,x_n$ being a context-setting fixed $n$-tuple of pairwise distinct variables --- which does not contain any free variables that have bound occurrences in $F$;
\item each $n$-ary general letter $P$ of $F$ to an (elementary or nonelementary) $\clthree$- formula $P^\heartsuit(x_1,\ldots,x_n)$ which does not contain any free variables that have bound occurrences in $F$;
\item each $n$-ary  function symbol $f$ of $F$ to a term $f^\heartsuit(x_1,\ldots,x_n)$ which does not contain any  variables that have bound occurrences in $F$. 
\end{itemize}

The above uniquely extends to  a  mapping that sends each term $\tau$ of $F$ to a term $\tau^\heartsuit$, and each subformula $H$ of $F$ to a $\clthree$-formula $H^\heartsuit$ by stipulating that: 
\begin{enumerate}
\item $x^\heartsuit= x$ (any variable $x$). 
\item Where $f$ is an $n$-ary function symbol and $\tau_1,\ldots,\tau_n$ are terms, $\bigl(f(\tau_1,\ldots,\tau_n)\bigr)^\heartsuit= f^\heartsuit(\tau_{1}^{\heartsuit},\ldots,\tau_{n}^{\heartsuit})$. 
\item $(\tau_1\equals \tau_2)^\heartsuit$ is $\tau_{1}^{\heartsuit}\equals \tau_{2}^{\heartsuit}$. 
\item Where $\mathfrak{L}$ is an $n$-ary nonlogical elementary or general letter  and $\tau_1,\ldots,\tau_n$ are terms, $\bigl(\mathfrak{L}(\tau_1,\ldots,\tau_n)\bigr)^\heartsuit = \mathfrak{L}^\heartsuit(\tau_{1}^{\heartsuit},\ldots,\tau_{n}^{\heartsuit})$. 
\item $^{\heartsuit}$ commutes with all logical operators:  
\begin{itemize}
\item $\twg^\heartsuit = \twg$; 
\item $\tlg^\heartsuit = \tlg$; 
\item $(\gneg E)^{\heartsuit}= \gneg E^{\heartsuit}$; 
\item $(E_1\mlc\ldots\mlc E_n)^{\heartsuit} = E^{\heartsuit}_{1}\mlc \ldots\mlc E^{\heartsuit}_n$; 
\item $(E_1\mld\ldots\mld E_n)^{\heartsuit}= E^{\heartsuit}_{1}\mld \ldots\mld E^{\heartsuit}_{n}$; 
\item $(E_1\adc\ldots\adc E_n)^{\heartsuit} = E^{\heartsuit}_{1}\adc \ldots\adc E^{\heartsuit}_n$; 
\item $(E_1\add\ldots\add E_n)^{\heartsuit} = E^{\heartsuit}_{1}\add \ldots\add E^{\heartsuit}_{n}$;
\item $(\cla x E)^{\heartsuit} = \cla x(E^{\heartsuit})$; 
\item $(\cle x E)^{\heartsuit} = \cle x(E^{\heartsuit})$; 
\item $(\ada x E)^{\heartsuit} = \ada x(E^{\heartsuit})$; 
\item $(\ade x E)^{\heartsuit} = \ade x(E^{\heartsuit})$.
\end{itemize}
\end{enumerate}

We say that a $\clthree$-formula $E$ is an {\bf instance}\label{iinstance2} of a $\clfour$-formula $F$, or that $E$ {\bf matches} $F$, iff $E=F^\heartsuit$ for some substitution $^\heartsuit$ for $F$. 

\begin{theorem}\label{ccs}
A $\clfour$-formula is provable in $\clfour$ iff all of its instances are provable in $\clthree$.
\end{theorem}

\begin{idea} The completeness part of this theorem is unnecessary for the purposes of the present paper, and its proof is omitted. For the soundness part, consider a  $\clfour$-provable formula $F$ and an arbitrary instance  $F^\heartsuit$ of it. We need to construct a $\clthree$-proof of $F^\heartsuit$. The idea here is to let such  
a proof  simulate the $\clfour$-proof of $F$. Speaking very roughly, simulating steps associated with $\add$-Choose, $\ade$-Choose and Wait is possible because these rules of $\clfour$ are also present in $\clthree$. As for the Match rule, it can be simulated by a certain ``deductive counterpart'' of the earlier seen copycat strategy. Namely, in the bottom-up view of the $\clthree$-proof under construction, every application of Wait that modifies a subformula originating from a matched (in the $\clfour$-proof) literal, should be followed by a symmetric application of $\add$-Choose or $\ade$-Choose in the subformula originating from the other matched literal  --- an application that evens out the two subformulas  so that one remains the negation of the other.   
\end{idea}

\begin{proof} Our proof will be focused on the soundness (``only if'') part of the theorem, as nothing in this paper relies on the completeness (``if'') part. 
We only want to point out that, essentially, the latter has been proven in Section 5 of \cite{Japtcs2}. Specifically, the proof of Lemma 5.1 of \cite{Japtcs2} proceeds by showing that, if $\clfour\not\vdash F$, then there is a $\clthree$-formula $\lceil F\rceil $ which is an instance of $F$ such that $\clthree\not\vdash\lceil F\rceil$. However, as noted earlier,  the logics under the names ``$\clthree$'' and ``$\clfour$''    are not exactly the same in \cite{Japtcs,Japtcs2} as they are here. Namely, \cite{Japtcs,Japtcs2} allowed constants in formulas while now we do not allow them. On the other hand, now we have  $\equals$ and function symbols in the language whereas the approach of \cite{Japtcs,Japtcs2} did not consider them, nor did it have the special-status variable $\bound$. Also, as we remember, in our present treatment $\ada,\ade$ mean $\ada^\bound,\ade^\bound$, whereas in \cite{Japtcs,Japtcs2} they meant properly $\ada,\ade$. Such technical differences, however, are minor, and have no impact on the relevant proofs. So, the above-mentioned proof from \cite{Japtcs2}, with just a few rather straightforward adjustments, goes through as a proof of the completeness part of the present theorem as well.

For the soundness part, we   extend  the language of $\clfour$ by adding to it a new sort of nonlogical letters   called {\bf hybrid}.\label{ihl} Each $n$-ary hybrid letter is  a pair $P_q$, where $P$ --- called its {\bf general component}\label{igc} --- is an $n$-ary general letter, and $q$ --- called its {\bf elementary component}\label{iec} --- is a nonlogical $n$-ary elementary letter. 
And vice versa: for every pair $(P,q)$ of letters of the above sort, we have an $n$-ary hybrid letter $P_q$. 
Formulas of this extended language,  to which we will be referring as {\bf hyperformulas},\label{ihf} are built in the same way as $\clfour$-formulas, with the difference that now atoms can be of any of the three --- elementary, general or hybrid --- sorts.  {\bf Surface occurrence}, (elementary, general, hybrid) {\bf literal} and similar concepts  straightforwardly extend  from $\clthree$- and $\clfour$-formulas to hyperformulas. Furthermore, concepts such as surface occurrence, positive occurrence, etc. extend from subformulas to parts of subformulas, such as letters occurring in them, in the obvious way. 

We say that a hyperformula $E$ is a {\bf $\clfour^\circ$-formula}\label{icl4cf}  iff, for every hybrid letter $P_q$ occurring in $E$, the following  conditions are satisfied:
\begin{enumerate}
\item $E$ has exactly two occurrences of $P_q$, where one occurrence is positive and the other occurrence is negative, and both occurrences are surface occurrences. We say that the corresponding two  literals --- where one looks like $P_q(\vec{\tau})$ and the other like $\gneg P_q(\vec{\theta})$ --- are {\bf matching}.\label{imatlit}
\item The elementary letter $q$ does not occur in $E$, nor is it the elementary component of any hybrid letter occurring in $E$ other than $P_q$. 
\end{enumerate}

 Of course, every $\clfour$-formula is also a $\clfour^\circ$-formula --- one with no hybrid letters.

The  {\bf elementarization}\label{ielz3}  \(\elz{E}\) of a $\clfour^\circ$-formula $E$ is the result of replacing, in $E$, each surface occurrence of the form $G_1\adc\ldots\adc G_n$ or $\ada x G$ by $\twg$, each surface occurrence of the form $G_1\add\ldots\add G_n$
or $\ade xG$  by  $\tlg$, every positive  surface occurrence of each general literal by $\tlg$, and every surface occurrence 
of each hybrid letter by the elementary component of that letter.  

We are going to employ a ``version'' of $\clfour$ called $\clfour^\circ$.\label{icl4c} Unlike $\clfour$ whose language consists only of $\clfour$-formulas, the language of $\clfour^\circ$ allows any $\clfour^\circ$-formulas. The axioms and rules of $\clfour^\circ$ are the same as those of  $\clfour$ ---  only, now applied to  any $\clfour^\circ$-formulas rather than just $\clfour$-formulas --- with the difference that the rule of Match is replaced by the following rule that we call {\bf Match$^\circ$}:\label{imatchc} 
\[\frac{F[P_q(\vec{\tau}),\gneg P_q(\vec{\theta})]}{F[P(\vec{\tau}),\gneg P(\vec{\theta})]},\]  
where $P$ is any $n$-ary general letter, $q$ is any $n$-ary elementary letter not occurring in the conclusion (neither independently nor as the elementary component of some hybrid letter), and $\vec{\tau}$,$\vec{\theta}$ are any $n$-tuples of terms.\vspace{7pt} 

{\bf Claim 1}.
{\em For any $\clfour$-formula $E$, \ if $\clfour\vdash E$, then $\clfour^\circ\vdash E$}.\vspace{7pt} 
 
\begin{subproof} The idea that underlies our proof of this claim is very simple: every application of Match naturally turns into an application of Match$^\circ$. 

Indeed,
consider any $\clfour$-proof of $E$. It can be  seen as a tree all of the leaves of which are labeled with axioms and every non-leaf node of which is   labeled with a formula that follows by one of the rules of $\clfour$ from (the labels of) its children, with $E$ being the label of the root. By abuse of terminology, here we identify the nodes of this tree with their labels, even though, of course, it may be the case that different nodes have the same label. For each node $G$ of the tree that is derived from its child $H$ by Match --- in particular, where $H$ is the result of replacing in $G$ a positive and a negative surface occurrences of an $n$-ary general letter $P$ by an $n$-ary nonlogical elementary letter $q$ --- do the following: replace $q$ by the hybrid letter $P_q$ in $H$ as well as in all of its descendants in the tree. It is not hard to see that this way we will get a 
$\clfour^\circ$-proof of $E$.  \end{subproof}

The concept of a {\bf substitution}\label{ish} $^\heartsuit$ for a $\clfour^\circ$-formula $E$, and the corresponding $\clthree$-formula $E^\heartsuit$, are defined in the same ways as for $\clfour$-formulas, treating each hybrid letter $P_q$ as a separate (not related to $P$ or any other $P_p$ with $p\not=q$) general letter. 

We say that a $\clthree$-formula $E$ is a {\bf TROW-premise}\label{itrow}  of a $\clthree$-formula $F$ (``TROW''=``Transitive Reflexive Ordinary Wait'') iff $E$ is $F$, or an ordinary Wait-premise of $F$, or an ordinary Wait-premise of an ordinary Wait-premise of $F$, or \ldots.

Let $E$ be a $\clfour^\circ$-formula with exactly $n$ positive surface occurrences of general literals, with those occurences being (not necessarily pairwise distinct literals) $G_1,\ldots,G_n$.  And let $^\heartsuit$ be a substitution for $E$. Then $E^\heartsuit$ can obviously be written as $H[G_{1}^{\heartsuit},\ldots,G_{n}^{\heartsuit}]$, where  $G_{1}^{\heartsuit},\ldots,G_{n}^{\heartsuit}$ are surface occurrences   originating from the occurrences of $G_{1},\ldots,G_{n}$ in $E$. Under these conditions, by a  {\bf $^\heartsuit$-quasiinstance}\label{iquasi} of $E$  we will mean any TROW-premise of $H[G_{1}^{\heartsuit},\ldots,G_{n}^{\heartsuit}]$ that can be written as $H[J_{1},\ldots,J_{n}]$. To summarize  in more intuitive terms, a $^\heartsuit$-quasiinstance of $E$ is a TROW-premise of $E^\heartsuit$ where all (if any) changes have taken place exclusively in subformulas ($G_{1}^{\heartsuit},\ldots,G_{n}^{\heartsuit}$) that originate from positive occurrences of general literals ($G_{1},\ldots,G_{n}$) in $E$. Of course, $E^\heartsuit$ is one of the $^\heartsuit$-quasiinstances of $E$.

By a (simply)  {\bf quasiinstance} of a  $\clfour^\circ$-formula $E$ we mean a $^\heartsuit$-quasiinstance of $E$ for some substitution $^\heartsuit$ for $E$. Note that every instance is a quasiinstance but not necessarily vice versa. \vspace{7pt} 

{\bf Claim 2}.
{\em For any $\clfour^\circ$-formula $E$,    if $\clfour^\circ\vdash E$, then every quasiinstance of $E$ is provable in $\clthree$.}\vspace{7pt} 
 
\begin{subproof} Consider any   $\clfour^\circ$-provable formula $E$.  We want to show that $\clthree$ proves any quasiinstance of $E$.  This will be done by induction on the length of the $\clfour^\circ$-proof of $E$; within the inductive step of this induction, we will use a second induction --- induction on the complexity (the number of logical connectives) of the quasiinstance of $E$ under consideration. Call the first induction {\em primary} and the second induction {\em secondary}. These adjectives will also be applied to the corresponding inductive hypotheses.

For the basis of the primary induction, assume  $E$ is an axiom of $\clfour^\circ$ (and hence of $\clthree$ as well), i.e. $E$ is a valid formula of classical logic. Consider any substitution $^\heartsuit$ for $E$. The formula $E^\heartsuit$ is an axiom of ($\clfour^\circ$ and) $\clthree$, because  classical validity is closed under applying substitutions. And, since $E$ is elementary,  $E^\heartsuit$ is the only $^\heartsuit$-quasiinstance of it. So, we are done. 

Below comes the inductive step of the primary induction, divided into three cases.\vspace{5pt}

{\em Case 1.} Assume $E$ is obtained from a premise $G$ by $\add$-Choose or $\ade$-Choose.  Consider any substitution $^\heartsuit$ for $E$. Obviously, $E^\heartsuit$ follows from $G^\heartsuit$ by the same rule;\footnote{To ensure that Convention \ref{jan26} is respected, here we can safely assume that, if $E$ is obtained by $\ade$-Choose and this rule (in the bottom-up view) introduced a fresh variable $s$, then $s$ has no (bound) occurrences in $G^\heartsuit$, or otherwise rename $s$ into some neutral variable.}  and, by the primary induction hypothesis, $\clthree\vdash G^\heartsuit$, so we have $\clthree \vdash E^\heartsuit$. Furthermore, what we just observed   extends to any other (other than $E^\heartsuit$) $^\heartsuit$-quasiinstance $H$ of $E$ as well: with some thought, one can see that such an $H$ follows from a certain (the corresponding) $^\heartsuit$-quasiinstance of $G$  by the same rule $\add$-Choose or $\ade$-Choose as $E$ follows from $G$.\vspace{5pt} 

{\em Case 2.} Assume $E$ is obtained from premises $\elz{E},G_1,\ldots,G_n$ by Wait. 
Consider any substitution $^\heartsuit$ for $E$ and any $^\heartsuit$-quasiinstance $H$ of $E$. We want to show that $H$ can be derived in $\clthree$ by Wait.

 The provability of the elementary formula $\elz{E}$ obviously means that it is an axiom, i.e., a valid formula of classical logic. Let $J_1,\ldots,J_k$ be all positive surface occurrences of general literals in $E$, and let $E'$ be the formula obtained from $E$ by replacing those occurrences by $q_1,\ldots,q_k$, where the $q_i$ are pairwise distinct $0$-ary elementary letters not occurring in $E$.  Observe that then $\elz{E'}$ differs from $\elz{E}$ in that, where the former has $k$ positive occurrences of $\tlg$ (originating  from  $J_1,\ldots,J_k$ when elementarizing $E$), the latter has the $k$ atoms  $q_1,\ldots,q_k$. It is known from classical logic that replacing positive occurrences of $\tlg$ by whatever formulas does not destroy validity. Hence, as $\elz{E}$ is valid, so is $\elz{E'}$.  Now, with some analysis, details of which are left to the reader, one can see that the formula $\elz{H}$ is a substitutional instance --- in both our present sense as well as in the classical sense --- of $\elz{E'}$. So, as an instance of a classically valid formula, $\elz{H}$ is classically valid, i.e. is an axiom of $\clthree$, and we thus have 
\begin{equation}\label{dec5}
\clthree\vdash \elz{H}.
\end{equation}
We now want to show that:
\begin{equation}\label{dec51}
\mbox{\em Whenever $H=H[K_1\adc\ldots\adc K_m]$ and $1\leq i\leq m$, we have $\clthree\vdash H[K_i]$.}
\end{equation}
Indeed, assuming the conditions of (\ref{dec51}), one of the following should be the case:
\begin{enumerate}
\item The occurrence of $K_1\adc\ldots\adc K_m$ in $H$ originates from a (surface) occurrence of a subformula $L_1\adc\ldots\adc L_m$ in $E$ (so that $K_1=L_{1}^\heartsuit$, $\ldots$, $K_m=L_{m}^\heartsuit$). Then, obviously, $H[K_i]$ is a  $^{\heartsuit}$-quasiinstance of one of the ordinary Wait-premises $G_j$ ($1\leq j\leq n$) of $E$. But then, by the primary induction hypothesis, we have $\clthree\vdash H[K_i]$.
\item The occurrence of $K_1\adc\ldots\adc K_m$ in $H$ originates from a (positive surface) occurrence of some general literal $L$ in $E$ (so that $K_1\adc\ldots\adc K_m$ has a surface occurrence in a TROW-premise of $L^\heartsuit$).  Note that then $H[K_i]$, just like $H$, is a $^\heartsuit$-quasiinstance of $E$. By the secondary induction hypothesis, the formula $H[K_i]$, as a quasiinstance of $E$ less complex than $H$ itself, is provable in $\clthree$.
\item The occurrence of $K_1\adc\ldots\adc K_m$ in $H$ originates from a (positive surface) occurrence of some hybrid literal $L$ in $E$  (so that $K_1\adc\ldots\adc K_m$ has a surface occurrence in $L^\heartsuit$). Then $H[K_i]$ contains a surface occurrence of the subformula $\gneg K_1\add\ldots\add \gneg K_m$, originating from the  occurrence of the matching hybrid literal $L'$ in $E$. Let $H'$ be the result of replacing that $\gneg K_1\add\ldots\add \gneg K_m$ by $\gneg K_i$ in $H[K_i]$. Obviously $H'$, just like $H$,  is a quasiinstance of $E$, but it is less complex than $H$. Hence, by the secondary induction hypothesis, $\clthree \vdash H'$. But $H[K_i]$ follows from $H'$ by $\add$-Choose. So,  $\clthree\vdash H[K_i]$.
\end{enumerate}
In all  cases we thus get $\clthree\vdash H[K_i]$, as desired.

In a very similar way, we can further show that 
\begin{equation}\label{dec52}
\mbox{\em Whenever $H=H[\ada x K(x)]$, we have $\clthree\vdash H[K(s)]$ for some variable $s$ not occurring in $H$.}
\end{equation}

Now, from (\ref{dec5}), (\ref{dec51}) and (\ref{dec52}), by Wait, we find the desired $\clthree \vdash H$.\vspace{5pt} 

{\em Case 3.} Suppose $P$ is a $k$-ary general letter, $q$ is a $k$-ary nonlogical elementary letter, $\tau_1,\ldots,\tau_k$, $\theta_1,\ldots,\theta_k$ are terms, \[E\ =\ E[P(\tau_1,\ldots,\tau_k),\ \gneg P(\theta_1,\ldots,\theta_k)]\] and it is obtained from the premise 
\begin{equation}\label{j30b}
 E[P_q(\tau_1,\ldots,\tau_k),\ \gneg P_q(\theta_1,\ldots,\theta_k)]
\end{equation}
by Match$^\circ$.
 Consider any substitution $^\heartsuit$ for $E$, and any $^\heartsuit$-quasiinstance  of $E$.  Obviously such a quasiinstance  can be written in the form
\begin{equation}\label{j30c} 
H[K_1(\tau_{1}^{\heartsuit},\ldots,\tau_{k}^{\heartsuit}),\ \gneg K_2(\theta_{1}^{\heartsuit},\ldots,\theta_{k}^{\heartsuit})],
\end{equation}
 where $H$ inherits the logical structure of $E$ (but probably adds some extra complexity to it), $K_1(\tau_{1}^{\heartsuit},\ldots,\tau_{k}^{\heartsuit})$ is a TROW-premise of $P^\heartsuit(\tau_{1}^{\heartsuit},\ldots,\tau_{k}^{\heartsuit})$ and $\gneg K_2(\theta_{1}^{\heartsuit},\ldots,\theta_{k}^{\heartsuit})$ is a TROW-premise of $\gneg P^\heartsuit(\theta_{1}^{\heartsuit},\ldots,\theta_{k}^{\heartsuit})$. With a little thought, one can see that there is a series of $\add$-Chooses and $\ade$-Chooses that we can apply --- in the bottom-up sense --- to (\ref{j30c}) to ``even out'' the $K_1(\tau_{1}^{\heartsuit},\ldots,\tau_{k}^{\heartsuit})$ and $\gneg K_2(\theta_{1}^{\heartsuit},\ldots,\theta_{k}^{\heartsuit})$ subformulas and bring (\ref{j30c})  to 
\begin{equation}\label{j30a}
H[K(\tau_{1}^{\heartsuit},\ldots,\tau_{k}^{\heartsuit}),\ \gneg K(\theta_{1}^{\heartsuit},\ldots,\theta_{k}^{\heartsuit})] 
\end{equation}
for a certain formula $K(x_1,\ldots,x_n)$. 
Let $^\diamondsuit$ be the substitution for $E$ which sends $P_q$ to $K(x_1,\ldots,x_n)$ and agrees with $^\heartsuit$ on everything else. With a little thought, we can see that (\ref{j30a}) is a $^\diamondsuit$-quasiinstance of (\ref{j30b}).   Hence, by the primary induction hypothesis,   $\clthree\vdash (\ref{j30a})$. Now, as we already know,   (\ref{j30c}) is obtained from (\ref{j30a}) using a series of $\add$-Chooses and $\ade$-Chooses. Hence (\ref{j30c}) --- which, as we remember, was an arbitrary quasiinstance of $E$ --- is provable in $\clthree$.\vspace{5pt}

The above Cases 1,2,3 complete the inductive step of our primary induction, and we conclude that, whenever $E$ is a $\clfour^\circ$-provable formula, every quasiinstance of it is provable in $\clthree$. 
\end{subproof}

To complete our proof of (the soundness part of) Theorem \ref{ccs}, assume $\clfour\vdash F$. Then, by \mbox{Claim 1}, $\clfour^\circ\vdash F$. Consider any substitution $^\heartsuit$ for $F$. $F^\heartsuit$ is a (quasi)instance of $F$  and hence, by Claim 2, $\clthree\vdash F^\heartsuit$. Since both $F$ and $^\heartsuit$ are arbitrary, we conclude that every instance of every $\clfour$-provable formula is provable in $\clthree$. 
\end{proof}

\section{The basic system of ptarithmetic  introduced}\label{ss11}

There can be various interesting systems of arithmetic based on computability logic (``{\bf clarithmetics}''),\label{iclarithmetic} depending on what language we consider, what fragment of CL is taken as a logical basis, and what extra-logical rules and axioms are employed. \cite{Japtowards} introduced three   systems of clarithmetic, named {\bf CLA1}, {\bf CLA2} and {\bf CLA3}, all based on the fragment {\bf CL12} (also introduced in \cite{Japtowards}) of computability logic. The basic one of them is {\bf CLA1}, with the other two systems being straightforward modifications of it through slightly extending ({\bf CLA2}) or limiting ({\bf CLA3}) the underlying nonlogical language. Unlike our present treatment, the underlying semantical concept for the systems of \cite{Japtowards} was computability-in-principle rather than efficient computability.  

The new system of clarithmetic introduced in this section, meant to axiomatize efficient computability of number-theoretic computational problems, is named $\arfour$.\label{ij31} The term ``{\bf ptarithmetic}''\label{iptarithmetic2} is meant to be a generic name for systems in this style, even though we often use it to refer to our present particular system $\arfour$ of ptarithmetic.  

The language of $\arfour$, whose formulas  we refer to as {\bf $\arfour$-formulas},\label{iarfff} is obtained from the language of $\clthree$ by removing all nonlogical predicate letters (thus only leaving the logical predicate letter $\equals $),  and also removing all but four function letters, which are:
 
\begin{itemize}
\item $zero$, $0$-ary. We will write $\zero$ for $zero$.
\item $successor$, unary. We will write $\tau\successor $ for $successor(\tau)$.
\item $sum$, binary. We will write $\tau_1\plus \tau_2$ for $sum(\tau_1,\tau_2)$.
\item $product$, binary. We will write $\tau_1\mult  \tau_2$ for $product(\tau_1,\tau_2)$.
\end{itemize}

From now on, when we just say ``formula'', we mean ``$\arfour$-formula'', unless otherwise specified or suggested by the context.

Formulas that have no free occurrences of variables are said to be {\bf sentences}.\label{isentence}

The concept of an interpretation explained earlier can now be restricted to interpretations that are only defined on $\zero$, $\successor $, $\plus $,  $\mult $ and $\equals$, as the present language has no other nonlogical function or predicate letters. Of such interpretations, the {\bf standard interpretation}\label{isi} $^\dagger$ is the one that interprets $\zero$ as (the $0$-ary function whose value is) $0$, interprets $\successor $ as the standard successor ($x\plus 1$) function, interprets $\plus $ as the sum function,   interprets $\mult $ as the product function, and interprets $\equals$ as the identity relation. Where $F$ is a $\arfour$-formula, the {\bf standard interpretation of} 
$F$ is the game $F^\dagger$, which we typically write simply as $F$ unless doing so may cause ambiguity. 
  
The {\bf axioms} of $\arfour$ are grouped into logical and nonlogical. 

The {\bf logical axioms}\label{ilax} of $\arfour$ are all elementary $\arfour$-formulas provable in classical first-order logic. That is, all axioms of $\clthree$ that are $\arfour$-formulas. 

As for the {\bf nonlogical axioms},\label{inlax} they are divided into what we call ``Peano'' and ``extra-Peano'' axioms. 

The {\bf Peano axioms}\label{ipeanax} of $\arfour$ are  all sentences matching the following seven schemes,\footnote{Only Axiom 7 is a scheme in the proper sense. Axioms 1-6 are ``schemes'' only in the sense that $x$ and $y$ are metavariables for variables rather than particular variables. These axioms can be painlessly turned into particular formulas by fixing some particular variables in the roles of $x$ and $y$. But why bother.} with $x,y,y_1,\ldots,y_n$ being any pairwise distinct variables other than $\bound$:

\begin{description}
\item[Axiom 1:] $\cla x(\zero\notequals x\successor )$
\item[Axiom 2:] $\cla x\cla y(x\successor \equals y\successor \mli x\equals y)$
\item[Axiom 3:] $\cla x(x\plus \zero\equals x)$
\item[Axiom 4:] $\cla x\cla y\bigl( x\plus y\successor \equals (x\plus y)\successor \bigr)$
\item[Axiom 5:] $\cla x(x\mult  \zero\equals \zero)$
\item[Axiom 6:] $\cla x\cla y\bigl(x\mult  y\successor \equals (x\mult  y)\plus x\bigr)$
\item[Axiom 7:] $\cla y_1\ldots\cla y_n\Bigl(F(\zero)\mlc \cla x\bigl(F(x)\mli F(x\successor )\bigr)\mli \cla xF(x)\Bigr)$, where $F(x)$ is any elementary formula and $y_1,\ldots,y_n$ are all of the variables occurring free in it and different from $\bound,x$. 
\end{description}

Before we present the extra-Peano axioms of $\arfour$, we need to agree on some notational matters. The language of $\arfour$ extends that of {\bf Peano Arithmetic}\label{ipa3} $\pa$ (see, for example, \cite{Hajek}) through adding to it $\adc,\add,\ade,\ada$. And the language of $\pa$ is known to be  very expressive, despite its nonlogical vocabulary officially being limited to only $\zero,\successor ,\plus ,\mult $. Specifically, it allows us to express, in a certain reasonable and standard way, all  recursive functions and relations, and beyond. Relying on the common knowledge of the power of the language of \pa,  we will be using standard expressions such as $x\mleq y$, $y\mgreater x$, etc. in formulas as abbreviations of the corresponding proper expressions of the language. Namely, in our metalanguage, $|x|$\label{iii} will refer to the length of (the binary numeral for the number represented by) $x$.\footnote{Warning: here we do not follow the standard convention, according to which $|0|$ is considered to be $0$ rather than $1$.} So, when we write, say, ``$|x|\mleq \bound$'', it is to be understood as an abbreviation of a standard formula of {\bf PA} saying that the size of $x$ does not exceed $\bound$. 

Where $\tau$ is a term, we will be using $\tau 0$ and $\tau 1$ as abbreviations for the terms $\zero\successor\successor\mult \tau$ and $(\zero\successor\successor\mult \tau)\successor$, respectively. The choice of this notation is related to the fact that, given any natural number $a$, the binary representation of $\zero\successor\successor \mult a$ (i.e., of $2a$) is nothing but the binary representation of $a$ with a ``$0$'' added on its right. Similarly, the binary representation of $(\zero\successor\successor\mult a)\successor$ is nothing but the binary representation of $a$ with a ``$1$'' added to it. Of course, here an exception is the case $a\equals 0$. It can be made an ordinary case by assuming that adding any number of $0$s at the beginning of a binary numeral $b$ results in a legitimate numeral representing the same number as $b$. 

The number $a0$ (i.e. $2a$) will be said to be the {\bf binary $0$-successor}\label{ibzs} of $a$, and $a1$ (i.e. $2a+1$) said to be the {\bf binary $1$-successor}\label{ibos} of $a$; in turn, we can refer to $a$ as the {\bf binary predecessor}\label{ibp} of $a0$ and $a1$. As for $a\successor $, we can refer to it as the {\bf unary successor}\label{ius} of $a$, and refer to $a$ as the {\bf unary predecessor}\label{iup} of $a\successor $. Every number has a binary predecessor, and every number except $0$ has a unary predecessor. 
Note that  the binary predecessor of a number is the result of deleting the last digit in its binary representation. Two exceptions are the numbers $0$ and $1$, both having $0$ as their binary predecessor.

Below and elsewhere, by a {\bf $\bound$-term}\label{ibterm} we mean a term of the official language of $\arfour$ containing no variables other than $\bound$. That is, a term exclusively built from $\bound, \zero, \successor , \plus ,\mult $.   

Now, the {\bf extra-Peano axioms}\label{iextra} of $\arfour$ are all formulas matching the following six schemes, where $s$ is any variable and $x$ is any variable other than $\bound,s$:  

\begin{description}
\item[Axiom 8:] $\ade x(x\equals \zero)$
\item[Axiom 9:] $s\equals \zero\add s\notequals \zero$
\item[Axiom 10:] $|s\successor |\mleq \bound\mli \ade x(x\equals s\successor )$
\item[Axiom 11:] $|s0|\mleq \bound\mli \ade x(x\equals s0)$
\item[Axiom 12:] $\ade x(s\equals x0\add s\equals x1)$
\item[Axiom 13:] $|s|\mleq \bound$ 
\end{description}

The {\bf rules of inference} are also divided into two groups: logical and nonlogical. 

The {\bf logical rules}\label{ilogr} of $\arfour$ are the  rules $\add$-Choose, $\ade$-Choose, Wait and Modus Ponens of Section \ref{ss7}.

And  there is a single {\bf nonlogical rule}\label{inlr} of inference, that we call {\bf Polynomial Time Induction} ({\bf PTI}),\label{ipti} in which $\tau$ is any $\bound$-term, $s$ is any non-$\bound$ variable, and $E(s),F(s)$ are any formulas:

\begin{center}
\begin{picture}(258,75)
\put(120,55){\bf PTI}
\put(0,35){$E(\zero)\mlc F(\zero)\hspace{45pt}  E(s)\mlc F(s)\mli E(s\successor)\adc \bigl(F(s\successor )\mlc E(s)\bigr)$}
\put(0,25){\line(1,0){258}}
\put(88,10){$s\mleq \tau\mli E(s)\mlc F(s)$}
\end{picture}
\end{center}

Here the left premise is called the {\bf basis} of induction, and the right premise called the {\bf inductive step}.

A formula $F$ is considered {\bf provable} in $\arfour$ iff there is a sequence of formulas, called a {\bf $\arfour$-proof} of $F$, where each formula is either a (logical or nonlogical) axiom, or  follows from some previous formulas by one of the (logical or nonlogical) rules of inference, and where the last formula is $F$. We write $\arfour\vdash F$ to say that $F$ is provable (has a proof) in $\arfour$, and $\arfour\not\vdash F$ to say the opposite.

In view of the following fact, an alternative way to present $\arfour$ would be to delete Axioms 1-7 together with all logical axioms and, instead, declare all theorems of {\bf PA} to be axioms of $\arfour$ along with Axioms 8-13:

\begin{fact}\label{nov7}
Every (elementary $\arfour$-) formula provable in $\pa$ 
is also provable in $\arfour$.
\end{fact}

\begin{proof} Suppose (the classical-logic-based) $\pa$ proves  $F$. By the deduction theorem for classical logic this means  that, for some nonlogical axioms $H_1,\ldots,H_n$ of $\pa$, the formula $H_1\mlc\ldots\mlc H_n\mli F$ is provable in classical first order logic. Hence $H_1\mlc\ldots\mlc H_n\mli F$ is a logical axiom of $\arfour$ and is thus provable in $\arfour$. But the nonlogical axioms of $\pa$ are nothing but the Peano axioms of $\arfour$. So, $\arfour$ proves each of the formulas $H_1,\ldots,H_n$.  Now, in view of the presence of the rule of Modus Ponens in $\arfour$, we find that $\arfour \vdash F$.
\end{proof}

The above fact, on which we will be implicitly relying in the sequel, allows us to construct ``lazy'' $\arfour$-proofs where some steps can be justified by simply indicating their provability in {\bf PA}. That is, we will  treat theorems of {\bf PA} as if they were axioms of $\arfour$. As {\bf PA} is well known and studied, we safely assume that the reader has a good feel of what it can prove, so we do not usually further justify {\bf PA}-provability claims that we make. A reader less familiar with {\bf PA}, can take it as a rule of thumb that, despite G\"{o}del's incompleteness theorems, 
{\bf PA} proves every true number-theoretic fact that a contemporary high school  student can establish, or that mankind was or could be aware of before 1931.
 
\begin{definition}\label{dd1} \ 

1. By an {\bf arithmetical problem} in this paper we mean a game $A$ such that, for some formula $F$ of the language of $\arfour$, \ $A= F^\dagger$ (remember that $^\dagger$ is the standard interpretation). Such a formula $F$ is said a {\bf representation}\label{irepresentation} of $A$.

2. We say that an arithmetical problem $A$ is {\bf provable} in $\arfour$ iff it has a $\arfour$-provable representation.  
\end{definition}

In these terms, the central result of the present paper sounds as follows:

\begin{theorem}\label{tt1}
An arithmetical problem has a polynomial time solution iff it is provable in $\arfour$. 

Furthermore, there is an effective procedure that takes an arbitrary $\arfour$-proof of an arbitrary formula $X$ and constructs a   
polynomial time solution for $X$ (for $X^\dagger$, that is). 
\end{theorem}

\begin{proof} The soundness (``if'') part of this theorem will be proven in Section \ref{sectsound}, and the completeness (``only if'') part in Section 
\ref{sectcompl}.
\end{proof}

\section{On the extra-Peano axioms of $\arfour$}

While the well known Peano axioms hardly require any explanations as their traditional meanings are fully preserved in our treatment, the extra-Peano axioms  of $\arfour$ may be worth briefly commenting on. Below we do so with the soundness  of $\arfour$ (the ``if'' part of Theorem \ref{tt1}) in mind, according to which every $\arfour$-provable formula expresses an efficiently (i.e. polynomial time) computable number-theoretic problem.

\subsection{Axiom 8} \[\ade x(x\equals \zero)\]
This axiom expresses our ability to efficiently name the number (constant) $0$. Nothing --- even such a ``trivial''  thing --- can be taken for granted when it comes to formal systems!

\subsection{Axiom 9} 
\[s\equals \zero\add s\notequals \zero\]
This axiom expresses our ability to efficiently tell whether any given number is $0$ or not. Yet another ``trivial'' thing that still has to be explicitly stated in the formal system. 

\subsection{Axiom 10} 
\[|s\successor |\mleq \bound\mli \ade x(x\equals s\successor )\]
This axiom establishes the efficient computability of the unary successor function (as long as the size of the value of the function does not exceed the bound $\bound$). Note that its classical counterpart $|s\successor |\mleq \bound\mli \cle x(x\equals s\successor )$ is simply a valid formula of classical first-order logic (because so is its consequent) and, as such, carries no information. Axiom 10, on the other hand, is not at all a logically valid formula, and does carry certain nontrivial information about the standard meaning of the successor function. A nonstandard meaning (interpretation) of $s\successor $ could be an intractable or even incomputable function. 

\subsection{Axiom 11}  
\[|s0|\mleq \bound\mli \ade x(x\equals s0)\]
Likewise, Axiom 11 establishes the efficient computability of the binary $0$-successor function. There is no need to state a similar axiom for the binary $1$-successor function, as can be seen from the following lemma:

\begin{lemma}\label{nov8}
$\arfour\vdash |s1|\mleq \bound\mli \ade x(x\equals s1)$.
\end{lemma}

\begin{proof} Informally, a proof of $|s1|\mleq \bound\mli \ade x(x\equals s1)$ would be based on the fact (known from \pa) that the binary $1$-successor of $s$ is nothing but the unary successor of the binary $0$-successor of $s$; the binary $0$-successor $r$ of $s$ can be found  using   Axiom 11; and the unary successor $u$ of that $r$ can be further found using Axiom 10.  

Here is a (``lazy'' in the earlier-mentioned sense) $\arfour$-proof formalizing the above argument:\vspace{9pt}

\noindent 1. $\begin{array}{l}
\twg\mlc \bigl(|s0|\mleq \bound\mli \tlg\bigr)
\mli\bigl(|s1|\mleq \bound \mli \tlg\bigr) 
\end{array}$  \ \ \pa\vspace{3pt}

\noindent 2. $\begin{array}{l}
\twg
\end{array}$  \ \ Logical axiom\vspace{3pt}

\noindent 3. $\begin{array}{l}
|s\successor |\mleq \bound\mli \ade x(x\equals s\successor ) 
\end{array}$  \ \ Axiom 10 \vspace{3pt}

\noindent 4. $\begin{array}{l}
\ada y\bigl(|y\successor |\mleq \bound\mli \ade x(x\equals y\successor )\bigr) 
\end{array}$  \ \ Wait: 2,3 \vspace{3pt}

\noindent 5. $\begin{array}{l}
 |s0|\mleq \bound\mli \ade x(x\equals s0) 
\end{array}$  \ \ Axiom 11\vspace{3pt}

\noindent 6. $\begin{array}{l}
 \bigl(|t\successor |\mleq \bound\mli \tlg\bigr)\mlc \bigl(|s0|\mleq \bound\mli  (t\equals s0)\bigr)
\mli\bigl(|s1|\mleq \bound \mli \tlg\bigr) 
\end{array}$  \ \ \pa \vspace{3pt}

\noindent 7. $\begin{array}{l}
 \bigl(|t\successor |\mleq \bound\mli  r\equals t\successor \bigr)\mlc \bigl(|s0|\mleq \bound\mli  (t\equals s0)\bigr)
\mli\bigl(|s1|\mleq \bound \mli  r\equals s1 \bigr) 
\end{array}$  \ \ \pa\vspace{3pt}

\noindent 8. $\begin{array}{l}
 \bigl(|t\successor |\mleq \bound\mli  r\equals t\successor \bigr)\mlc \bigl(|s0|\mleq \bound\mli  (t\equals s0)\bigr)
\mli\bigl(|s1|\mleq \bound \mli \ade x(x\equals s1)\bigr) 
\end{array}$  \ \  $\ade$-Choose: 7\vspace{3pt}

\noindent 9. $\begin{array}{l}
 \bigl(|t\successor |\mleq \bound\mli \ade x(x\equals t\successor )\bigr)\mlc \bigl(|s0|\mleq \bound\mli  (t\equals s0)\bigr)
\mli\bigl(|s1|\mleq \bound \mli \ade x(x\equals s1)\bigr) 
\end{array}$  \ \  Wait: 6,8\vspace{3pt}

\noindent 10. $\begin{array}{l}

\ada y\bigl(|y\successor |\mleq \bound\mli \ade x(x\equals y\successor )\bigr)\mlc \bigl(|s0|\mleq \bound\mli  (t\equals s0)\bigr)
\mli\bigl(|s1|\mleq \bound \mli \ade x(x\equals s1)\bigr) 
\end{array}$  \ \ $\ade$-Choose: 9 \vspace{3pt}

\noindent 11. $\begin{array}{l}

\ada y\bigl(|y\successor |\mleq \bound\mli \ade x(x\equals y\successor )\bigr)\mlc \bigl(|s0|\mleq \bound\mli \ade x(x\equals s0)\bigr)
\mli\bigl(|s1|\mleq \bound \mli \ade x(x\equals s1)\bigr) 
\end{array}$  \ \ Wait:  1,10 \vspace{3pt}

\noindent 12.  $\begin{array}{l}
|s1|\mleq \bound \mli \ade x(x\equals s1)
\end{array}$ \ \ MP: 4,5,11 
\end{proof}

This was our first experience with generating a formal $\arfour$-proof. We will  do quite some more exercising with  $\arfour$-proofs later in order to start seeing that behind every informal argument in the style of the one given at the beginning of the proof of Lemma \ref{nov8} is a ``real'', formal proof. 

\subsection{Axiom 12} 
\begin{equation}\label{f4}
\ade x(s\equals x0\add s\equals x1)
\end{equation}
Let us compare the above with three other, ``similar'' formulas:  
\begin{equation}\label{f3}
\cle x(s\equals x0\add s\equals x1)
\end{equation}
\begin{equation}\label{f2}
\ade x(s\equals x0\mld s\equals x1)
\end{equation}
\begin{equation}\label{f1}
\cle x(s\equals x0\mld s\equals x1)
\end{equation}
All four formulas ``say the same'' about the arbitrary number represented by $s$, but in different ways. (\ref{f1}) is the weakest, least informative, of the four.  It says that  $s$ has a binary predecessor $x$, and that $s$ is even (i.e., is the binary $0$-successor of its binary predecessor) or odd (i.e., is the binary $1$-successor of its binary predecessor). This is an almost trivial piece of information. (\ref{f2}) and (\ref{f3}) carry stronger information. According to (\ref{f2}), $s$ not just merely {\em has} a binary predecessor $x$, but  such a predecessor can be actually and efficiently {\em found}. (\ref{f3}) strengthens (\ref{f1}) in another way. It says that $s$ can be efficiently determined to be even or odd. As for (\ref{f4}), which is Axiom 12 proper, it is the strongest. It carries two pieces of good news at once: we can efficiently find the binary predecessor $x$ of $s$ and, simultaneously, tell whether $s$ is even or odd.

\subsection{Axiom 13} \[|s|\mleq \bound\] 
Remember that our semantics considers only bounded valuations, meaning that the size of the number represented by a (free) variable $s$ will never exceed the bound represented by the variable $\bound$. Axiom 13 simply states this fact. Note that this is the only elementary formula among the extra-Peano axioms. 

In view of the above-said, whenever we say ``an arbitrary $s$'' in an informal argument, unless otherwise suggested by the context, it is always to be understood as an arbitrary $s$ whose size does not exceed the bound $\bound$.

Due to Axiom 13, $\arfour$ proves that the bound is nonzero: 

\begin{lemma}\label{zer}
$\arfour\vdash \bound\notequals\zero$.
\end{lemma}
\begin{proof} No binary numeral is of length $0$ and, of course, $\pa$ knows this. Hence  $\pa\vdash |s|\mleq\bound\mli\bound\notequals\zero$. From here and Axiom 13, by Modus Ponens,    $\arfour\vdash \bound\notequals\zero$.
\end{proof}

The formula of the following lemma is similar to Axiom 8, only it is about $\zero\successor$ instead of $\zero$.

\begin{lemma}\label{zersuc}
$\arfour\vdash \ade x(x\equals\zero\successor)$.
\end{lemma}
\begin{proof} 

\noindent 1. $\begin{array}{l}
  \bound\notequals\zero  
\end{array}$  \ \   Lemma \ref{zer} \vspace{3pt}

\noindent 2. $\begin{array}{l}
\ade x(x\equals \zero) 
\end{array}$  \ \   Axiom 8 \vspace{3pt}

\noindent 3. $\begin{array}{l}
\twg 
\end{array}$  \ \  Logical axiom \vspace{3pt}

\noindent 4. $\begin{array}{l}
|s\successor|\mleq\bound\mli\ade x(x\equals s\successor)
\end{array}$  \ \   Axiom 10 \vspace{3pt}

\noindent 5. $\begin{array}{l}
\ada y\bigl(|y\successor|\mleq\bound\mli\ade x(x\equals y\successor)\bigr)
\end{array}$  \ \  Wait: 3,4 \vspace{3pt}

\noindent 6. $\begin{array}{l}
  \bound\notequals\zero \mlc\tlg \mlc\twg \mli  \tlg
\end{array}$  \ \  Logical axiom \vspace{3pt}

\noindent 7. $\begin{array}{l}
  \bound\notequals\zero \mlc w\equals \zero \mlc  (|w\successor|\mleq\bound\mli\tlg ) \mli   \tlg
\end{array}$  \ \  $\pa$ \vspace{3pt}

\noindent 8. $\begin{array}{l}
  \bound\notequals\zero \mlc w\equals \zero \mlc  (|w\successor|\mleq\bound\mli v\equals w\successor  ) \mli   v\equals\zero\successor 
\end{array}$  \ \  $\pa$ \vspace{3pt}

\noindent 9. $\begin{array}{l}
  \bound\notequals\zero \mlc w\equals \zero \mlc  (|w\successor|\mleq\bound\mli v\equals w\successor  ) \mli   \ade x(x\equals\zero\successor)
\end{array}$  \ \  $\ade$-Choose: 8 \vspace{3pt}

\noindent 10. $\begin{array}{l}
  \bound\notequals\zero \mlc w\equals \zero \mlc \bigl(|w\successor|\mleq\bound\mli\ade x(x\equals w\successor)\bigr) \mli   \ade x(x\equals\zero\successor)
\end{array}$  \ \  Wait: 7,9 \vspace{3pt}

\noindent 11. $\begin{array}{l}
  \bound\notequals\zero \mlc w\equals \zero \mlc\ada y\bigl(|y\successor|\mleq\bound\mli\ade x(x\equals y\successor)\bigr) \mli   \ade x(x\equals\zero\successor)
\end{array}$  \ \ $\ade$-Choose: 10 \vspace{3pt}

\noindent 12. $\begin{array}{l}
  \bound\notequals\zero \mlc \ade x(x\equals \zero)\mlc\ada y\bigl(|y\successor|\mleq\bound\mli\ade x(x\equals y\successor)\bigr) \mli   \ade x(x\equals\zero\successor)
\end{array}$  \ \  Wait: 6,11 \vspace{3pt}

\noindent 13. $\begin{array}{l}
\ade x(x\equals\zero\successor)
\end{array}$  \ \  MP: 1,2,5,12 \vspace{3pt}
\end{proof}

\section{On the Polynomial Time Induction  rule}\label{ss14}

\begin{center}
\begin{picture}(258,55)
\put(0,35){$E(\zero)\mlc F(\zero)\hspace{45pt}  E(s)\mlc F(s)\mli E(s\successor)\adc \bigl(F(s\successor )\mlc E(s)\bigr)$}
\put(0,25){\line(1,0){258}}
\put(88,10){$s\mleq \tau\mli E(s)\mlc F(s)$}
\end{picture}
\end{center}

Induction is the cornerstone of every system of arithmetic. The many versions of formal arithmetic studied in the literature  (see \cite{Hajek}) mainly differ in varying --- typically weakening --- the unrestricted induction of the basic $\pa$, which is nothing but our Axiom 7.  
In $\arfour$, induction comes in two forms: Axiom 7, and the above-displayed PTI rule. Axiom 7, along with the other axioms of \pa, is taken to preserve the full power of \pa.
But it is limited to elementary formulas and offers no inductive mechanism applicable to computational problems in general. The role of  PTI  is  to provide such a missing mechanism.  

A naive attempt to widen the induction of $\pa$ would be to remove, from Axiom 7, the condition requiring that $F(x)$ be an elementary formula. This would be a terribly wrong idea though.   The resulting scheme would not even be a scheme of computable problems, let alone efficiently computable problems. Weakening the resulting scheme by additionally  replacing  
the blind quantifiers with choice quantifiers, resulting in (a scheme equivalent to) 
\begin{equation}\label{nov14a}
F(\zero)\mlc \ada x\bigl(F(x)\mli F(x\successor )\bigr)\mli \ada xF(x),
\end{equation}
would not fix the problem, either. The intuitive reason why (\ref{nov14a}) is unsound with respect to the semantics of computability logic, even if the underlying concept of interest is computability-in-principle without any regard for efficiency, is the following. In order to solve  $F(s)$ for an arbitrary $s$ (i.e., solve the problem $\ada xF(x)$), one would need to ``modus-ponens'' $F(x)\mli F(x\successor )$ with $F(\zero)$ to compute $F(1)$, then further ``modus-ponens'' $F(x)\mli F(x\successor )$ with $F(1)$ to compute $F(2)$, etc. up to $F(s)$. This would thus require $s$ ``copies'' of the resource $F(x)\mli F(x\successor )$. But the trouble is that only one copy of this resource   is available in the antecedent of (\ref{nov14a})!

The problem that we just pointed out can be neutralized by taking the following {\bf rule} instead of the  {\em formula} (scheme) (\ref{nov14a}): 
\[\frac{F(\zero)\hspace{45pt}  \ada x\bigl(F(x)\mli F(x')\bigr)}{\ada xF(x)}.\]
Taking into account that both semantically and syntactically $\ada xY(x)$ (in isolation) is equivalent to just $Y(s)$, we prefer to rewrite the above in the following form:
\begin{equation}\label{nov14c}
\frac{F(\zero)\hspace{35pt}   F(s)\mli F(s')}{F(s)}.
\end{equation}
Unlike the situation with (\ref{nov14a}), the resource $F(s)\mli F(s')$ comes in an unlimited supply in (\ref{nov14c}). As a rule, (\ref{nov14c}) assumes that the premise $F(s)\mli F(s')$ has already been proven. If proven, we know how to solve it. And if we know how to solve it, we can solve it as many times as needed. In contrast, in the case of (\ref{nov14a}) we do not really know how to solve the corresponding problem of the antecedent, but rather we rely on the environment to demonstrate such a solution; and the environment is obligated to do so only once.

(\ref{nov14c}) can indeed be shown to be a computability-preserving rule.  As we remember, however, we are concerned with efficient computability rather than computability-in-principle. And, in this respect, (\ref{nov14c}) is not sound. Roughly, the reason is the following: the way of computing $F(s)$ offered by (\ref{nov14c}) would require performing at least as many MP-style steps as the numeric value of $s$ (rather than the dramatically smaller {\em size} of $s$). This would yield a computational complexity exponential in the size of $s$. (\ref{nov14c}) can be made sound by limiting $s$ to ``sufficiently small'' numbers as done below, where  $\tau$ is an arbitrary  $\bound$-term: 
\begin{equation}\label{nov14b}
\frac{F(\zero)\hspace{35pt}   F(s)\mli F(s')}{s\mleq \tau\mli F(s)}.
\end{equation}  
Here the value of $\tau$, being a $(\bound,\zero,\successor,\plus,\mult)$-combination, is guaranteed to be polynomial in (the value of) $\bound$. Hence, we are no longer getting an exponential complexity of computation. This, by the way, explains the presence of   ``$s\mleq\tau$'' in the conclusion of PTI. Unlike (\ref{nov14a}) and (\ref{nov14c}), (\ref{nov14b}) is indeed sound with respect to our present semantics of efficient computability. 

A problem with (\ref{nov14b}), however, is that it is not strong enough --- namely, not as strong as PTI, and with (\ref{nov14b}) instead of PTI, we cannot achieve the earlier promised extensional completeness of $\arfour$. What makes PTI stronger than (\ref{nov14b}) is that its right premise is weaker. Specifically, while the right premise of (\ref{nov14b}) requires the ability to compute $F(s\successor)$ only using $F(s)$ as a computational resource, the right premise of PTI allows using the additional resource $E(s)$ in such a computation. 

Note that, in a classical context, identifying the two sorts of conjunction, there would be no difference between (\ref{nov14b}) and PTI. First of all, the (sub)conjunct $E(s)$ in the consequent of the right premise of PTI would be meaningless and hence could be deleted, as it is already present in the antecedent. Second, the conjunction of $E(s)$ and $F(s)$ could be thought of as one single formula of induction, and thus PTI would become simply (\ref{nov14b}). 

Our context is not classical though, and the difference between PTI and (\ref{nov14b}) is huge. First of all, we cannot think of ``the conjunction'' of $E(s)$ and $F(s)$ as a single formula of induction, for that ``conjunction'' is $\adc$ in the consequent of the right premise while $\mlc$ elsewhere. For simplicity, consider the case $E(s)=F(s)$. 
Also, let us ignore the technicality imposed by the presence of  ``$E(s)$'' in the consequent of the right premise of PTI. 
Then that premise would look like $F(s)\mlc F(s)\mli F(s\successor)\adc F(s\successor)$ which, taking into account that $X\adc X$ is equivalent to $X$, would be essentially the same as simply  $F(s)\mlc F(s)\mli F(s\successor)$. This is a much weaker premise than the premise $F(s)\mli F(s\successor)$ of (\ref{nov14b}). It signifies that  computing a single copy of $F(s\successor)$ requires computing two copies of $F(s)$. By back-propagating this effect, it would eventually mean that computing $F(s)$ requires computing an exponential number of copies of $F(\zero)$, even when $s$ is ``small enough'' such as $s\mleq \tau$. 

The above sort of an explosion is avoided in PTI due to the presence of $E(s)$ in the consequent of the right premise --- the ``technical detail'' that we have ignored so far. The reemergence of $E(s)$ in the consequent of that premise makes this resource ``recyclable''. Even though computing  $F(s\successor)$  still requires computing {\em both} $E(s)$ {\em and} $F(s)$, a new copy of $E(s)$ comes ``for free'' as a side-product of this computation, and hence can be directly passed to another, parallel computation of  $F(s\successor)$. Such and all other parallel computations would thus require a new copy of $F(s)$ but not a new copy of $E(s)$, as they get the required resource $E(s)$ from the neighboring computation. So, a group of $n$ parallel computations of $ F(s\successor)$ would require $n$ copies of $F(s)$ and only one copy of $E(s)$. This essentially cuts the demand on resources at each step    (for each $s$) by half, and the eventual number of copies of $E(\zero)\mlc F(\zero)$ to be computed will be of the order of $s$ rather than $2^s$. How this effect is exactly achieved will be clear after reading the following section.

\section{The soundness of $\arfour$}\label{sectsound}
This section is devoted to proving the soundness part of Theorem \ref{tt1}. It means showing that any $\arfour$-provable formula $X$ (identified with its standard interpretation $X^\dagger$) has a polynomial time solution, and that, furthermore, such a solution for $X$ can be effectively extracted from any $\arfour$-proof of $X$. 

We prove the above by  induction on the lengths of $\arfour$-proofs.  

Consider any $\arfour$-provable formula $X$. 

For the basis of induction, assume $X$ is an axiom of $\arfour$. Let us say that an elementary $\arfour$-formula $G$ is {\bf true}\label{itrue} iff, for any bounded valuation $e$, $e[G]$ is  true in the standard arithmetical sense, i.e., $\win{G^\dagger}{e}\emptyrun=\pp$. 

If $X$ is a logical axiom or a Peano axiom, then it is a true elementary formula and therefore is ``computed'' by a machine that makes no moves at all. The same holds for the case when $X$ is Axiom 13, remembering that, for any bounded valuation $e$, the size of $e(s)$ (whatever variable $s$) never exceeds $e(\bound)$.

If $X$ is $\ade x(x\equals\zero)$ (Axiom 8), then it is computed by a machine that makes the move $0$ and never makes any moves after that. 

If $X$ is $s\equals\zero\add s\notequals\zero$ (Axiom 9), then it is computed by a machine that reads the value $e(s)$ of $s$ from the valuation tape and, depending on whether that value is $0$ or not, makes the move $0$ or $1$, respectively.

If $X$ is $|s\successor|\mleq\bound\mli\ade x(x\equals s\successor)$ (Axiom 10),  it is computed by a machine that reads the value $e(s)$ of $s$ from the valuation tape, then finds (the binary numeral) $c$ with $c\equals e(s)\plus 1$, compares its size with $e(\bound)$ (the latter also read from the valuation tape) and, if $|c|\mleq e(\bound)$, makes $1.c$ as its only move in the game.  

Similarly, if $X$ is $|s0|\mleq\bound\mli\ade x(x\equals s0)$ (Axiom 11),  it is computed by a machine that reads the value $e(s)$ of $s$ from the valuation tape, then finds (the binary numeral) $c$ with $c=e(s)0$, compares its size with $e(\bound)$  and, if $|c|\mleq e(\bound)$, makes $1.c$ as its only move in the game.   

Finally, if $X$ is $\ade x(s\equals x0\add s\equals x1)$ (Axiom 12), it is computed by a machine that reads the value $e(s)$ of $s$ from the valuation tape, then finds the binary predecessor $c$ of $e(s)$, and makes the two moves $c$ and $0$ or $c$ and $1$, depending whether the last digit of $e(s)$ is $0$ or $1$, respectively. 

Needless to point out that, in all of the above cases, the machines that solve the axioms run in polynomial time. And, of course, such machines can be constructed effectively.

For the inductive step,  suppose $X$ is obtained from premises $X_1,\ldots,X_k$ by one of the four logical rules. By the induction hypothesis, we know how to (effectively) construct a polynomial time solution for each $X_i$. Then, by the results of \mbox{Section \ref{ss7}} on the uniform-constructive soundness of the four logical rules, we also know how to construct a polynomial time solution for $X$.

Finally, suppose $X$ is $s\mleq\tau\mli E(s)\mlc F(s)$, where $\tau$ is a $\bound$-term, and $X$ is obtained by PTI as follows: 
\[\frac{E(\zero )\mlc F(\zero )\hspace{20pt}   E(s)\mlc F(s)\mli E(s\successor)\adc \bigl(F(s\successor )\mlc E(s)\bigr) }{s\mleq\tau\mli E(s)\mlc F(s)}.\]

By the induction hypothesis, the following two problems have polynomial time solutions --- and, furthermore, we know how to construct such solutions:

\begin{equation}\label{sept200}
E(\zero )\mlc F(\zero );
\end{equation}
\begin{equation}\label{sept20}
 E(s)\mlc F(s)\mli E(s\successor)\adc \bigl(F(s \successor )\mlc E(s)\bigr). 
\end{equation}
Then the same holds for the following four problems:
\begin{equation}\label{sept200q}
E(\zero );
\end{equation}
\begin{equation}\label{sept200w}
F(\zero );
\end{equation}
\begin{equation}\label{sept20e}
 E(s)\mlc F(s)\mli  E(s \successor ); 
\end{equation}
\begin{equation}\label{sept20r}
 E(s)\mlc F(s)\mli  E(s)\mlc  F(s\successor ) . 
\end{equation}

For (\ref{sept200q}) and (\ref{sept200w}), this is so because $\clfour\vdash P_1\mlc P_2\mli P_i$ ($i=1,2$), whence $\clthree$ proves both $E(\zero)\mlc F(\zero)$ $\mli E(\zero)$ and $E(\zero)\mlc F(\zero)\mli F(\zero)$,  whence --- by the uniform-constructive soundness of $\clthree$ --- we know how to construct polynomial time solutions for these two problems, whence --- by the polynomial time solvability of (\ref{sept200}) and the closure of this property (in the strong sense of Theorem \ref{april22}) under Modus Ponens --- we also know how to construct polynomial time solutions for $E(\zero)$ and $F(\zero)$. With (\ref{sept20}) instead of (\ref{sept200}),  the arguments  for (\ref{sept20e}) and (\ref{sept20r}) are similar, the first one relying  on the fact that $\clfour$ proves $(P_1\mli P_2\adc Q)\mli (P_1\mli P_2)$, and the second one relying on the fact that  $\clfour$ proves $\bigl(P_1\mli P_2\adc (Q_1\mlc Q_2)\bigr)\mli (P_1\mli Q_2\mlc Q_1)$.

Throughout the rest of this proof, assume some arbitrary bounded valuation $e$ to be fixed. Correspondingly, when we write $\bound$ or $\tau$, they are to be understood as $e(\bound)$ or $e(\tau)$. As always, saying ``polynomial'' means ``polynomial in $\bound$''.

For a formula $G$ and a positive integer $n$, we will be using the abbreviation \[\pst^n G\label{ipst}\] for the $\mlc$-conjunction  $G\mlc\ldots\mlc G$ of $n$ copies of $G$. If here $n=1$, $\pst^n G$ simply means $G$.\vspace{7pt} 

{\bf Claim 1}.{\em  
For any integer  $k\in\{1,\ldots,\tau\}$,    the following problem has a polynomial time solution which, in turn, can be constructed in polynomial time:}
\begin{equation}\label{sep25}E(s)\mlc \pst^{k\plus 1}F(s)\mli E(s\successor )\mlc \pst^{k}F(s\successor ).\end{equation}

\begin{subproof}  In this proof and later,  we use the term ``{\bf synchronizing}''\label{imatching} to mean applying copycat between two (sub)games of the form $A$ and $\gneg A$. This means copying one player's moves made in $A$ as the other player's moves in $\gneg A$, and vice versa. The effect achieved this way is that the games to which $A$ and $\gneg A$  eventually evolve (the final positions hit by them, that is) will be of the form $A'$ and $\gneg A'$, that is, one will remain the negation of the other, so that one will be won by a given player iff the other is lost by the same player. We already saw an application of this idea/technique in the proof of Theorem \ref{april22}. Partly for this reason and partly because now we are dealing with a more complicated case, our present proof 
 will be given in less detail than the proof of Theorem \ref{april22} was. 

Here is a solution/strategy for (\ref{sep25}). While playing the real play of (\ref{sep25}) on  valuation $e$, also play, in parallel, one imaginary copy of (\ref{sept20e}) and $k$ imaginary copies of (\ref{sept20r}) on the same valuation $e$, using the strategies for (\ref{sept20e}) and (\ref{sept20r}) whose existence we already know. In  this mixture of the real and imaginary plays, do the following:
\begin{itemize}
\item Synchronize the $F(s)$ of the antecedent of each $i$th copy of (\ref{sept20r}) with the $i$th conjunct of the  $\pst^{k\plus 1}F(s)$ part of the antecedent of (\ref{sep25}). 
\item Synchronize the $E(s)$ of the antecedent of the first copy of (\ref{sept20r}) with the $E(s)$ of the antecedent of (\ref{sep25}). 
\item Synchronize the $E(s)$ of the antecedent of each  copy $\#(i\plus 1)$ of (\ref{sept20r}) with the $E(s)$ of the consequent of copy $\#i$ of (\ref{sept20r}).  
\item Synchronize the $E(s)$ of the antecedent of (the single copy of)  (\ref{sept20e}) with the $E(s)$ of the consequent of copy $\#k$ of (\ref{sept20r}).
\item Synchronize the $F(s)$ of the antecedent of   (\ref{sept20e}) with the last conjunct of the   $\pst^{k\plus 1}F(s)$ part of the antecedent of (\ref{sep25}).
\item Synchronize the $E(s\successor )$ of the consequent  of (\ref{sept20e}) with the $E(s\successor )$ of the consequent of (\ref{sep25}). 
\item Synchronize the $F(s\successor )$ of  the consequent  of each copy $\#i$ of (\ref{sept20r}) with the $i$th conjunct of the $\pst^{k}F(s\successor )$ part of the consequent of (\ref{sep25}).
\end{itemize}

Below is an illustration of such synchronization arrangements --- indicated by arcs --- for the case $k=3$:

\begin{center}
\begin{picture}(336,120)

\put(0,110){$(\ref{sep25})$:}
\put(50,110){$E(s)\mlc F(s)\mlc F(s)\mlc F(s)\mlc F(s)\mli E(s\successor )\mlc F(s\successor )\mlc F(s\successor)\mlc F(s\successor)$}
\put(0,85){$(\ref{sept20r})_1$:}
\put(50,85){$ E(s)\mlc F(s)\mli E(s)\mlc   F(s\successor )$}
\put(0,60){$(\ref{sept20r})_2$:}
\put(50,60){$ E(s)\mlc F(s)\mli E(s)\mlc   F(s\successor )$}
\put(0,35){$(\ref{sept20r})_3$:}
\put(50,35){$ E(s)\mlc F(s)\mli E(s)\mlc   F(s\successor )$}
\put(0,10){$(\ref{sept20e})$:}
\put(50,10){$ E(s)\mlc F(s)\mli  E(s\successor )$}
\put(60,105){\line(0,-1){10}}
\put(90,105){\line(0,-1){10}}
\put(120,105){\line(-4,-5){28}}
\put(151,105){\line(-1,-1){60}}
\put(180,105){\line(-1,-1){84}}
\put(124,30){\line(-6,-1){60}}
\put(124,55){\line(-6,-1){60}}
\put(124,80){\line(-6,-1){60}}
\put(142,12){\line(1,0){71}}
\put(213,12){\line(0,1){94}}
\put(172,37){\line(1,0){144}}
\put(316,37){\line(0,1){68}}
\put(172,62){\line(1,0){110}}
\put(282,62){\line(0,1){44}}
\put(172,87){\line(1,0){74}}
\put(246,87){\line(0,1){19}}

\end{picture}
\end{center}

Of course, the strategy that we have just described  can be constructed effectively and, in fact, in polynomial time, from the strategies for (\ref{sept20e}) and (\ref{sept20r}). Furthermore, since the latter run in polynomial time, obviously so does our present one. It is left to the reader to verify that our strategy indeed wins (\ref{sep25}).
\end{subproof}

Now, the sought polynomial time solution for 
\begin{equation}\label{oct7}
 s\mleq \tau\mli E(s)\mlc F(s)  
\end{equation}
on valuation $e$  will go like this. Read the value $d=e(s)$ of $s$ from the valuation tape. Also read the value  of $\bound$    and, using it, compute the value $c$ of $\tau$. Since $\tau$ is a $(\zero,\successor,\plus,\times)$-combination of $\bound$, computing $c$ only takes a polynomial amount of steps.  If $d\mgreater c$, do nothing --- you are the winner (again, comparing $d$ with $c$, of course, takes only a polynomial amount of steps). Otherwise,  using the strategy from Claim 1, for each $a\in\{0,\ldots,d-1\}$, play (a single copy of) the imaginary game $G_a$ on valuation $e$, defined by 
\[ G_a \ = \  E(a)\mlc \pst^{d-a\plus 1}F(a)\mli E(a\successor )\mlc \pst^{d-a}F(a\successor ).\]
Namely, the effect of playing $G_a$ on valuation $e$ is achieved by playing  $E(s)\mlc \pst^{d-a\plus 1}F(s)\mli E(s\successor )\mlc \pst^{d-a}F(s\successor )$ on the valuation $e'$ which sends $s$ to $a$ and agrees with $e$ on all other variables. 
In addition, using the strategy for (\ref{sept200q}), play a single imaginary copy of $E(\zero )$ on $e$, and, using the strategy for (\ref{sept200w}), play $d\plus 1$  imaginary copies of $F(\zero)$ on $e$. In this mixture of imaginary plays and the real play of (\ref{oct7}), do the following:
\begin{itemize}
\item Synchronize the above $E(\zero)$ and $F(\zero)$s with the corresponding conjuncts   of the antecedent of $G_0$. 
\item Synchronize the antecedent of each $G_{i\plus 1}$ with the consequent of $G_i$. 
\item Synchronize the consequent of $G_{d-1}$ with the consequent of (\ref{oct7}). 
\end{itemize}
Below is an illustration of these synchronization arrangements for the case $d=11$ (decimal $3$): 

\begin{center}
\begin{picture}(333,209)

\put(50,189){$11\mleq \tau$}
\put(74,189){$\mli$}
\put(90,189){$ \underbrace{E(11)\mlc F(11)}  $} 
\put(50,148){$ \underbrace{E(10)\mlc  F(10)\mlc F(10)} $}
\put(152,148){$\mli$}
\put(170,148){$  \overbrace{E(11 )\mlc  F(11 )} $}
\put(50,101){$ \underbrace{E(1 )\mlc  F(1 )\mlc F(1 )\mlc F(1 )} $}
\put(169,101){$\mli$}
\put(186,101){$\overbrace{ E(10 )  \mlc F(10)\mlc F(10 ) }$}
\put(50,50){$E(0 )\mlc  F(0 )\mlc F(0 )\mlc F(0 ) \mlc F(0 )$}
\put(195,50){$\mli$}
\put(209,50){$ \overbrace{E(1  )  \mlc F(1 )\mlc F(1  )\mlc F(1  )} $}
\put(50,10){$ E(\zero)$}
\put(80,10){$ F(\zero)$}
\put(110,10){$ F(\zero)$}
\put(140,10){$ F(\zero)$}
\put(170,10){$ F(\zero)$}
\put(0,189){$(\ref{oct7})$:}
\put(0,148){$G_{10}$:}
\put(0,101){$G_{1 }$:}
\put(0,50){$G_{0 }$:}
\put(60,22){\line(0,1){21}}
\put(90,22){\line(0,1){21}} 
\put(120,22){\line(0,1){21}}
\put(150,22){\line(0,1){21}}
\put(180,22){\line(0,1){21}}
\put(265,65){\line(-6,1){158}}
\put(234,116){\line(-6,1){135}}
\put(201,163){\line(-5,1){80}}
\end{picture}
\end{center}

Again, with some thought, one can see that our strategy --- which, of course, can be constructed effectively --- runs in polynomial time, and it indeed wins (\ref{oct7}), as desired. 

\section{Some admissible logical rules of $\arfour$} 
When we say that a given rule is {\bf admissible} in $\arfour$, we mean that, whenever all premises of any given instance of the rule are provable in $\arfour$, so is the conclusion.

 This section is devoted to observing the admissibility of a number of rules. From our admissibility proofs it can be seen that these rules are admissible not only in $\arfour$ but also in any $\clthree$-based applied theory in general. This is the reason why these rules can be called ``logical''.  Such rules  can and will be used as shortcuts in $\arfour$-proofs. Many of such rules can be further strengthened, but in this paper --- for the sake of simplicity and at the expense of (here) unnecessary generality --- we present them only in forms  that (and as much as they) will be actually used in our further treatment.  

In the formulations of some of the rules we use the expression
\[ E^{\vee} [F].\]
It means the same as the earlier-used $E[F]$, i.e., a formula $E$ with a fixed positive surface occurrence of a subformula $F$; only, in $E^{\vee} [F]$, the additional (to being a positive surface occurrence) condition on the occurrence of $F$ is that this occurrence is not in the scope of any operator other than 
$\mld$.    

\subsection{$\clfour$-Instantiation} 
\[\frac{}{\ \ F\ \ },\]
where $F$ is any $\arfour$-formula which is an instance of some $\clfour$-provable formula $E$.  

Unlike all other rules given in the present section, this one, as we see, takes no premises. It is a ``rule'' that simply allows us to jump to a formula $F$ as long as it is an instance of a $\clfour$-provable formula. 

\begin{fact}\label{Instclosure}
$\clfour$-Instantiation is admissible in $\arfour$.
\end{fact}

\begin{proof} Assume a $\arfour$-formula $F$ is an instance of some $\clfour$-provable formula. Then, by Theorem \ref{ccs}, $\clthree\vdash F$.  $\clthree$ is an analytic system,   in the sense that it never introduces into premises any function or predicate letters  that are not present in the conclusion. So, all formulas  involved in the $\clthree$-proof of $F$ will be $\arfour$-formulas. This includes the axioms used in the proof. But such axioms are also axioms of $\arfour$. And $\arfour$  has all inference rules that $\clthree$ does. Hence, the above $\clthree$-proof of $F$ will be a $\arfour$-proof of $F$ as well. 
\end{proof}

\subsection{Transitivity (TR)}\label{itr}
\[\frac{E_1\mli F\hspace{30pt}F\mli E_2}{E_1\mli E_2}\]

\begin{fact}\label{Transclosure}
Transitivity is admissible in $\arfour$.
\end{fact}

\begin{proof} Assume 
\begin{equation}\label{nov6a}
\mbox{\em $\arfour\vdash E_1\mli F$ \ and \  $\arfour\vdash F\mli E_2$}.
\end{equation}
$\clfour$ proves $(P_1\mli Q)\mlc (Q\mli P_2)\mli (P_1\mli P_2)$ (it is derived from the classical tautology $(p_1\mli q)\mlc (q\mli p_2)\mli$ $(p_1\mli p_2)$ by Match applied three times). Hence, by $\clfour$-Instantiation, 
\begin{equation}\label{nov6b}
 \arfour\vdash (E_1\mli F)\mlc (F\mli E_2)\mli (E_1\mli E_2).
\end{equation}
Now, from (\ref{nov6a}) and (\ref{nov6b}), by Modus Ponens, we get the desired $\arfour\vdash E_1\mli E_2$.
\end{proof}

\subsection{$\ada$-Elimination}\label{iadael}
\[\frac{\ \ada x F(x)\ }{F(s)},\]
where $x$ is any variable, $F(x)$ is any formula,  $s$ is any variable not bound in the premise, and $F(s)$ is the result of replacing all free occurrences of $x$ by $s$ in $F(x)$.

\begin{fact}\label{adaelclosure}
$\ada$-Elimination is admissible in $\arfour$.
\end{fact}

\begin{proof} Assume $\arfour\vdash \ada x F(x)$. $p(s)\mli p(s)$ is classically valid and hence, by Match, $\clfour\vdash P(s)\mli P(s)$. From here, by $\ade$-Choose, $\clfour\vdash \ada xP(x)\mli P(s)$. Then, by $\clfour$-Instantiation,  $\arfour\vdash  \ada x F(x)\mli F(s)$. Now, by Modus Ponens, $\arfour\vdash   F(s)$. 
\end{proof}

\subsection{$\add$-Elimination}\label{iaddel}
\[\frac{F_1\add\ldots\add F_n\hspace{20pt}F_1\mli E\hspace{20pt}\ldots\hspace{20pt}F_n\mli E}{E}\]

\begin{fact}\label{addelclosure}
$\add$-Elimination is admissible in $\arfour$.
\end{fact}

\begin{proof} Assume $\arfour$ proves all premises. For each $i\in\{1,\ldots,n\}$, the formula
\[p_i\mlc (\tlg\mli \twg)\mlc\ldots\mlc (\tlg\mli \twg)\mlc (p_i\mli q)\mlc (\tlg\mli \twg) \mlc\ldots \mlc (\tlg\mli\twg )\mli q\]
is a classical tautology and hence an axiom of $\clfour$. By Wait from the above, we   have
\[\clfour\vdash p_i\mlc (P_1\mli Q)\mlc\ldots \mlc (P_{i-1}\mli Q) \mlc (p_i\mli q)\mlc  (P_{i+1}\mli Q) \mlc \ldots \mlc (P_n\mli Q)\mli q.\]
Now, by Match applied twice, we get
\[\clfour\vdash P_i\mlc (P_1\mli Q)\mlc\ldots\mlc (P_n\mli Q)\mli Q.\]
We also have 
\[\clfour\vdash \tlg \mlc (\tlg\mli \twg)\mlc\ldots\mlc (\tlg\mli \twg)\mli \tlg\]
because the above formula is a classical tautology. From the last two facts, by Wait, we find
\[\clfour \vdash (P_1\add\ldots\add P_n)\mlc (P_1\mli Q)\mlc\ldots\mlc (P_n\mli Q)\mli Q\]
and hence, by $\clfour$-Instantiation,  
\[\arfour\vdash (F_1\add\ldots\add F_n)\mlc (F_1\mli E)\mlc\ldots\mlc (F_n\mli E)\mli E.\]
As all of the conjuncts of the antecedent of the above formula are $\arfour$-provable by our original assumption, Modus Ponens yields $\arfour\vdash E$.
\end{proof}

As an aside, one could show that the present rule with $\mld$ instead of $\add$, while admissible in classical logic, is not admissible in $\arfour$ or $\clthree$-based applied theories in general.

\subsection{Weakening}\label{iwk}
\[\frac{E^{\vee}[G_1\mld\ldots\mld G_m\mld  H_1\mld\ldots\mld H_n]}{E^{\vee} [G_1\mld\ldots\mld G_m\mld F\mld H_1\mld\ldots\mld H_n]},\]
where $m,n\geq 0$ and $m+n\not=0$.

\begin{fact}\label{Weakeningclosure}
Weakening is admissible in $\arfour$.
\end{fact}

\begin{proof} Assume $\arfour$ proves the premise. It is not hard to see that $\mbox{\em Premise}\mli\mbox{\em Conclusion}$  can be obtained by $\clfour$-Instantiation, so it is also provable in $\arfour$. Hence, by Modus Ponens, $\arfour$ proves the conclusion.
\end{proof}

 \subsection{$\adc$-Introduction}\label{iadcintro}
\[\frac{E^{\vee}[F_1]\hspace{30pt}\ldots\hspace{30pt}E^{\vee}[F_n]}{E^{\vee} [F_1\adc\ldots\adc F_n]}\]

\begin{fact}\label{adcintroclosure}
$\adc$-Introduction is admissible in $\arfour$.
\end{fact}

\begin{proof} Assume $\arfour$ proves each of the $n$ premises.   Let $G$ be the $\mld$-disjunction of all subformulas of $E^{\vee}[F_1\adc\ldots\adc F_n]$,   other than the indicated occurrence of $F_1\adc\ldots\adc F_n$, that do not occur in the scope of any operators other than $\mld$ and whose main operator (if nonatomic) is not $\mld$. 
We want to first verify the rather expected fact that $\arfour\vdash F_i \mld G$ for each $i$ (expected, because, modulo the  associativity of $\mld$, the formulas $E^{\vee}[F_i]$ and $ F_i\mld G$ are the same). Indeed, $E^{\vee}[F_i]\mli F_i\mld G$ can be easily seen to be obtainable by $\clfour$-Instantiation. Then, $ F_i\mld G$ follows by Modus Ponens. In a similar manner one can show that whenever $\arfour\vdash (F_1\adc\ldots\adc F_n)\mld G$, we also have $\arfour \vdash E^{\vee} [F_1\adc\ldots\adc F_n]$.
So, in order to complete our proof of Fact \ref{adcintroclosure},  it would suffice to show that 
\begin{equation}\label{nov6f}
\arfour\vdash (F_1\adc\ldots\adc F_n)\mld G.
\end{equation} 

From $\arfour\vdash F_1\mld G$, \ldots, $\arfour\vdash F_1\mld G$ and the obvious fact that $\arfour\vdash\twg$, by Wait, we get 
\begin{equation}\label{nov6g}
\arfour\vdash  (F_1\mld G)\adc\ldots\adc (F_n\mld G).
\end{equation}

Next,  $p \mld q \mli   p \mld q $ is an axiom of $\clfour$. From it, by Match applied twice, we get  
$ \clfour\vdash  P_i \mld Q \mli   P_i \mld Q $ (any $i\in\{1,\ldots,n\}$). Now, by $\add$-Choose, we get  \[ \clfour\vdash   ( P_1 \mld Q  ) \adc\ldots\adc  ( P_n \mld Q  )  \mli   P_i\mld Q .\] From here and from (the obvious) $\clfour\vdash \twg\mli\twg\mld\tlg$, by Wait, we get \[ \clfour\vdash ( P_1 \mld Q  ) \adc\ldots\adc  ( P_n \mld Q  ) \mli  (P_1\adc\ldots\adc P_n)\mld Q .\] The above, by $\clfour$-Instantiation, yields 
\begin{equation}\label{nov6h}
\arfour\vdash ( F_1 \mld G  ) \adc\ldots\adc  ( F_n \mld G  ) \mli  (F_1\adc\ldots\adc F_n)\mld G .
\end{equation}  
Now, the desired (\ref{nov6f}) follows  from (\ref{nov6g}) and (\ref{nov6h}) by Modus Ponens.  
\end{proof}

It is worth pointing out  that the present rule with $\mlc$ instead of $\adc$, while admissible in classical logic, is not admissible in $\arfour$ or $\clthree$-based applied theories in general.

\subsection{$\ada$-Introduction}\label{iadaintro}
\[\frac{E^{\vee}[F(s)]}{E^{\vee} [\ada xF(x)]},\]
where $x$ is any (non-$\bound$) variable, $F(x)$ is any formula, $s$ is any non-$\bound$ variable not occurring in the conclusion, and $F(s)$ is the result of replacing all free occurrences of $x$ by $s$ in $F(x)$.

\begin{fact}\label{adaintroclosure}
$\ada$-Introduction is admissible in $\arfour$.
\end{fact}

\begin{proof} Assume $\arfour\vdash E^{\vee}[F(s)]$. Let $G$ be the $\mld$-disjunction of all subformulas of $E^{\vee} [\ada xF(x)]$, other than the indicated occurrence of $\ada xF(x)$, that do not occur in the scope of any operators other than $\mld$ and whose main operator (if nonatomic) is not $\mld$.
As in the previous subsection, we can easily find that  $\arfour\vdash F(s)\mld G$, and that  whenever  $\arfour\vdash \ada xF(x)\mld G$, we also have $\arfour \vdash E^{\vee} [\ada xF(x)]$. So, in order to complete our proof of Fact \ref{adaintroclosure},  it would suffice to show that

\begin{equation}\label{nov6c}
\arfour\vdash \ada xF(x)\mld G.
\end{equation} 

From $\arfour\vdash F(s)\mld G$ and the obvious fact that $\arfour\vdash\twg$, by Wait, we get 
\begin{equation}\label{nov6d}
\arfour\vdash \ada y\bigl(F(y)\mld G\bigr),
\end{equation} 
where $y$ is a ``fresh'' variable --- a variable not occurring in $F(s)\mld G$. 

Next,  $  p(t)\mld q \mli   p(t)\mld q $ is  an axiom of $\clfour$. From it, by Match applied twice, we find that $\clfour$ proves   
$P(t)\mld Q \mli   P(t)\mld Q $. Now, by $\ade$-Choose, we get  $ \clfour\vdash \ada y\bigl( P(y)\mld Q\bigr) \mli   P(t)\mld Q $. From here and from (the obvious) $\clfour\vdash \twg\mli\twg\mld\tlg$, by Wait, we get $ \clfour\vdash \ada y\bigl( P(y)\mld Q\bigr) \mli   \ada xP(x)\mld Q $. This, by $\clfour$-Instantiation, yields
\begin{equation}\label{nov6e}
\arfour\vdash \ada y\bigl(F(y)\mld G\bigr)\mli \ada xF(x)\mld G.
\end{equation}  
Now, the desired (\ref{nov6c}) follows  from (\ref{nov6d}) and (\ref{nov6e}) by Modus Ponens.  
\end{proof}

We are again pointing out that the present rule with $\cla$ instead of $\ada$, while admissible in classical logic, is not admissible in $\arfour$ or $\clthree$-based applied theories in general.

\section{Formal versus informal arguments in $\arfour$}\label{sss16}
 We have already seen a couple of nontrivial formal $\arfour$-proofs, and will see more later. However, continuing forever in this style will be hardly possible. Little by little, we will need to start trusting and relying on informal arguments in the style of the argument found at the beginning of the proof of Lemma \ref{nov8}, or the arguments that we employed when discussing the PTI rule in Section \ref{ss14}. Just as in {\bf PA}, formal proofs in $\arfour$ tend to be long, and generating them in every case can be an arduous job. The practice of dealing with informal proofs or descriptions instead of detailed formal ones is  familiar not only from the metatheory of $\pa$ or similar systems. The same practice is adopted, say, when dealing with Turing machines, where full  transition diagrams are typically replaced by high-level informal descriptions, relying on the reader's understanding that, if necessary, every such description can be turned into a real Turing machine.

In the nearest few sections we will continue generating formal proofs, often accompanied with underlying informal arguments  to get used to such arguments and  see that they are always translatable into formal ones. As we advance, however, our reliance on informal arguments and the degree of our ``laziness'' will gradually increase, and in later sections we may stop producing formal proofs altogether. 

The informal language and methods of reasoning induced by computability logic and clarithmetic or ptarithmetic in particular, are in the   painful initial   process of forming and, at this point, can be characterized as ``experimental''. They cannot be concisely or fully {\em explained}, but rather they should be {\em learned} through experience and practicing, not unlike the way one learns a foreign language. A reader who initially does not find some of our informal $\arfour$-arguments very clear or helpful, should not feel disappointed. Both the readers and the author should simply keep trying their best. Greater fluency and better understanding will come gradually and inevitably. 

At this point we only want to make one general remark on the informal $\arfour$-arguments that will be employed. Those arguments will often proceed in terms of game-playing and problem-solving instead of theorem-proving, or will be some kind of a mixture of these two. That is, a way to show how to prove a formula $F$ will often be to show how to win/solve the game/problem $F$. The legitimacy of this approach is related to the fact that the logic $\clthree$ underlying $\arfour$ is a logic of problem-solving and, as such, is complete (Theorem \ref{main}). That is, whenever a problem $F$ can be solved in a way that relies merely on the logical structure of $F$ --- and perhaps also those of some axioms of $\arfour$ --- then we have a guarantee that $F$ can as well be proven. Basic problem-solving steps are very directly simulated (translated through) the rules of $\clthree$ or some derivative rules in the style of the rules of the previous section, with those rules seen bottom-up (in the ``from conclusion to premises'' direction). For instance, a step such as ``choose the $i$th disjunct in the subformula/subgame $F_1\add\ldots\add F_n$''  translates as a  bottom-up application of $\add$-Choose which replaces  $F_1\add\ldots\add F_n$ by $F_i$; a step such as ``specify $x$ as $s$ in $\ade xF(x)$'' translates as a  bottom-up application of  $\ade$-Choose; a step such as ``wait till the environment specifies a value $s$ for $x$ in $\ada xF(x)$'' translates as a  bottom-up application of $\ada$-Introduction; etc. Correspondingly, an informally described winning/solution strategy for $F$   can usually be seen as a relaxed, bottom-up description of a formal proof of $F$. 

\section{Some admissible induction rules of $\arfour$}

The present section introduces a few new admissible rules of induction. These rules are weaker than PTI, but are still useful in that, in many cases, they may offer greater convenience than PTI does.  
\subsection{WPTI}\label{inpti} Here we reproduce rule (\ref{nov14b}) discussed in Section \ref{ss14}, and baptize it as ``{\bf WPTI}'' (``W'' for ``Weak''):
\[\frac{F(\zero)\hspace{25pt}   F(s)\mli F(s\successor )}{s\mleq \tau\mli F(s)},\]
where $s$ is any non-$\bound$ variable, $F(s)$ is any formula, and $\tau$ is any $\bound$-term.

\begin{theorem}\label{nov9a}
WPTI is admissible in $\arfour$.
\end{theorem}

\begin{idea} WPTI is essentially nothing but PTI with $\twg$ in the role of  $E(s)$.
\end{idea}

\begin{proof}    
Assume $s$, $F(s)$, $\tau$ are as stipulated in the  rule, and $\arfour$ proves both $F(\zero)$ and $F(s)\mli F(s\successor )$.  
The following formula matches the $\clfour$-provable  $(  P\mli Q)\mli \bigl( \twg\mlc P\mli \twg\adc(Q\mlc\twg)\bigr)$  and hence, by $\clfour$-Instantiation,  is  provable in $\arfour$:  
\begin{equation}\label{nov9ab}
\bigl( F(s)\mli F(s\successor )\bigr)\mli \Bigl( \twg\mlc F(s)\mli \twg\adc\bigl(F(s\successor)\mlc\twg \bigr)\Bigr).\end{equation}
By Modus Ponens from $F(s)\mli F(s\successor )$ and (\ref{nov9ab}), we find that $\arfour$ proves 
\begin{equation}\label{nov9ac}
 \twg\mlc F(s)\mli \twg\adc\bigl(F(s\successor )\mlc\twg\bigr).\end{equation}

Similarly, $F(\zero)\mli \twg\mlc F(\zero)$ is obviously provable in $\arfour$ by $\clfour$-Instantiation. Modus-ponensing this with our assumption $\arfour\vdash F(\zero)$ yields $\arfour\vdash \twg\mlc F(\zero)$. From here and (\ref{nov9ac}), by PTI with $\twg$ in the role of $E(s)$, we find that $\arfour$ proves $s\mleq\tau\mli \twg\mlc F(s)$. But $\arfour$ also proves  $ \twg\mlc F(s)\mli F(s)$ because this is an instance of the {\bf CL4}-provable $\twg\mlc P\mli P$. Hence, by Transitivity, $\arfour\vdash s\mleq\tau\mli F(s)$, as desired.  
\end{proof}

\subsection{BSI}\label{ibsi} What we  call {\bf BSI} ({\bf B}inary-{\bf S}uccessor-based {\bf I}nduction) is the following rule, where $s$ is any non-$\bound$ variable and $F(s)$ is any formula:
\[\frac{F(\zero)\hspace{43pt} F(s)\ \mli \ F(s0)\adc F(s1)}{F(s)}.\]

\begin{theorem}\label{nov9b}
BSI is admissible in $\arfour$.
\end{theorem}

\begin{idea} We manage to  reduce BSI to WPTI with $\ada x\bigl(|x|\mleq s\mli F(x)\bigr)$ in the role of $F(s)$ of the latter.
\end{idea}

\begin{proof} Assume $s$, $F(s)$ are as stipulated in the  rule, 
\begin{equation}\label{nov9bb}
\arfour\vdash F(\zero)\end{equation}
and 
\begin{equation}\label{nov9ba}
 \arfour\vdash  F(s)\mli F(s0)\adc F(s1) .\end{equation}

Let us observe right now that, by $\ada$-Introduction,  (\ref{nov9ba}) immediately implies 
\begin{equation}\label{nov9baa}
 \arfour\vdash  \ada x\bigl(F(x)\mli F(x0)\adc F(x1)\bigr) .\end{equation}

The goal is to verify that $\arfour\vdash F(s)$. 

An outline of our strategy for achieving this goal is that we take the formula $\ada x\bigl(|x|\mleq t\mli F(x)\bigr)$ --- let us denote it by $G(t)$ --- and  show that both $G(\zero)$ and $ G(t)\mli G(t\successor)$ are provable. This, by (the already shown to be admissible) WPTI, allows us to immediately conclude that $t\mleq\bound\mli G(t)$ is also provable, which, in turn, implies that so is $\ada y \bigl(y\mleq\bound\mli G(y)\bigr)$, and hence $\bound\mleq\bound \mli G(\bound)$, and hence $G(\bound)$. \ $G(\bound)$ asserts that, for any ($\ada$) given $x$ whose length does not exceed $\bound$, we can solve $F(x)$. But the length of no $x$ that we consider exceeds $\bound$, so that $G(\bound)$, in fact, simply says that (we can solve) $F(x)$. Formalizing this argument in $\arfour$ and taking $s$ for $x$ yields the desired conclusion $\arfour\vdash F(s)$.  

In following the above outline, we first claim that $\arfour\vdash G(\zero)$, i.e., 
\begin{equation}\label{nov10a}
 \arfour\vdash \ada x\bigl(|x|\mleq \zero\mli F(x)\bigr).\end{equation}
An informal argument here is that, since no constant is of length $0$,  $|x|\mleq \zero$ is false, and hence the problem $|x|\mleq \zero\mli F(x)$ is automatically ``solved'' (i.e., won without any moves by $\pp$) no matter what $F(x)$ is. Formally, $\pa$ and hence $\arfour$ proves the true fact $\gneg |v|\mleq \zero$. $\arfour$ also proves $\gneg |v|\mleq \zero \mli \bigl(|v|\mleq \zero\mli F(v)\bigr)$, as this is an instance of the {\bf CL4}-provable $\gneg p\mli(p\mli Q)$. Then, by Modus Ponens, $\arfour \vdash |v|\mleq \zero\mli F(v)$, whence, by $\ada$-Introduction,  $\arfour \vdash \ada x\bigl(|x|\mleq \zero\mli F(x)\bigr)$, as desired.   

Our next goal is to show that $ \arfour\vdash  G(t)\mli G(t\successor )$, i.e., 
\begin{equation}\label{nov10b}
 \arfour\vdash   \ada x\bigl(|x|\mleq t\mli F(x)\bigr)\mli \ada x\bigl(|x|\mleq t\successor\mli F(x)\bigr).\end{equation}
This can be done by showing the $\arfour$-provability of   
\begin{equation}\label{nov10c}
\ada x\bigl(|x|\mleq t\mli F(x)\bigr)\mli  |v|\mleq t\successor \mli F(v) ,\end{equation}
from which (\ref{nov10b}) follows by $\ada$-Introduction. 

Let us  first try to justify (\ref{nov10c}) informally. Consider any $t$, $v$ with  $|v|\mleq t\successor $, and also assume that (a single copy of) the resource $\ada x\bigl(|x|\mleq t\mli F(x)\bigr)$ is at our disposal. The goal is to establish $F(v)$.  
$F(0)$ is immediate by (\ref{nov9bb}). In turn, by (\ref{nov9ba}), $F(0)$ easily implies $F(1)$. Thus, we are done for the case  $v\mleq 1$.   Suppose now $v\mgreater 1$.  Then (unlike the case  $v\mleq 1$), remembering that $|v|\mleq t\successor $,  $v$ must have a binary predecessor $r$ with $|r|\mleq t$. By Axiom 12, we can actually find such an $r$ and, furthermore,  tell whether $v\equals r0$ or  $v\equals r1$. Specifying $x$ as $r$ in the   antecedent of (\ref{nov10c}), we can bring it down to the resource $|r|\mleq t\mli F(r)$ and --- as we already know that $|r|\mleq t$ --- essentially to the resource $F(r)$.  By (\ref{nov9baa}), the resource $\ada x\bigl(F(x)\mli F(x0)\adc F(x1)\bigr)$ and hence $ F(r)\mli F(r0)\adc F(r1)$ is also available. This is a resource that consumes $F(r)$ and generates $F(r0)\adc F(r1)$. Feeding to its consumption needs\footnote{Do you see or feel a possible application of MP, or TR, or Match behind this informal phrase?} our earlier-obtained $F(r)$, we thus get the resource $F(r0)\adc F(r1)$.  As noted earlier, we know precisely whether $v\equals r0$ or $v\equals r1$. So, by choosing the corresponding $\adc$-conjunct, we can further turn $F(r0)\adc F(r1)$ into the sought $F(v)$.  

Strictly verifying  (\ref{nov10c}) is quite some task, and we break in into several subtasks/subgoals. 

Our first subgoal is to show that $\arfour$ proves the following:

\begin{equation}\label{nov11a}
v\equals \zero\add v\equals \zero\successor \add v\mgreater \zero\successor ,\end{equation}
implying our ability to (efficiently) tell whether $v$ is $0$, $1$, or greater than $1$. For simplicity considerations, in our earlier informal justification of (\ref{nov10c}), we, in a sense, cheated by taking this ability for granted --- or, rather, by not really mentioning the need for it at all. Some additional evidence of such ``cheating'' can be discovered after reading the later parts of the present proof as well. 

Informally, an argument for (\ref{nov11a}) goes like this. Due to Axiom 12, we can find the binary predecessor $r$ of $v$. Moreover, due to the same axiom, we can tell whether $v\equals r0$ or $v\equals r1$. Using Axiom 9, we can further tell whether $r\equals 0$ or $r\notequals 0$. So, we will know precisely which of the four combinations $v\equals r0\mlc r\equals 0$, $v\equals r1\mlc r\equals 0$, $v\equals r0\mlc r\notequals 0$, $v\equals r1\mlc r\notequals 0$ is the case. From \pa, we also know that in the first case we have $v\equals 0$, in the second case we have $v\equals 1$, and in the third and the fourth cases we have $v\mgreater 1$. So, one of $v\equals 0$, $v\equals 1$,  $v\mgreater 1$ will be true and, moreover, we will be able to actually tell which one is true. 
   
Below is a full formalization of this argument:\vspace{7pt}

\noindent 1. $\begin{array}{l}
s\equals \zero\add s\notequals \zero
\end{array}$  \ \ Axiom 9  \vspace{3pt}

\noindent 2. $\begin{array}{l}
\ada x (x\equals \zero\add x\notequals \zero)
\end{array}$  \ \ $\ada$-Introduction: 1  \vspace{3pt}

\noindent 3. $\begin{array}{l}
\ade x(v\equals x0\add v\equals x1)
\end{array}$  \ \ Axiom 12  \vspace{3pt}

\noindent 4. $\begin{array}{l}
r\equals \zero\mlc v\equals r0 \ \ \mli\ \  v\equals \zero
\end{array}$  \ \   $\pa$   \vspace{3pt}

\noindent 5. $\begin{array}{l}
r\equals \zero\mlc v\equals r0 \ \ \mli\ \  v\equals \zero\add v\equals \zero\successor \add v\mgreater \zero\successor 
\end{array}$  \ \   $\add$-Choose: 4   \vspace{3pt}

\noindent 6. $\begin{array}{l}
r\notequals \zero\mlc v\equals r0 \ \ \mli\ \  v\mgreater \zero\successor 
\end{array}$  \ \  $\pa$   \vspace{3pt}

\noindent 7. $\begin{array}{l}
r\notequals \zero\mlc v\equals r0 \ \ \mli\ \  v\equals \zero\add v\equals \zero\successor \add v\mgreater \zero\successor 
\end{array}$  \ \   $\add$-Choose: 6   \vspace{3pt}

\noindent 8. $\begin{array}{l}
(r\equals \zero\add r\notequals \zero)\mlc v\equals r0 \ \ \mli\ \  v\equals \zero\add v\equals \zero\successor \add v\mgreater \zero\successor 
\end{array}$  \ \  $\adc$-Introduction: 5,7  \vspace{3pt}

\noindent 9. $\begin{array}{l}
\ada x (x\equals \zero\add x\notequals \zero)\mlc v\equals r0 \ \ \mli\ \  v\equals \zero\add v\equals \zero\successor \add v\mgreater \zero\successor 
\end{array}$  \ \ $\ade$-Choose: 8   \vspace{3pt}

\noindent 10. $\begin{array}{l}
r\equals \zero\mlc v\equals r1 \ \ \mli\ \  v\equals \zero\successor 
\end{array}$  \ \ $\pa$   \vspace{3pt}

\noindent 11. $\begin{array}{l}
r\equals \zero\mlc v\equals r1 \ \ \mli\ \  v\equals \zero\add v\equals \zero\successor \add v\mgreater \zero\successor 
\end{array}$  \ \ $\add$-Choose: 10   \vspace{3pt}

\noindent 12. $\begin{array}{l}
r\notequals \zero\mlc v\equals r1 \ \ \mli\ \   v\mgreater \zero\successor 
\end{array}$  \ \   $\pa$   \vspace{3pt}

\noindent 13. $\begin{array}{l}
r\notequals \zero\mlc v\equals r1 \ \ \mli\ \  v\equals \zero\add v\equals \zero\successor \add v\mgreater \zero\successor 
\end{array}$  \ \ $\add$-Choose: 12   \vspace{3pt}

\noindent 14. $\begin{array}{l}
(r\equals \zero\add r\notequals \zero)\mlc v\equals r1 \ \ \mli\ \  v\equals \zero\add v\equals \zero\successor \add v\mgreater \zero\successor 
\end{array}$  \ \ $\adc$-Introduction: 11,13  \vspace{3pt}

\noindent 15. $\begin{array}{l}
\ada x (x\equals \zero\add x\notequals \zero)\mlc v\equals r1 \ \ \mli\ \  v\equals \zero\add v\equals \zero\successor \add v\mgreater \zero\successor 
\end{array}$  \ \ $\ade$-Choose: 14   \vspace{3pt}

\noindent 16. $\begin{array}{l}
\ada x (x\equals \zero\add x\notequals \zero)\mlc (v\equals r0\add v\equals r1) \ \ \mli\ \  v\equals \zero\add v\equals \zero\successor \add v\mgreater \zero\successor 
\end{array}$  \ \  $\adc$-Introduction: 9,15  \vspace{3pt}

\noindent 17. $\begin{array}{l}
\ada x (x\equals \zero\add x\notequals \zero)\mlc \ade x(v\equals x0\add v\equals x1) \ \ \mli\ \  v\equals \zero\add v\equals \zero\successor \add v\mgreater \zero\successor 
\end{array}$  \ \ $\ada$-Introduction: 16  \vspace{3pt}

\noindent 18. $\begin{array}{l}
v\equals \zero\add v\equals \zero\successor \add v\mgreater \zero\successor 
\end{array}$  \ \ MP: 2,3,17  \vspace{9pt}

The theoremhood of (\ref{nov11a}) thus has been verified. 

Our next subgoal is to show that each disjunct of (\ref{nov11a}) implies (\ref{nov10c}), that is, that each of the following formulas is provable in $\arfour$: 
\begin{equation}\label{nov11b}
v\equals \zero\mli  \ada x\bigl(|x|\mleq t\mli F(x)\bigr)\mli  |v|\mleq t\successor \mli F(v)  \end{equation}
\begin{equation}\label{nov11c}
v\equals \zero\successor \mli   \ada x\bigl(|x|\mleq t\mli F(x)\bigr)\mli  |v|\mleq t\successor \mli F(v)  \end{equation}
\begin{equation}\label{nov11d}
v\mgreater \zero\successor \mli   \ada x\bigl(|x|\mleq t\mli F(x)\bigr)\mli  |v|\mleq t\successor \mli F(v)\  \end{equation}

To see the provability of (\ref{nov11b}), observe that {\bf CL4} proves the formula 
\begin{equation}\label{sch1}
P(f)\mli g\equals f\mli P(g) .\end{equation} The formula 
\(F(\zero)\mli v\equals \zero\mli F(v) \)
is an instance of (\ref{sch1}) and therefore is provable in   $\arfour$. By (\ref{nov9bb}), $F(\zero)$ is also provable. Hence, by Modus 
Ponens, $\arfour\vdash  v\equals \zero\mli F(v) $. From here, by Weakening applied twice, we find the desired $\arfour\vdash (\ref{nov11b})$.  

The $\arfour$-provability of (\ref{nov11c}) is established as follows:\vspace{7pt}

\noindent 1. $\begin{array}{l}
\ade x(x\equals\zero) 
\end{array}$  \ \  Axiom 8 \vspace{3pt}

\noindent 2. $\begin{array}{l}
F(\zero)\mli s\equals\zero\mli F(s)
\end{array}$  \ \ $\clfour$-Instantiation, matches (\ref{sch1})  \vspace{3pt}

\noindent 3. $\begin{array}{l}
s\equals\zero\mli F(s)
\end{array}$  \ \ MP:  (\ref{nov9bb}),2  \vspace{3pt}

\noindent 4. $\begin{array}{l}
s\equals\zero\mli  F(s0)\adc F(s1)
\end{array}$  \ \ TR:  3, (\ref{nov9ba})  \vspace{3pt}

\noindent 5. $\begin{array}{l}
F(s0)\adc F(s1)\mli F(s1)
\end{array}$  \ \ $\clfour$-Instantiation, matches $P\adc Q\mli Q$\vspace{3pt}

\noindent 6. $\begin{array}{l}
s\equals\zero\mli  F(s1)
\end{array}$  \ \ TR:  4,5  \vspace{3pt}

\noindent 7. $\begin{array}{l}
\bigl(s\equals\zero\mli  F(s1)\bigr)\mli \bigl(s\equals\zero\mli  F(\zero 1)\bigr) 
\end{array}$  \ \ $\clfour$-Instantiation,  matches $ (s\equals f\mli P (g(s) ) )\mli  (s\equals f\mli P  (g(f) ) )$\vspace{3pt}

\noindent 8. $\begin{array}{l}
s\equals\zero\mli  F(\zero 1)  
\end{array}$  \ \ MP: 6,7 \vspace{3pt}

\noindent 9. $\begin{array}{l}
\ade x(x\equals\zero)\mli  F(\zero 1) 
\end{array}$  \ \ \ \ $\ada$-Introduction:  8 \vspace{3pt}

\noindent 10. $\begin{array}{l}
F(\zero 1) 
\end{array}$  \ \  MP: 1,9\vspace{3pt}

\noindent 11. $\begin{array}{l}
\zero 1 \equals \zero\successor
\end{array}$  \ \  $\pa$ \vspace{3pt}

\noindent 12. $\begin{array}{l}
F(\zero 1)\mlc \zero 1 \equals \zero\successor\mli F(\zero\successor)
\end{array}$  \ \  $\clfour$-Instantiation, matches $P(f)\mlc f\equals g\mli P(g)$ \vspace{3pt}

\noindent 13. $\begin{array}{l}
F(\zero\successor)
\end{array}$  \ \  MP: 10,11,12 \vspace{3pt}

\noindent 14. $\begin{array}{l}
F(\zero\successor)\mli v\equals \zero\successor\mli F(v)
\end{array}$  \ \  $\clfour$-Instantiation, matches (\ref{sch1}) \vspace{3pt}

\noindent 15. $\begin{array}{l}
 v\equals \zero\successor\mli F(v)
\end{array}$  \ \  MP: 13,14 \vspace{3pt}

\noindent 16. $\begin{array}{l}
v\equals \zero\successor \mli   \ada x\bigl(|x|\mleq t\mli F(x)\bigr)\mli  |v|\mleq t\successor \mli F(v)
\end{array}$  \ \ Weakening (twice): 15 \vspace{7pt}

Finally, to construct a proof of  (\ref{nov11d}), observe  that the following formula is valid in classical logic:
\[
\begin{array}{l}
\bigl(p_1(f) \mlc p_2(f)    \mli   p_3\bigr)\ \mli \    
 \bigl( p_4\mli p_5(f)\bigr)\mlc v\equals f\ \mli \ 
 p_2(v) \ \mli  \ \bigl(p_3\mli p_4\bigr)\ \mli \ p_1(v) \mli \ p_5(v). 
\end{array}\] 
Hence, by Match applied twice, $\clfour$ proves 
\begin{equation}\label{nov11g}
\begin{array}{l}
\bigl(p_1(f) \mlc p_2(f)    \mli  p_3\bigr)\ \mli \    
 \bigl( P\mli Q(f)\bigr)\mlc v\equals f\ \mli \ 
 p_2(v) \ \mli  \ \bigl(p_3\mli P\bigr)\ \mli \ p_1(v) \mli \ Q(v). 
\end{array}\end{equation}  
The following two formulas are instances of (\ref{nov11g}), and are therefore provable in $\arfour$:
\begin{equation}\label{nov11e}
\bigl(|r0|\mleq t\successor \mlc r0 \mgreater \zero \successor \mli |r|\mleq t\bigr)  \mli    
 \bigl(  F(r)\mli F(r0)\bigr)\mlc v\equals  r0  \mli  
 v\mgreater \zero \successor   \mli   \bigl(|r|\mleq t\mli F(r)\bigr)  \mli   |v|\mleq t\successor  \mli   F(v). 
\end{equation} 
\begin{equation}\label{nov11f}
\bigl(|r1|\mleq t\successor \mlc r1 \mgreater \zero \successor \mli |r|\mleq t\bigr) \mli   
 \bigl(  F(r)\mli F(r1)\bigr)\mlc v\equals  r1 \mli    
 v\mgreater \zero \successor  \mli  \bigl(|r|\mleq t \mli  F(r)\bigr) \mli  |v|\mleq t\successor  \mli  F(v). 
\end{equation} 

Now, the following sequence is a $\arfour$-proof of (\ref{nov11d}):\vspace{7pt}

\noindent 1. $\begin{array}{l}
\ade x(v\equals x0\add v\equals x1)
\end{array}$  \ \ Axiom 12    \vspace{3pt}

\noindent 2. $\begin{array}{l}
|r0|\mleq t\successor \mlc r0 \mgreater \zero \successor \mli |r|\mleq t
\end{array}$  \ \  $\pa$ \vspace{3pt}

\noindent 3. $\begin{array}{l}
\bigl(  F(r)\mli F(r0)\bigr)\mlc v\equals r0 \mli 
 v\mgreater \zero \successor \mli  \bigl(|r|\mleq t\mli F(r)\bigr)\mli  |v|\mleq t\successor \mli F(v) 
\end{array}$  \ \ MP: (\ref{nov11e}),2   \vspace{3pt}

\noindent 4. $\begin{array}{l}
\bigl(F(r)\mli F(r0)\adc F(r1)\bigr)\mlc v\equals r0 \mli
 v\mgreater \zero \successor \mli  \bigl(|r|\mleq t\mli F(r)\bigr)\mli  |v|\mleq t\successor \mli F(v) 
\end{array}$  \ \  $\add$-Choose: 3  \vspace{3pt}

\noindent 5. $\begin{array}{l}
|r1|\mleq t\successor \mlc r1 \mgreater \zero \successor \mli |r|\mleq t
\end{array}$  \ \  $\pa$  \vspace{3pt}

\noindent 6. $\begin{array}{l}
\bigl(  F(r)\mli F(r1)\bigr)\mlc v\equals r1 \mli 
 v\mgreater \zero \successor \mli \bigl(|r|\mleq t\mli F(r)\bigr)\mli  |v|\mleq t\successor \mli F(v) 
\end{array}$  \ \    MP: (\ref{nov11f}),5   \vspace{3pt}

\noindent 7. $\begin{array}{l}
\bigl(  F(r)\mli F(r0)\adc F(r1)\bigr)\mlc v\equals r1 \mli 
 v\mgreater \zero \successor \mli  \bigl(|r|\mleq t\mli F(r)\bigr)\mli  |v|\mleq t\successor \mli F(v) 
\end{array}$  \ \  $\add$-Choose: 8  \vspace{3pt}

\noindent 8. $\begin{array}{l}
\bigl(  F(r)\mli F(r0)\adc F(r1)\bigr)\mlc (v\equals r0\add v\equals r1) \mli
 v\mgreater \zero \successor \mli \bigl(|r|\mleq t\mli F(r)\bigr)\mli  |v|\mleq t\successor \mli F(v) 
\end{array}$  \ \  $\adc$-Intro: 4,7  \vspace{3pt}

\noindent 9. $\begin{array}{l}
\ada x\bigl(  F(x)\mli F(x0)\adc F(x1)\bigr)\mlc (v\equals r0\add v\equals r1) \mli
 v\mgreater \zero \successor \mli  \ada x\bigl(|x|\mleq t\mli F(x)\bigr)\mli  |v|\mleq t\successor \mli F(v) 
\end{array}$    \mbox{$\ade$-Chooses: 8}\vspace{3pt}

\noindent 10. $\begin{array}{l}
\ada x\bigl( F(x)\mli F(x0)\adc F(x1)\bigr)\mlc \ade x(v\equals x0\add v\equals x1) \mli
 v\mgreater \zero \successor \mli  \ada x\bigl(|x|\mleq t\mli F(x)\bigr)\mli  |v|\mleq t\successor \mli F(v) 
\end{array}$    \mbox{$\ada$-Intro: 9}\vspace{3pt}

\noindent 11. $\begin{array}{l}
v\mgreater \zero \successor \mli  \ada x\bigl(|x|\mleq t\mli F(x)\bigr)\mli  |v|\mleq t\successor \mli F(v) 
\end{array}$  \ \  MP: (\ref{nov9baa}),1,10  \vspace{9pt}

The provability of each of the three formulas (\ref{nov11b}), (\ref{nov11c}) and (\ref{nov11b}) has now been verified. From these three facts and the provability of (\ref{nov11a}), by $\add$-Elimination, we find that $\arfour$ proves $(\ref{nov10c})$. This, in turn, as noted earlier, implies  $(\ref{nov10b})$. Now, from  $(\ref{nov10a})$ and $(\ref{nov10b})$, by WPTI, we find that  
\[\arfour\vdash t\mleq \bound\mli \ada x\bigl(|x|\mleq t\mli F(x)\bigr).\]
The above, by $\ada$-Introduction, yields
\(\arfour\vdash \ada y\Bigl(y\mleq \bound\mli \ada x\bigl(|x|\mleq y\mli F(x)\bigr)\Bigr),\) from which, by $\ada$-Elimination,   
\(\arfour\vdash \bound\mleq \bound\mli \ada x\bigl(|x|\mleq \bound\mli F(x)\bigr).\)
But  $\pa\vdash \bound\mleq\bound$. So, by Modus Ponens,  
\(\arfour \vdash \ada x\bigl(|x|\mleq \bound\mli F(x)\bigr),\) from which, by $\ada$-Elimination, 
 $\arfour\vdash |s|\mleq \bound\mli F(s)$. This, together with Axiom 13, by Modus ponens, yields the desired conclusion  $\arfour\vdash  F(s)$.
\end{proof}

\subsection{An illustration of BSI in work} In this section we prove one $\arfour$-provability fact which, with the soundness of $\arfour$ in mind,   formally establishes the efficient decidability of the equality predicate. The proof of this fact presents a good exercise on using BSI, and  may help the reader appreciate the convenience offered by this rule, which is often a more direct and intuitive tool for efficiency-preserving inductive reasoning than PTI is.  

\begin{lemma}\label{nov18a}
$\arfour\vdash \ada x \ada y(y\equals x\add y\notequals x)$. 
\end{lemma}

\begin{idea} Using BSI, prove  $\ada y(y\equals s\add y\notequals s)$, from which the target formula follows  by $\ada$-Introduction. 
\end{idea}

\begin{proof} Let us first give an informal justification for $\ada x \ada y(y\equals x\add y\notequals x)$. We proceed by BSI-induction on $s$, where the formula $F(s)$ of induction is $\ada y(y\equals s\add y\notequals s)$. By Axiom 9, for an arbitrary $y$, we can tell whether $y\equals \zero$ or $y\notequals \zero$. This takes care of the basis (left premise) 
\begin{equation}\label{lpre}
\ada y(y\equals \zero\add y\notequals \zero)\end{equation}
 of induction. For the inductive step (right premise)
\begin{equation}\label{rpre}
\ada y(y\equals s\add y\notequals s)\mli \ada y(y\equals s0\add y\notequals s0)\adc \ada y(y\equals s1\add y\notequals s1) ,\end{equation}
 assume the resource $\ada y(y\equals s\add y\notequals s)$ is at our disposal. We need to show that we can solve 
\[\ada y(y\equals s0\add y\notequals s0)\adc\ada y(y\equals s1\add y\notequals s1),\] i.e., either {\em one} of the problems $\ada y(y\equals s0\add y\notequals s0)$ and $\ada y(y\equals s1\add y\notequals s1) $. Let us for now look at the first problem. Consider an arbitrary $y$. Axiom 12 allows us to find the binary predecessor $r$ of $y$ and also tell whether $y\equals r0$ or $y\equals r1$. If $y\equals r1$, then we already know that $y\notequals s0$ (because $s0$ is even while $r1$ is odd). And if $y\equals r0$, then $y\equals s0$ --- i.e. $r0\equals s0$ --- iff $r\equals s$. But whether $r\equals s$ we can figure out using (the available single copy of) the resource  $\ada y(y\equals s\add y\notequals s)$. To summarize, in any case we can tell whether $y\equals s0$ or $y\notequals s0$, meaning that we can solve $\ada y(y\equals s0\add y\notequals s0)$. The case of $\ada y(y\equals s1\add y\notequals s1)$ is handled in a similar way.
Then, by BSI, (\ref{lpre}) and (\ref{rpre}) imply $\ada y(s\equals y\add s\notequals y)$, which, in turn (by $\ada$-Introduction), implies $\ada x\ada y(x\equals y\add x\notequals y)$.

The above informal argument can be formalized as follows:\vspace{9pt}

\noindent 1. $\begin{array}{l}
s\equals\zero\add s\notequals \zero
\end{array}$  \ \ Axiom 9   \vspace{3pt}

\noindent 2. $\begin{array}{l}
\ada y(y\equals\zero\add y\notequals \zero)
\end{array}$  \ \ $\ada$-Introduction: 1  \vspace{3pt}

\noindent 3. $\begin{array}{l}
\ade x(t\equals x0\add t\equals x1)
\end{array}$  \ \ Axiom 12   \vspace{3pt}

\noindent 4. $\begin{array}{l}
t\equals r0  \mli  r\equals s  \mli  t\equals s0 
\end{array}$  \ \ Logical axiom  \vspace{3pt}

\noindent 5. $\begin{array}{l}
t\equals r0  \mli  r\equals s  \mli  t\equals s0\add t\notequals s0  
\end{array}$  \ \  $\add$-Choose: 4  \vspace{3pt}

\noindent 6. $\begin{array}{l}
t\equals r0  \mli   r\notequals s \mli   t\notequals s0  
\end{array}$  \ \ \pa   \vspace{3pt}

\noindent 7. $\begin{array}{l}
t\equals r0  \mli   r\notequals s \mli  t\equals s0\add t\notequals s0  
\end{array}$  \ \   $\add$-Choose: 6  \vspace{3pt}

\noindent 8. $\begin{array}{l}
t\equals r0  \mli  r\equals s\add r\notequals s \mli  t\equals s0\add t\notequals s0  
\end{array}$  \ \  $\adc$-Introduction: 5,7  \vspace{3pt}

\noindent 9. $\begin{array}{l}
t\equals r0  \mli \ada y(y\equals s\add y\notequals s)\mli  t\equals s0\add t\notequals s0  
\end{array}$  \ \  $\ade$-Choose: 8  \vspace{3pt}

\noindent 10. $\begin{array}{l}
t\equals r1 \mli   t\notequals s0  
\end{array}$  \ \  $\pa$  \vspace{3pt}

\noindent 11. $\begin{array}{l}
t\equals r1 \mli \ada y(y\equals s\add y\notequals s)\mli   t\notequals s0  
\end{array}$  \ \  Wakening: 10  \vspace{3pt}

\noindent 12. $\begin{array}{l}
t\equals r1 \mli \ada y(y\equals s\add y\notequals s)\mli  t\equals s0\add t\notequals s0  
\end{array}$  \ \  $\add$-Choose: 11 \vspace{3pt}

\noindent 13. $\begin{array}{l}
t\equals r0\add t\equals r1 \mli \ada y(y\equals s\add y\notequals s)\mli  t\equals s0\add t\notequals s0  
\end{array}$  \ \  $\adc$-Introduction: 9,12  \vspace{3pt}

\noindent 14. $\begin{array}{l}
 \ade x(t\equals x0\add t\equals x1)\mli \ada y(y\equals s\add y\notequals s)\mli  t\equals s0\add t\notequals s0  
\end{array}$  \ \  $\ada$-Introduction: 13  \vspace{3pt}

\noindent 15. $\begin{array}{l}
\ada y(y\equals s\add y\notequals s)\mli  t\equals s0\add t\notequals s0  
\end{array}$  \ \  MP: 3,14  \vspace{3pt}

\noindent 16. $\begin{array}{l}
\ada y(y\equals s\add y\notequals s)\mli \ada y(y\equals s0\add y\notequals s0) 
\end{array}$  \ \ $\ada$-Introduction: 15   \vspace{3pt}

\noindent 17. $\begin{array}{l}
t\equals r1  \mli  r\equals s  \mli  t\equals s1 
\end{array}$  \ \ Logical axiom  \vspace{3pt}

\noindent 18. $\begin{array}{l}
t\equals r1  \mli  r\equals s  \mli  t\equals s1\add t\notequals s1  
\end{array}$  \ \  $\add$-Choose: 17  \vspace{3pt}

\noindent 19. $\begin{array}{l}
t\equals r1  \mli   r\notequals s \mli   t\notequals s1  
\end{array}$  \ \ $\pa$   \vspace{3pt}

\noindent 20. $\begin{array}{l}
t\equals r1  \mli   r\notequals s \mli  t\equals s1\add t\notequals s1  
\end{array}$  \ \   $\add$-Choose: 19  \vspace{3pt}

\noindent 21. $\begin{array}{l}
t\equals r1  \mli  r\equals s\add r\notequals s \mli  t\equals s1\add t\notequals s1  
\end{array}$  \ \  $\adc$-Introduction: 18,20  \vspace{3pt}

\noindent 22. $\begin{array}{l}
t\equals r1  \mli \ada y(y\equals s\add y\notequals s)\mli  t\equals s1\add t\notequals s1  
\end{array}$  \ \  $\ade$-Choose: 21  \vspace{3pt}

\noindent 23. $\begin{array}{l}
t\equals r0 \mli    t\notequals s1  
\end{array}$  \ \  $\pa$  \vspace{3pt}

\noindent 24. $\begin{array}{l}
t\equals r0 \mli \ada y(y\equals s\add y\notequals s)\mli   t\notequals s1  
\end{array}$  \ \  Weakening: 23  \vspace{3pt}

\noindent 25. $\begin{array}{l}
t\equals r0 \mli \ada y(y\equals s\add y\notequals s)\mli  t\equals s1\add t\notequals s1  
\end{array}$  \ \  $\add$-Choose: 24  \vspace{3pt}

\noindent 26. $\begin{array}{l}
t\equals r0\add t\equals r1 \mli \ada y(y\equals s\add y\notequals s)\mli  t\equals s1\add t\notequals s1  
\end{array}$  \ \  $\adc$-Introduction: 25,22  \vspace{3pt}

\noindent 27. $\begin{array}{l}
 \ade x(t\equals x0\add t\equals x1)\mli \ada y(y\equals s\add y\notequals s)\mli  t\equals s1\add t\notequals s1  
\end{array}$  \ \  $\ada$-Introduction: 26  \vspace{3pt}

\noindent 28. $\begin{array}{l}
\ada y(y\equals s\add y\notequals s)\mli  t\equals s1\add t\notequals s1  
\end{array}$  \ \  MP: 3,27  \vspace{3pt}

\noindent 29. $\begin{array}{l}
\ada y(y\equals s\add y\notequals s)\mli \ada y(y\equals s1\add y\notequals s1) 
\end{array}$  \ \ $\ada$-Introduction: 28   \vspace{3pt}

\noindent 30. $\begin{array}{l}
\ada y(y\equals s\add y\notequals s)\mli \ada y(y\equals s0\add y\notequals s0)\adc \ada y(y\equals s1\add y\notequals s1)
\end{array}$  \ \ $\adc$-Introduction: 16,29   \vspace{3pt}
 
\noindent 31. $\begin{array}{l}
\ada y(y\equals s\add y\notequals s) 
\end{array}$  \ \ BSI: 2,30   \vspace{3pt}

\noindent 32. $\begin{array}{l}
\ada x\ada y(y\equals x\add y\notequals x) 
\end{array}$  \ \ $\ada$-Introduction: 31  \vspace{7pt}
\end{proof}
  
\subsection{ PTI+, WPTI+ and BSI+}

The conclusion of PTI limits $s$ to ``very small'' values --- those that do not exceed (the value of) some $\bound$-term $\tau$. On the other hand, the right premise of the rule does not impose the corresponding restriction $s\mless\tau$ on $s$, and appears to be  stronger than necessary. Imposing the additional condition $|s\successor |\mleq \bound$ on $s$ in that premise also seems reasonable, because the size of $s$ in the conclusion cannot exceed $\bound$ anyway, and hence there is no need to prove the induction hypothesis for  the cases with $|s\successor|\mgreater \bound$.\footnote{The condition $s\mless \tau$ would not automatically imply $|s\successor|\mleq \bound$: in pathological cases where $\bound$ is ``very small'', it may happen that the first condition holds but the second condition is still violated.}  So, one might ask why we did not state PTI in the following, seemingly stronger, form --- call it ``{\bf PTI+}'':\label{iptiplus}
\[\frac{E(\zero)\mlc F(\zero)\hspace{45pt} s\mless\tau\mlc |s\successor|\mleq\bound \mlc E(s)\mlc F(s)\mli E(s\successor)\adc \bigl(  F(s\successor )\mlc E(s)\bigr)}{s\mleq \tau\mli E(s)\mlc F(s)}\]
(with the same additional conditions  as in PTI.)

The answer is very simple: PTI+, while being esthetically (or from the point of view of simplicity) inferior to PTI, does not really offer any greater deductive power, as implied by the forthcoming Theorem \ref{nov99}.

The following two rules --- call them {\bf WPTI+}\label{irp} (left) and {\bf BSI+} (right) --- are  pseudostrengthenings of WPTI and BSI in the same sense as PTI+ is a pseudostrengthening of PTI:

  \begin{center}
\begin{picture}(398,30)
\put(0,17){$F(\zero)\hspace{25pt}   s\mless\tau\mlc |s\successor|\mleq\bound \mlc F(s)\mli F(s\successor )$}
\put(0,11){\line(1,0){170}}
\put(58,0){$s\mleq \tau\mli F(s)$}

\put(220,17){$F(\zero)\hspace{25pt} |s0|\mleq\bound\mlc F(s) \mli F(s0)\adc F(s1)$}
\put(220,11){\line(1,0){178}}
\put(298,0){$ F(s)$}
\end{picture}
\end{center}
where $s$ is any variable different from $\bound$, $F(s)$ is any formula, and $\tau$ is any $\bound$-term.

\begin{theorem}\label{na}
WPTI+ is admissible in $\arfour$.
\end{theorem}

\begin{idea} WPTI+ is essentially a special case of WPTI with  $|s|\mleq \bound\mli s\mleq\tau \mli F(s)$ in the role of   $F(s)$. \end{idea}

\begin{proof}  
Assume $s$, $F(s)$, $\tau$ are as stipulated in the rule,  
\begin{equation}\label{nbb}
\arfour\vdash  F(\zero)\end{equation}
and 
\begin{equation}\label{nba}
 \arfour\vdash  s\mless\tau\mlc |s\successor |\mleq \bound \mlc  F(s)\mli  F(s\successor   ) .\end{equation}

Our goal is to verify that $\arfour\vdash s\mleq \tau\mli  F(s)$. 

From (\ref{nbb}), by Weakening applied twice, we get   
\begin{equation}\label{n19a}
\arfour\vdash  |\zero  |\mleq \bound\mli \zero\mleq\tau   \mli  F(\zero) .\end{equation}

Next,   observe that \[\clfour\vdash\bigl(q_1\mlc q_2\mli p_1\mlc p_2\mlc p_3\bigr)\mlc\bigl(p_1\mlc q_2\mlc  P\mli  Q\bigr) \mli  \bigl(p_3\mli p_2\mli P\bigr)\mli   \bigl(q_2\mli q_1\mli Q\bigr).\] Hence, by $\clfour$-Instantiation, we have 
\begin{equation}\label{n19b}
\begin{array}{l}
\arfour \vdash \bigl(s\successor\mleq \tau\mlc |s\successor |\mleq\bound \mli s\mless\tau\mlc s\mleq\tau\mlc |s|\mleq \bound\bigr)\mlc \bigl(s\mless\tau\mlc |s\successor |\mleq\bound \mlc  F(s)\mli  F(s\successor)\bigr)  \mli \\ \bigl(|s|\mleq\bound\mli s \mleq\tau\mli F(s )\bigr)\mli   \bigl(|s\successor |\mleq\bound\mli s\successor\mleq\tau\mli F(s\successor)\bigr).
\end{array} \end{equation}
We also have $\pa\vdash s\successor\mleq \tau \mlc |s\successor |\mleq\bound \mli  s\mless\tau\mlc s\mleq\tau\mlc |s|\mleq \bound$. This, together with (\ref{nba}) and (\ref{n19b}), by Modus Ponens, yields 
\begin{equation}\label{n19r}
\arfour \vdash   \bigl(|s|\mleq\bound\mli s \mleq\tau\mli F(s )\bigr)\mli   \bigl(|s\successor |\mleq\bound\mli s\successor\mleq\tau\mli F(s\successor)\bigr) .\end{equation}

From (\ref{n19a}) and (\ref{n19r}), by WPTI, we get 
\begin{equation}\label{n19t}
\arfour \vdash   s \mleq\tau\mli \bigl(|s|\mleq \bound\mli s \mleq\tau\mli F(s )\bigr)   .\end{equation}
But $\clfour\vdash \bigl(p\mli(q\mli p\mli Q)\bigr)\mli(q\mli p\mli Q)$ and hence, by $\clfour$-Instantiation,  
\[\arfour\vdash \Bigl(s \mleq\tau\mli \bigl(|s|\mleq\bound\mli s \mleq\tau\mli F(s )\bigr)\Bigr)\mli 
 \bigl(|s|\mleq\bound\mli  s \mleq\tau\mli F(s )\bigr) .\]
 Modus-ponensing the above with (\ref{n19t}) yields $\arfour\vdash |s|\mleq \bound\mli s \mleq\tau\mli F(s )$. Now, remembering Axiom 13, by Modus Ponens, we get  the desired $\arfour\vdash s \mleq\tau\mli F(s )$.
 \end{proof}

Note that the above proof established something stronger than what Theorem \ref{na} states. Namely, our proof of the admissibility of WPTI+ relied on WPTI without appealing to PTI. This means that WPTI+ would remain admissible even if $\arfour$ had WPTI instead of PTI. It is exactly this fact that justifies the qualification ``pseudostrengthening of WPTI'' that we gave to WPTI+. The same applies to the other two pseudostrengthening rules PTI+ and BSI+ discussed in this subsection.

\begin{theorem}\label{nb}
BSI+ is admissible in $\arfour$.
\end{theorem}

\begin{idea} 
BSI+ reduces to BSI with $|s|\mleq\bound\mli F(s)$ in the role of $F(s)$. 
\end{idea}

\begin{proof} 
Assume $s$, $F(s)$ are as stipulated in the rule, 
\begin{equation}\label{mbb}
\arfour\vdash  F(\zero)\end{equation}
and 
\begin{equation}\label{mba}
 \arfour\vdash |s0|\mleq\bound\mlc F(s) \mli F(s0)\adc F(s1) .\end{equation}

Our goal is to verify that $\arfour\vdash  F(s ) $. 

From (\ref{mbb}), by Weakening,  we have  
\begin{equation}\label{m19a}
\arfour\vdash  |\zero|\mleq\bound\mli  F(\zero) .\end{equation}

Next, in a routine (analytic) syntactic exercise, one can show that
\[\begin{array}{l}
\clfour \vdash \bigl(p_0\mld p_1\mli p_0\mlc q \bigr)  \mlc \bigl(p_0\mlc  Q\mli P_0\adc P_1\bigr)  \mli
  \bigl(q\mli Q\bigr)\mli   \bigl(p_0\mli P_0\bigr)\adc\bigl(p_1\mli P_1\bigr).  \end{array} \]
 Hence, by $\clfour$-Instantiation,  
\begin{equation}\label{m19b}
\begin{array}{l}
\arfour \vdash \bigl(|s0|\mleq\bound\mld |s1|\mleq\bound\mli |s0|\mleq\bound\mlc |s|\mleq\bound \bigr)  \mlc \bigl(|s0|\mleq\bound\mlc  F(s)\mli  F(s0)\adc F(s1)\bigr)  \mli\\
  \bigl(|s| \mleq\bound\mli F(s )\bigr)\mli   \bigl(|s0|\mleq\bound\mli F(s0)\bigr)\adc\bigl(|s1|\mleq\bound\mli F(s1)\bigr).  \end{array} \end{equation}
But   $\pa\vdash   |s0|\mleq\bound\mld |s1|\mleq\bound\mli |s0|\mleq\bound\mlc |s|\mleq\bound $.  This, together with (\ref{mba}) and (\ref{m19b}), by Modus Ponens, yields \[\arfour\vdash \bigl(|s| \mleq\bound\mli F(s )\bigr)\mli   \bigl(|s0|\mleq\bound\mli F(s0)\bigr)\adc\bigl(|s1|\mleq\bound\mli F(s1)\bigr).\]
The above and (\ref{m19a}), by BSI, yield $\arfour \vdash |s| \mleq\bound\mli F(s )$. Finally, modus-ponensing the latter with Axiom 13, we get the desired
$\arfour \vdash  F(s )$. 
 \end{proof}

\begin{theorem}\label{nov99}
PTI+ is admissible in $\arfour$.
\end{theorem}

\begin{idea} PTI+ reduces to PTI with $|s|\mleq\bound\mli s\mleq\tau\mli E(s)$ and $|s|\mleq\bound\mli s\mleq\tau\mli F(s)$ in the roles of $E(s)$ and $F(s)$, respectively. 

The present theorem will never be relied upon later, so, a reader satisfied with this explanation can safely omit the technical proof given below. 
\end{idea} 

\begin{proof}  
Assume $s$, $E(s)$, $F(s)$, $\tau$ are as stipulated in the rule,  
\begin{equation}\label{nov99bb}
\arfour\vdash E(\zero)\mlc F(\zero)\end{equation}
and 
\begin{equation}\label{nov99ba}
 \arfour\vdash  s\mless\tau \mlc |s\successor |\mleq\bound \mlc E(s)\mlc F(s)\mli E(s\successor)\adc \bigl(F(s\successor )\mlc E(s)\bigr) .\end{equation}

Our goal is to show that $\arfour\vdash s\mleq \tau\mli E(s)\mlc F(s)$. 

$\clfour$ proves $P\mlc Q\mli\bigl(p\mli q\mli P )\mlc (p\mli q\mli Q)$ and hence, by $\clfour$-Instantiation, 
\[\arfour \vdash E(\zero)\mlc F(\zero)\mli\bigl(|\zero|\mleq\bound\mli \zero \mleq\tau\mli E(\zero)\bigr)\mlc\bigl(|\zero|\mleq\bound\mli\zero \mleq\tau\mli F(\zero)\bigr).\]
Modus-ponensing the above with (\ref{nov99bb}) yields  
\begin{equation}\label{nov19a}
\arfour\vdash \bigl(|\zero|\mleq\bound\mli\zero\mleq\tau\mli E(\zero)\bigr)\mlc\bigl(|\zero |\mleq\bound\mli\zero\mleq\tau\mli F(\zero)\bigr).\end{equation}

Next, in a routine syntactic exercise we observe that  \[\clfour\vdash \gneg (p\mlc q)\mli Q\mlc P_1\mli (q\mli p\mli P_2)\adc\bigl((q\mli p\mli P_3)\mlc Q\bigr).\] 
Hence, by $\clfour$-Instantiation,
\begin{equation}\label{nov19b}
\begin{array}{l}
\arfour \vdash \gneg(   s\successor\mleq\tau\mlc |s\successor |\mleq \bound)\mli  \bigl(|s |\mleq \bound\mli s \mleq\tau\mli E(s )\bigr) \mlc \bigl(|s |\mleq \bound\mli s \mleq\tau\mli F(s )\bigr)\mli\\
\bigl(|s\successor |\mleq \bound\mli s\successor \mleq\tau\mli E(s\successor )\bigr)\adc \Bigl(\bigl(|s\successor |\mleq \bound\mli s\successor\mleq\tau\mli F(s\successor)\bigr)\mlc\bigl(|s|\mleq \bound\mli s\mleq\tau\mli E(s)\bigr)\Bigr).\end{array}
\end{equation}

In another syntactic exercise we find that 
\begin{equation}\label{jan13b}
\begin{array}{l}
 \clfour\vdash   (p_2\mlc q_2 \mli p_1\mlc q_1\mlc p_0 ) \mlc \bigl(p_0\mlc q_2\mlc P_1\mlc Q_1\mli P_2\adc (Q_2\adc P_1)\bigr)\mli p_2\mlc q_2\mli \\ 
(q_1\mli p_1\mli P_1)\mlc (q_1\mli p_1\mli Q_1)\mli (q_2\mli p_2\mli P_2)\adc\bigl((q_2\mli p_2\mli Q_2)\adc(q_1\mli p_1\mli P_1)\bigr).
\end{array}\end{equation}  
Since this ``exercise'' is longer than the previous one, below we provide a full proof of (\ref{jan13b}):\vspace{7pt}

\noindent 1. $\begin{array}{l}
(p_2\mlc q_2 \mli p_1\mlc q_1\mlc p_0 ) \mlc  (p_0\mlc q_2\mlc p_3\mlc q_3\mli\twg )\mli  p_2\mlc q_2\mli\\   
(q_1\mli p_1\mli p_3)\mlc (q_1\mli p_1\mli q_3)\mli \twg 
\end{array}$ \ \  Tautology\vspace{3pt}

\noindent 2. $\begin{array}{l}
(p_2\mlc q_2 \mli p_1\mlc q_1\mlc p_0 ) \mlc  (p_0\mlc q_2\mlc p_3\mlc q_3\mli p_4  )\mli p_2\mlc q_2\mli \\ 
(q_1\mli p_1\mli p_3)\mlc (q_1\mli p_1\mli q_3)\mli (q_2\mli p_2\mli p_4) 
\end{array}$   \ \ Tautology \vspace{3pt}

\noindent 3. $\begin{array}{l}
(p_2\mlc q_2 \mli p_1\mlc q_1\mlc p_0 ) \mlc  (p_0\mlc q_2\mlc p_3\mlc q_3\mli P_2  )\mli p_2\mlc q_2\mli \\ 
(q_1\mli p_1\mli p_3)\mlc (q_1\mli p_1\mli q_3)\mli (q_2\mli p_2\mli P_2) 
\end{array}$ \ \  Match: 2  \vspace{3pt}

\noindent 4. $\begin{array}{l}
(p_2\mlc q_2 \mli p_1\mlc q_1\mlc p_0 ) \mlc \bigl(p_0\mlc q_2\mlc p_3\mlc q_3\mli P_2\adc (Q_2\mlc P_1)\bigr)\mli p_2\mlc q_2\mli \\ 
(q_1\mli p_1\mli p_3)\mlc (q_1\mli p_1\mli q_3)\mli (q_2\mli p_2\mli P_2) 
\end{array}$  \ \ $\add$-Choose: 3  \vspace{3pt}

\noindent 5. $\begin{array}{l}
(p_2\mlc q_2 \mli p_1\mlc q_1\mlc p_0 ) \mlc  (p_0\mlc q_2\mlc p_3\mlc q_3\mli  q_4\mlc p_4  )\mli p_2\mlc q_2\mli \\ 
(q_1\mli p_1\mli p_3)\mlc (q_1\mli p_1\mli q_3)\mli  (q_2\mli p_2\mli q_4)\mlc(q_1\mli p_1\mli p_4)  
\end{array}$  \ \ Tautology  \vspace{3pt}

\noindent 6. $\begin{array}{l}
(p_2\mlc q_2 \mli p_1\mlc q_1\mlc p_0 ) \mlc  (p_0\mlc q_2\mlc p_3\mlc q_3\mli  Q_2\mlc P_1  )\mli p_2\mlc q_2\mli \\ 
(q_1\mli p_1\mli p_3)\mlc (q_1\mli p_1\mli q_3)\mli  (q_2\mli p_2\mli Q_2)\mlc(q_1\mli p_1\mli P_1)  
\end{array}$ \ \  Match (twice): 5  \vspace{3pt}

\noindent 7. $\begin{array}{l}
(p_2\mlc q_2 \mli p_1\mlc q_1\mlc p_0 ) \mlc \bigl(p_0\mlc q_2\mlc p_3\mlc q_3\mli P_2\adc (Q_2\mlc P_1)\bigr)\mli p_2\mlc q_2\mli \\ 
(q_1\mli p_1\mli p_3)\mlc (q_1\mli p_1\mli q_3)\mli  (q_2\mli p_2\mli Q_2)\mlc(q_1\mli p_1\mli P_1)  
\end{array}$  $\add$-Choose: 6  \vspace{3pt}

\noindent 8. $\begin{array}{l}
(p_2\mlc q_2 \mli p_1\mlc q_1\mlc p_0 ) \mlc \bigl(p_0\mlc q_2\mlc p_3\mlc q_3\mli P_2\adc (Q_2\mlc P_1)\bigr)\mli p_2\mlc q_2\mli \\ 
(q_1\mli p_1\mli p_3)\mlc (q_1\mli p_1\mli q_3)\mli (q_2\mli p_2\mli P_2)\adc\bigl((q_2\mli p_2\mli Q_2)\mlc(q_1\mli p_1\mli P_1)\bigr) 
\end{array}$ \ \ Wait: 1,4,7   \vspace{3pt}

\noindent 9. $\begin{array}{l}
(p_2\mlc q_2 \mli p_1\mlc q_1\mlc p_0 ) \mlc \bigl(p_0\mlc q_2\mlc P_1\mlc Q_1\mli P_2\adc (Q_2\mlc P_1)\bigr)\mli p_2\mlc q_2\mli \\ 
(q_1\mli p_1\mli P_1)\mlc (q_1\mli p_1\mli Q_1)\mli (q_2\mli p_2\mli P_2)\adc\bigl((q_2\mli p_2\mli Q_2)\mlc(q_1\mli p_1\mli P_1)\bigr) 
\end{array}$    \mbox{Match (twice): 8}\vspace{7pt}

The  formula below matches the formula of (\ref{jan13b})  and therefore, by $\clfour$-Instantiation,  
\begin{equation}\label{jan13a}
\begin{array}{l}
\arfour \vdash   (s\successor\mleq\tau \mlc |s\successor |\mleq\bound  \mli s\mleq \tau \mlc |s|\mleq\bound \mlc s\mless\tau ) \mlc\\
 \Bigl( s\mless\tau \mlc |s\successor |\mleq\bound\mlc E(s) \mlc F(s)\mli E(s\successor) \adc \bigl(F(s\successor)\mlc E(s)\bigr)\Bigr)\mli\\ 
 s\successor\mleq\tau \mlc |s\successor |\mleq\bound \mli   
\bigl(|s|\mleq\bound \mli  s\mleq \tau \mli E(s)\bigr)\mlc \bigl(|s|\mleq\bound  \mli  s\mleq \tau \mli F(s)\bigr)\mli \\
\bigl(|s\successor|\mleq\bound \mli s\successor\mleq \tau \mli E(s\successor) \bigr)\adc\Bigl(\bigl(|s\successor |\mleq\bound \mli s\successor\mleq\tau \mli F(s\successor)\bigr)\mlc\bigl(|s|\mleq\bound \mli s\mleq\tau \mli E(s)\bigr)\Bigr).\end{array}
\end{equation}

Obviously we have $\pa \vdash   s\successor\mleq\tau \mlc |s\successor |\mleq\bound  \mli s\mleq \tau \mlc |s|\mleq\bound \mlc s\mless\tau $. This fact, together with (\ref{nov99ba}) and (\ref{jan13a}), by Modus Ponens, implies 
\begin{equation}\label{jan13c}
\begin{array}{l}
\arfour \vdash   
s\successor\mleq\tau \mlc |s\successor |\mleq\bound \mli   
\bigl(|s|\mleq\bound \mli  s\mleq \tau \mli E(s)\bigr)\mlc \bigl(|s|\mleq\bound  \mli  s\mleq \tau \mli F(s)\bigr)\mli \\
\bigl(|s\successor|\mleq\bound \mli s\successor\mleq \tau \mli E(s\successor) \bigr)\adc\Bigl(\bigl(|s\successor |\mleq\bound \mli s\successor\mleq\tau \mli F(s\successor)\bigr)\mlc\bigl(|s|\mleq\bound \mli s\mleq\tau \mli E(s)\bigr)\Bigr).\end{array}
\end{equation}

According to the forthcoming Lemmas  \ref{com} and  \ref{minus}, whose proofs (as any other proofs in this paper) do not rely  on PTI+,  we have:
\begin{equation}\label{jan13d}
\mbox{\em For any  term $\theta$, $\arfour\vdash \gneg |\theta|\mleq\bound\add \ade z(z\equals\theta)$;}
\end{equation}
\begin{equation}\label{jan13e}
 \arfour\vdash \ada x\ada y\ade z(x\equals y\plus z\add y\equals x\plus z).
\end{equation}

Below is a proof of the fact that 
\begin{equation}\label{jan13f}
\arfour\vdash \gneg (s\successor\mleq\tau \mlc |s\successor |\mleq\bound)\add (s\successor\mleq\tau \mlc |s\successor |\mleq\bound) :\vspace{5pt}
\end{equation}

\noindent 1. $\begin{array}{l}
\gneg |s\successor |\mleq \bound\add\ade z(z\equals s\successor)  
\end{array}$  \ \ (\ref{jan13d}) with $\theta=s\successor$     \vspace{3pt}

\noindent 2. $\begin{array}{l}
\gneg |s\successor |\mleq \bound\mli \gneg (s\successor\mleq\tau \mlc |s\successor |\mleq\bound) 
\end{array}$  \ \  Logical axiom    \vspace{3pt}

\noindent 3. $\begin{array}{l}
\gneg |s\successor |\mleq \bound\mli \gneg (s\successor\mleq\tau \mlc |s\successor |\mleq\bound)\add (s\successor\mleq\tau \mlc |s\successor |\mleq\bound) 
\end{array}$  \ \ $\add$-Choose: 2     \vspace{3pt}

\noindent 4. $\begin{array}{l}
\gneg |\tau|\mleq \bound\add\ade z(z\equals \tau) 
\end{array}$  \ \ (\ref{jan13d}) with $\theta=\tau$   \vspace{3pt}

\noindent 5. $\begin{array}{l}
|r|\mleq \bound 
\end{array}$  \ \   Axiom 13  \vspace{3pt}

\noindent 6. $\begin{array}{l}
|r|\mleq \bound\mli \gneg |\tau|\mleq \bound\mli r\equals s\successor   \mli s\successor\mleq\tau \mlc |s\successor |\mleq\bound  
\end{array}$  \ \   $\pa$   \vspace{3pt}

\noindent 7. $\begin{array}{l}
\gneg |\tau|\mleq \bound\mli r\equals s\successor   \mli s\successor\mleq\tau \mlc |s\successor |\mleq\bound  
\end{array}$  \ \  MP: 5,6    \vspace{3pt}

\noindent 8. $\begin{array}{l}
\gneg |\tau|\mleq \bound\mli r\equals s\successor   \mli \gneg (s\successor\mleq\tau \mlc |s\successor |\mleq\bound)\add (s\successor\mleq\tau \mlc |s\successor |\mleq\bound) 
\end{array}$  \ \  $\add$-Choose: 7   \vspace{3pt}

\noindent 9. $\begin{array}{l}
\ade z(t\equals r\plus z\add r\equals t\plus z) 
\end{array}$  \ \   $\ada$-Elimination (twice): (\ref{jan13e})  \vspace{3pt}

\noindent 10. $\begin{array}{l}
|r|\mleq \bound\mlc t\equals r\plus v \mli t\equals \tau \mli r\equals s\successor   \mli   s\successor\mleq\tau \mlc |s\successor |\mleq\bound  
\end{array}$  \ \  $\pa$    \vspace{3pt}

\noindent 11. $\begin{array}{l}
|r|\mleq \bound\mlc t\equals r\plus v \mli t\equals \tau \mli r\equals s\successor   \mli \gneg (s\successor\mleq\tau \mlc |s\successor |\mleq\bound)\add (s\successor\mleq\tau \mlc |s\successor |\mleq\bound) 
\end{array}$  \ \  $\add$-Choose: 10   \vspace{3pt}

\noindent 12. $\begin{array}{l}
v\equals\zero\add v\notequals\zero 
\end{array}$  \ \   Axiom 8  \vspace{3pt}

\noindent 13. $\begin{array}{l}
v\equals\zero \mli
|r|\mleq \bound\mlc r\equals t\plus v \mli t\equals \tau \mli r\equals s\successor   \mli   s\successor\mleq\tau \mlc |s\successor |\mleq\bound  
\end{array}$  \ \  $\pa$    \vspace{3pt}

\noindent 14. $\begin{array}{l}
v\equals\zero \mli
|r|\mleq \bound\mlc r\equals t\plus v \mli t\equals \tau \mli r\equals s\successor   \mli \gneg (s\successor\mleq\tau \mlc |s\successor |\mleq\bound)\add (s\successor\mleq\tau \mlc |s\successor |\mleq\bound) 
\end{array}$  \ \   $\add$-Choose: 13   \vspace{3pt}

\noindent 15. $\begin{array}{l}
v\notequals\zero \mli
|r|\mleq \bound\mlc r\equals t\plus v \mli t\equals \tau \mli r\equals s\successor   \mli \gneg (s\successor\mleq\tau \mlc |s\successor |\mleq\bound) 
\end{array}$  \ \  $\pa$    \vspace{3pt}

\noindent 16. $\begin{array}{l}
v\notequals\zero \mli
|r|\mleq \bound\mlc r\equals t\plus v \mli t\equals \tau \mli r\equals s\successor   \mli \gneg (s\successor\mleq\tau \mlc |s\successor |\mleq\bound)\add (s\successor\mleq\tau \mlc |s\successor |\mleq\bound) 
\end{array}$  \ \     $\add$-Choose: 15   \vspace{3pt}

\noindent 17. $\begin{array}{l}
|r|\mleq \bound\mlc r\equals t\plus v \mli t\equals \tau \mli r\equals s\successor   \mli \gneg (s\successor\mleq\tau \mlc |s\successor |\mleq\bound)\add (s\successor\mleq\tau \mlc |s\successor |\mleq\bound) 
\end{array}$  \ \  $\add$-Elimination: 12,14,16   \vspace{3pt}

\noindent 18. $\begin{array}{l}
|r|\mleq \bound\mlc(t\equals r\plus v\add r\equals t\plus v)\mli t\equals \tau \mli r\equals s\successor   \mli \gneg (s\successor\mleq\tau \mlc |s\successor |\mleq\bound)\add (s\successor\mleq\tau \mlc |s\successor |\mleq\bound) 
\end{array}$  \ \ $\adc$-Introduction:  11,17\vspace{3pt}

\noindent 19. $\begin{array}{l}
|r|\mleq \bound\mlc \ade z(t\equals r\plus z\add r\equals t\plus z)\mli t\equals \tau \mli r\equals s\successor   \mli \gneg (s\successor\mleq\tau \mlc |s\successor |\mleq\bound)\add (s\successor\mleq\tau \mlc |s\successor |\mleq\bound) 
\end{array}$   $\ada$-Introduction:  18\vspace{3pt}

\noindent 20. $\begin{array}{l}
t\equals \tau \mli r\equals s\successor   \mli \gneg (s\successor\mleq\tau \mlc |s\successor |\mleq\bound)\add (s\successor\mleq\tau \mlc |s\successor |\mleq\bound) 
\end{array}$  \ \   MP: 5,9,19   \vspace{3pt}

\noindent 21. $\begin{array}{l}
\ade z(z\equals \tau)\mli r\equals s\successor   \mli \gneg (s\successor\mleq\tau \mlc |s\successor |\mleq\bound)\add (s\successor\mleq\tau \mlc |s\successor |\mleq\bound) 
\end{array}$  \ \   $\ada$-Introduction: 20   \vspace{3pt}

\noindent 22. $\begin{array}{l}
r\equals s\successor   \mli \gneg (s\successor\mleq\tau \mlc |s\successor |\mleq\bound)\add (s\successor\mleq\tau \mlc |s\successor |\mleq\bound) 
\end{array}$  \ \  $\add$-Elimination: 4,8,21    \vspace{3pt}

\noindent 23. $\begin{array}{l}
\ade z(z\equals s\successor)  \mli \gneg (s\successor\mleq\tau \mlc |s\successor |\mleq\bound)\add (s\successor\mleq\tau \mlc |s\successor |\mleq\bound) 
\end{array}$  \ \   $\ada$-Introduction: 22   \vspace{3pt}

\noindent 24. $\begin{array}{l}
\gneg (s\successor\mleq\tau \mlc |s\successor |\mleq\bound)\add (s\successor\mleq\tau \mlc |s\successor |\mleq\bound) 
\end{array}$  \ \  $\add$-Elimination: 1,3,23   \vspace{7pt}

Now, from (\ref{jan13f}), (\ref{nov19b}) and (\ref{jan13c}), by $\add$-Elimination, we get 
\begin{equation}\label{jan13h}
\begin{array}{l}
\arfour \vdash    
\bigl(|s|\mleq\bound \mli  s\mleq \tau \mli E(s)\bigr)\mlc \bigl(|s|\mleq\bound  \mli  s\mleq \tau \mli F(s)\bigr)\mli \\
\bigl(|s\successor|\mleq\bound \mli s\successor\mleq \tau \mli E(s\successor) \bigr)\adc\Bigl(\bigl(|s\successor |\mleq\bound \mli s\successor\mleq\tau \mli F(s\successor)\bigr)\mlc\bigl(|s|\mleq\bound \mli s\mleq\tau \mli E(s)\bigr)\Bigr).\end{array}
\end{equation}
From (\ref{nov19a}) and (\ref{jan13h}), by PTI, we get  
\begin{equation}\label{jan13i}
\begin{array}{l}
\arfour \vdash    
s\mleq\tau\mli \bigl(|s|\mleq\bound \mli  s\mleq \tau \mli E(s)\bigr)\mlc \bigl(|s|\mleq\bound  \mli  s\mleq \tau \mli F(s)\bigr).\end{array}
\end{equation}
But $\clfour$ obviously proves $\bigl(p\mli(q\mli p\mli P)\mlc (q\mli p\mli Q)\bigr)\mli (q\mli p\mli P\mlc Q)  $ and hence, by $\clfour$-Instantiation,
\[\begin{array}{l}
\arfour \vdash    
\Bigl(s\mleq\tau\mli \bigl(|s|\mleq\bound \mli  s\mleq \tau \mli E(s)\bigr)\mlc \bigl(|s|\mleq\bound  \mli  s\mleq \tau \mli F(s)\bigr)\Bigr)\mli 
\bigl( |s|\mleq\bound \mli  s\mleq \tau \mli E(s) \mlc  F(s) \bigr).\end{array}\]
Modus-ponensing the above with (\ref{jan13i}) yields $ \arfour\vdash |s|\mleq\bound \mli  s\mleq \tau \mli E(s) \mlc  F(s)$, further modus-ponensing which with Axiom 13 yields the desired $ \arfour\vdash  s\mleq \tau \mli E(s) \mlc  F(s)$.
\end{proof}

\subsection{BPI}  For any formula $E(s)$, we let $E(\lfloor s/2\rfloor)$\label{ifloor} stand for the formula $\cla z\bigl( s\equals z0\mld s\equals z1 \mli E(z)\bigr)$, asserting that $E$ holds for the binary predecessor of $s$.

One last rule of induction that we are going to look at is what we call {\bf BPI} ({\bf B}inary-{\bf P}redecessor-based {\bf I}nduction):\label{ibpi}
\[\frac{F(\zero )\hspace{20pt} F(\lfloor s/2\rfloor)\mli F(s)}{F(s)},\]
where $s$ is any non-$\bound$ variable and   $F(s)$ is any formula.\footnote{Those familiar with bounded arithmetics will notice a resemblance between BPI and the version of induction axiom known as PIND (\cite{Buss,Hajek}). An important difference, however, is that PIND assumes $s$ to be (actually or potentially) 
$\cla$-bound, while in our case $s$, as a free variable, can be seen as $\adc$-bound but by no means as $\cla$-bound.}  

 This rule could be characterized as an ``alternative formulation of BSI+'', and is apparently equivalent to the latter in the sense that replacing PTI with BSI+ in ptarithmetic yields the same class of provable formulas as replacing PTI with BPI.  One direction of this equivalence is immediately implied by our proof of the following theorem.

\begin{theorem}\label{nov88}
BPI is admissible in $\arfour$.
\end{theorem}

\begin{idea} As noted, BPI is essentially the same as BSI+.
\end{idea}

\begin{proof}   
Assume $s$, $F(s)$ are as stipulated in the rule,   
\begin{equation}\label{n20a}
\arfour\vdash  F(\zero)\end{equation}
and 
\begin{equation}\label{n20b}
 \arfour\vdash F(\lfloor s/2\rfloor)\mli F(s) .\end{equation}

Our goal is to verify that $\arfour\vdash  F(s)$.

We observe that  
\[\clfour \vdash \bigl(p\mli  t\equals f(s)\bigr) \mlc \Bigl( \cla z\bigl(t\equals f(z)\mld q\mli P(z)\bigr)\mli Q(t)\Bigr)\mli p\mlc P(s)\mli Q\bigl((f(s)\bigr) \]
(bottom-up, apply Match twice and you will hit a classically valid formula). By $\ade$-Choose, this yields 
\[\clfour\vdash \bigl(p\mli  t\equals f(s)\bigr) \mlc \ada x\Bigl( \cla z\bigl(x\equals f(z)\mld q\mli P(z)\bigr)\mli Q(x)\Bigr)\mli p\mlc P(s)\mli Q\bigl((f(s)\bigr). \]
The above, together with the obvious fact $\clfour\vdash \bigl(p\mli  \tlg\bigr) \mlc\twg \mli p\mlc\twg\mli \tlg$, by Wait, yields 
\[\clfour\vdash \Bigl(p\mli  \ade x\bigl(x\equals f(s)\bigr)\Bigr) \mlc \ada x\Bigl( \cla z\bigl(x\equals f(z)\mld q\mli P(z)\bigr)\mli Q(x)\Bigr)\mli p\mlc P(s)\mli Q\bigl((f(s)\bigr). \]
Hence, by $\clfour$-Instantiation, we have
\[
\arfour\vdash \bigl(|s0|\mleq \bound\mli  \ade x(x\equals s0)\bigr) \mlc \ada x\Bigl( \cla z\bigl(x\equals z0\mld x\equals z1\mli F(z)\bigr)\mli F(x)\Bigr)\mli |s0|\mleq\bound\mlc F(s)\mli F(s0),\]
which we abbreviate as 
\begin{equation}\label{kata1}
\arfour\vdash \bigl(|s0|\mleq \bound\mli  \ade x(x\equals s0)\bigr) \mlc \ada x\bigl( F(\lfloor x/2\rfloor)\mli F(x)\bigr)\mli |s0|\mleq\bound\mlc F(s)\mli F(s0). 
\end{equation}
In a similar way we find that 
\begin{equation}\label{kata2}
\arfour\vdash \bigl(|s0|\mleq \bound\mli  \ade x(x\equals s1)\bigr) \mlc \ada x\bigl( F(\lfloor x/2\rfloor)\mli F(x)\bigr)\mli |s0|\mleq\bound\mlc F(s)\mli F(s1). 
\end{equation}
Now, we construct a sought $\arfour$-proof of $F(s)$ as follows:\vspace{9pt}

\noindent 1. $\begin{array}{l}
\ada x\bigl(F(\lfloor x/2\rfloor)\mli F(x)\bigr)
\end{array}$  \ \  $\ada$-Introduction: (\ref{n20b})   \vspace{3pt}

\noindent 2. $\begin{array}{l}
|s0|\mleq \bound\mli \ade x(x\equals s0)
\end{array}$  \ \  Axiom 11   \vspace{3pt}

\noindent 3. $\begin{array}{l}
 |s0|\mleq\bound\mlc F(s)\mli F(s0) 
\end{array}$  \ \ MP: 2,1,(\ref{kata1})     \vspace{3pt}

\noindent 4. $\begin{array}{l}
|s0|\mleq \bound\mli   |s1|\mleq \bound 
\end{array}$  \ \  $\pa$   \vspace{3pt}

\noindent 5. $\begin{array}{l}
|s1|\mleq \bound\mli \ade x(x\equals s1)
\end{array}$  \ \  Lemma \ref{nov8}  \vspace{3pt}

\noindent 6. $\begin{array}{l}
|s0|\mleq \bound\mli \ade x(x\equals s1)
\end{array}$  \ \  TR: 4,5  \vspace{3pt}

\noindent 7. $\begin{array}{l}
 |s0|\mleq\bound\mlc F(s)\mli F(s1) 
\end{array}$  \ \ MP: 6,1,(\ref{kata2})     \vspace{3pt}

\noindent 8. $\begin{array}{l}
 |s0|\mleq\bound\mlc F(s)\mli F(s0)\adc F(s1) 
\end{array}$  \ \ $\adc$-Introduction: 3,7     \vspace{3pt}

\noindent 9. $\begin{array}{l}
F(s) 
\end{array}$  \ \ BSI+: (\ref{n20a}),8   
\end{proof}

\section{Efficient computability through \arfour-provability}\label{s17}

In this section we establish several $\arfour$-provability facts. In view of the soundness of $\arfour$, each such fact tells us about the efficient solvability  of the associated number-theoretic  computational problem.

\subsection{The efficient computability of logarithm} 
The term ``logarithm'' in the title of this subsection refers to the size of the binary numeral for a given number, which happens to be an integer approximation of the (real) base-2 logarithm of that number. 

\begin{lemma}\label{comlen}
 $\arfour\vdash \ada x\ade y(y\equals |x|)$.
\end{lemma}

\begin{proof} 
An outline of our proof is that $\ada x\ade y(y\equals |x|)$ follows by $\ada$-Introduction from $\ade y(y\equals |s|)$, and the latter will  be proven by BPI. Let us first try to justify the two premises of BPI informally. 

From $\pa$, we know that the size of $\zero$ is $\zero\successor$, and the value of $\zero\successor$ can be found using Lemma \ref{zersuc}. 
This allows us to resolve $\ade y(y\equals |\zero|)$, which is the basis of our BPI-induction.

The inductive step looks like \[ \ade y(y\equals |\lfloor s/2\rfloor|)\mli \ade y(y\equals |s|).\] Resolving it means telling the size of $s$ (in the consequent) while knowing (from the antecedental resource) the size $r$ of the binary predecessor $\lfloor s/2\rfloor$ of $s$. As was established earlier in the proof of Theorem \ref{nov9b}, we can tell whether $s$ equals $\zero$, $\zero\successor$, or neither. If $s\equals\zero$ or $s\equals\zero\successor$, then its size is the value of $\zero\successor$ which, as pointed out in the previous paragraph, we know how to compute. Otherwise, the size of $s$ is $r\successor$, which we can compute using Axiom 10. In all cases we thus can tell the size of $s$, and thus we can resolve the consequent of the above-displayed inductive step. 

Below is a formal counterpart of the above argument.\vspace{9pt}

\noindent 1. $\begin{array}{l}
\ade x(x\equals\zero\successor) 
\end{array}$  \ \ Lemma \ref{zersuc} \vspace{3pt}

\noindent 2. $\begin{array}{l}
v\equals\zero\successor \mli v\equals |\zero| 
\end{array}$  \ \ $\pa$ \vspace{3pt}

\noindent 3. $\begin{array}{l}
v\equals\zero\successor \mli \ade y(y\equals |\zero|)
\end{array}$  \ \ $\ade$-Choose: 2 \vspace{3pt}

\noindent 4. $\begin{array}{l}
\ade x(x\equals\zero\successor)\mli \ade y(y\equals |\zero|)
\end{array}$  \ \ $\ada$-Introduction: 3 \vspace{3pt}

\noindent 5. $\begin{array}{l}
\ade y(y\equals |\zero|)
\end{array}$  \ \ MP: 1,4 \vspace{3pt}

\noindent 6. $\begin{array}{l}
s\equals \zero\add s\equals \zero\successor \add s\mgreater \zero\successor
\end{array}$  \ \ (\ref{nov11a}),  established in the proof of Theorem \ref{nov9b}  \vspace{3pt}

\noindent 7. $\begin{array}{l}
v\equals\zero\successor \mli s\equals \zero\mli v\equals |s| 
\end{array}$  \ \  $\pa$  \vspace{3pt}

\noindent 8. $\begin{array}{l}
v\equals\zero\successor \mli s\equals \zero\mli \ade y(y\equals |s|)
\end{array}$  \ \ $\ade$-Choose: 7   \vspace{3pt}

\noindent 9. $\begin{array}{l}
\ade x(x\equals\zero\successor)\mli s\equals \zero\mli \ade y(y\equals |s|)
\end{array}$  \ \  $\ada$-Introduction: 8  \vspace{3pt}

\noindent 10. $\begin{array}{l}
s\equals \zero\mli \ade y(y\equals |s|)
\end{array}$  \ \  MP: 1,9  \vspace{3pt}

\noindent 11. $\begin{array}{l}
s\equals \zero\mli \ade y(y\equals |\lfloor s/2\rfloor|)\mli \ade y(y\equals |s|)
\end{array}$  \ \ Weakening: 10  \vspace{3pt}

\noindent 12. $\begin{array}{l}
v\equals\zero\successor \mli s\equals \zero\successor\mli  v\equals |s| 
\end{array}$  \ \ $\pa$  \vspace{3pt}

\noindent 13. $\begin{array}{l}
v\equals\zero\successor \mli s\equals \zero\successor\mli \ade y(y\equals |s|)
\end{array}$  \ \ $\ade$-Choose: 12  \vspace{3pt}

\noindent 14. $\begin{array}{l}
\ade x(x\equals\zero\successor)\mli s\equals \zero\successor\mli   \ade y(y\equals |s|)
\end{array}$  \ \  $\ada$-Introduction: 13   \vspace{3pt}

\noindent 15. $\begin{array}{l}
s\equals \zero\successor\mli   \ade y(y\equals |s|)
\end{array}$  \ \ MP: 1,14  \vspace{3pt}

\noindent 16. $\begin{array}{l}
s\equals \zero\successor\mli \ade y(y\equals |\lfloor s/2\rfloor|)\mli \ade y(y\equals |s|)
\end{array}$  \ \ Weakening: 15  \vspace{3pt}

\noindent 17. $\begin{array}{l}
|s|\mleq\bound 
\end{array}$  \ \ Axiom 13  \vspace{3pt}

\noindent 18. $\begin{array}{l}
|s\successor |\mleq\bound\mli\ade x(x\equals s\successor) 
\end{array}$  \ \ Axiom 10  \vspace{3pt}

\noindent 19. $\begin{array}{l}
 \ada y\bigl(|y\successor |\mleq\bound\mli\ade x(x\equals y\successor)\bigr)
\end{array}$  \ \ $\ada$-Introduction: 18  \vspace{3pt}

\noindent 20. $\begin{array}{l}
|s|\mleq\bound\mlc  \bigl(|r\successor |\mleq\bound\mli\tlg\bigr)\mli s\mgreater \zero\successor\mli r\equals |\lfloor s/2\rfloor| \mli \tlg
\end{array}$  \ \ $\pa$  \vspace{3pt}

\noindent 21. $\begin{array}{l}
|s|\mleq\bound\mlc   (|r\successor |\mleq\bound\mli w\equals r\successor  )\mli s\mgreater \zero\successor\mli r\equals |\lfloor s/2\rfloor| \mli w\equals |s| 
\end{array}$  \ \ $\pa$   \vspace{3pt}

\noindent 22. $\begin{array}{l}
|s|\mleq\bound\mlc   (|r\successor |\mleq\bound\mli w\equals r\successor  )\mli s\mgreater \zero\successor\mli r\equals |\lfloor s/2\rfloor| \mli \ade y(y\equals |s|)
\end{array}$  \ \ $\ade$-Choose: 21 \vspace{3pt}

\noindent 23. $\begin{array}{l}
|s|\mleq\bound\mlc  \bigl(|r\successor |\mleq\bound\mli\ade x(x\equals r\successor)\bigr)\mli s\mgreater \zero\successor\mli r\equals |\lfloor s/2\rfloor| \mli \ade y(y\equals |s|)
\end{array}$  \ \ Wait: 20,22  \vspace{3pt}

\noindent 24. $\begin{array}{l}
|s|\mleq\bound\mlc \ada y\bigl(|y\successor |\mleq\bound\mli\ade x(x\equals y\successor)\bigr)\mli s\mgreater \zero\successor\mli r\equals |\lfloor s/2\rfloor| \mli \ade y(y\equals |s|)
\end{array}$  \ \ $\ade$-Choose: 23  \vspace{3pt}

\noindent 25. $\begin{array}{l}
|s|\mleq\bound\mlc \ada y\bigl(|y\successor |\mleq\bound\mli\ade x(x\equals y\successor)\bigr)\mli s\mgreater \zero\successor\mli \ade y(y\equals |\lfloor s/2\rfloor|)\mli \ade y(y\equals |s|)
\end{array}$  \ \ $\ada$-Introduction: 24  \vspace{3pt}

\noindent 26. $\begin{array}{l}
s\mgreater \zero\successor\mli \ade y(y\equals |\lfloor s/2\rfloor|)\mli \ade y(y\equals |s|)
\end{array}$  \ \ MP: 17,19,25  \vspace{3pt}

\noindent 27. $\begin{array}{l}
\ade y(y\equals |\lfloor s/2\rfloor|)\mli \ade y(y\equals |s|)
\end{array}$  \ \ $\add$-Elimination:  6,11,16,26 \vspace{3pt} 

\noindent 28. $\begin{array}{l}
\ade y(y\equals |s|)
\end{array}$  \ \ BPI: 5,27 \vspace{3pt} 
\end{proof}

\begin{lemma}\label{comeqbound}
 $\arfour\vdash \ada x (|x|\equals \bound\add |x|\mless \bound)$.
\end{lemma}

\begin{proof} An informal argument for $\ada x (|x|\equals \bound\add |x|\mless \bound)$ is the following. Given an arbitrary $x$, we can find a $t$ with $t\equals |x|$ using Lemma \ref{comlen}. Lemma \ref{nov18a} allows us to tell whether $t\equals\bound$ or $t\notequals\bound$. In the second case, in view of Axiom 13, we have $t\mless\bound$. Thus, we can tell whether $t\equals \bound$ or $t\mless\bound$, i.e., whether $|x|\equals \bound$ or $|x|\mless \bound$. This means that we can resolve   $|x|\equals \bound\add |x|\mless \bound$. Formally, we have:\vspace{9pt}

\noindent 1. $\begin{array}{l}
\ada x\ade y(y=|x|) 
\end{array}$  \ \ Lemma \ref{comlen}     \vspace{3pt}

\noindent 2. $\begin{array}{l}
\ada x\ada y(y\equals x\add y\notequals x) 
\end{array}$  \ \ Lemma \ref{nov18a}    \vspace{3pt}

\noindent 3. $\begin{array}{l}
t\equals |s| \mlc   t\equals \bound  \mli |s|\equals \bound 
\end{array}$  \ \  Logical axiom    \vspace{3pt}

\noindent 4. $\begin{array}{l}
t\equals |s| \mlc   t\equals \bound  \mli |s|\equals \bound\add |s|\mless \bound 
\end{array}$  \ \  $\add$-Choose: 3    \vspace{3pt}

\noindent 5. $\begin{array}{l}
|s|\mleq\bound 
\end{array}$  \ \ Axiom 13     \vspace{3pt}

\noindent 6. $\begin{array}{l}
|s|\mleq\bound\mli t\equals |s| \mlc   t\notequals \bound \mli   |s|\mless \bound 
\end{array}$  \ \ \pa     \vspace{3pt}

\noindent 7. $\begin{array}{l}
t\equals |s| \mlc   t\notequals \bound \mli   |s|\mless \bound 
\end{array}$  \ \ MP: 5,6     \vspace{3pt}

\noindent 8. $\begin{array}{l}
t\equals |s| \mlc   t\notequals \bound \mli |s|\equals \bound\add |s|\mless \bound 
\end{array}$  \ \  $\add$-Choose: 7    \vspace{3pt}

\noindent 9. $\begin{array}{l}
t\equals |s| \mlc  (t\equals \bound \add t\notequals \bound)\mli |s|\equals \bound\add |s|\mless \bound 
\end{array}$  \ \ $\adc$-Introduction: 4,8     \vspace{3pt}

\noindent 10. $\begin{array}{l}
t\equals |s| \mlc \ada x\ada y(y\equals x\add y\notequals x)\mli |s|\equals \bound\add |s|\mless \bound 
\end{array}$  \ \ $\ade$-Choose (twice): 9     \vspace{3pt}

\noindent 11. $\begin{array}{l}
 \ade y(y\equals |s|)\mlc \ada x\ada y(y\equals x\add y\notequals x)\mli |s|\equals \bound\add |s|\mless \bound 
\end{array}$  \ \  $\ada$-Introduction: 10    \vspace{3pt}

\noindent 12. $\begin{array}{l}
\ada x\ade y(y\equals |x|)\mlc \ada x\ada y(y\equals x\add y\notequals x)\mli |s|\equals \bound\add |s|\mless \bound 
\end{array}$  \ \ $\ade$-Choose: 11     \vspace{3pt}

\noindent 13. $\begin{array}{l}
 |s|\equals \bound\add |s|\mless \bound 
\end{array}$  \ \ MP: 1,2,12     \vspace{3pt}

\noindent 14. $\begin{array}{l}
\ada x (|x|\equals \bound\add |x|\mless \bound)
\end{array}$  \ \ $\ada$-Introduction: 13    \vspace{3pt}
\end{proof}

\subsection{The efficient computability of unary successor}

In our subsequent treatment we will  be using  the abbreviation 
\[E\adi F\]
for the expression $\gneg E\add F$. The operator $\adi$ thus can be called 
{\bf choice implication}.\label{icimpl} 

When omitting parentheses, $\adi$ will have the same precedence level as $\add$, so that, say, $E\adi F\mli G$ should be understood as $(E\adi F)\mli G$ rather than $E\adi (F\mli G)$.  

The following lemma strengthens (the $\ada$-closure of) Axiom 10 by replacing $\mli$ with $\adi$.

\begin{lemma}\label{comsuc}
 $\arfour\vdash \ada x \bigl(|x\successor|\mleq \bound\adi \ade y (y\equals x\successor)\bigr)$.
\end{lemma}

\begin{proof} An informal argument for $\ada x \bigl(|x\successor|\mleq \bound\adi \ade y (y\equals x\successor)\bigr)$ goes like this. Given an arbitrary $x$, using Lemma \ref{comeqbound}, we can   figure out whether $|x|\equals\bound$ or $|x|\mless\bound$. 

If $|x|\mless\bound$, then (by $\pa$)    $|x\successor|\mleq\bound$. Then, using Axiom 10, we can find  a $t$ with $t\equals x\successor$. In this case, $|x\successor|\mleq \bound\adi \ade y (y\equals x\successor)$ will be resolved by choosing its right component $\ade y (y\equals x\successor)$   and then specifying $y$ as $t$ in it.

Suppose now  $|x|\equals\bound$. Then, by $\pa$, $|x\successor|\mleq\bound$ if and only if $x$ is even. And Axiom 12 allows us to tell whether $x$ is even or odd. If $x$ is even, we resolve $|x\successor|\mleq \bound\adi \ade y (y\equals x\successor)$ as in the previous case. And if $x$ is odd, then $|x\successor|\mleq \bound\adi \ade y (y\equals x\successor)$, i.e. $\gneg |x\successor|\mleq \bound\add \ade y (y\equals x\successor)$, is resolved by choosing its left component $\gneg |x\successor|\mleq \bound$. 

The following is a formalization of the above argument.\vspace{9pt} 

 \noindent 1. $\begin{array}{l}
\ada x (|x|\equals \bound\add |x|\mless \bound)
\end{array}$  \ \ Lemma \ref{comeqbound}     \vspace{3pt}

 \noindent 2. $\begin{array}{l}
|s|\equals \bound\add |s|\mless \bound 
\end{array}$  \ \ $\ada$-Elimination: 1     \vspace{3pt}

 \noindent 3. $\begin{array}{l}
\ade x(s\equals x0\add s\equals x1)
\end{array}$  \ \ Axiom 12     \vspace{3pt}

 \noindent 4. $\begin{array}{l}
 |s\successor |\mleq \bound\mli \ade x(x\equals s\successor) 
\end{array}$  \ \ Axiom 10     \vspace{3pt}

\noindent 5. $\begin{array}{l}
\bigl(|s\successor |\mleq \bound\mli \tlg \bigr)\mli s\equals r0 \mli |s |\equals \bound\mli \tlg  
\end{array}$  \ \  \pa    \vspace{3pt}

\noindent 6. $\begin{array}{l}
\bigl(|s\successor |\mleq \bound\mli  t\equals s\successor  \bigr)\mli s\equals r0 \mli |s |\equals \bound\mli t\equals s\successor   
\end{array}$  \ \  \pa    \vspace{3pt}

\noindent 7. $\begin{array}{l}
\bigl(|s\successor |\mleq \bound\mli  t\equals s\successor  \bigr)\mli s\equals r0 \mli |s |\equals \bound\mli \ade y (y\equals s\successor)  
\end{array}$  \ \  $\ade$-Choose: 6    \vspace{3pt}

\noindent 8. $\begin{array}{l}
\bigl(|s\successor |\mleq \bound\mli \ade x(x\equals s\successor) \bigr)\mli s\equals r0 \mli |s |\equals \bound\mli \ade y (y\equals s\successor)  
\end{array}$  \ \  Wait: 5,7    \vspace{3pt}

\noindent 9. $\begin{array}{l}
s\equals r0 \mli |s |\equals \bound\mli \ade y (y\equals s\successor)  
\end{array}$  \ \  MP: 4,8    \vspace{3pt}

\noindent 10. $\begin{array}{l}
s\equals r0 \mli |s |\equals \bound\mli |s\successor|\mleq \bound\adi\ade y (y\equals s\successor)  
\end{array}$  \ \  $\add$-Choose: 9    \vspace{3pt}

 \noindent 11. $\begin{array}{l}
 s\equals r1 \mli |s |\equals \bound\mli \gneg |s\successor|\mleq \bound
\end{array}$  \ \  \pa    \vspace{3pt}

 \noindent 12. $\begin{array}{l}
 s\equals r1 \mli |s |\equals \bound\mli |s\successor|\mleq \bound\adi\ade y (y\equals s\successor)  
\end{array}$  \ \    $\add$-Choose: 11    \vspace{3pt}

 \noindent 13. $\begin{array}{l}
s\equals r0\add s\equals r1 \mli |s |\equals \bound\mli |s\successor|\mleq \bound\adi\ade y (y\equals s\successor)  
\end{array}$  \ \   $\adc$-Introduction: 10,12   \vspace{3pt}

 \noindent 14. $\begin{array}{l}
\ade x(s\equals x0\add s\equals x1)\mli |s |\equals \bound\mli |s\successor|\mleq \bound\adi\ade y (y\equals s\successor)  
\end{array}$  \ \    $\ada$-Introduction: 13   \vspace{3pt}

 \noindent 15. $\begin{array}{l}
|s |\equals \bound\mli |s\successor|\mleq \bound\adi\ade y (y\equals s\successor)  
\end{array}$  \ \  MP: 3,14    \vspace{3pt}

 \noindent 16. $\begin{array}{l}
\bigl(|s\successor |\mleq \bound\mli\tlg\bigr)\mli |s |\mless \bound\mli \tlg  
\end{array}$  \ \    \pa  \vspace{3pt}

 \noindent 17. $\begin{array}{l}
\bigl(|s\successor |\mleq \bound\mli t\equals s\successor \bigr)\mli |s |\mless \bound\mli t\equals s\successor   
\end{array}$  \ \  \pa   \vspace{3pt}

 \noindent 18. $\begin{array}{l}
\bigl(|s\successor |\mleq \bound\mli t\equals s\successor \bigr)\mli |s |\mless \bound\mli \ade y (y\equals s\successor)  
\end{array}$  \ \  $\ade$-Choose: 17   \vspace{3pt}

 \noindent 19. $\begin{array}{l}
\bigl(|s\successor |\mleq \bound\mli \ade x(x\equals s\successor)\bigr)\mli |s |\mless \bound\mli \ade y (y\equals s\successor)  
\end{array}$  \ \  Wait: 16,18   \vspace{3pt}

 \noindent 20. $\begin{array}{l}
|s |\mless \bound\mli \ade y (y\equals s\successor)  
\end{array}$  \ \   MP: 4,19  \vspace{3pt}

 \noindent 21. $\begin{array}{l}
|s |\mless \bound\mli |s\successor|\mleq \bound\adi\ade y (y\equals s\successor)  
\end{array}$  \ \  $\add$-Choose: 20  \vspace{3pt}

\noindent 22. $\begin{array}{l}
|s\successor|\mleq \bound\adi \ade y (y\equals s\successor)  
\end{array}$  \ \  $\add$-Elimination: 2,15,21  \vspace{3pt}

\noindent 23. $\begin{array}{l}
\ada x \bigl(|x\successor|\mleq \bound\adi \ade y (y\equals x\successor)\bigr)
\end{array}$  \ \  $\ada$-Introduction: 22  \vspace{3pt}
\end{proof}

\subsection{The efficient computability of binary $0$-successor} 

The following lemma  strengthens Axiom 11 in the same way as Lemma \ref{comsuc} strengthens Axiom 10. 
 
\begin{lemma}\label{combzs}
 $\arfour\vdash \ada x (|x0|\mleq \bound\adi \ade y (y\equals x0)\bigr)$.
\end{lemma}

\begin{proof} Informally, the argument underlying our formal proof of $\ada x (|x0|\mleq \bound\adi \ade y (y\equals x0)\bigr)$ is the following.  Consider an arbitrary $x$. Using Lemma \ref{comeqbound}, we can tell whether $|x|\equals\bound$ or $|x|\mless\bound$. If $|x|\mless\bound$, then $|x0|\mleq\bound$ and, using Axiom 11, we can find a $t$ with $t\equals x0$. We then resolve $|x0|\mleq \bound\adi \ade y (y\equals x0)$ by choosing its right $\adi$-component and specifying $y$ as $t$ in it. Suppose now $|x|\equals\bound$. Using Axiom 9, we can tell whether $x$ is $0$ or not. If $x$ is $0$, then $|x0|\mleq \bound\adi \ade y (y\equals x0)$ is resolved by choosing its right $\adi$-component and specifying $y$ as $x$ in it. Otherwise, if $x\notequals\zero$, then the size of $x0$ exceeds  $\bound$. So, 
$|x0|\mleq \bound\adi \ade y (y\equals x0)$ is resolved by choosing its left $\adi$-component $\gneg |x0|\mleq \bound $. Formally, we have:\vspace{9pt}

\noindent 1. $\begin{array}{l}
\ada x(|x |\equals \bound\add |x|\mless \bound) 
\end{array}$  \ \ Lemma \ref{comeqbound}    \vspace{3pt}

\noindent 2. $\begin{array}{l}
|s |\equals \bound\add |s|\mless \bound 
\end{array}$  \ \ $\ada$-Elimination: 1    \vspace{3pt}

\noindent 3. $\begin{array}{l}
 s  \equals\zero\add s\notequals \zero 
\end{array}$  \ \ Axiom 9    \vspace{3pt}

\noindent 4. $\begin{array}{l}
 s  \equals\zero\mli |s |\equals \bound\mli   s\equals s0  
\end{array}$  \ \ $\pa$    \vspace{3pt}

\noindent 5. $\begin{array}{l}
 s  \equals\zero\mli |s |\equals \bound\mli   \ade y (y\equals s0) 
\end{array}$  \ \  $\ade$-Choose: 4   \vspace{3pt}

\noindent 6. $\begin{array}{l}
 s  \equals\zero\mli |s |\equals \bound\mli |s0|\mleq \bound\adi \ade y (y\equals s0) 
\end{array}$  \ \     $\add$-Choose: 5   \vspace{3pt}

\noindent 7. $\begin{array}{l}
 s  \notequals\zero\mli |s |\equals \bound\mli \gneg |s0|\mleq \bound 
\end{array}$  \ \   \pa  \vspace{3pt}

\noindent 8. $\begin{array}{l}
 s  \notequals\zero\mli |s |\equals \bound\mli |s0|\mleq \bound\adi \ade y (y\equals s0) 
\end{array}$  \ \     $\add$-Choose: 7   \vspace{3pt}

\noindent 9. $\begin{array}{l}
|s |\equals \bound\mli |s0|\mleq \bound\adi \ade y (y\equals s0) 
\end{array}$  \ \ $\add$-Elimination: 3,6,8    \vspace{3pt}

\noindent 10. $\begin{array}{l}
|s 0|\mleq \bound\mli \ade x(x\equals s0) 
\end{array}$  \ \ Axiom 11    \vspace{3pt}

\noindent 11. $\begin{array}{l}
\bigl(|s0 |\mleq \bound\mli \tlg\bigr)\mli |s|\mless \bound \mli \tlg
\end{array}$  \ \  $\pa$  \vspace{3pt}

\noindent 12. $\begin{array}{l}
\bigl(|s 0|\mleq \bound\mli t\equals s0 \bigr)\mli |s|\mless \bound \mli  t\equals s0  
\end{array}$  \ \  $\pa$   \vspace{3pt}

\noindent 13. $\begin{array}{l}
\bigl(|s 0|\mleq \bound\mli t\equals s0 \bigr)\mli |s|\mless \bound \mli  \ade y (y\equals s0) 
\end{array}$  \ \  $\ade$-Choose: 12   \vspace{3pt}

\noindent 14. $\begin{array}{l}
\bigl(|s0 |\mleq \bound\mli \ade x(x\equals s0)\bigr)\mli |s|\mless \bound \mli  \ade y (y\equals s0) 
\end{array}$  \ \  Wait: 11,13   \vspace{3pt}

\noindent 15. $\begin{array}{l}
|s|\mless \bound \mli   \ade y (y\equals s0) 
\end{array}$  \ \ MP: 10,14   \vspace{3pt}

\noindent 16. $\begin{array}{l}
|s|\mless \bound \mli  |s0|\mleq \bound\adi \ade y (y\equals s0) 
\end{array}$  \ \ $\add$-Choose: 15   \vspace{3pt}

\noindent 17. $\begin{array}{l}
|s0|\mleq \bound\adi \ade y (y\equals s0) 
\end{array}$  \ \  $\add$-Elimination: 2,9,16   \vspace{3pt}

\noindent 18. $\begin{array}{l}
\ada x (|x0|\mleq \bound\adi \ade y (y\equals x0)\bigr)
\end{array}$  \ \  $\ada$-Introduction: 17   \vspace{3pt}
\end{proof}

\subsection{The efficient computability of binary $1$-successor}

The following lemma is the same to Lemma \ref{nov8} as  Lemma \ref{combzs} is to Axiom 11.  

\begin{lemma}\label{combos}
 $\arfour\vdash \ada x (|x1|\mleq \bound\adi \ade y (y\equals x1)\bigr)$.
\end{lemma}

\begin{proof} The argument underlying our formal proof of $\ada x (|x1|\mleq \bound\adi \ade y (y\equals x1)\bigr)$ is the following.  Consider an arbitrary $x$. Using Lemma \ref{combzs}, we can tell whether the size of $x0$ exceeds $\bound$, or else find a $t$ with $t\equals x0$. In the first case we resolve $|x1|\mleq \bound\adi \ade y (y\equals x1)$ by choosing its left component $\gneg |x1|\mleq \bound$, because (we know from $\pa$ that) $|x0|=|x1|$. In the second case, using Axiom 10, we find an $r$ with $r\equals t\successor$. This axiom is applicable here because $|t\successor |\mleq\bound $; $|t\successor |\mleq\bound $, in turn, is true because $|t|\mleq \bound$ (by Axiom 13) and $t$ is even, so the unary successor of $t$ is of the same size as  $t$ itself.    Note that (as $\pa$ can help us to figure out), in the present case, $r\equals x1$. So, we can resolve 
$|x1|\mleq \bound\adi \ade y (y\equals x1)$ by choosing its right component and then specifying $y$ as $r$ in it. Formally, we have:\vspace{9pt}

\noindent 1. $\begin{array}{l}
\ada x (\gneg |x0|\mleq \bound\add \ade y (y\equals x0)\bigr)
\end{array}$  \ \  Lemma \ref{combzs}  \vspace{3pt}

\noindent 2. $\begin{array}{l}
\gneg |s0|\mleq \bound\add \ade y (y\equals s0) 
\end{array}$  \ \  $\ada$-Elimination: 1  \vspace{3pt}

\noindent 3. $\begin{array}{l}
\gneg |s0|\mleq \bound \mli \gneg |s1|\mleq \bound 
\end{array}$  \ \ \pa  \vspace{3pt}

\noindent 4. $\begin{array}{l}
\gneg |s0|\mleq \bound \mli |s1|\mleq \bound\adi \ade y (y\equals s1)
\end{array}$  \ \  $\add$-Choose: 3  \vspace{3pt}

\noindent 5. $\begin{array}{l}
|t|\mleq\bound
\end{array}$  \ \  Axiom 13 \vspace{3pt}

\noindent 6. $\begin{array}{l}
|t\successor |\mleq\bound\mli\ade x(x\equals t\successor )
\end{array}$  \ \  Axiom 10 \vspace{3pt}

\noindent 7. $\begin{array}{l}
|t|\mleq\bound\mlc \bigl(|t\successor |\mleq\bound\mli\tlg\bigr)\mli t\equals s0 \mli   \tlg 
\end{array}$  \ \  $\pa$ \vspace{3pt}

\noindent 8. $\begin{array}{l}
|t|\mleq\bound\mlc \bigl(|t\successor |\mleq\bound\mli  r\equals t\successor  \bigr)\mli t\equals s0 \mli   r\equals s1 
\end{array}$  \ \ $\pa$  \vspace{3pt}

\noindent 9. $\begin{array}{l}
|t|\mleq\bound\mlc \bigl(|t\successor |\mleq\bound\mli  r\equals t\successor  \bigr)\mli t\equals s0 \mli   \ade y (y\equals s1) 
\end{array}$  \ \  $\ade$-Choose: 8  \vspace{3pt}

\noindent 10. $\begin{array}{l}
|t|\mleq\bound\mlc \bigl(|t\successor |\mleq\bound\mli\ade x(x\equals t\successor )\bigr)\mli t\equals s0 \mli   \ade y (y\equals s1) 
\end{array}$  \ \  Wait: 7,9  \vspace{3pt}

\noindent 11. $\begin{array}{l}
t\equals s0 \mli   \ade y (y\equals s1) 
\end{array}$  \ \ MP: 5,6,10   \vspace{3pt}

\noindent 12. $\begin{array}{l}
 \ade y (y\equals s0)\mli   \ade y (y\equals s1) 
\end{array}$  \ \ $\ada$-Introduction: 11   \vspace{3pt}

\noindent 13. $\begin{array}{l}
 \ade y (y\equals s0)\mli |s1|\mleq \bound\adi \ade y (y\equals s1) 
\end{array}$  \ \  $\add$-Choose: 12  \vspace{3pt}

\noindent 14. $\begin{array}{l}
|s1|\mleq \bound\adi \ade y (y\equals s1) 
\end{array}$  \ \ $\add$-Elimination: 2,4,13   \vspace{3pt}

\noindent 15. $\begin{array}{l}
\ada x (|x1|\mleq \bound\adi \ade y (y\equals x1)\bigr)
\end{array}$  \ \  $\ada$-Introduction: 14  \vspace{3pt}
\end{proof}

\subsection{The efficient computability of addition}

\begin{lemma}\label{comad}
 $\arfour\vdash \ada x\ada y\bigl(|x\plus y|\mleq \bound\adi \ade z (z\equals x\plus y)\bigr)$.
\end{lemma}

\begin{proof} The main idea behind our proof of $\ada x\ada y\bigl(|x\plus y|\mleq \bound\adi \ade z (z\equals x\plus y)\bigr)$, which proceeds by BSI induction, is the fact --- known from $\pa$ --- that the sum of two numbers can be ``easily'' found from the sum of the binary predecessors of those numbers. Specifically, observe that we have: 
\begin{description}
\item[(i)] $s0\plus r0\equals(s\plus r)0$, because $2s\plus 2r\equals 2(s\plus r)$;
\item[(ii)] $s0\plus r1\equals(s\plus r)1$, because $2s\plus (2r\plus 1)\equals 2(s\plus r)\plus 1$;
\item[(iii)] $s1\plus r0\equals(s\plus r)1$, because $(2s\plus 1)\plus 2r\equals 2(s\plus r)\plus 1$;
\item[(iv)] $s1\plus r1\equals\bigl((s\plus r)1\bigr)\successor$, because $(2s\plus 1)\plus (2r\plus 1)\equals \bigl(2(s\plus r)\plus 1\bigr)\plus 1$.
\end{description}

The formula of induction is $\ada y\bigl(|s\plus y|\mleq \bound\adi \ade z (z\equals s\plus y)\bigr)$ (from which the target formula immediately follows by $\ada$-Introduction). 

The basis $\ada y\bigl(|\zero\plus y|\mleq \bound\adi \ade z (z\equals \zero \plus y)\bigr)$ of induction can be established/resolved rather easily, by choosing the right component of the $\adi$ combination and selecting the value of $z$ to be the same as the value of $y$. 

In resolving the inductive step 
\[ \ada y\bigl(|s\plus y|\mleq \bound\adi \ade z (z\equals s\plus y)\bigr) \mli \ada y\bigl(|s0\plus y|\mleq \bound\adi \ade z (z\equals s0\plus y)\bigr)  \adc \ada y\bigl(|s1\plus y|\mleq \bound\adi \ade z (z\equals s1\plus y)\bigr) ,\]
we wait for the environment to select a $\adc$-conjunct in the consequent (bottom-up $\adc$-Introduction) and then select a value $t$ for $y$ in it (bottom-up $\ada$-Introduction). Let us say the left conjunct is selected, meaning that the inductive step will be brought down to (i.e. the premise we are talking about will be) 
\[ \ada y\bigl(|s\plus y|\mleq \bound\adi \ade z (z\equals s\plus y)\bigr) \mli  \bigl(|s0\plus t|\mleq \bound\adi \ade z (z\equals s0\plus t)\bigr) .\]
Using Axiom 12, we can find the binary predecessor $r$ of $t$, and also figure out whether $t$ is $r0$ or $r1$. Let us say $t\equals r0$. Then we specify $y$ as $r$ in the antecedent of the above formula, after which the problem we need to resolve is, in fact, 
\[  \bigl(|s\plus r|\mleq \bound\adi \ade z (z\equals s\plus r)\bigr) \mli \bigl(|s0\plus r0|\mleq \bound\adi \ade z (z\equals s0\plus r0)\bigr) .\]
Here we can wait till the environment selects one of the $\adi$-components in the antecedent. If the left component is selected, we can resolve the problem by selecting the left $\adi$-component in the consequent, because, if $|s\plus r|$ exceeds $\bound$, then ``even more so'' does  $|s0\plus r0|$. Otherwise, if the right component is selected, then we further wait till the environment also selects a value $u$ for $z$ there, after which the problem will be brought down to 
\[  u\equals s\plus r  \mli \bigl(|s0\plus r0|\mleq \bound\adi \ade z (z\equals s0\plus r0)\bigr) .\]
But from the earlier observation (i) we know that $s0\plus r0=(s\plus r)0$. So, the above problem is, in fact, nothing but 
\[  u\equals s\plus r  \mli \bigl(|u0|\mleq \bound\adi \ade z (z\equals u0)\bigr),\]
which --- whose consequent, that is --- we can resolve using Lemma \ref{combzs}.

The remaining three possibilities of the above scenario are similar, but will rely on observation (ii), (iii) or (iv) instead of (i), and Lemma     \ref{combos} instead of \ref{combzs}. The case corresponding to (iv), in addition, will also use Lemma \ref{comsuc}.

Below is a formal counterpart of the above argument in full detail: \vspace{9pt}

\noindent 1. $\begin{array}{l}
s\equals \zero \plus s
\end{array}$  \ \ $\pa$   \vspace{3pt}

\noindent 2. $\begin{array}{l}
\ade z (z\equals \zero \plus s)
\end{array}$  \ \ $\ade$-Choose: 1   \vspace{3pt}

\noindent 3. $\begin{array}{l}
|\zero \plus s|\mleq \bound\adi \ade z (z\equals \zero \plus s)
\end{array}$  \ \  $\add$-Choose: 2  \vspace{3pt}

\noindent 4. $\begin{array}{l}
\ada y\bigl(|\zero \plus y|\mleq \bound\adi \ade z (z\equals \zero \plus y)\bigr)
\end{array}$  \ \  $\ada$-Introduction: 3  \vspace{3pt}

\noindent 5. $\begin{array}{l}
 t\equals r0   \mli   \gneg |s\plus r|\mleq \bound  \mli  \gneg |s0\plus t|\mleq \bound 
\end{array}$  \ \ \pa   \vspace{3pt}

\noindent 6. $\begin{array}{l}
 t\equals r1   \mli   \gneg |s\plus r|\mleq \bound  \mli  \gneg |s0\plus t|\mleq \bound 
\end{array}$  \ \ \pa   \vspace{3pt}

\noindent 7. $\begin{array}{l}
 t\equals r0\add t\equals r1  \mli   \gneg |s\plus r|\mleq \bound  \mli  \gneg |s0\plus t|\mleq \bound 
\end{array}$  \ \  $\adc$-Introduction: 5,6  \vspace{3pt}

\noindent 8. $\begin{array}{l}
 t\equals r0\add t\equals r1  \mli   \gneg |s\plus r|\mleq \bound  \mli  |s0\plus t|\mleq \bound\adi \ade z (z\equals s0\plus t)  
\end{array}$  \ \  $\add$-Choose: 7  \vspace{3pt}

\noindent 9. $\begin{array}{l}
\ada x (|x0|\mleq \bound\adi \ade y (y\equals x0)\bigr)  
\end{array}$  \ \  Lemma \ref{combzs} \vspace{3pt}

\noindent 10. $\begin{array}{l}
|u0|\mleq \bound\adi \ade y (y\equals u0)  
\end{array}$  \ \  $\ada$-Elimination: 9 \vspace{3pt}

\noindent 11. $\begin{array}{l}
 \gneg |u0|\mleq \bound \mli t\equals r0   \mli    u\equals s\plus r   \mli  \gneg |s0\plus t|\mleq \bound 
\end{array}$  \ \  $\pa$  \vspace{3pt} (observation (i))

\noindent 12. $\begin{array}{l}
 \gneg |u0|\mleq \bound \mli t\equals r0   \mli    u\equals s\plus r   \mli  |s0\plus t|\mleq \bound\adi \ade z (z\equals s0\plus t)  
\end{array}$  \ \  $\add$-Choose: 11  \vspace{3pt}

\noindent 13. $\begin{array}{l}
w\equals u0  \mli t\equals r0   \mli    u\equals s\plus r   \mli  w\equals s0\plus t   
\end{array}$  \ \  $\pa$ (observation (i))  \vspace{3pt}

\noindent 14. $\begin{array}{l}
w\equals u0  \mli t\equals r0   \mli    u\equals s\plus r   \mli  \ade z (z\equals s0\plus t)  
\end{array}$  \ \ $\ade$-Choose: 13   \vspace{3pt}

\noindent 15. $\begin{array}{l}
\ade y (y\equals u0) \mli t\equals r0   \mli    u\equals s\plus r   \mli  \ade z (z\equals s0\plus t)  
\end{array}$  \ \ $\ada$-Introduction: 14   \vspace{3pt}

\noindent 16. $\begin{array}{l}
\ade y (y\equals u0) \mli t\equals r0   \mli    u\equals s\plus r   \mli  |s0\plus t|\mleq \bound\adi \ade z (z\equals s0\plus t)  
\end{array}$  \ \  $\add$-Choose: 15  \vspace{3pt}

\noindent 17. $\begin{array}{l}
 |u0|\mleq \bound\adi \ade y (y\equals u0) \mli t\equals r0   \mli    u\equals s\plus r   \mli  |s0\plus t|\mleq \bound\adi \ade z (z\equals s0\plus t)  
\end{array}$  \ \  $\adc$-Introduction: 12,16 \vspace{3pt}

\noindent 18. $\begin{array}{l}
 t\equals r0   \mli    u\equals s\plus r   \mli  |s0\plus t|\mleq \bound\adi \ade z (z\equals s0\plus t)  
\end{array}$  MP: 10,17\ \    \vspace{3pt}

\noindent 19. $\begin{array}{l}
\ada x (|x1|\mleq \bound\adi \ade y (y\equals x1)\bigr)  
\end{array}$  \ \  Lemma \ref{combos} \vspace{3pt}

\noindent 20. $\begin{array}{l}
|u1|\mleq \bound\adi \ade y (y\equals u1)  
\end{array}$  \ \  $\ada$-Elimination: 19 \vspace{3pt}

\noindent 21. $\begin{array}{l}
 \gneg |u1|\mleq \bound \mli t\equals r1   \mli    u\equals s\plus r   \mli  \gneg |s0\plus t|\mleq \bound 
\end{array}$  \ \  \pa  \vspace{3pt} (observation (ii))

\noindent 22. $\begin{array}{l}
 \gneg |u1|\mleq \bound \mli t\equals r1   \mli    u\equals s\plus r   \mli  |s0\plus t|\mleq \bound\adi \ade z (z\equals s0\plus t)  
\end{array}$  \ \  $\add$-Choose: 21  \vspace{3pt}

\noindent 23. $\begin{array}{l}
w\equals u1  \mli t\equals r1   \mli    u\equals s\plus r   \mli  w\equals s0\plus t   
\end{array}$  \ \  $\pa$ (observation (ii))  \vspace{3pt}

\noindent 24. $\begin{array}{l}
w\equals u1  \mli t\equals r1   \mli    u\equals s\plus r   \mli  \ade z (z\equals s0\plus t)  
\end{array}$  \ \ $\ade$-Choose: 23   \vspace{3pt}

\noindent 25. $\begin{array}{l}
\ade y (y\equals u1) \mli t\equals r1   \mli    u\equals s\plus r   \mli  \ade z (z\equals s0\plus t)  
\end{array}$  \ \ $\ada$-Introduction: 24   \vspace{3pt}

\noindent 26. $\begin{array}{l}
\ade y (y\equals u1) \mli t\equals r1   \mli    u\equals s\plus r   \mli  |s0\plus t|\mleq \bound\adi \ade z (z\equals s0\plus t)  
\end{array}$  \ \  $\add$-Choose: 25  \vspace{3pt}

\noindent 27. $\begin{array}{l}
 |u1|\mleq \bound\adi \ade y (y\equals u1) \mli t\equals r1   \mli    u\equals s\plus r   \mli  |s0\plus t|\mleq \bound\adi \ade z (z\equals s0\plus t)  
\end{array}$  \ \  $\adc$-Introduction: 22,26 \vspace{3pt}

\noindent 28. $\begin{array}{l}
 t\equals r1   \mli    u\equals s\plus r   \mli  |s0\plus t|\mleq \bound\adi \ade z (z\equals s0\plus t)  
\end{array}$  MP: 20,27\ \    \vspace{3pt}

\noindent 29. $\begin{array}{l}
 t\equals r0\add t\equals r1  \mli    u\equals s\plus r   \mli  |s0\plus t|\mleq \bound\adi \ade z (z\equals s0\plus t)  
\end{array}$  \ \ $\adc$-Introduction: 18,28   \vspace{3pt}

\noindent 30. $\begin{array}{l}
 t\equals r0\add t\equals r1  \mli    \ade z (z\equals s\plus r)  \mli  |s0\plus t|\mleq \bound\adi \ade z (z\equals s0\plus t)  
\end{array}$  \ \ $\ada$-Introduction: 29   \vspace{3pt}   \vspace{3pt}

\noindent 31. $\begin{array}{l}
 t\equals r0\add t\equals r1  \mli   |s\plus r|\mleq \bound\adi \ade z (z\equals s\plus r)  \mli  |s0\plus t|\mleq \bound\adi \ade z (z\equals s0\plus t)  
\end{array}$  \ \ $\adc$-Introduction: 8,30  \vspace{3pt}

\noindent 32. $\begin{array}{l}
 t\equals r0\add t\equals r1  \mli \ada y\bigl(|s\plus y|\mleq \bound\adi \ade z (z\equals s\plus y)\bigr) \mli  |s0\plus t|\mleq \bound\adi \ade z (z\equals s0\plus t)  
\end{array}$  \ \ $\ade$-Choose: 31   \vspace{3pt}

\noindent 33. $\begin{array}{l}
\ade x(t\equals x0\add t\equals x1) \mli \ada y\bigl(|s\plus y|\mleq \bound\adi \ade z (z\equals s\plus y)\bigr) \mli  |s0\plus t|\mleq \bound\adi \ade z (z\equals s0\plus t)  
\end{array}$  \ \  $\ada$-Introduction: 32  \vspace{3pt}

\noindent 34. $\begin{array}{l}
\ade x(t\equals x0\add t\equals x1)  
\end{array}$  \ \  Axiom 12  \vspace{3pt}

\noindent 35. $\begin{array}{l}
\ada y\bigl(|s\plus y|\mleq \bound\adi \ade z (z\equals s\plus y)\bigr) \mli  |s0\plus t|\mleq \bound\adi \ade z (z\equals s0\plus t)  
\end{array}$  \ \  MP: 34,33  \vspace{3pt}

\noindent 36. $\begin{array}{l}
\ada y\bigl(|s\plus y|\mleq \bound\adi \ade z (z\equals s\plus y)\bigr) \mli \ada y\bigl(|s0\plus y|\mleq \bound\adi \ade z (z\equals s0\plus y)\bigr)  
\end{array}$  \ \  $\ada$-Introduction: 35  \vspace{3pt}

\noindent 37. $\begin{array}{l}
 t\equals r0   \mli   \gneg |s\plus r|\mleq \bound  \mli  \gneg |s1\plus t|\mleq \bound 
\end{array}$  \ \ \pa   \vspace{3pt}

\noindent 38. $\begin{array}{l}
 t\equals r1   \mli   \gneg |s\plus r|\mleq \bound  \mli  \gneg |s1\plus t|\mleq \bound 
\end{array}$  \ \ \pa   \vspace{3pt}

\noindent 39. $\begin{array}{l}
 t\equals r0\add t\equals r1  \mli   \gneg |s\plus r|\mleq \bound  \mli  \gneg |s1\plus t|\mleq \bound 
\end{array}$  \ \  $\adc$-Introduction: 37,38  \vspace{3pt}

\noindent 40. $\begin{array}{l}
 t\equals r0\add t\equals r1  \mli   \gneg |s\plus r|\mleq \bound  \mli  |s1\plus t|\mleq \bound\adi \ade z (z\equals s1\plus t)  
\end{array}$  \ \  $\add$-Choose: 39  \vspace{3pt}

 \noindent 41. $\begin{array}{l}
 \gneg |u1|\mleq \bound \mli t\equals r0   \mli    u\equals s\plus r   \mli  \gneg |s1\plus t|\mleq \bound 
\end{array}$  \ \  $\pa$  \vspace{3pt} (observation (iii))

\noindent 42. $\begin{array}{l}
 \gneg |u1|\mleq \bound \mli t\equals r0   \mli    u\equals s\plus r   \mli  |s1\plus t|\mleq \bound\adi \ade z (z\equals s1\plus t)  
\end{array}$  \ \  $\add$-Choose: 41  \vspace{3pt}

\noindent 43. $\begin{array}{l}
w\equals u1  \mli t\equals r0   \mli    u\equals s\plus r   \mli  w\equals s1\plus t   
\end{array}$  \ \  $\pa$ (observation (iii))  \vspace{3pt}

\noindent 44. $\begin{array}{l}
w\equals u1  \mli t\equals r0   \mli    u\equals s\plus r   \mli  \ade z (z\equals s1\plus t)  
\end{array}$  \ \ $\ade$-Choose: 43   \vspace{3pt}

\noindent 45. $\begin{array}{l}
\ade y (y\equals u1) \mli t\equals r0   \mli    u\equals s\plus r   \mli  \ade z (z\equals s1\plus t)  
\end{array}$  \ \ $\ada$-Introduction: 44   \vspace{3pt}

\noindent 46. $\begin{array}{l}
\ade y (y\equals u1) \mli t\equals r0   \mli    u\equals s\plus r   \mli  |s1\plus t|\mleq \bound\adi \ade z (z\equals s1\plus t)  
\end{array}$  \ \  $\add$-Choose: 45  \vspace{3pt}

\noindent 47. $\begin{array}{l}
 |u1|\mleq \bound\adi \ade y (y\equals u1) \mli t\equals r0   \mli    u\equals s\plus r   \mli  |s1\plus t|\mleq \bound\adi \ade z (z\equals s1\plus t)  
\end{array}$  \ \  $\adc$-Introduction: 42,46 \vspace{3pt}

\noindent 48. $\begin{array}{l}
 t\equals r0   \mli    u\equals s\plus r   \mli  |s1\plus t|\mleq \bound\adi \ade z (z\equals s1\plus t)  
\end{array}$  MP: 20,47\ \    \vspace{3pt}

 \noindent 49. $\begin{array}{l}
 \gneg |u1|\mleq \bound \mli t\equals r1   \mli    u\equals s\plus r   \mli  \gneg |s1\plus t|\mleq \bound 
\end{array}$  \ \  $\pa$  \vspace{3pt} (observation (iv))

\noindent 50. $\begin{array}{l}
 \gneg |u1|\mleq \bound \mli t\equals r1   \mli    u\equals s\plus r   \mli  |s1\plus t|\mleq \bound\adi \ade z (z\equals s1\plus t)  
\end{array}$  \ \  $\add$-Choose: 49  \vspace{3pt}

\noindent 51. $\begin{array}{l}
\ada x\bigl(|x\successor|\mleq\bound\adi\ade y(y\equals x\successor)\bigr) 
\end{array}$  \ \   Lemma \ref{comsuc}  \vspace{3pt}

\noindent 52. $\begin{array}{l}
|w\successor|\mleq\bound\adi\ade y(y\equals w\successor)  
\end{array}$  \ \   $\ada$-Elimination: 51  \vspace{3pt}

\noindent 53. $\begin{array}{l}
\gneg |w\successor|\mleq\bound \mli w\equals u1  \mli t\equals r1  \mli    u\equals s\plus r   \mli  \gneg |s1\plus t|\mleq \bound   
\end{array}$  \ \  \pa    \vspace{3pt} (observation (iv))

\noindent 54. $\begin{array}{l}
\gneg |w\successor|\mleq\bound \mli w\equals u1  \mli t\equals r1  \mli    u\equals s\plus r   \mli  |s1\plus t|\mleq \bound\adi \ade z (z\equals s1\plus t)  
\end{array}$  \ \  $\add$-Choose: 53   \vspace{3pt}

\noindent 55. $\begin{array}{l}
v\equals w\successor \mli w\equals u1  \mli t\equals r1  \mli    u\equals s\plus r   \mli   v\equals s1\plus t   
\end{array}$  \ \  \pa   \vspace{3pt} (observation (iv))

\noindent 56. $\begin{array}{l}
v\equals w\successor \mli w\equals u1  \mli t\equals r1  \mli    u\equals s\plus r   \mli  \ade z (z\equals s1\plus t)  
\end{array}$  \ \ $\ade$-Choose: 55    \vspace{3pt}

\noindent 57. $\begin{array}{l}
 \ade y(y\equals w\successor)\mli w\equals u1  \mli t\equals r1  \mli    u\equals s\plus r   \mli  \ade z (z\equals s1\plus t)  
\end{array}$  \ \ $\ada$-Introduction: 56    \vspace{3pt}

\noindent 58. $\begin{array}{l}
 \ade y(y\equals w\successor)\mli w\equals u1  \mli t\equals r1  \mli    u\equals s\plus r   \mli  |s1\plus t|\mleq \bound\adi \ade z (z\equals s1\plus t)  
\end{array}$  \ \   $\add$-Choose: 57  \vspace{3pt}

\noindent 59. $\begin{array}{l}
|w\successor|\mleq\bound\adi\ade y(y\equals w\successor)\mli w\equals u1  \mli t\equals r1  \mli    u\equals s\plus r   \mli  |s1\plus t|\mleq \bound\adi \ade z (z\equals s1\plus t)  
\end{array}$  \ \  $\adc$-Introduction: 54,58   \vspace{3pt}

\noindent 60. $\begin{array}{l}
w\equals u1  \mli t\equals r1  \mli    u\equals s\plus r   \mli  |s1\plus t|\mleq \bound\adi \ade z (z\equals s1\plus t)  
\end{array}$  \ \  MP: 52,59  \vspace{3pt}

\noindent 61. $\begin{array}{l}
 \ade y (y\equals u1) \mli t\equals r1  \mli    u\equals s\plus r   \mli  |s1\plus t|\mleq \bound\adi \ade z (z\equals s1\plus t)  
\end{array}$  \ \  $\ada$-Introduction: 60   \vspace{3pt}

\noindent 62. $\begin{array}{l}
|u1|\mleq \bound\adi \ade y (y\equals u1) \mli t\equals r1  \mli    u\equals s\plus r   \mli  |s1\plus t|\mleq \bound\adi \ade z (z\equals s1\plus t)  
\end{array}$  \ \ $\adc$-Introduction: 50,61   \vspace{3pt}

\noindent 63. $\begin{array}{l}
t\equals r1  \mli    u\equals s\plus r   \mli  |s1\plus t|\mleq \bound\adi \ade z (z\equals s1\plus t)  
\end{array}$  \ \ MP: 20,62   \vspace{3pt}

\noindent 64. $\begin{array}{l}
 t\equals r0\add t\equals r1  \mli    u\equals s\plus r   \mli  |s1\plus t|\mleq \bound\adi \ade z (z\equals s1\plus t)  
\end{array}$  \ \ $\adc$-Introduction: 48,63   \vspace{3pt}

\noindent 65. $\begin{array}{l}
 t\equals r0\add t\equals r1  \mli    \ade z (z\equals s\plus r)  \mli  |s1\plus t|\mleq \bound\adi \ade z (z\equals s1\plus t)  
\end{array}$  \ \ $\ada$-Introduction: 64   \vspace{3pt}   \vspace{3pt}

\noindent 66. $\begin{array}{l}
t\equals r0\add t\equals r1  \mli   |s\plus r|\mleq \bound\adi \ade z (z\equals s\plus r)  \mli   |s1\plus t|\mleq \bound\adi \ade z (z\equals s1\plus t) 
\end{array}$  \ \  $\adc$-Introduction: 40,65  \vspace{3pt}

\noindent 67. $\begin{array}{l}
t\equals r0\add t\equals r1  \mli \ada y\bigl(|s\plus y|\mleq \bound\adi \ade z (z\equals s\plus y)\bigr) \mli   |s1\plus t|\mleq \bound\adi \ade z (z\equals s1\plus t) 
\end{array}$  \ \  $\ada$-Choose: 66  \vspace{3pt}

\noindent 68. $\begin{array}{l}
\ade x(t\equals x0\add t\equals x1) \mli \ada y\bigl(|s\plus y|\mleq \bound\adi \ade z (z\equals s\plus y)\bigr) \mli   |s1\plus t|\mleq \bound\adi \ade z (z\equals s1\plus t) 
\end{array}$  \ \  $\ada$-Introduction: 67  \vspace{3pt}

\noindent 69. $\begin{array}{l}
\ada y\bigl(|s\plus y|\mleq \bound\adi \ade z (z\equals s\plus y)\bigr) \mli   |s1\plus t|\mleq \bound\adi \ade z (z\equals s1\plus t) 
\end{array}$  \ \ MP: 34,68   \vspace{3pt}

\noindent 70. $\begin{array}{l}
\ada y\bigl(|s\plus y|\mleq \bound\adi \ade z (z\equals s\plus y)\bigr) \mli   \ada y\bigl(|s1\plus y|\mleq \bound\adi \ade z (z\equals s1\plus y)\bigr)
\end{array}$  \ \ $\ada$-Introduction: 69   \vspace{3pt}

\noindent 71. $\begin{array}{l}
\ada y\bigl(|s\plus y|\mleq \bound\adi \ade z (z\equals s\plus y)\bigr) \mli \ada y\bigl(|s0\plus y|\mleq \bound\adi \ade z (z\equals s0\plus y)\bigr)  \adc \ada y\bigl(|s1\plus y|\mleq \bound\adi \ade z (z\equals s1\plus y)\bigr)
\end{array}$  

 \hspace{340pt}$\adc$-Introduction: 36,70    \vspace{3pt}

\noindent 72. $\begin{array}{l}
\ada y\bigl(|s\plus y|\mleq \bound\adi \ade z (z\equals s\plus y)\bigr)  
\end{array}$  \ \ BSI: 4,71     \vspace{3pt}

\noindent 73. $\begin{array}{l}
\ada x\ada y\bigl(|x\plus y|\mleq \bound\adi \ade z (z\equals x\plus y)\bigr)  
\end{array}$  \ \ $\ada$-Introduction: 72     \vspace{3pt}
\end{proof}

\subsection{The efficient computability of multiplication}

The following lemma is fully analogous to the lemma of the previous subsection, with the difference that this one is about multiplication instead of addition. Morally, the proof of this lemma is also very similar to the proof of its counterpart. But, as multiplication is somewhat more complex than addition, technically a formal proof here would be considerably longer than the 73-step proof of Lemma \ref{comad}, and producing it  would be no fun. For this reason, we limit ourselves to only an informal proof. As noted earlier,  sooner or later it would be necessary to abandon the luxury of generating formal proofs, anyway.  

\begin{lemma}\label{commul}
 $\arfour\vdash \ada x\ada y\bigl(|x\mult  y|\mleq \bound\adi \ade z (z\equals x\mult  y)\bigr)$.
\end{lemma}

\begin{proof} By BSI induction on $s$, we want to prove \(\ada y\bigl(|s\mult  y|\mleq \bound\adi \ade z (z\equals s\mult  y)\bigr),\) from which the target formula follows by $\ada$-Introduction. 

The {\em basis}
\begin{equation}\label{n24a} 
\ada y\bigl(|\zero\mult  y|\mleq \bound\adi \ade z (z\equals \zero\mult  y)\bigr)
\end{equation}
of induction is simple: for whatever $y$, since $\zero\equals \zero \mult  y$,  the problem $|\zero\mult  y|\mleq \bound\adi \ade z (z\equals \zero\mult  y)$   is resolved by choosing the right $\adi$-component and specifying  $z$ as the value of $\zero$.  Our ability to produce such a value is guaranteed by Axiom 8.

The {\em inductive step} is 
\begin{equation}\label{n24b} 
\ada y\bigl(|s\mult  y|\mleq \bound\adi \ade z (z\equals s\mult  y)\bigr)\mli \ada y\bigl(|s0\mult  y|\mleq \bound\adi \ade z (z\equals s0\mult  y)\bigr)\adc\ada y\bigl(|s1\mult  y|\mleq \bound\adi \ade z (z\equals s1\mult  y)\bigr).
\end{equation} 

In justifying it, we rely on the following facts --- call them ``{\em observations}'' for subsequent references --- provable in {\bf PA}:
\begin{description}
\item[(i)] $s0\mult  r0\equals (s\mult  r)00$, \ \ because $2s\mult 2r\equals 4(s\mult r)$;
\item[(ii)] $s0\mult  r1\equals (s\mult  r)00\plus s0$,  \ \ because $2s\mult (2r\plus 1)\equals 4(s\mult r)\plus 2s$;
\item[(iii)] $s1\mult  r0\equals (s\mult  r)00\plus r0$, \ \ because $(2s\plus 1)\mult 2r\equals 4(s\mult r)\plus 2r$;
\item[(iv)] $s1\mult  r1\equals (s\mult  r)00\plus (s\plus r)1$, \ \ because $(2s\plus 1)\mult (2r\plus 1)\equals 4(s\mult r)\plus \bigl(2(s\plus r)\plus 1\bigr)$.
\end{description} 

In resolving (\ref{n24b}), at  the beginning we wait till the environment selects one of the two $\adc$-conjuncts in the consequent, and also a value $t$ for $y$ there. What we see as a ``beginning'' here is, in fact, the end of the proof of (\ref{n24b}) for, as pointed out in Section \ref{sss16}, such proofs correspond to  winning strategies only when they are read bottom-up. And, as we know,  the steps corresponding to selecting a $\adc$-conjunct and selecting $t$ for $y$ are (bottom-up) $\adc$-Introduction and $\ada$-Introduction. Then, using Axiom 12, we   find the binary predecessor $r$ of $t$. Furthermore, the same axiom will simultaneously allow us to tell whether $t\equals r0$ or $t\equals r1$. We immediately specify (bottom-up $\ade$-Choose) $y$ as $r$ in the antecedent of (\ref{n24b}). We thus have the following four possibilities to consider now, depending on whether the left or the right $\adc$-conjunct was selected in the consequent of (\ref{n24b}), and whether $t\equals r0$ or $t\equals r1$. In each case we will have a different problem to resolve.\vspace{5pt}

{\em Case 1}: The problem to resolve (essentially) is   
\begin{equation}\label{c1} 
\ |s\mult  r|\mleq \bound\adi \ade z (z\equals s\mult  r) \mli  |s0\mult  r0|\mleq \bound\adi \ade z (z\equals s0\mult  r0)  .
\end{equation} 
Pretending for a while --- for simplicity --- that no values that we are going to deal with have sizes exceeding $\bound$, here is our strategy.  Using the resource provided by the antecedent of (\ref{c1}), we find the  product $w$ of $s$ and $r$. Then, using the resource provided by Lemma \ref{combzs} (which, unlike the resource provided by the antecedent of (\ref{c1}), comes in an unlimited supply) twice, we find the value $v$ of $w00$, i.e. of $(s\times r)00$. In view of observation (i), that very $v$ will be (equal to) $s0\mult r0$, so (\ref{c1}) can be resolved by choosing the right $\adi$-component in its consequent and specifying $z$ as $v$. 

The above, however, was a simplified scenario. In a complete scenario without ``cheating'', what may happen is that, while using the antecedent of 
(\ref{c1}) in computing $s\times r$, or while --- after that --- using Lemma \ref{combzs} in (first) computing $(s\times r)0$ and (then) $(s\times r)00$, we   discover that the size of the to-be-computed value exceeds $\bound$ and hence the corresponding resource (the antecedent of (\ref{c1}), or Lemma \ref{combzs}) does not really allow us to compute that value. Such a corresponding resource, however, does allow us to tell that the size of the sought value has exceeded $\bound$. And, in that case, (\ref{c1}) is resolved by choosing the left component $\gneg |s0\mult  r0|\mleq \bound$ of its consequent.\vspace{5pt}  

{\em Case 2}: The problem to resolve is   
\begin{equation}\label{c2} 
\ |s\mult  r|\mleq \bound\adi \ade z (z\equals s\mult  r) \mli  |s0\mult  r1|\mleq \bound\adi \ade z (z\equals s0\mult  r1)  .
\end{equation} 
Here and in the remaining cases, as was done in the first paragraph of Case 1, we will continue pretending that  no values that we  deal with have sizes exceeding $\bound$. Violations of this simplifying assumption  will be handled in the way explained in the second paragraph of Case 1.

Here, we fist compute (the value of) $(s\times r)00$ exactly as we did in Case 1. Exploiting Lemma \ref{combzs} one more time, we also compute $s0$. Using these values, we then employ Lemma \ref{comad} to compute $(s\times r)00\plus s0$, and use the computed value to specify $z$ in the consequent of (\ref{c2}) (after first choosing the right $\adi$-component there, of course). Observation (ii) guarantees success.\vspace{5pt}

{\em Case 3}: The problem to resolve is   
\begin{equation}\label{c3} 
\ |s\mult  r|\mleq \bound\adi \ade z (z\equals s\mult  r) \mli  |s1\mult  r0|\mleq \bound\adi \ade z (z\equals s1\mult  r0)  .
\end{equation} 

This case is very similar to the previous one, with the only difference that Lemma \ref{combzs} will be used to compute $r0$ rather than $s0$, and the success of the strategy will be guaranteed by observation (iii) rather than (ii).\vspace{5pt} 

{\em Case 4}: The problem to resolve is   
\begin{equation}\label{c4} 
\ |s\mult  r|\mleq \bound\adi \ade z (z\equals s\mult  r) \mli  |s1\mult  r1|\mleq \bound\adi \ade z (z\equals s1\mult  r1)  .
\end{equation} 
First, we  compute  $(s\times r)00$ exactly as  in Case 1. Using Lemma \ref{comad}, we also compute $s\plus r$ and then, using Lemma \ref{combos}, compute $(s\plus r)1$. With the values of $(s\times r)00$ and $(s\plus r)1$ now known, Lemma \ref{comad} allows us to compute the value of $(s\times r)00\plus (s\plus r)1$. Finally, using the resulting value to specify $z$ in the consequent of (\ref{c4}), we achieve success. It is guaranteed by observation (iv).  
\end{proof}

\subsection{The efficient computability of all explicitly polynomial functions} By ``explicitly polynomial functions'' in the title of this subsection we mean functions represented by terms of the language of $\arfour$. Such functions are ``explicitly polynomial'' because they, along with variables, are only allowed to use $\zero$, $\successor$, $\plus$ and $\times$.

\begin{lemma}\label{com}
For any\footnote{In view of Convention \ref{jan26}, it is implicitly assumed here that $\tau$ does not contain $z$, for otherwise the formula would have both bound and free occurrences of $z$. Similarly, since $z$ is quantified, it cannot be $\bound$.} term $\tau$, 
 $\arfour\vdash  |\tau|\mleq \bound\adi \ade z (z\equals \tau) $.
\end{lemma}

\begin{proof} We prove this lemma by (meta)induction on the complexity of $\tau$. The following Cases 1 and 2 comprise the basis of this induction, and   Cases 3-5 the inductive step.\vspace{10pt}

{\em Case 1}: $\tau$ is   a variable $t$.  In this case the formula $|\tau|\mleq \bound\adi \ade z (z\equals \tau)$, i.e. $|t|\mleq \bound\adi \ade z (z\equals t)$, immediately follows from the logical axiom $t\equals t$ by $\ade$-Choose and then $\add$-Choose.\vspace{10pt}

{\em Case 2}: $\tau$ is   $\zero$.  Then  $|\zero|\mleq \bound\adi \ade z (z\equals \zero)$  follows in a single step from Axiom 8 by   $\add$-Choose.\vspace{10pt}

{\em Case 3}: $\tau$ is $ \theta\successor $ for some term $\theta$. By the induction hypothesis, $\arfour$ proves  
\begin{equation}\label{d1}
 |\theta|\mleq \bound\adi \ade z (z\equals \theta). 
\end{equation}
Our goal is to establish the $\arfour$-provability of $ |{\theta}\successor |\mleq \bound\adi \ade z (z\equals \theta\successor ) $, which is done as follows:\vspace{9pt}

\noindent 1. $\begin{array}{l} 
 \ada x\bigl(|x\successor|\mleq\bound\adi\ade y(y\equals x\successor)\bigr) 
\end{array}$  \ \  Lemma \ref{comsuc}  \vspace{3pt}

\noindent 2. $\begin{array}{l} 
 \gneg |\theta|\mleq \bound   \mli  \gneg |\theta\successor |\mleq \bound  
\end{array}$  \ \ \pa  \vspace{3pt}

\noindent 3. $\begin{array}{l} 
 \gneg |\theta|\mleq \bound \mlc \ada x\bigl(|x\successor|\mleq\bound\adi\ade y(y\equals x\successor)\bigr)\mli  \gneg |\theta\successor |\mleq \bound\ 
\end{array}$  \ \  Weakening: 2 \vspace{3pt}

\noindent 4. $\begin{array}{l} 
 \gneg |\theta|\mleq \bound \mlc \ada x\bigl(|x\successor|\mleq\bound\adi\ade y(y\equals x\successor)\bigr)\mli  |\theta\successor |\mleq \bound\adi \ade z (z\equals \theta\successor )
\end{array}$  \ \ $\add$-Choose: 3   \vspace{3pt}

\noindent 5. $\begin{array}{l} 
s\equals \theta \mlc  \gneg |s\successor|\mleq\bound  \mli  \gneg |\theta\successor |\mleq \bound 
\end{array}$  \ \ Logical axiom   \vspace{3pt}

\noindent 6. $\begin{array}{l} 
s\equals \theta \mlc  \gneg |s\successor|\mleq\bound  \mli  |\theta\successor |\mleq \bound\adi \ade z (z\equals \theta\successor )
\end{array}$  \ \  $\add$-Choose: 5  \vspace{3pt}

\noindent 7. $\begin{array}{l} 
s\equals \theta \mlc  t\equals s\successor  \mli    t\equals \theta\successor  
\end{array}$  \ \  Logical axiom  \vspace{3pt}

\noindent 8. $\begin{array}{l} 
s\equals \theta \mlc  t\equals s\successor  \mli    \ade z (z\equals \theta\successor )
\end{array}$  \ \  $\ade$-Choose: 7  \vspace{3pt}

\noindent 9. $\begin{array}{l} 
s\equals \theta \mlc   \ade y(y\equals s\successor) \mli    \ade z (z\equals \theta\successor )
\end{array}$  \ \  $\ada$-Introduction: 8  \vspace{3pt}

\noindent 10. $\begin{array}{l} 
s\equals \theta \mlc   \ade y(y\equals s\successor) \mli  |\theta\successor |\mleq \bound\adi \ade z (z\equals \theta\successor )
\end{array}$  \ \ $\add$-Choose: 9  \vspace{3pt}   

\noindent 11. $\begin{array}{l} 
s\equals \theta \mlc  \bigl(|s\successor|\mleq\bound\adi\ade y(y\equals s\successor)\bigr) \mli  |\theta\successor |\mleq \bound\adi \ade z (z\equals \theta\successor )
\end{array}$  \ \ $\adc$-Introduction: 6,10   \vspace{3pt}

\noindent 12. $\begin{array}{l} 
s\equals \theta \mlc \ada x\bigl(|x\successor|\mleq\bound\adi\ade y(y\equals x\successor)\bigr)\mli  |\theta\successor |\mleq \bound\adi \ade z (z\equals \theta\successor )
\end{array}$  \ \  $\ade$-Choose: 11  \vspace{3pt}

\noindent 13. $\begin{array}{l} 
\ade z (z\equals \theta)\mlc \ada x\bigl(|x\successor|\mleq\bound\adi\ade y(y\equals x\successor)\bigr)\mli  |\theta\successor |\mleq \bound\adi \ade z (z\equals \theta\successor )
\end{array}$  \ \  $\ada$-Introduction: 12  \vspace{3pt}

\noindent 14. $\begin{array}{l} 
 \bigl(|\theta|\mleq \bound\adi \ade z (z\equals \theta)\bigr)\mlc \ada x\bigl(|x\successor|\mleq\bound\adi\ade y(y\equals x\successor)\bigr)\mli  |\theta\successor |\mleq \bound\adi \ade z (z\equals \theta\successor )
\end{array}$  \ \  $\adc$-Introduction: 4,13  \vspace{3pt}

\noindent 15. $\begin{array}{l} 
 |{\theta}\successor |\mleq \bound\adi \ade z (z\equals \theta\successor )
\end{array}$  \ \ MP: (\ref{d1}),1,14   \vspace{10pt}
 
{\em Case 4}: $\tau$ is $ \theta_{1}\plus\theta_2 $ for some terms $\theta_1$ and $\theta_2$. By the induction hypothesis, $\arfour$ proves  both of the following formulas:
\begin{equation}\label{d4a}
 |\theta_1 |\mleq \bound\adi \ade z (z\equals \theta_1 ); 
\end{equation}
\begin{equation}\label{d4b}
 |\theta_2 |\mleq \bound\adi \ade z (z\equals \theta_2 ). 
\end{equation}
Our goal is to establish the $\arfour$-provability of 
\( |\theta_1 \plus\theta_2|\mleq \bound\adi \ade z (z\equals \theta_1\plus\theta_2 ),\)
which is done as follows:\vspace{9pt}

\noindent 1. $\begin{array}{l} 
\ada x\ada y\bigl(|x\plus y|\mleq\bound\adi\ade z(z\equals x\plus y)\bigr)
\end{array}$  \ \ Lemma \ref{comad}   \vspace{3pt}

\noindent 2. $\begin{array}{l} 
\gneg  |\theta_1|\mleq \bound  \mli  
 \gneg |\theta_1\plus\theta_2|\mleq \bound 
\end{array}$  \ \ \pa   \vspace{3pt}

\noindent 3. $\begin{array}{l} 
\gneg  |\theta_1|\mleq \bound \mli  
|\theta_1\plus\theta_2|\mleq \bound\adi \ade z (z\equals \theta_1\plus\theta_2 )
\end{array}$  \ \  $\add$-Choose:  2  \vspace{3pt}

\noindent 4. $\begin{array}{l} 
\gneg  |\theta_1|\mleq \bound \mlc \bigl( |\theta_2|\mleq \bound\adi \ade z (z\equals \theta_2)\bigr)\mlc \ada x\ada y\bigl(|x\plus y|\mleq\bound\adi\ade z(z\equals x\plus y)\bigr)\\\mli   
 |\theta_1\plus\theta_2|\mleq \bound\adi \ade z (z\equals \theta_1\plus\theta_2 )
\end{array}$  
\ \ \ Weakenings: 3\vspace{3pt}

\noindent 5. $\begin{array}{l} 
 \gneg |\theta_2|\mleq \bound  \mli  
 \gneg |\theta_1\plus\theta_2|\mleq \bound 
\end{array}$  \ \ \pa \vspace{3pt}

\noindent 6. $\begin{array}{l} 
 \gneg |\theta_2|\mleq \bound \mli  
 |\theta_1\plus\theta_2|\mleq \bound\adi \ade z (z\equals \theta_1\plus\theta_2 )
\end{array}$  \ \ $\add$-Choose: 5    \vspace{3pt}

\noindent 7. $\begin{array}{l} 
 \ade z (z\equals \theta_1)\mlc  \gneg |\theta_2|\mleq \bound \mlc \ada x\ada y\bigl(|x\plus y|\mleq\bound\adi\ade z(z\equals x\plus y)\bigr)\mli  
 |\theta_1\plus\theta_2|\mleq \bound\adi \ade z (z\equals \theta_1\plus\theta_2 )
\end{array}$  \  \mbox{Weakenings: 6}\vspace{3pt} 

\noindent 8. $\begin{array}{l} 
t_1\equals \theta_1\mlc  t_2\equals \theta_2\mlc  \gneg |t_1\plus t_2|\mleq\bound \mli  
\gneg  |\theta_1\plus\theta_2|\mleq \bound 
\end{array}$  \ \ Logical axiom \vspace{3pt}

\noindent 9. $\begin{array}{l} 
t_1\equals \theta_1\mlc  t_2\equals \theta_2\mlc  \gneg |t_1\plus t_2|\mleq\bound \mli  
 |\theta_1\plus\theta_2|\mleq \bound\adi \ade z (z\equals \theta_1\plus\theta_2 )
\end{array}$  \ \ $\add$-Choose: 8 \vspace{3pt}

\noindent 10. $\begin{array}{l} 
t_1\equals \theta_1\mlc  t_2\equals \theta_2\mlc   t\equals t_1\plus t_2 \mli  
t\equals \theta_1\plus\theta_2 
\end{array}$  \ \ Logical axiom \vspace{3pt}

\noindent 11. $\begin{array}{l} 
t_1\equals \theta_1\mlc  t_2\equals \theta_2\mlc   t\equals t_1\plus t_2 \mli  
 \ade z (z\equals \theta_1\plus\theta_2 )
\end{array}$  \ \ $\ade$-Choose: 10 \vspace{3pt}

\noindent 12. $\begin{array}{l} 
t_1\equals \theta_1\mlc  t_2\equals \theta_2\mlc   \ade z(z\equals t_1\plus t_2)\mli  
 \ade z (z\equals \theta_1\plus\theta_2 )
\end{array}$  \ \ $\ada$-Introduction: 11 \vspace{3pt}

\noindent 13. $\begin{array}{l} 
t_1\equals \theta_1\mlc  t_2\equals \theta_2\mlc   \ade z(z\equals t_1\plus t_2)\mli  
 |\theta_1\plus\theta_2|\mleq \bound\adi \ade z (z\equals \theta_1\plus\theta_2 )
\end{array}$  \ \ $\add$-Choose: 12 \vspace{3pt}

\noindent 14. $\begin{array}{l} 
t_1\equals \theta_1\mlc  t_2\equals \theta_2\mlc  \bigl(|t_1\plus t_2|\mleq\bound\adi\ade z(z\equals t_1\plus t_2)\bigr)\mli  
 |\theta_1\plus\theta_2|\mleq \bound\adi \ade z (z\equals \theta_1\plus\theta_2 )
\end{array}$  \ \ $\adc$-Introduction: 9,13 \vspace{3pt}

\noindent 15. $\begin{array}{l} 
t_1\equals \theta_1\mlc  t_2\equals \theta_2\mlc \ada x\ada y\bigl(|x\plus y|\mleq\bound\adi\ade z(z\equals x\plus y)\bigr)\mli  
 |\theta_1\plus\theta_2|\mleq \bound\adi \ade z (z\equals \theta_1\plus\theta_2 )
\end{array}$  \ \ $\ade$-Chooses: 14 \vspace{3pt}

\noindent 16. $\begin{array}{l} 
 \ade z (z\equals \theta_1)\mlc   \ade z (z\equals \theta_2)\mlc \ada x\ada y\bigl(|x\plus y|\mleq\bound\adi\ade z(z\equals x\plus y)\bigr)\\ \mli  
 |\theta_1\plus\theta_2|\mleq \bound\adi \ade z (z\equals \theta_1\plus\theta_2 )
\end{array}$  
\ \ \ $\ada$-Introductions: 15\vspace{3pt}

\noindent 17. $\begin{array}{l} 
 \ade z (z\equals \theta_1)\mlc  \bigl(|\theta_2|\mleq \bound\adi \ade z (z\equals \theta_2)\bigr)\mlc \ada x\ada y\bigl(|x\plus y|\mleq\bound\adi\ade z(z\equals x\plus y)\bigr)\\ \mli   
 |\theta_1\plus\theta_2|\mleq \bound\adi \ade z (z\equals \theta_1\plus\theta_2 )
\end{array}$  
\ \ \ $\adc$-Introduction: 7,16 \vspace{3pt}

\noindent 18. $\begin{array}{l} 
\bigl( |\theta_1|\mleq \bound\adi \ade z (z\equals \theta_1)\bigr)\mlc \bigl( |\theta_2|\mleq \bound\adi \ade z (z\equals \theta_2)\bigr)\mlc\\ \ada x\ada y\bigl(|x\plus y|\mleq\bound\adi\ade z(z\equals x\plus y)\bigr)\mli  
 |\theta_1\plus\theta_2|\mleq \bound\adi \ade z (z\equals \theta_1\plus\theta_2 )
\end{array}$  
\ \ \ $\adc$-Introduction: 4,17\vspace{3pt}

\noindent 19. $\begin{array}{l} 
|\theta_1 \plus\theta_2|\mleq \bound\adi \ade z (z\equals \theta_1\plus\theta_2 ) 
\end{array}$  \ \  MP: (\ref{d4a}),(\ref{d4b}),1,18  \vspace{9pt}

{\em Case 5}: $\tau$ is $ \theta_{1}\mult\theta_2 $ for some terms $\theta_1$ and $\theta_2$. Here we only outline a  proof/solution  for the target 
\( |\theta_1 \mult\theta_2|\mleq \bound\adi \ade z (z\equals \theta_1\mult\theta_2 )\). 
Using the induction hypothesis (\ref{d4a}) and Axiom 9,\footnote{Strictly speaking, Axiom 13 or Lemma \ref{zer} will also be needed here to be sure that, if the size of $\theta_1$ exceeds $\bound$, then $\theta_1\notequals \zero$.} we figure our whether $\theta_1\equals \zero$ or not. If $\theta_1\equals\zero$, then $\theta_1\mult\theta_2$ is also $\zero$, and we solve the target  by choosing its right component  $\ade z (z\equals \theta_1\mult\theta_2 )$ and then naming the value of $\zero$ (which is found using Axiom 8) for $z$. Suppose now $\theta_1\notequals\zero$. Then we do for $\theta_2$ the same as what we  did for $\theta_1$, and figure out whether  $\theta_2\equals\zero$ or $\theta_2\notequals\zero$. If $\theta_2\equals\zero$, we solve the target as we did in the case $\theta_1\equals\zero$. Suppose now $\theta_2$, just like $\theta_1$, does not equal to $\zero$. Note that then the proof given in  Case 4 goes through for our present case virtually without any changes, only with ``$\mult$'' instead of ``$\plus$'' and ``Lemma \ref{commul}'' instead of ``Lemma \ref{comad}''. Indeed, the only steps of that proof that would be generally incorrect for $\mult$ instead of $\plus$ are steps 2 and 5. Namely, the formula of step 2 is false when $\theta_2\equals \zero$, and the formula of step 5 is false when $\theta_1\equals \zero$. But, in the case that we are considering, these possibilities have been handled separately and by now are already ruled out.   
 \end{proof}

  \subsection{The efficient computability of subtraction}
The formula of the following lemma, as a computational problem, is about finding the difference $z$  between any two numbers $x$ and $y$ and then telling whether this difference is $x\mminus y$ or $y\mminus x$.  

\begin{lemma}\label{minus}
 $\arfour\vdash \ada x\ada y\ade z(x\equals y\plus z\add y\equals x\plus z)$.  
\end{lemma}

\begin{proof} As we did in the case of Lemma \ref{commul}, showing a proof idea or sketch instead of  a detailed formal proof would be sufficient here. By BSI+ induction on $s$, we want to prove \(\ada y\ade z(s\equals y\plus z\add y\equals s\plus z),\) from which the target formula follows by $\ada$-Introduction. 

The {\em basis}
\begin{equation}\label{dn24a} 
\ada y\ade z(\zero\equals y\plus z\add y\equals \zero\plus z)
\end{equation}
of induction is proven as follows:\vspace{9pt}

\noindent 1. $\begin{array}{l} 
t\equals \zero\plus t
\end{array}$ \ \  $\pa$\vspace{3pt}

\noindent 2. $\begin{array}{l} 
\zero\equals t\plus t\add t\equals \zero\plus t
\end{array}$ \ \  $\add$-Choose: 1 \vspace{3pt}

\noindent 3. $\begin{array}{l} 
\ade z(\zero\equals t\plus z\add t\equals \zero\plus z)
\end{array}$ \ \  $\ade$-Choose: 2 \vspace{3pt}

\noindent 4. $\begin{array}{l} 
\ada y\ade z(\zero\equals y\plus z\add y\equals \zero\plus z)
\end{array}$ \ \  $\ada$-Introduction: 3 \vspace{9pt}

The {\em inductive step} is 
\begin{equation}\label{dn24b} 
|s0|\mleq\bound\mlc \ada y\ade z(s\equals y\plus z\add y\equals s\plus z)\mli \ada y\ade z(s0\equals y\plus z\add y\equals s0\plus z)\adc \ada y\ade z(s1\equals y\plus z\add y\equals s1\plus z).
\end{equation} 
To prove (\ref{dn24b}), it would be sufficient to prove the following two formulas, from which (\ref{dn24b}) follows by $\adc$-Introduction:
\begin{equation}\label{dn24c} 
|s0|\mleq\bound\mlc \ada y\ade z(s\equals y\plus z\add y\equals s\plus z)\mli \ada y\ade z(s0\equals y\plus z\add y\equals s0\plus z) ;
\end{equation}
\begin{equation}\label{dn24d} 
|s0|\mleq\bound\mlc \ada y\ade z(s\equals y\plus z\add y\equals s\plus z)\mli   \ada y\ade z(s1\equals y\plus z\add y\equals s1\plus z).
\end{equation}
Let us focus on (\ref{dn24c}) only, as the case with (\ref{dn24d}) is similar. (\ref{dn24c}) follows from the following formula by $\ada$-Introduction:
\begin{equation}\label{dn24e} 
|s0|\mleq\bound\mlc \ada y\ade z(s\equals y\plus z\add y\equals s\plus z)\mli  \ade z(s0\equals t\plus z\add t\equals s0\plus z).
\end{equation}
A strategy for the above, which can eventually be translated into a bottom-up $\arfour$-proof, is the following. Using Axiom 12, we find the binary predecessor $r$ of $t$, and also determine whether $t\equals r0$ or $t\equals r1$. 

Consider the case of $t\equals r0$.  Solving (\ref{dn24e}) in this case essentially means solving 
\begin{equation}\label{dn24eeee} 
|s0|\mleq\bound\mlc \ada y\ade z(s\equals y\plus z\add y\equals s\plus z)\mli  \ade z(s0\equals r0\plus z\add r0\equals s0\plus z).
\end{equation}
 We can solve the above by using the second conjunct of the antecedent (specifying $y$ as $r$ in it) to find a $w$ such that $s\equals r\plus w$ or $r\equals s\plus w$, with ``or'' here being a choice one, meaning that we will actually know which of the two alternatives is the case. Let us say the case is $s\equals r\plus w$ (with the other case being similar). From $\pa$ we know that, if $s\equals r\plus w$, then $s0\equals r0\plus w0$. So, in order to solve the consequent of (\ref{dn24eeee}), it would be sufficient to specify $z$ as the value $u$ of $w0$, and then choose the left $\add$-disjunct   $s0\equals r0\plus u$ of the resulting formula. Such a $u$ can be computed using Axiom 11: $|w0|\mleq\bound\mli\ade x(x\equals w0)$, whose antecedent is true because, according to the first conjunct of the antecedent of (\ref{dn24eeee}), the size of $s0$ --- and hence of $w0$ --- does not exceed $\bound$.

The remaining case of $t\equals r1$ is similar, but it 
additionally requires proving $\ada x\bigl(x\notequals\zero\adi\ade y (x\equals y\successor)\bigr)$ (the efficient computability of unary predecessor), doing which is left as an exercise for the reader.
\end{proof}

 \subsection{The efficient computability of ``$x$'s $y$th  bit''}
For a natural numbers $n$ and $i$ --- as always identified with the corresponding binary numerals --- we will write $(n)_i\equals 0$ for a formula saying that $|n|\mgreater i $ and bit $\# i$ of $n$ is $0$. Similarly for $(n)_i\equals 1$. In either case the count of the bits of $n$ starts from $0$ rather than $1$, and proceeds from left to right rather than (as more common in the literature) from right to left. So, for instance, if $n=100$, then $1$ is its bit $\#0$, and the $0$s are its bits $\#1$ and $\#2$.

\begin{lemma}\label{combit}
 $\arfour\vdash \ada x\ada y \bigl(|x|\mgreater y\adi (x)_y\equals 0\add (x)_y\equals 1\bigr) $.  
\end{lemma}

\begin{proof} We limit ourselves to providing an informal argument within $\arfour$. The target formula follows by $\ada$-Introduction from
$\ada y \bigl(|s|\mgreater y\adi (s)_y\equals 0\add (s)_y\equals 1\bigr) $, and the latter we prove by BSI.  

The basis of induction is 
\begin{equation}\label{f13a}
\ada y \bigl(|\zero|\mgreater y\adi (\zero)_y\equals 0\add (\zero)_y\equals 1\bigr).
\end{equation}
Solving it is easy. Given any $y$, using Axiom 9, figure out whether $y\equals \zero$ or $y\notequals\zero$. If $y\equals \zero$, then resolve (\ref{f13a}) by choosing $(\zero)_y\equals 0$ in it. Otherwise, choose $\gneg |\zero|\mgreater y$.

The inductive step is 
\begin{equation}\label{f13b}
\begin{array}{l}
\ada y \bigl(|s|\mgreater y\adi (s)_y\equals 0\add (s)_y\equals 1\bigr)\mli \\
\ada y \bigl(|s0|\mgreater y\adi (s0)_y\equals 0\add (s0)_y\equals 1\bigr)\adc \ada y \bigl(|s1|\mgreater y\adi (s1)_y\equals 0\add (s1)_y\equals 1\bigr) .
\end{array}\end{equation}
Solving it is not hard, either. It means solving the following two problems,  from which (\ref{f13b}) follows by first applying $\ade$-Choose, then $\ada$-Introduction and then $\adc$-Introduction:
\begin{equation}\label{f13c}
\begin{array}{l}
|s|\mgreater r \adi (s)_r\equals 0\add (s)_r\equals 1 \mli 
|s0|\mgreater r\adi (s0)_r\equals 0\add (s0)_r\equals 1 ;
\end{array}\end{equation}
\begin{equation}\label{f13d}
\begin{array}{l}
 |s|\mgreater r \adi (s)_r\equals 0\add (s)_r\equals 1 \mli 
|s1|\mgreater r\adi (s1)_r\equals 0\add (s1)_r\equals 1 .
\end{array}\end{equation}

To solve (\ref{f13c}), wait till the environment selects one of the three $\add$-disjuncts in the antecedent. If $(s)_r\equals 0$ is selected, then select $(s0)_r\equals 0$ in the consequent and you are done. Similarly, if $(s)_r\equals 1$ is selected, then select $(s0)_r\equals 1$ in the consequent. Suppose now
$\gneg |s|\mgreater r $ is selected.  In this case, using Lemma \ref{comlen}, find the value of $|s|$ and then, using Lemma \ref{nov18a}, figure out whether 
$|s|\equals r$ or $|s|\notequals r$. If  $|s|\equals r$, then select  $(s0)_r\equals 0$ in the consequent of (\ref{f13c}); otherwise, if $|s|\notequals r$, select  
$\gneg |s0|\mgreater r$ there. 

The problem  (\ref{f13d}) is solved in a similar way, with the difference that, where in the previous case we selected $(s0)_r\equals 0$, now $ (s1)_r\equals 1 $ should be selected.
\end{proof}

\section{Two more induction rules}
This section establishes the closure of $\arfour$ under  two additional variations of the PTI and WPTI rules. These variations are not optimal as they could be made  stronger,\footnote{For instance,  the rule of Lemma \ref{noct17e} can be easily strengthened by weakening the consequent of its right premise to the more PTI-style 
$E(t\successor)\adc (F(t\successor)\mlc E(t) ) $, and/or strengthening the antecedent of that premise by adding the conjuncts $R$ and  $w\mleq t\mless\tau$ (on the additional condition that $t$ does not occur in $R$).} nor are they natural enough to deserve special names. But these two rules, exactly in their present forms, will be relied upon later in Section \ref{s19}. Thus, the present section is a purely technical one, and a less technically-minded reader may want to omit the proofs of its results.  
 
\begin{lemma}\label{noct17e}
The following rule is admissible in $\arfour$:
\[\frac{R\mli E( w )\mlc F( w )\hspace{20pt}|t\successor|\mleq\bound\mlc E(t)\mlc F(t)\mli E(t)\mlc\bigl(E(t\successor)\adc F(t\successor)\bigr)}{R\mlc w \mleq t\mleq \tau\mli E(t)\mlc F(t)},\]
where $R$ is any elementary formula, $w$ is any variable, $t$ is any variable other than $\bound$,   $\tau$ is any $\bound$-term, $E(t),F(t)$ are any formulas,  $E( w )$ (resp. $E(t\successor)$) is the result of replacing in $E(t)$ all free occurrences of $t$ by $ w $ (resp. $t\successor$), and similarly for $F( w ),F(t\successor)$. 
\end{lemma}

\begin{idea} We manage reduce this rule to PTI  by taking $R\mlc | w \plus s|\mleq\bound \mli E( w \plus s)$ and $R\mlc | w \plus s|\mleq\bound \mli F( w \plus s)$ in the roles of the formulas $E(s)$ and $F(s)$ of the latter.
\end{idea}

\begin{proof} Assume all conditions of the rule, and assume its premises   are provable, i.e., 
\begin{equation}\label{j14a}
\arfour\vdash R\mli E( w )\mlc F( w );
\end{equation}
\begin{equation}\label{j14b}
\arfour\vdash |t\successor|\mleq\bound\mlc E(t)\mlc F(t)\mli E(t)\mlc\bigl(E(t\successor)\adc F(t\successor)\bigr).
\end{equation}
Our goal is to show that $\arfour\vdash R\mlc w \mleq t\mleq \tau\mli E(t)\mlc F(t)$.

Let us agree on the following abbreviations:
\[\tilde{E}(s)\ = \ R\mlc | w \plus s|\mleq\bound \mli E( w \plus s); \hspace{25pt} \tilde{F}(s)\ = \ R\mlc | w \plus s|\mleq\bound \mli F( w \plus s).\]

As easily seen, we have
\[\clfour \vdash f\equals w\mlc \bigl(p\mli P(w)\mlc Q(w)\bigr)\mli \bigl(p\mlc q\mli P(f)\bigr)\mlc \bigl(p\mlc q\mli Q(f)\bigr) \]
and hence, by $\clfour$-Instantiation,  
\[\arfour \vdash  w \plus\zero\equals  w \mlc \bigl(R\mli E( w )\mlc F( w )\bigr)\mli \bigl(R\mlc | w \plus \zero|\mleq\bound\mli E( w \plus\zero)\bigr)\mlc \bigl(R\mlc | w \plus \zero|\mleq\bound\mli F( w \plus\zero)\bigr). \]
The above, together with (\ref{j14a}) and the obvious fact $\pa\vdash  w \plus\zero\equals  w $, by Modus Ponens, yields 
\[\arfour\vdash \bigl(R\mlc | w \plus \zero|\mleq\bound \mli E( w \plus\zero)\bigr)\mlc\bigl(R\mlc | w \plus \zero|\mleq\bound \mli F( w \plus\zero)\bigr),\]
i.e., using our abbreviations, 
\begin{equation}\label{j14c}
\arfour\vdash \tilde{E}(\zero)\mlc \tilde{F}(\zero). 
\end{equation}

With a little effort, the following can be seen to be a valid formula of classical logic: 
\[ \begin{array}{l} 
\bigl(|t\successor|\mleq\bound\mlc p_1(t)\mlc q_1(t)\mli p_2(t)\mlc q_2(t\successor) \bigr)\mli \\
t\equals  w \plus s  \mli 
 | w \plus s\successor|\mleq\bound \mlc | w \plus s|\mleq\bound \mlc  ( w \plus s)\successor \equals w\plus s\successor   \mli \\
 \bigl(R\mlc | w \plus s|\mleq\bound \mli  p_1( w \plus s)\bigr)\mlc \bigl(R\mlc | w \plus s|\mleq\bound \mli q_1( w \plus s)\bigr) \mli\\
 \bigl(R\mlc | w \plus s|\mleq\bound \mli p_2( w \plus s)\bigr)\mlc \bigl(R\mlc | w \plus s\successor |\mleq\bound \mli q_2( w \plus s\successor)\bigr). 
\end{array} \]
Applying Match four times to the above formula, we find that $\clfour$ proves
\begin{equation}\label{j14f}
\begin{array}{l} 
 \bigl(|t\successor|\mleq\bound\mlc P_1(t)\mlc Q_1(t)\mli P_2(t)\mlc Q_2(t\successor) \bigr)\mli \\
t\equals  w \plus s  \mli 
 | w \plus s\successor|\mleq\bound \mlc | w \plus s|\mleq\bound \mlc ( w \plus s)\successor \equals w\plus s\successor   \mli \\
 \bigl(R\mlc | w \plus s|\mleq\bound \mli  P_1( w \plus s)\bigr)\mlc \bigl(R\mlc | w \plus s|\mleq\bound \mli Q_1( w \plus s)\bigr) \mli\\
 \bigl(R\mlc | w \plus s|\mleq\bound \mli P_2( w \plus s)\bigr)\mlc \bigl(R\mlc | w \plus s\successor |\mleq\bound \mli Q_2( w \plus s\successor)\bigr). 
\end{array}
\end{equation}

Now we claim that 
\begin{equation}\label{j14d}
\begin{array}{l}
\arfour\vdash s\mleq\tau\mli \tilde{E}(s)\mlc\tilde{F}(s) .
\end{array}
\end{equation}
Below comes a justification of this claim:\vspace{7pt}

\noindent 1. $\begin{array}{l} 
\gneg | w \plus s\successor |\mleq\bound\add \ade z(z\equals  w \plus s\successor )
\end{array}$ \ \  Lemma \ref{com} \vspace{3pt}

\noindent 2. $\begin{array}{l} 
 \gneg | w \plus s\successor  | \mleq\bound \mli \tilde{E}(s)\mlc\tilde{F}(s) \mli  \tilde{E}(s)\mlc\Bigl(\bigl(R\mlc | w \plus s\successor|\mleq\bound \mli E( w \plus s\successor)\bigr)\adc \bigl(R\mlc | w \plus s\successor|\mleq\bound \mli F( w \plus s\successor)\bigr)\Bigr)
\end{array}$ \ \  

\hspace{95pt} $\clfour$-Instantiation, instance of  $ \gneg p \mli P \mlc Q_1 \mli P \mlc\bigl( (q\mlc p \mli Q_2 )\adc  (q\mlc p \mli Q_3 )\bigr)$\vspace{3pt}

\noindent 3. $\begin{array}{l} 
 \gneg | w \plus s\successor  | \mleq\bound \mli \tilde{E}(s)\mlc\tilde{F}(s) \mli  \tilde{E}(s)\mlc \bigl(\tilde{E}(s\successor) \adc  \tilde{F}(s\successor)\bigr) 
\end{array}$ \ \  abbreviating 2 \vspace{3pt}

\noindent 4. $\begin{array}{l} 
|v|\mleq \bound   
\end{array}$ \ \ Axiom 13 \vspace{3pt}

\noindent 5. $\begin{array}{l} 
|v|\mleq \bound\mli    v\equals  w \plus s\successor  \mli | w \plus s\successor|\mleq\bound \mlc | w \plus s|\mleq\bound \mlc ( w \plus s)\successor \equals w\plus s\successor 
\end{array}$ \ \ $\pa$ \vspace{3pt}

\noindent 6. $\begin{array}{l} 
 v\equals  w \plus s\successor  \mli | w \plus s\successor|\mleq\bound \mlc | w \plus s|\mleq\bound \mlc ( w \plus s)\successor \equals w\plus s\successor   
\end{array}$ \ \ MP: 4,5 \vspace{3pt}

\noindent 7. $\begin{array}{l} 
\gneg | w \plus s  |\mleq\bound\add \ade z(z\equals  w \plus s  )
\end{array}$ \ \  Lemma \ref{com} \vspace{3pt}

\noindent 8. $\begin{array}{l} 
\gneg | w \plus s  |\mleq\bound \mli  | w \plus s\successor|\mleq\bound \mlc | w \plus s|\mleq\bound\mlc ( w \plus s)\successor \equals w\plus s\successor    \mli   \tilde{E}(s)\mlc\tilde{F}(s) \mli \tilde{E}(s)\mlc\bigl(\tilde{E}(s\successor)\adc \tilde{F}(s\successor)\bigr)
\end{array}$ \ \  

\hspace{210pt}$\clfour$-Instantiation, instance of $\gneg p\mli q_1\mlc p\mlc q_2\mli Q$  \vspace{3pt}

\noindent 9. $\begin{array}{l} 
\bigl(|t\successor|\mleq\bound\mlc E(t)\mlc F(t)\mli E(t)\mlc E(t\successor) \bigr)\mli \\
t\equals  w \plus s  \mli 
 | w \plus s\successor|\mleq\bound \mlc | w \plus s|\mleq\bound \mlc ( w \plus s)\successor \equals w\plus s\successor    \mli \\
 \bigl(R\mlc | w \plus s|\mleq\bound \mli  E( w \plus s)\bigr)\mlc \bigl(R\mlc | w \plus s|\mleq\bound \mli F( w \plus s)\bigr) \mli\\
 \bigl(R\mlc | w \plus s|\mleq\bound \mli E( w \plus s)\bigr)\mlc \bigl(R\mlc | w \plus s\successor |\mleq\bound \mli E( w \plus s\successor)\bigr)
\end{array}$ \ \    $\clfour$-Instantiation, instance of (\ref{j14f})  \vspace{3pt}

\noindent 10. $\begin{array}{l} 
\bigl(|t\successor|\mleq\bound\mlc E(t)\mlc F(t)\mli E(t)\mlc E(t\successor) \bigr)\mli\\
t\equals  w \plus s  \mli 
 | w \plus s\successor|\mleq\bound \mlc | w \plus s|\mleq\bound \mlc ( w \plus s)\successor \equals w\plus s\successor    \mli   \tilde{E}(s)\mlc\tilde{F}(s) \mli \tilde{E}(s)\mlc  \tilde{E}(s\successor)
\end{array}$ \ \   abbreviating 9  \vspace{3pt}

\noindent 11. $\begin{array}{l} 
\Bigl(|t\successor|\mleq\bound\mlc E(t)\mlc F(t)\mli E(t)\mlc\bigl(E(t\successor)\adc F(t\successor)\bigr)\Bigr)\mli\\
t\equals  w \plus s  \mli 
 | w \plus s\successor|\mleq\bound \mlc | w \plus s|\mleq\bound \mlc ( w \plus s)\successor \equals w\plus s\successor   \mli   \tilde{E}(s)\mlc\tilde{F}(s) \mli \tilde{E}(s)\mlc  \tilde{E}(s\successor)
\end{array}$ \ \  $\add$-Choose: 10   \vspace{3pt}

\noindent 12. $\begin{array}{l} 
\bigl(|t\successor|\mleq\bound\mlc E(t)\mlc F(t)\mli E(t)\mlc F(t\successor) \bigr)\mli \\
t\equals  w \plus s  \mli 
 | w \plus s\successor|\mleq\bound \mlc | w \plus s|\mleq\bound \mlc ( w \plus s)\successor \equals w\plus s\successor    \mli \\
 \bigl(R\mlc | w \plus s|\mleq\bound \mli  E( w \plus s)\bigr)\mlc \bigl(R\mlc | w \plus s|\mleq\bound \mli F( w \plus s)\bigr) \mli\\
 \bigl(R\mlc | w \plus s|\mleq\bound \mli E( w \plus s)\bigr)\mlc \bigl(R\mlc | w \plus s\successor |\mleq\bound \mli F( w \plus s\successor)\bigr)
\end{array}$ \ \   $\clfour$-Instantiation, instance of (\ref{j14f})  \vspace{3pt}

\noindent 13. $\begin{array}{l} 
\bigl(|t\successor|\mleq\bound\mlc E(t)\mlc F(t)\mli E(t)\mlc F(t\successor) \bigr)\mli\\
t\equals  w \plus s  \mli 
 | w \plus s\successor|\mleq\bound \mlc | w \plus s|\mleq\bound \mlc ( w \plus s)\successor \equals w\plus s\successor    \mli   \tilde{E}(s)\mlc\tilde{F}(s) \mli \tilde{E}(s)\mlc  \tilde{F}(s\successor)
\end{array}$ \ \   abbreviating 12  \vspace{3pt}

\noindent 14. $\begin{array}{l} 
\Bigl(|t\successor|\mleq\bound\mlc E(t)\mlc F(t)\mli E(t)\mlc\bigl(E(t\successor)\adc F(t\successor)\bigr)\Bigr)\mli\\
t\equals  w \plus s  \mli 
 | w \plus s\successor|\mleq\bound \mlc | w \plus s|\mleq\bound \mlc ( w \plus s)\successor \equals w\plus s\successor   \mli   \tilde{E}(s)\mlc\tilde{F}(s) \mli \tilde{E}(s)\mlc  \tilde{F}(s\successor)
\end{array}$ \ \  $\add$-Choose: 10   \vspace{3pt}

\noindent 15. $\begin{array}{l} 
\Bigl(|t\successor|\mleq\bound\mlc E(t)\mlc F(t)\mli E(t)\mlc\bigl(E(t\successor)\adc F(t\successor)\bigr)\Bigr)\mli\\
t\equals  w \plus s  \mli 
 | w \plus s\successor|\mleq\bound \mlc | w \plus s|\mleq\bound \mlc ( w \plus s)\successor \equals w\plus s\successor   \mli\\   \tilde{E}(s)\mlc\tilde{F}(s) \mli \tilde{E}(s)\mlc\bigl(\tilde{E}(s\successor)\adc \tilde{F}(s\successor)\bigr)
\end{array}$ \ \  $\adc$-Introduction: 11,14   \vspace{3pt}

\noindent 16. $\begin{array}{l} 
t\equals  w \plus s  \mli 
 | w \plus s\successor|\mleq\bound \mlc | w \plus s|\mleq\bound \mlc ( w \plus s)\successor \equals w\plus s\successor    \mli    \tilde{E}(s)\mlc\tilde{F}(s) \mli \tilde{E}(s)\mlc\bigl(\tilde{E}(s\successor)\adc \tilde{F}(s\successor)\bigr)
\end{array}$ \ \ MP: (\ref{j14b}),15\vspace{3pt}

\noindent 17. $\begin{array}{l} 
\ade z(z\equals  w \plus s  )\mli \\
 | w \plus s\successor|\mleq\bound \mlc | w \plus s\mleq\bound|\mlc ( w \plus s)\successor \equals w\plus s\successor    \mli   \tilde{E}(s)\mlc\tilde{F}(s) \mli  \tilde{E}(s)\mlc\bigl(\tilde{E}(s\successor)\adc \tilde{F}(s\successor)\bigr)
\end{array}$  \mbox{$\ada$-Introduction: 16}\vspace{3pt}

\noindent 18. $\begin{array}{l} 
 | w \plus s\successor|\mleq\bound \mlc | w \plus s|\mleq\bound \mlc ( w \plus s)\successor \equals w\plus s\successor    \mli   \tilde{E}(s)\mlc\tilde{F}(s) \mli \tilde{E}(s)\mlc\bigl(\tilde{E}(s\successor)\adc \tilde{F}(s\successor)\bigr)
\end{array}$ \ \ $\add$-Elimination: 7,8,17\vspace{3pt}

\noindent 19. $\begin{array}{l} 
v\equals  w \plus s\successor  \mli   \tilde{E}(s)\mlc\tilde{F}(s) \mli \tilde{E}(s)\mlc\bigl(\tilde{E}(s\successor)\adc \tilde{F}(s\successor)\bigr)
\end{array}$ \ \ TR: 6,18 \vspace{3pt}

\noindent 20. $\begin{array}{l} 
  \ade z(z\equals  w \plus s\successor )\mli   \tilde{E}(s)\mlc\tilde{F}(s) \mli \tilde{E}(s)\mlc\bigl(\tilde{E}(s\successor)\adc \tilde{F}(s\successor)\bigr)
\end{array}$ \ \ $\ada$-Introduction: 19 \vspace{3pt}

\noindent 21. $\begin{array}{l} 
\tilde{E}(s)\mlc\tilde{F}(s) \mli \tilde{E}(s)\mlc\bigl(\tilde{E}(s\successor)\adc \tilde{F}(s\successor)\bigr)
\end{array}$ \ \  $\add$-Elimination: 1,3,20 \vspace{3pt}

\noindent 22. $\begin{array}{l} 
\tilde{E}(s)\mlc\bigl(\tilde{E}(s\successor)\adc \tilde{F}(s\successor)\bigr)\mli \tilde{E}(s\successor)\adc\bigl(\tilde{F}(s\successor)\mlc \tilde{E}(s)\bigr)
\end{array}$ \ \  $\clfour$-Instantiation  \vspace{3pt}

\noindent 23. $\begin{array}{l} 
\tilde{E}(s)\mlc\tilde{F}(s) \mli \tilde{E}(s\successor)\adc\bigl(\tilde{F}(s\successor)\mlc \tilde{E}(s)\bigr)
\end{array}$ \ \  TR: 21,22 \vspace{3pt}

\noindent 24. $\begin{array}{l} 
 s\mleq\tau\mli \tilde{E}(s)\mlc\tilde{F}(s)
\end{array}$ \ \  PTI: (\ref{j14c}), 23 \vspace{7pt}

The following is a disabbreviation of (\ref{j14d}):
\[
\begin{array}{l}
\arfour\vdash s\mleq\tau\mli \bigl(R\mlc | w \plus s|\mleq\bound \mli E( w \plus s)\bigr)\mlc\bigl(R\mlc | w \plus s|\mleq\bound \mli F( w \plus s)\bigr) .
\end{array}
\]
It is easy to see that, by $\clfour$-Instantiation, we also have 
\[
\begin{array}{l}
\arfour\vdash \Bigl(s\mleq\tau\mli \bigl(R\mlc | w \plus s|\mleq\bound \mli E( w \plus s)\bigr)\mlc\bigl(R\mlc | w \plus s|\mleq\bound \mli F( w \plus s)\bigr)\Bigr)\mli\\
s\mleq\tau\mlc R\mlc   | w \plus s|\mleq\bound \mli E( w \plus s)\mlc  F( w \plus s) .
\end{array}
\]
Hence, by Modus Ponens,
\begin{equation}\label{j14g}
\begin{array}{l}
\arfour\vdash s\mleq\tau\mlc R\mlc   | w \plus s|\mleq\bound \mli E( w \plus s)\mlc  F( w \plus s) .
\end{array}
\end{equation}

Now, the following sequence is an $\arfour$-proof of the target formula $R\mlc w \mleq t\mleq \tau\mli E(t)\mlc F(t)$, which completes our proof of the present lemma:\vspace{7pt}

\noindent 1. $\begin{array}{l} 
\ada x\ada y\ade z(x\equals y\plus z\add y\equals x\plus z)
\end{array}$ \ \  Lemma \ref{minus} \vspace{3pt}

\noindent 2. $\begin{array}{l} 
\ade z(w\equals t\plus z\add t\equals w\plus z)
\end{array}$ \ \  $\ada$-Elimination (twice): 1 \vspace{3pt}

\noindent 3. $\begin{array}{l} 
 s\equals\zero\add s\notequals\zero
\end{array}$ \ \ Axiom 8  \vspace{3pt}

\noindent 4. $\begin{array}{l} 
 s\equals\zero\mli (w\equals t\plus s   \mli t\equals w )   
\end{array}$ \ \ $\pa$  \vspace{3pt}

\noindent 5. $\begin{array}{l} 
\bigl(R \mli E( w )\mlc F( w )\bigr) \mli(w\equals t\plus s   \mli t\equals w )\mli \bigl(w\equals t\plus s \mli   R \mli E(t)\mlc F(t)\bigr)   
\end{array}$ \ \  $\clfour$-Instantiation\vspace{3pt}

\noindent 6. $\begin{array}{l} 
(w\equals t\plus s    \mli t\equals w )\mli \bigl(w\equals t\plus s \mli   R \mli E(t)\mlc F(t)\bigr)   
\end{array}$ \ \  MP: (\ref{j14a}),5\vspace{3pt}

\noindent 7. $\begin{array}{l} 
 s\equals\zero\mli  w\equals t\plus s \mli  R \mli E(t)\mlc F(t)  
\end{array}$ \ \  TR: 4,6\vspace{3pt}

\noindent 8. $\begin{array}{l} 
 s\equals\zero\mli  w\equals t\plus s  \mli R\mlc w \mleq t\mleq\tau \mli E(t)\mlc F(t)   
\end{array}$ \ \  Weakening: 7  \vspace{3pt}

\noindent 9. $\begin{array}{l} 
 s\notequals\zero\mli w\equals t\plus s   \mli \gneg w \mleq t\mleq \tau
\end{array}$ \ \ $\pa$  \vspace{3pt}

\noindent 10. $\begin{array}{l} 
 s\notequals\zero\mli w\equals t\plus s   \mli   R\mlc w \mleq t\mleq \tau\mli E(t)\mlc F(t)
\end{array}$ \ \ Weakenings: 9  \vspace{3pt}

\noindent 11. $\begin{array}{l} 
w\equals t\plus s   \mli   R\mlc w \mleq t\mleq \tau\mli E(t)\mlc F(t)
\end{array}$ \ \ $\add$-Elimination: 3,8,10  \vspace{3pt}

\noindent 12. $\begin{array}{l} 
|t|\mleq\bound
\end{array}$ \ \ Axiom 13  \vspace{3pt}

\noindent 13. $\begin{array}{l} 
t\equals w\plus s     \mlc w \mleq t\mleq \tau\mli s\mleq\tau
\end{array}$ \ \ $\pa$  \vspace{3pt}

\noindent 14. $\begin{array}{l} 
|t|\mleq\bound\mlc (t\equals w\plus s \mlc   w \mleq t\mleq \tau\mli s\mleq\tau)\mlc 
\bigl(s\mleq\tau\mlc R\mlc   | w \plus s|\mleq\bound \mli E( w \plus s)\mlc  F( w \plus s)\bigr) \\
 \mli t\equals w\plus s   \mli   R\mlc w \mleq t\mleq \tau\mli E(t)\mlc F(t)
\end{array}$ \ \ $\clfour$-Instantiation\vspace{3pt}

\noindent 15. $\begin{array}{l} 
t\equals w\plus s  \mli   R\mlc w \mleq t\mleq \tau\mli E(t)\mlc F(t)
\end{array}$ \ \ MP: 12,13,(\ref{j14g}),14  \vspace{3pt}

\noindent 16. $\begin{array}{l} 
w\equals t\plus s\add t\equals w\plus s     \mli   R\mlc w \mleq t\mleq \tau\mli E(t)\mlc F(t)
\end{array}$ \ \ $\adc$-Introduction: 11,15  \vspace{3pt}

\noindent 17. $\begin{array}{l} 
\ade z(w\equals t\plus z\add t\equals w\plus z)\mli    R\mlc w \mleq t\mleq \tau\mli E(t)\mlc F(t)
\end{array}$ \ \ $\ada$-Introduction: 16  \vspace{3pt}

\noindent 18. $\begin{array}{l} 
 R\mlc w \mleq t\mleq \tau\mli E(t)\mlc F(t)
\end{array}$ \ \ MP: 2,17  \vspace{7pt} \end{proof}

\begin{lemma}\label{noct17ee}
The following rule is admissible in $\arfour$:
\[\frac{R\mli  F( w )\hspace{20pt}R\mlc  w \mleq t\mless \tau \mlc  F(t)\mli  F(t\successor) }{R\mlc w \mleq t\mleq \tau\mli   F(t)},\]
where $R$ is any elementary formula, $w$ is any variable, $t$ is any variable not occurring in $R$ and different from $\bound$, $F(t)$ is any formula,   $\tau$ is any $\bound$-term, and $F( w )$ (resp. $F(t\successor)$)  is the result of replacing in $F(t)$ all free occurrences of $t$ by $ w $ (resp. $t\successor$). 
\end{lemma}

\begin{idea} This rule can be reduced to the rule of Lemma \ref{noct17e}  by taking $\twg$  and $R\mlc w \mleq t\mleq \tau\mli   F(t)$ in the roles of $E(t)$  and $F(t)$ of the latter, respectively.  
\end{idea}

\begin{proof} Assume all conditions of the rule, and assume its premises  are provable, i.e., 
\begin{equation}\label{j14aa}
\arfour\vdash R\mli  F( w );
\end{equation}
\begin{equation}\label{j14ba}
\arfour\vdash R\mlc  w \mleq t\mless \tau \mlc  F(t)\mli  F(t\successor).
\end{equation}
Our goal is to show that $\arfour\vdash R\mlc w \mleq t\mleq \tau\mli   F(t)$.

Let us agree on the following abbreviation:
\[ \tilde{F}(t)\ = \ R\mlc w \mleq t\mleq \tau\mli   F(t).\]

From (\ref{j14aa}), by Weakening, we have  
\begin{equation}\label{j15a}
\arfour\vdash R\mli \tilde{F}( w ).
\end{equation}

We now claim that 
\begin{equation}\label{j15b}
\arfour\vdash |t\successor|\mleq\bound\mlc \twg \mlc \tilde{F}(t)\mli \twg\mlc\bigl(\twg\adc \tilde{F}(t\successor)\bigr).
\end{equation}
This claim is justified a follows:\vspace{7pt}

\noindent 1. $\begin{array}{l} 
\gneg |t\successor|\mleq\bound\add \ade z(z\equals t\successor)
\end{array}$ \ \ Lemma \ref{com}   \vspace{3pt}

\noindent 2. $\begin{array}{l} 
|w|\mleq\bound
\end{array}$ \ \ Axiom 13   \vspace{3pt}

\noindent 3. $\begin{array}{l} 
|w|\mleq\bound\mli \gneg |t\successor|\mleq\bound \mli t\successor\notequals w
\end{array}$ \ \ Logical axiom  \vspace{3pt}

\noindent 4. $\begin{array}{l} 
\gneg |t\successor|\mleq\bound \mli t\successor\notequals w
\end{array}$ \ \ MP: 2,3  \vspace{3pt}

\noindent 5. $\begin{array}{l} 
\gneg |t\successor|\mleq\bound \mli t\successor\equals w\add t\successor\notequals w
\end{array}$ \ \ $\add$-Choose: 4  \vspace{3pt}

\noindent 6. $\begin{array}{l} 
\ada x\ada y(y\equals x\add y\notequals x) 
\end{array}$ \ \  Lemma \ref{nov18a}  \vspace{3pt}

\noindent 7. $\begin{array}{l} 
r\equals w\add r\notequals w 
\end{array}$ \ \ $\ada$-Elimination (twice): 6   \vspace{3pt}

\noindent 8. $\begin{array}{l} 
r\equals w\mli r\equals t\successor  \mli  t\successor\equals w
\end{array}$ \ \ Logical axiom   \vspace{3pt}

\noindent 9. $\begin{array}{l} 
r\equals w\mli r\equals t\successor  \mli t\successor\equals w\add t\successor\notequals w
\end{array}$ \ \  $\add$-Choose: 8  \vspace{3pt}

\noindent 10. $\begin{array}{l} 
r\notequals w\mli r\equals t\successor  \mli  t\successor\notequals w
\end{array}$ \ \  Logical axiom  \vspace{3pt}

\noindent 11. $\begin{array}{l} 
 r\notequals w\mli r\equals t\successor  \mli t\successor\equals w\add t\successor\notequals w
\end{array}$ \ \  $\add$-Choose: 10  \vspace{3pt}

\noindent 12. $\begin{array}{l} 
 r\equals t\successor\mli t\successor\equals w\add t\successor\notequals w
\end{array}$ \ \  $\add$-Elimination: 7,9,11  \vspace{3pt}

\noindent 13. $\begin{array}{l} 
\ade z(z\equals t\successor) \mli t\successor\equals w\add t\successor\notequals w
\end{array}$ \ \  $\ada$-Introduction: 12  \vspace{3pt}

\noindent 14. $\begin{array}{l} 
 t\successor\equals w\add t\successor\notequals w
\end{array}$ \ \  $\add$-Elimination: 1,5,13  \vspace{3pt}

\noindent 15. $\begin{array}{l} 
(R\mli F(w)\bigr)\mli t\successor\equals w\mli \bigl(R\mli F(t\successor)\bigr)
\end{array}$ \ \  $\clfour$-Instantiation \vspace{3pt}

\noindent 16. $\begin{array}{l} 
t\successor\equals w\mli \bigl(R\mli F(t\successor)\bigr)
\end{array}$ \ \ MP: (\ref{j14aa}),15  \vspace{3pt}

\noindent 17. $\begin{array}{l} 
t\successor\equals w\mli\bigl(R\mlc w \mleq t\mleq \tau\mli F(t)\bigr)\mli\bigl(R\mlc w \mleq t\successor\mleq \tau\mli  F(t\successor)\bigr)
\end{array}$ \ \ Weakenings: 16  \vspace{3pt}

\noindent 18. $\begin{array}{l} 
t\successor\notequals w\mli (w \mleq t\successor\mleq \tau\mli   w \mleq t\mleq \tau \mlc w \mleq t\mless \tau) 
\end{array}$ \ \  $\pa$ \vspace{3pt}

\noindent 19. $\begin{array}{l} 
\bigl(R\mlc  w \mleq t\mless \tau \mlc  F(t)\mli  F(t\successor)\bigr)\mli  \\
 ( w \mleq t\successor\mleq \tau\mli   w \mleq t\mleq \tau \mlc w \mleq t\mless \tau ) \mli 
\bigl(R\mlc w \mleq t\mleq \tau\mli F(t)\bigr)\mli\bigl(R\mlc w \mleq t\successor\mleq \tau\mli  F(t\successor)\bigr)
\end{array}$ \ \   

\hspace{34pt}$\clfour$-Instantiation, instance of $ (q\mlc p_3 \mlc  P\mli Q ) \mli  ( p_1\mli  p_2 \mlc p_3 )\mli  
 (q\mlc p_2\mli P )\mli (q\mlc p_1\mli  Q )$ \vspace{3pt}

\noindent 20. $\begin{array}{l} 
 ( w \mleq t\successor\mleq \tau\mli   w \mleq t\mleq \tau \mlc w \mleq t\mless \tau ) \mli 
\bigl(R\mlc w \mleq t\mleq \tau\mli F(t)\bigr)\mli\bigl(R\mlc w \mleq t\successor\mleq \tau\mli  F(t\successor)\bigr)
\end{array}$ \ \  MP: (\ref{j14ba}),19\vspace{3pt}

\noindent 21. $\begin{array}{l} 
t\successor\notequals w\mli\bigl(R\mlc w \mleq t\mleq \tau\mli F(t)\bigr)\mli\bigl(R\mlc w \mleq t\successor\mleq \tau\mli  F(t\successor)\bigr)
\end{array}$ \ \ TR: 18,20   \vspace{3pt}

\noindent 22. $\begin{array}{l} 
\bigl(R\mlc w \mleq t\mleq \tau\mli F(t)\bigr)\mli\bigl(R\mlc w \mleq t\successor\mleq \tau\mli  F(t\successor)\bigr)
\end{array}$ \ \ $\add$-Elimination: 14,17,21  \vspace{3pt}

\noindent 23. $\begin{array}{l} 
\tilde{F}(t)\mli \tilde{F}(t\successor)
\end{array}$ \ \ abbreviating 22 \vspace{3pt}

\noindent 24. $\begin{array}{l} 
 |t\successor|\mleq\bound\mlc \twg \mlc \tilde{F}(t)\mli  \tilde{F}(t\successor)
\end{array}$ \ \ Weakenings: 23 \vspace{3pt}

\noindent 25. $\begin{array}{l} 
 \tilde{F}(t\successor)\mli \twg\mlc\bigl(\twg\adc \tilde{F}(t\successor)\bigr)
\end{array}$ \ \ $\clfour$-Instantiation \vspace{3pt}

\noindent 26. $\begin{array}{l} 
 |t\successor|\mleq\bound\mlc \twg \mlc \tilde{F}(t)\mli \twg\mlc\bigl(\twg\adc \tilde{F}(t\successor)\bigr)
\end{array}$ \ \ TR: 24,25 \vspace{7pt}

From (\ref{j15a}) and (\ref{j15b}), by the rule of Lemma \ref{noct17e}, we get $\arfour\vdash R\mlc w\mleq t\mleq\tau\mli \twg\mlc \tilde{F}(t)$. Of course (by $\clfour$-Instantiation) $\arfour\vdash \twg\mlc\tilde{F}(t)\mli\tilde{F}(t)$, so, by Transitivity,  $\arfour\vdash R\mlc w\mleq t\mleq\tau\mli \  \tilde{F}(t)$. Disabbreviating the latter, we thus have
\[\arfour\vdash R\mlc w\mleq t\mleq\tau\mli R\mlc w\mleq t\mleq\tau\mli  F(t).\]
We also have 
\[\arfour\vdash (R\mlc w\mleq t\mleq\tau\mli R\mlc w\mleq t\mleq\tau\mli  F(t)\bigr)\mli \bigl(R\mlc w\mleq t\mleq\tau\mli  F(t)\bigr) \]
(the above formula is an instance of the obviously $\clfour$-provable $(p\mli p\mli Q)\mli(p\mli Q)$). So, by Modus Ponens, we find that $\arfour$ proves the desired $R\mlc w\mleq t\mleq\tau\mli  F(t)$.
\end{proof}

\section{The extensional completeness of $\arfour$}\label{sectcompl}\label{s19}
This section is devoted to proving the completeness part of Theorem \ref{tt1}. It means showing that, for any arithmetical problem $A$ that has a polynomial time solution,  there is a theorem of $\arfour$ which, under the standard interpretation, equals (``expresses'') $A$. 

So, let us pick an arbitrary polynomial-time-solvable arithmetical problem $A$. By definition, $A$ is an arithmetical problem because, for some formula $X$ of the language of $\arfour$, $A=X^\dagger$. For the rest of this section, we fix such a formula \[X,\label{ix}\]  and fix \[{\cal X}\label{ixxx}\] as an HPM that solves $A$ (and hence $X^\dagger$) in polynomial time. Specifically, we assume that $\cal X$ runs in time 
\[\xi(\bound),\label{ixi}\]
where $\xi(\bound)$, which we also fix for the rest of this section and which sometimes can be written simply as $\xi$, is a $\bound$-term (a term containing no variables other than $\bound$).  

$X$ may not necessarily be provable in $\arfour$, and our goal is to construct another formula $\overline{X}$ for which, just like for $X$, we have $A=\overline{X}^\dagger$ and which, perhaps unlike $X$, is  provable in $\arfour$.

Remember our convention about identifying formulas  of ptarithmetic with (the games that are) their standard interpretations. So, in the sequel, just as we have done so far, we shall typically write $E,F,\ldots$ to mean either $E,F,\ldots$ or $E^\dagger,F^\dagger,\ldots$.  
Similar conventions apply to terms as well. In fact, we have just used this convention when saying that $\cal X$ runs in time $\xi$. What was really meant was that it runs in time $\xi^\dagger$.  

\subsection{Preliminary insights}\label{gggg}
 Our proof is long and, in the process of going through it, it is easy to get lost in the forest and stop seeing it for the trees. Therefore, it might be worthwhile to try to get some preliminary insights into the basic idea behind this proof before venturing into its details.

Let us consider the simplest nontrivial special case  where $X$ is \[Y(x)\add Z(x) \] for some elementary formulas $Y(x)$ and $Z(x)$
(perhaps $Z(x)$ is $\gneg Y(x)$, in which case $X$ expresses   an ordinary decision problem --- the problem of deciding  the predicate $Y(x)$). 

The assertion ``$\cal X$ does not win $X$ in time $\xi$'' can be formalized in the language of $\pa$ through as a certain formula $\mathbb{L}$. Then we let the earlier mentioned  $\overline{X}$ be the formula 
\[\bigl(Y(x)\mld \mathbb{L}\bigr)\add\bigl(Z(x)\mld\mathbb{L}\bigr).\]
Since $\cal X$ {\em does} win game $X$ in time $\xi$, $\mathbb{L}$ is false. Hence $Y(x)\mld \mathbb{L}$ is equivalent to $Y(x)$, and  $Z(x)\mld \mathbb{L}$ is equivalent to $Z(x)$. This means that $\overline{X}$ and $X$, as games, are the same, that is,  $\overline{X}^\dagger=X^\dagger$. It now remains to understand why $\arfour\vdash \overline{X}$.  

A central lemma here is one establishing that the work of $\cal X$ is ``{\em provably traceable}''.\label{iprovtr1} Roughly, this means the provability of the fact that, for any time moment $t\mleq \xi(\bound)$, we can tell (``can tell'' formally indicated with $\add$ or $\ade$ applied to the possible alternatives) the state in which $\cal X$ will be, the locations of its three scanning heads, and the content of any of the cells of any of the three tapes. Letting $\cal X$ work for $\xi(\bound)$ steps,  one of the following four eventual scenarios should take place, and the provable traceability of the work of $\cal X$ can be shown to imply that $\arfour$ proves the $\add$-disjunction of formulas describing those scenarios:

\begin{description}
\item[Scenario 1:] $\cal X$ makes the move $0$ (and no other moves).
\item[Scenario 2:] $\cal X$ makes the move $1$ (and no other moves).
\item[Scenario 3:] $\cal X$ does not make any moves.  
\item[Scenario 4:] $\cal X$ makes an illegal move (perhaps after   first making a legal move $0$ or $1$). 
\end{description}

In the case of Scenario 1, the play over $\overline{X}$ hits $Y(x)\mld \mathbb{L}$. And $\arfour$ --- in fact, $\pa$ --- proves that, in this case, $Y(x)\mld \mathbb{L}$ is true. The truth of $Y(x)\mld \mathbb{L}$ is indeed very easily established: if it was false, then $Y(x)$ should be false, but then the play of $\cal X$ over $X$ (which, as a game, is the same as $\overline{X}$) hits the false $Y(x)$ and hence is lost, but then $\mathbb{L}$ is true, but then $Y(x)\mld \mathbb{L}$ is true. Thus, 
\(\arfour\vdash (\mbox{\em Scenario 1})\mli Y(x)\mld \mathbb{L}$, from which, by $\add$-Choose, $\arfour\vdash (\mbox{\em Scenario 1})\mli \overline{X}$.  

The case of Scenario 2 is symmetric.

In the case of Scenario  3, ($\arfour$ proves that) $\cal X$ loses, i.e.   $\mathbb{L}$ is true, and hence, say, $Y(x)\mld \mathbb{L}$ (or $Z(x)\mld \mathbb{L}$ if you like) is true. That is, $\arfour\vdash   (\mbox{\em Scenario 3})\mli Y(x)\mld \mathbb{L}$, from which, by $\add$-Choose, $\arfour\vdash (\mbox{\em Scenario 3})\mli\overline{X}$.  

The case of Scenario 4 is similar. 

Thus, for each $i\in\{1,2,3,4\}$, $\arfour\vdash (\mbox{\em Scenario i})\mli \overline{X}$.  
And, as we also have  
\[\arfour\vdash (\mbox{\em Scenario 1})\add (\mbox{\em Scenario 2})\add (\mbox{\em Scenario 3})\add (\mbox{\em Scenario 4}),\]
by $\add$-Elimination, we find the desired $\arfour\vdash \overline{X}$.

The remaining question to clarify  is how the provable traceability of the work of $\cal X$ is achieved. This is where PTI comes into play. In the roles of the two formulas $E$ and $F$ of that rule we employ certain nonelementary formulas $\mathbb{E}$ and $\mathbb{F}$. With $t$ being the ``current time'', $\mathbb{E}(t)$ is a formula which, as a resource, allows us to tell ($\add$ or $\ade$) the current state of $\cal X$, and ($\mlc$) the locations of its three heads, and ($\mlc$) the contents of the three cells under the three heads. And $\mathbb{F}(t)$ allows us, for any ($\ada$) cell of any ($\mlc$) tape, to tell ($\add$) its current content. 

In order to resolve $\mathbb{F}(t\successor)$ --- that is, to tell the content of any ($\ada$) given cell $\#c$ at time $t\plus 1$ --- all we need  to know is the state of $\cal X$, the content of cell $\#c$, the locations of the scanning heads (perhaps only one of them), and the contents of the three cells scanned  by the three heads at time $t$. The content of cell $\#c$ at time $t$ can be obtained from (a single copy of) the resource $\mathbb{F}(t)$, and the rest of the above information from (a single copy of) the resource $\mathbb{E}(t)$. $\arfour$ is aware of this, and proves $\mathbb{E}(t)\mlc\mathbb{F}(t)\mli \mathbb{F}(t\successor)$.

Similarly, it turns out that, in order to resolve $\mathbb{E}(t\successor)$, a single copy of $\mathbb{E}(t)$ and a single copy of $\mathbb{F}(t)$ are sufficient, and $\arfour$, being aware of this, proves $\mathbb{E}(t)\mlc\mathbb{F}(t)\mli \mathbb{E}(t\successor)$.

The above two provabilities, by $\adc$-Introduction, imply  $\arfour\vdash \mathbb{E}(t)\mlc\mathbb{F}(t)\mli \mathbb{E}(t\successor)\adc \mathbb{F}(t\successor) $. This is almost the inductive step of PTI. What is missing is a $\mlc$-conjunct $\mathbb{E}(t)$ in the consequent. Not to worry. Unlike $\mathbb{F}(t)$, $\mathbb{E}(t)$ is a recyclable resource due to the fact that it does not contain $\adc$ or $\ada$ (albeit it contains  $\add,\ade$). Namely, once we learn --- from the antecedental resource $\mathbb{E}(t)$ --- about the state of $\cal X$, the locations of the three scanning heads and the cell contents at those locations  at time $t$, we can use/recycle that information and ``return/resolve back''  $\mathbb{E}(t)$ in the consequent. A syntactic equivalent --- or rather consequence --- of what we just said is that the provability of
$ \mathbb{E}(t)\mlc\mathbb{F}(t)\mli \mathbb{E}(t\successor)\adc \mathbb{F}(t\successor) $ implies the  provability of $ \mathbb{E}(t)\mlc\mathbb{F}(t)\mli  \bigl(\mathbb{E}(t\successor)\adc \mathbb{F}(t\successor)\bigr)\mlc\mathbb{E}(t) $, and hence also the provability of the weaker $ \mathbb{E}(t)\mlc\mathbb{F}(t)\mli  \mathbb{E}(t\successor)\adc\bigl(\mathbb{F}(t\successor)\mlc \mathbb{E}(t)\bigr) $. 

Thus,  $\arfour\vdash \mathbb{E}(t)\mlc\mathbb{F}(t)\mli  \mathbb{E}(t\successor)\adc\bigl(\mathbb{F}(t\successor)\mlc \mathbb{E}(t)\bigr) $. We also have $\arfour\vdash \mathbb{E}(\zero)\mlc \mathbb{F}(\zero)$, as this formula is essentially just a description of the initial configuration of the machine. Then, by PTI, $\arfour\vdash t\mleq \xi(\bound)\mli \mathbb{E}(t)\mlc\mathbb{F}(t)$. This is exactly what we meant by the provable traceability of the work of $\cal X$. 

The above was about the pathologically simple case of $X=Y(x)\add Z(x)$, and the general case will be much more complex, of course. Among other things, provable traceability would have to account for the possibility of the environment making moves now and then. And showing the provability of $\overline{X}$ would require a certain metainduction on its complexity, which we did not need in the present case. But the   idea that we have just tried to explain would still remain valid and central, only requiring certain --- nontrivial but doable ---  adjustments and refinements.

\subsection{The overline notation}

Throughout the rest of this section, we assume that the formula $X$ has no free occurrences of variables other than $\bound$. There is no loss of generality in making such an assumption, because, if $X$ does not satisfy this condition, it can be replaced by the  both semantically and deductively equivalent $\ada$-closure of it over all free variables different from $\bound$. 

We shall sometimes find it helpful to write $X$ as 
\[X(\bound).\]
When, after that, writing $X(b)$ (where $b$ is a constant), one should keep in mind that it means the result of substituting $\bound$ by $b$ in $X(\bound)$ not only where we explicitly see $\bound$, but also in choice quantifiers $\ada$ and $\ade$, which, as we remember, are lazy ways to write $\ada^{\bound}$ and $\ade^{\bound}$. So, for instance, if $X(\bound)$ is $\ade xE(x)$ and $c\not=b$, then $X(c)$ is not the same as  $X(b)$ even if $\bound$ does not occur in $E(x)$, because the former is  $\ade^c xE(x)$ and the latter is $\ade^b x E(x)$. The same  applies to any  formula written in the form $F(\bound,\ldots)$, of course. 


Let us say that a formula is {\bf safe}\label{isafe} iff no two occurrences of quantifiers in it bind the same variable. For simplicity and also without loss of generality, we further assume that the formula $X$  is safe  (otherwise make it safe by renaming variables). 

Since $X$ has no free variables other than $\bound$, for simplicity we can limit our considerations to  valuations that send every non-$\bound$ variable to $0$. We call such valuations {\bf standard}\label{istandard} and use a special notation for them. Namely, for an   integer $b$, we write 
\[e_b\label{ieb}\]
for the valuation such that $e_b(\bound)=b$ and, for any other variable $v$, $e_b(v)=0$.


By a {\bf politeral}\label{ipoliteral} of a formula we mean a positive occurrence of a literal in it. 
While a politeral is not merely a literal but a literal $L$  {\em together} with a fixed occurrence, we shall often refer to it just by the name $L$ of the literal, assuming that it is clear from the context which (positive) occurrence of $L$ is meant.

We assume that the reader is sufficiently familiar with G\"{o}del's technique of encoding and arithmetizing. Using that technique, 
we can construct a sentence   \[\mathbb{L}\label{illl}\] of the language of $\pa$     which asserts --- more precisely, implies --- ``$\cal X$ does not win $X$ in time $\xi$''.  

Namely, let $E_1(\bound,\vec{x}),\ldots,E_n(\bound,\vec{x})$ be  all   subformulas of $X$, where all free variables of each $E_i(\bound,\vec{x}) $ are among $\bound,\vec{x}$ (but not necessarily vice versa).
Then the above sentence $\mathbb{L}$   is a natural formalization of the following statement: 
\begin{quote} {\em ``There is a (finite)  run $\Gamma$ generated by $\cal X$ on some standard bounded valuation $e_b$   such that: 
\begin{enumerate}
\item $\pp$'s time in $\Gamma$ is not smaller than $\xi(b)$, or 
\item $\Gamma$ is a $\pp$-illegal run of $X(b)$, or
\item  $\Gamma$ is a legal run of $X(b)$ and there is a tuple $\vec{c}$ of constants ($\vec{c}$ of the same length as $\vec{x}$) such that:
\begin{itemize}
\item  $\seq{\Gamma}X(b)=E_1(b,\vec{c})$, and we have $\gneg  \elz{E_1(b,\vec{c})} $ (i.e., $ \elz{E_1(b,\vec{c})} $ is false),  
\item  or $\ldots$, or 
\item $\seq{\Gamma}X(b)=E_n(b,\vec{c})$, and we have $\gneg  \elz{E_n(b,\vec{c})} $ (i.e., $ \elz{E_n(b,\vec{c})} $ is false).'' 
\end{itemize}
\end{enumerate} } 
\end{quote}

As we remember, our goal is to construct a formula $\overline{X}$ which expresses the same problem as $X$ does and which is provable in $\arfour$. For any   formula $E$ --- including $X$ --- we let 
\[\overline{E}\label{ipver}\]
be the result of replacing in $E$ every politeral $L$ by $L\mld\mathbb{L}$.

\begin{lemma}\label{august12}
Any  literal  $L $ is equivalent (in the standard model of arithmetic) to  $L\mld\mathbb{L}$. 
\end{lemma}

\begin{proof} That $L$ implies $L\mld\mathbb{L}$ is immediate, as the former is a disjunct of the latter. For the opposite direction, suppose     $L\mld\mathbb{L}$ is true at a given valuation $e$. Its second disjunct cannot be true, because $\cal X$ {\em does} win $X$ in time $\xi$, contrary to what $\mathbb{L}$ asserts.  So, the first disjunct, i.e. $L$, is true. 
\end{proof}

\begin{lemma}\label{august12a}
For any  formula $E$, including $X$, we have  $E^\dagger=\overline{E}^\dagger$. 
\end{lemma}

\begin{proof} Immediately from Lemma \ref{august12} by induction on the complexity of $E$.
\end{proof}

In view of the above lemma, what now remains to do for the completion of our completeness proof is to show that $\arfour\vdash\overline{X}$. The rest of the present section is entirely devoted to this task.

\subsection{This and that} 
\begin{lemma}\label{jan4d}
For any    formula $E$, $\arfour\vdash \mathbb{L} \mli \overline{E}$. 
\end{lemma}

\begin{idea} $\overline{E}$ is a logical combination of ``quasipoliterals'' of the form $L\mld\mathbb{L}$. Under the assumption (of the truth of) $\mathbb{L}$, each such quasipoliteral becomes true and, correspondingly, $\overline{E}$ essentially becomes a logical combination of $\twg$s. Any such combination is very easy to solve/prove. \end{idea}

\begin{proof} We prove this lemma by induction on the complexity of $E$.

If $E$ has the form $E[H_1\add\ldots\add H_n]$, then, by the induction hypothesis, $\arfour\vdash \mathbb{L} \mli  \overline{E[H_1]}$. From here, by $\add$-Choose, we get the desired $\arfour\vdash \mathbb{L} \mli  \overline{E[H_1\add\ldots\add H_n]}$.

Quite similarly, if $E$ has the form $E[\ade x H(x)]$, then, by the induction hypothesis, $\arfour\vdash \mathbb{L} \mli  \overline{E[H(v)]}$  (for whatever variable $v$ you like). From here, by $\ade$-Choose, we get  $\arfour\vdash \mathbb{L}\mli  \overline{E[\ade x H(x)]}$.

Now assume $E$ has no surface occurrences of $\add$- and $\ade$-subformulas. The formula $\elz{\overline{E}}$ is a $(\mlc,\mld,\cla,\cle)$-combination of $\twg$s (originating from $\adc$- and $\ada$-subformulas when elementarizing $\overline{E}$) and formulas $L \mld\mathbb{L}$ (originating from $L$ when transferring from $E$ to $\overline{E}$) where $L$ is a politeral of $E$. $\twg$ is true.  If $\mathbb{L}$ is true, then each    $L\mld\mathbb{L}$ is also true no matter what the values of the variables of $L$ are (if $L$ contains any variables at all). Therefore, clearly, $\elz{\overline{E}}$, as a $(\mlc,\mld,\cla,\cle)$-combination of (always) true formulas, is  true.  Formalizing this argument in $\pa$ and hence in $\arfour$ yields $\arfour\vdash \mathbb{L} \mli  \elz{\overline{E}}$, which, taking into account that $\mathbb{L}$ is an elementary formula and hence $\mathbb{L}=\elz{\mathbb{L}}$, is the same as to say that   
\begin{equation}\label{jan4a}
\arfour\vdash \elz{\mathbb{L} \mli  \overline{E}}.
\end{equation}

Suppose $E$ has the form $E[H_1\adc\ldots\adc H_n]$. Then, by the induction hypothesis, $\arfour$  proves $\mathbb{L}\mli \overline{E[H_i]}$ for each $i\in\{1,\ldots,n\}$. Similarly, suppose 
$E$ has the form $E[\ada xH(x)]$. Let $v$ be a variable different from $\bound$ and not occurring in $E[\ada xH(x)]$. Then, again by the induction hypothesis, 
 $\arfour$  proves $\mathbb{L}\mli \overline{E[H(v)]}$. These observations, together with (\ref{jan4a}), by Wait, yield the desired $\arfour\vdash \mathbb{L}\mli \overline{E}$.
\end{proof}

We shall say that a run generated by the machine $\cal X$ is {\bf prompt}\label{iprompt} iff  $\oo$'s time in it is $0$. In   a prompt run, 
the environment always reacts to a move by $\cal X$  instantaneously (on the same clock cycle as on which $\cal X$ moved), or does not react at all.
An exception is clock cycle $\#0$, on which the environment can move even if $\cal X$ did not move. Such runs are convenient to deal with, because in them $\pp$'s time equals the timestamp of the last move. And this, in turn, means that no moves by either player are made at any time greater or equal to $\xi(b)$, where $b$ is the value assigned to $\bound$ by the valuation spelled on the valuation tape of the machine. 

By our assumption, $\cal X$ wins $X$ (in time $\xi$), meaning   that every run $\Gamma$ generated by $\cal X$ on a bounded valuation $e$ is a $\pp$-won run of $e[X] $, including the cases when $\Gamma$ is prompt and $e$ is standard. This allows us to focus on  prompt runs and standard valuations only.  Specifically, we are going to show that $\overline{X}$ is provable because $\cal X$ wins (in time $\xi$) every prompt run of $X$ on every standard bounded valuation.  

Further, for our present purposes, environment's possible strategies can be understood as (limited to) fixed/predetermined behaviors  seen as finite sequences of   moves with non-decreasing timestamps. Let us call such sequences {\bf counterbehaviors}.\label{icnb} The meaning of a counterbehavior \[\seq{(\alpha_1,t_1),(\alpha_2,t_2),\ldots,(\alpha_n,t_n)}\] is that the environment makes move $\alpha_1$ at time $t_1$, move $\alpha_2$ at time $t_2$, \ldots,  move $\alpha_n$ at time $t_n$. If two consecutive moves have the same timestamp, the moves are assumed to be made (appear in the run) in the same order as they are listed in the counterbehavior.

Given a standard valuation $e$  and a counterbehavior $C=\seq{(\alpha_1,t_1),\ldots,(\alpha_n,t_n)}$, by the {\bf $(C,e)$-branch}\label{icebranch} we mean the $e$-computation branch of $\cal X$ where the environment acts according to $C$ --- that is, makes move $\alpha_1$ at time $t_1$, \ldots, move $\alpha_n$ at time $t_n$. And the {\bf $(C,e)$-run}\label{icerun} is the run spelled by this branch.

For natural numbers $b$ and $d$, we  say that a counterbehavior $C$ is {\bf $(b,d)$-adequate}\label{iadequate} iff  the following three conditions are satisfied:
\begin{enumerate}
\item the $(C,e_b)$-run is  not a $\oo$-illegal run of $X(b)$;
\item the $(C,e_b)$-run is  prompt;
\item the timestamp of the last move of $C$ (if $C$ is nonempty) is less than $d$.
\end{enumerate}  
Thus, ``$C$ is $(b,d)$-adequate'' means that, using this counterbehavior against $\cal X$ with  $e_b$ on the valuation tape of the latter, the environment 
 has played legally (condition 1), acted fast/promptly (condition 2), and made all (if any) moves before  time $d$ (condition 3). 

Just as any finite objects,  counterbehaviors can be encoded through natural numbers. The {\bf code} ({\bf G\"{o}del number}) of an object $O$ will be denoted by \[\ulcorner O\urcorner .\label{ign}\] Under any encoding, 
 the size of the code of a counterbehavior of interest will generally exceed the value of $\bound$. But this is not going to be a problem as we will quantify counterbehaviors using blind rather than choice quantifiers.

For convenience, we assume that every natural number is the code of some counterbehavior. This allows us to terminologically identify counterbehaviors 
with their codes, and say phrases like ``$a$ is a $(b,d)$-adequate counterbehavior'' --- as done below --- which should be understood as ``Where $C$ is the counterbehavior with $a=\ulcorner C\urcorner$,
$C$ is $(b,d)$-adequate''. Similarly, ``the $(a,e)$-branch'' (or ``the $(a,e)$-run'') will mean ``the $(C,e)$-branch (or $(C,e)$-run) where $C$ is the counterbehavior with $a=\ulcorner C\urcorner$''.

Let $E=E(\bound,\vec{s})$ be a   formula all of whose  free variables are among  $\bound,\vec{s}$ (but not necessarily vice versa). We 
 will write  
\[\mathbb{W}^E(z,t_1,t_2,\bound,\vec{s})\]
to denote an elementary formula whose free variables are exactly (the pairwise distinct) $z,t_1,t_2,\bound,\vec{s}$, and which is a natural arithmetization of the predicate which,  for any constants $a,d_1,d_2,b,\vec{c}$, holds --- that is, $\mathbb{W}^E(a,d_1,d_2,b,\vec{c})$ is true --- iff  the following conditions are satisfied:
\begin{itemize}
\item $0\mless d_1\mleq d_2\mleq \xi(b)$;
\item $a$ is a $(b,d_1)$-adequate counterbehavior;  
\item where $\Phi$ is the initial segment of the $(a,e_b)$-run obtained from the latter by deleting all moves except those whose timestamps are  less than  $d_1$, $\Phi$ is a legal position of  $E(b,\vec{c})$;  
\item for the above  $\Phi$,  we have $\seq{\Phi}X(b)=E(b,\vec{c})$, and $\pa$ proves this fact;
\item either $d_1=1$ or, in the $(a,e_b)$-branch, $\cal X$ has made some move at time $d_1-1$ (so that the effect of that move first took place at time $d_1$);
\item for any $k$ with $d_1\mleq k\mless d_2$, no move is made at time $k$ (i.e. no move has the timestamp $k$) in the $(a,e_b)$-run. 
\end{itemize}

Thus, in the context of the $(a,e_b)$-branch,    $\mathbb{W}^E(a,d_1,d_2,b,\vec{c})$ says that, exactly by time $d_1$,\footnote{Only considering nonzero times in this context.} the play has hit the position  $E(b,\vec{c})$, and that this position has remained stable (there were  no moves to change it) throughout the interval $[d_1,d_2]$. It does not rule out that a move was made at time $d_2$ but, as we remember, the effect of such a move will take place by time $d_2+1$ rather than $d_2$.  

It may be worthwhile to comment on the meaning of the above for the special case where $t_2$ is $\xi(\bound)$. Keeping in mind that $\cal X$ runs in time $\xi(\bound)$, the formula 
\[\mathbb{W}^E(z,t,\xi(\bound),\bound,\vec{s}),\label{idoubleu}\]
for any given values $a,d,b,\vec{c}$ for $z,t,\bound,\vec{s}$, asserts --- or rather implies --- that, in the scenario of the $(a,e_b)$-branch, at time $d$, the play (position to which $X$ has evolved) hits $E(b,\vec{c})$ and remains stable ever after, so that  $E(b,\vec{c})$ is the final, ultimate position of the play. 

We say that a formula $E$ or the corresponding game is {\bf critical}\label{icritical} iff one of the following conditions is satisfied: 
\begin{itemize}
\item $E$ is a $\add$- or $\ade$-formula;
\item $E$ is $\cla y G$ or $\cle y G$, and $G$ is critical;
\item $E$ is a $\mld$-disjunction, with all disjuncts critical;
\item $E$ is a $\mlc$-conjunction, with at least one conjunct critical.
\end{itemize}

The importance of the above concept is related to the fact that ($\pa$ knows that)   a given legal run of $\overline{X}$ is lost by $\cal X$ if and only if the eventual formula/position hit by that run is  critical. 

\begin{lemma}\label{august20b}
Assume $E=E(\bound,\vec{s})$ is a non-critical  formula all of whose  free variables are among  $\bound,\vec{s}$. Further assume   $\theta, \omega,\vec{\psi}$ are any terms ($\vec{\psi}$ of the same length as $\vec{s}$),   
and $z$ is a variable not occurring in these terms or in $E$.  
Then
\[\arfour\vdash \cle z\mathbb{W}^E(z,\theta,\xi(\omega),\omega,\vec{\psi})  \mli \elz{\overline{E(\omega,\vec{\psi})}}.\] 
\end{lemma}

\begin{idea} The antecedent of the above formula implies that some run of $X$ generated by $\cal X$ yields the non-critical eventual position $E(\omega,\vec{\psi})$. If $\elz{E(\omega,\vec{\psi})}$ is true, then so is $\elz{\overline{E(\omega,\vec{\psi})}}$. Otherwise, if $\elz{E(\omega,\vec{\psi})}$ is false, $\cal X$ has lost, so $\mathbb{L}$ is true. But the truth of the formula $\mathbb{L}$, which is disjuncted with every politeral of $\overline{E(\omega,\vec{\psi})}$, easily implies the truth of $\elz{\overline{E(\omega,\vec{\psi})}}$. This argument is formalizable in $\pa$. 
\end{idea}

\begin{proof} Assume the conditions of the lemma. Argue in $\pa$. Consider arbitrary values of $\theta$, $\omega$,  $\vec{\psi}$, which we continue writing as $\theta$, $\omega$,  $\vec{\psi}$. Suppose, for a contradiction, that the ultimate position --- that is, the position reached by the time $\xi(\omega)$ --- of some play of $\cal X$ over $X$  is $E(\omega,\vec{\psi})$ (i.e., $\cle z\mathbb{W}^E(z,\theta,\xi(\omega),\omega,\vec{\psi})$ is true) but
 $\elz{\overline{E(\omega,\vec{\psi})}}$ is false. The falsity of $\elz{\overline{E(\omega,\vec{\psi}})}$ implies the falsity of $\elz{E(\omega,\vec{\psi})}$. This is so because the only difference between the two formulas is that, wherever the latter has some politeral $L$, the former has a disjunction containing $L$ as a disjunct.  

But ending with an  ultimate position whose elementarization is false means that $\cal X$ does not win $X$ in time $\xi$ (remember Lemma \ref{new1}). In other words, 
\begin{equation}\label{jan4b}
\mbox{\em $\mathbb{L}$ is true.}
\end{equation}

Consider any non-critical   formula $G$.  By induction on the complexity of $G$, we are going to show that $\elz{\overline {G}}$ is true for any values of its free variables. Indeed:

If $G$ is a literal, then  $\elz{\overline {G}}$ is $G\mld \mathbb{L}$ which, by (\ref{jan4b}), is true. 

If $G$ is $H_1\adc\ldots \adc H_n$ or $\ada xH(x)$, then $\elz{\overline {G}}$ is $\twg$ and is thus true.

$G$ cannot be $H_1\add\ldots \add H_n$ or $\ade xH(x)$, because then it would be critical. 

If $G$ is $\cla yH(y)$ or $\cle y H(y)$, then $\elz{\overline {G}}$ is $\cla y\elz{\overline{H(y)}}$ or $\cle y\elz{\overline{H(y)}}$. In either case $\elz{\overline{G}}$ is true because, by the induction hypothesis, $\elz{\overline{H(y)}}$ is true for every value of its free variables, including variable $y$.

 If $G$ is $H_1\mlc\ldots \mlc H_n$, then the formulas $H_1,\ldots,H_n$ are non-critical. Hence, by the induction hypothesis, $\elz{\overline{H_1}},\ldots,\elz{\overline{H_n}}$ are  true. Hence so is  $\elz{\overline{H_1}}\mlc\ldots\mlc\elz{\overline{H_n}}$ which, in turn, is nothing but $\elz{\overline{G}}$. 

Finally, if  $G$ is $H_1\mld\ldots \mld H_n$, then one of the formulas $H_i$ is non-critical. Hence, by the induction hypothesis, $\elz{\overline{H_i}}$ is  true. Hence so is  $\elz{\overline{H_1}}\mld\ldots\mld\elz{\overline{H_n}}$ which, in turn, is nothing but $\elz{\overline{G}}$.

Thus, for any non-critical formula $G$, $\elz{\overline {G}}$ is true. This includes the case $G= E(\omega,\vec{\psi}) $  which, however,  contradicts our earlier observation that $\elz{E(\omega,\vec{\psi})}$ is false.  
\end{proof}

\begin{lemma}\label{august20a}
Assume $E=E(\bound,\vec{s})$ is a  critical   formula all of whose  free variables are among  $\bound,\vec{s}$. Further assume   $\theta,\omega, \vec{\psi}$ are any terms ($\vec{\psi}$ of the same length as $\vec{s}$),  and $z$ is a variable not occurring in these terms or in $E$. Then
\[\arfour\vdash \cle z\mathbb{W}^E(z,\theta,\xi(\omega),\omega,\vec{\psi})  \mli \overline{E(\omega,\vec{\psi})}.\] 
\end{lemma}

\begin{proof} Assume the conditions of the lemma. By induction on complexity, one can easily see that the elementarization of any critical formula is false. Thus, $\elz{E(\omega,\vec{\psi})}$ is false. Arguing further as we did in the proof of Lemma \ref{august20b} when deriving (\ref{jan4b}), we find that, if  $\cle z \mathbb{W}^E(z,\theta,\xi(\omega),\omega,\vec{\psi})$ is true,  then so is $\mathbb{L}$. And this argument can be formalized in $\pa$, so that we have 
\[\arfour\vdash \cle z \mathbb{W}^E(z,\theta,\xi(\omega),\omega,\vec{\psi})\mli \mathbb{L}.\]
The above, together with Lemma \ref{jan4d}, by Transitivity, implies 
$\arfour\vdash \cle z\mathbb{W}^E(z,\theta,\xi(\omega),\omega,\vec{\psi})  \mli \overline{E(\omega,\vec{\psi})}$.               
\end{proof}

\subsection{Taking care of the case of small bounds} 
$|\xi(\bound)|$ is logarithmic in $\bound$ and hence,  generally, it will be    much  smaller   than $\bound$. However, there are exceptions.   For instance, when $\bound=1$ and $\xi(\bound)=\bound\plus\bound$, the size of $\xi(\bound)$ is $2$, exceeding $\bound$. Such exceptions will   only occur in a finite number of cases,   where $\bound$ is ``very  small''. These pathological cases --- the cases with $\gneg |\xi(\bound)|\mleq \bound$ --- require a separate handling, which we present in this subsection. The main result here is Lemma \ref{pathology6}, according to which  $\arfour$ proves $\gneg |\xi(\bound)|\mleq \bound\mli \overline{X}$, i.e. proves the target $\overline{X}$ on the assumption that we are dealing with a pathologically small $\bound$. The remaining, ``normal'' case of $|\xi(\bound)|\mleq \bound$ will be taken care of later in Subsection \ref{jjjj}.

For a natural number $n$, by the {\bf formal numeral for}\label{ifornum} $n$, denoted $\hat{n}$, we will mean some standard variable-free term representing $n$. For clarity, let us say that the formal numeral for zero is $\zero 0$, the formal numeral for one is $\zero 1$, the formal numeral for two is $\zero 10$, the formal numeral for three is $\zero 11$, the formal numeral for four is $\zero 100$, etc. 

The above-mentioned provability of $\gneg |\xi(\bound)|\mleq \bound\mli \overline{X}$ will be established through showing (Lemma \ref{pathology5}) that, for each particular positive integer $b$, including all of the finitely many $b$'s with $\gneg |\xi(\hat{b})|\mleq \hat{b}$, $\arfour$ proves $\bound\equals\hat{b}\mli\overline{X}$. But we need a little preparation first. 

\begin{lemma}\label{pathology3}
Let $r$ be any variable, $b$ any positive integer,   and $N$  the set of all natural numbers $a$ with $|a|\mleq b$. Then 
\[\arfour\vdash \bound\equals \hat{b}\mli \add\{r\equals \hat{a}\ |\ a\in N\}.\]
\end{lemma}

\begin{idea} On the assumption $\bound\equals \hat{b}$ and due to Axiom 13, $\arfour$ knows that, whatever $r$ is, its size cannot exceed $\hat{b}$. In other words, it knows that $r$ has to be one of the elements of $N$. The main technical part of our proof of the lemma is devoted to showing that this knowledge is, in fact, constructive, in the sense that $\arfour$ can tell exactly which ($\add$) element of $N$ the number $r$ is.   
\end{idea}

\begin{proof} Assume the conditions of the lemma. Obviously we have
\[\pa\vdash |r|\mleq\bound\mli \bound\equals \hat{b}\mli \mld \{r\equals \hat{a}\ |\ a\in N\},\]
modus-ponensing which with Axiom 13 yields
\begin{equation}\label{dec10}
\arfour\vdash \bound\equals \hat{b}\mli \mld\{r\equals \hat{a}\ |\ a\in N\}.
\end{equation}

Next, consider any $a\in N$. We claim that 
\begin{equation}\label{dec11}
\arfour\vdash \bound\equals \hat{b}\mli r\equals\hat{a}\add r\notequals\hat{a}, 
\end{equation}
which is justified as follows:\vspace{7pt}

\noindent 1. $\begin{array}{l} 
\gneg |\hat{a}|\mleq \bound\add \ade z (z\equals \hat{a})
\end{array}$ \ \  Lemma \ref{com} \vspace{3pt}

\noindent 2. $\begin{array}{l} 
\gneg |\hat{a}|\mleq \bound \mli \bound\notequals \hat{b}
\end{array}$ \ \  $\pa$ \vspace{3pt}

\noindent 3. $\begin{array}{l} 
\gneg |\hat{a}|\mleq \bound \mli \bigl(\bound\equals \hat{b}\mli \ade z(z\equals \hat{a})\bigr)
\end{array}$ \ \  Weakening: 2\vspace{3pt}

\noindent 4. $\begin{array}{l} 
\ade z (z\equals \hat{a}) \mli \bigl(\bound\equals \hat{b}\mli \ade z(z\equals \hat{a})\bigr)
\end{array}$ \ \  $\clfour$-Instantiation, instance of $P\mli(q\mli P)$ \vspace{3pt}

\noindent 5. $\begin{array}{l} 
\bound\equals \hat{b}\mli \ade z(z\equals \hat{a})
\end{array}$ \ \  $\add$-Elimination: 1,3,4\vspace{3pt}

\noindent 6. $\begin{array}{l} 
\ada x\ada y(y\equals x\add y\notequals x)  
\end{array}$ \ \  Lemma \ref{nov18a}\vspace{3pt}

\noindent 7. $\begin{array}{l} 
s\equals r\add s\notequals r  
\end{array}$ \ \  $\ada$-Elimination (twice): 6\vspace{3pt}

\noindent 8. $\begin{array}{l} 
s\equals r\mli s\equals \hat{a}\mli  r\equals\hat{a}
\end{array}$ \ \   Logical axiom \vspace{3pt}

\noindent 9. $\begin{array}{l} 
s\equals r\mli s\equals \hat{a}\mli  r\equals\hat{a}\add r\notequals\hat{a}
\end{array}$ \ \  $\add$-Choose: 8\vspace{3pt}

\noindent 10. $\begin{array}{l} 
s\notequals r\mli s\equals \hat{a}\mli   r\notequals\hat{a}
\end{array}$ \ \ Logical axiom \vspace{3pt}

\noindent 11. $\begin{array}{l} 
s\notequals r\mli s\equals \hat{a}\mli  r\equals\hat{a}\add r\notequals\hat{a}
\end{array}$ \ \  $\add$-Choose: 10\vspace{3pt}

\noindent 12. $\begin{array}{l} 
s\equals \hat{a}\mli  r\equals\hat{a}\add r\notequals\hat{a}
\end{array}$ \ \  $\add$-Elimination: 7,9,11\vspace{3pt}

\noindent 13. $\begin{array}{l} 
\ade z(z\equals \hat{a})\mli  r\equals\hat{a}\add r\notequals\hat{a}
\end{array}$ \ \  $\ada$-Introduction: 12\vspace{3pt}

\noindent 14. $\begin{array}{l} 
\bound\equals \hat{b}\mli r\equals\hat{a}\add r\notequals\hat{a}
\end{array}$ \ \  TR: 5,13\vspace{7pt}

Now, with a little thought, the formula 
\[ (\bound\equals \hat{b}\mli \mld\{r\equals \hat{a}\ |\ a\in N\})\mlc \mlc\{\bound\equals \hat{b}\mli r\equals\hat{a}\add r\notequals\hat{a}\ | \ a\in N\} \mli
( \bound\equals \hat{b}\mli \add\{r\equals \hat{a}\ |\ a\in N\})\]
can be seen to be provable in $\clthree$ and hence in $\arfour$. Modus-ponensing the above with (\ref{dec10}) and (\ref{dec11}) yields the desired $\arfour\vdash \bound\equals \hat{b}\mli \add\{r\equals \hat{a}\ |\ a\in N\}$. 
  \end{proof}

\begin{lemma}\label{pathology2}
Let $r$ be any variable, $b$  any positive integer, and $E(r)$ any  formula.
Assume that, for each natural number $a$ with $|a|\mleq b$, $\arfour\vdash  E(\hat{a}) $. Then $\arfour\vdash \bound\equals \hat{b}\mli  E(r) $. 
\end{lemma}
\begin{proof} Assume the conditions of the lemma. Let $N$ be the set of all numbers $a$ with $|a|\mleq b$.  Consider any $a\in N$. Clearly,  by $\clfour$-Instantiation, $\arfour\vdash E(\hat{a})\mli r\equals \hat{a}\mli  E(r) $. Modus-ponensing this with the assumption $\arfour\vdash  E(\hat{a}) $ yields  $\arfour\vdash r\equals \hat{a}\mli  E(r) $. This holds for all $a\in N$, so, by $\adc$-Introduction, $\arfour\vdash \add\{r\equals\hat{a}\ |\ a\in N\}\mli  E(r) $.  But, by Lemma \ref{pathology3}, $\arfour\vdash \bound\equals \hat{b}\mli \add\{r\equals \hat{a}\ |\ a\in N\}$. Hence, by Transitivity, $\arfour\vdash \bound\equals \hat{b}\mli   E(r) $.
\end{proof}

Below and elsewhere, for a tuple $\vec{c}=c_1,\ldots,c_n$ of constants, $\vec{\hat{c}}$ stands for the tuple $\hat{c}_1,\ldots,\hat{c}_n$. 

\begin{lemma}\label{pathology1}
Assume $E=E(\bound,\vec{s})$ is a   formula all of whose free variables are among $\bound,\vec{s}$, \ $b$ is any positive integer,  and   $a$,$d_1,d_2$,$\vec{c}$ are any natural numbers ($\vec{c}$ of the same length as $\vec{s}$). Then  
$\mathbb{W}^E(\hat{a},\hat{d}_1,\hat{d}_2,\hat{b},\vec{\hat{c}})$ is true iff it is provable in $\pa$.
\end{lemma}
\begin{proof} $\pa$ only proves true sentences.  $\pa$ is also known to prove   all  ``mechanically   verifiable'' (of complexity $\Sigma^{0}_{1}$, to be precise) true sentences such as $\mathbb{W}^E(\hat{a},\hat{d}_1,\hat{d}_2,\hat{b},\vec{\hat{c}})$ is if true. \end{proof}

\begin{lemma}\label{pathology4}
Under the conditions of Lemma \ref{pathology1}, 
 if
$\mathbb{W}^E(\hat{a},\hat{d}_1, \hat{d}_2,\hat{b},\vec{\hat{c}} )$  is true, then \(\arfour\vdash\bound\equals\hat{b}\mli    \overline{E(\hat{b},\vec{\hat{c}})}.\)  \end{lemma} 

\begin{idea} In  the context of the $(a,e_b)$-branch, the assumptions of the lemma imply that, at some ($d_1$) point, the play hits the position $E(b,\vec{c})$. $\cal X$ may or may not make further moves to modify this position. 

If a move is made, it brings us to a new position expressed through a simpler 
 formula, from which $E(b,\vec{c})$ follows by $\adc$-Choose or $\ada$-Choose. This allows us to apply the induction hypothesis to that formula, and then find the provability of  $\bound\equals\hat{b}\mli \overline{E(\hat{b},\vec{\hat{c}})}$ by the corresponding Choose rule.   

Suppose now no moves are made, so that the play ends as $E(b,\vec{c})$. This position has to be non-critical, or otherwise $\cal X$ would be the loser. Then Lemmas \ref{august20b} and \ref{pathology1} allow us to find that the elementarization of the target formula is provable. Appropriately manipulating the induction hypothesis, we manage to find the provability of all additional  premises from which the target formula follows by Wait.   
\end{idea}

\begin{proof} Our proof proceeds by induction on the complexity of $E(\bound,\vec{s})$.  Assume $\mathbb{W}^E(\hat{a},\hat{d}_1,\hat{d}_2,\hat{b},\vec{\hat{c}})$ is true. We separately consider the following two cases.\vspace{10pt}

{\em Case 1}:  $\mathbb{W}^E(\hat{a},\hat{d}_1,\xi(\hat{b}),\hat{b},\vec{\hat{c}})$  is not true. On the other hand, by our assumption,  $\mathbb{W}^E(\hat{a},\hat{d}_1,\hat{d}_2,\hat{b},\vec{\hat{c}})$ is true. The latter implies that, in the $(a,e_b)$-branch, the play reaches (by time $d_1$)  the position  $E(\hat{b},\vec{\hat{c}})$ which persists up to time $d_2$;  and the former implies that this situation changes sometime afterwards (the latest by time $\xi(b)$). So, a move is made at some time $m$ with $d_2\mleq m\mless \xi(b)$. Such a move $\beta$ (the earliest one if there are several) cannot be made by the environment, because, as implied by the assumption $\mathbb{W}^E(\hat{a},\hat{d}_1,\hat{d}_2,\hat{b},\vec{\hat{c}})$, $a$ is a $(b,d_1)$-adequate counterbehavior. So, $\beta$ is a move by $\cal X$. Since $\cal X$ wins $X$, $\beta$ cannot be an illegal move of the play. 
It is obvious that then  one of the following conditions holds:
\begin{description}
\item[(i)]  There is a formula $H=H(\bound,\vec{s})$ which  is the result of replacing in $E(\bound, \vec{s} )$ a surface occurrence of a subformula $G_1\add\ldots \add G_n$ by one of the $G_i$'s, such that  $\mathbb{W}^H(\hat{a},\hat{m}\successor,\hat{m}\successor,\hat{b},\vec{\hat{c}})$ is true.
\item[(ii)] There is formula $H=H(\bound,\vec{s},r)$, where $r$ is a variable not occurring in $E(\bound, \vec{s} )$, such that $H(\bound,\vec{s},r)$
is 
 the result of replacing in $E(\bound, \vec{s} )$ a surface occurrence of a subformula $\ade yG(y)$ by $G(r)$, and $\mathbb{W}^H(\hat{a},\hat{m}\successor,\hat{m}\successor,\hat{b},\vec{\hat{c}},\hat{k})$ is true for some constant  $k$ with $|k|\mleq b$.                            . 
\end{description}

Thus,  $H(b,\vec{c})$ (in case (i)) or $H(b,\vec{c},k)$ (in case (ii)) is the game/position to which $E(b,\vec{c})$ is brought down by the above-mentioned legal labmove $\pp\beta$.  

Assume condition (i) holds. By the induction hypothesis, $\arfour\vdash\bound\equals\hat{b}\mli \overline{H(\hat{b},\vec{\hat{c}})}$. Then, by $\add$-Choose, $\arfour\vdash \bound\equals\hat{b}\mli \overline{E(\hat{b},\hat{\vec{c}})}$.

Assume now condition (ii) holds. Again, by the induction hypothesis, 
\begin{equation}\label{jan6a}
\arfour\vdash \bound\equals\hat{b}\mli \overline{H(\hat{b},\vec{\hat{c}},\hat{k} )}.
\end{equation}
Obviously $\clfour\vdash \bigl(p\mli Q(f)\bigr)
\mli\bigl(r\equals f\mli 
p\mli Q(r)\bigr)$ whence,  by $\clfour$-Instantiation,   
\[\arfour\vdash 
\bigl(\bound\equals\hat{b}\mli \overline{H(\hat{b},\vec{\hat{c}},\hat{k} )}\bigr)
\mli\bigl(r\equals\hat{k}\mli 
\bound\equals\hat{b}\mli  \overline{H(\hat{b},\vec{\hat{c}} , r  )}\bigr). \]
Modus-ponensing the above with (\ref{jan6a}) yields 
\[
 \arfour\vdash r\equals\hat{k}\mli 
\bigl(\bound\equals\hat{b}\mli  \overline{H(\hat{b},\vec{\hat{c}} , r  )}\bigr) 
\]
from which,  by $\ade$-Choose,
\[\arfour\vdash  r\equals\hat{k} \mli \bigl(\bound\equals\hat{b}\mli  \overline{E(\hat{b},\vec{\hat{c}})}\bigr)\]
and then, by $\ada$-Introduction,
\begin{equation}\label{jan6c}
\arfour\vdash  \ade z(z\equals\hat{k}) \mli \bigl(\bound\equals\hat{b}\mli  \overline{E(\hat{b},\vec{\hat{c}})}\bigr).
\end{equation}

We also have 
\begin{equation}\label{jan6d}
\arfour\vdash \bound\equals\hat{b}\mli \ade z(z\equals\hat{k}),
\end{equation}
justified as follows:\vspace{7pt}

\noindent 1. $\begin{array}{l} 
\gneg |\hat{k}|\mleq \bound\add \ade z (z\equals \hat{k})
\end{array}$ \ \  Lemma \ref{com} \vspace{3pt}

\noindent 2. $\begin{array}{l} 
\gneg |\hat{k}|\mleq \bound\mli \bound\notequals\hat{b} 
\end{array}$ \ \  $\pa$ \vspace{3pt}

\noindent 3. $\begin{array}{l} 
\gneg |\hat{k}|\mleq \bound\mli \bound\equals\hat{b}\mli   \ade z (z\equals \hat{k})
\end{array}$ \ \  Weakening: 2\vspace{3pt}

\noindent 4. $\begin{array}{l} 
\ade z (z\equals \hat{k}) \mli  \bound\equals\hat{b}\mli   \ade z (z\equals \hat{k})
\end{array}$ \ \  $\clfour$-Instantiation, instance of $P\mli q\mli P$ \vspace{3pt}

\noindent 6. $\begin{array}{l} 
 \bound\equals\hat{b}\mli  \ade z(z\equals\hat{k})
\end{array}$ \ \  $\add$-Elimination: 1,3,4 \vspace{7pt}

From (\ref{jan6d}) and (\ref{jan6c}), by Transitivity, $\arfour\vdash \bound\equals\hat{b} \mli  \bound\equals\hat{b}\mli  \overline{E(\hat{b},\vec{\hat{c}})} $. But, by $\clfour$-Instantiation, we have 
\[\arfour\vdash \bigl(\bound\equals\hat{b} \mli  \bound\equals\hat{b}\mli  \overline{E(\hat{b},\vec{\hat{c}})} \bigr)\mli \bigl(\bound\equals\hat{b}\mli  \overline{E(\hat{b},\vec{\hat{c}})}\bigr)
\]
(this matches $(p\mli p\mli Q)\mli(p\mli Q)$). Hence, by Modus Ponens, we find $\arfour\vdash \bound\equals\hat{b}\mli  \overline{E(\hat{b},\vec{\hat{c}})}$, as desired.\vspace{10pt}

{\em Case 2}: $\mathbb{W}^E(\hat{a},\hat{d}_1,\xi(\hat{b}),\hat{b},\vec{\hat{c}})$  is true. Then, by Lemma \ref{pathology1}, $\arfour$ proves 
$\mathbb{W}^E(\hat{a},\hat{d}_1,\xi(\hat{b}),\hat{b},\vec{\hat{c}})$. $\arfour$ also proves the following formula because it is a logical axiom:   
\[\mathbb{W}^E(\hat{a},\hat{d}_1,\xi(\hat{b}),\hat{b},\vec{\hat{c}})\mli\cle z\mathbb{W}^E(z,\hat{d}_1,\xi(\hat{b}),\hat{b},\vec{\hat{c}}). \] 
Hence, by Modus Ponens,
\begin{equation}\label{jan5a}
\arfour\vdash \cle z\mathbb{W}^E(z,\hat{d}_1,\xi(\hat{b}),\hat{b},\vec{\hat{c}}) .
\end{equation}
$\mathbb{W}^E(\hat{a},\hat{d}_1,\xi(\hat{b}),\hat{b},\vec{\hat{c}})$ implies that $E(b,\vec{c})$ is the final position of the play over $X$ according to the scenario of the $(a,e_b)$-branch.  Note that, therefore, $E(\bound,\vec{s})$ cannot be critical. This is so because, as observed earlier,  the elementarization of any critical formula is  false, and having such a formula as the final position in some play would make $\cal X$ lose, contrary to our assumption that $\cal X$ (always) wins $X$.   Therefore, by Lemma \ref{august20b}, 
\[\arfour\vdash \cle z\mathbb{W}^E(z,\hat{d}_1,\xi(\hat{b}),\hat{b},\vec{\hat{c}})  \mli \elz{\overline{E(\hat{b},\vec{\hat{c}})}}.\]
Modus-ponensing the above with (\ref{jan5a}) yields $\arfour\vdash \elz{\overline{E(\hat{b},\vec{\hat{c}})}}$, from which, by Weakening, 
$\arfour\vdash \bound\equals\hat{b}\mli$ $ \elz{\overline{E(\hat{b},\vec{\hat{c}})}}$, which is the same as to say that 
\begin{equation}\label{jan5b}
\arfour\vdash \elz{\bound\equals\hat{b}\mli \overline{E(\hat{b},\vec{\hat{c}})}} .
\end{equation}

{\bf Claim 1.} {\em Assume $\overline{E(\hat{b},\vec{\hat{c}})}$ has the form $H[G_1\adc\ldots \adc G_m]$, and $i\in\{1,\ldots,m\}$. Then $\arfour\vdash \bound\equals \hat{b}\mli H[G_i]$.}\vspace{5pt}

\begin{subproof} Assume the conditions of the claim. Let $F=F(\bound,\vec{s})$ be the    formula such that $\overline{F(\hat{b},\vec{\hat{c}})}=H[G_i]$. Let $\beta$ be the environment's move that brings $H[G_1\adc\ldots \adc G_m]$ to $H[G_i]$.   Let $k$ be the (code of the) counterbehavior obtained by appending the timestamped move $(\beta,d_1\mminus 1)$ to (the counterbehavior whose code is) $a$. 
 Since $\mathbb{W}^E(\hat{a},\hat{d}_1, \hat{d}_2,\hat{b},\vec{\hat{c}})$ is true,    obviously $\mathbb{W}^{F}(\hat{k},\hat{d}_1, \hat{d}_1,\hat{b},\vec{\hat{c}})$ also has to be true.  Then, by the induction hypothesis, $\arfour\vdash \bound\equals\hat{b}\mli \overline{F(\hat{b},\vec{\hat{c}})}$, i.e. $\arfour\vdash \bound\equals\hat{b}\mli H[G_i]$. \end{subproof}

{\bf Claim 2.} {\em Assume $\overline{E(\hat{b},\vec{\hat{c}})}$ has the form $H[\ada yG(y)]$, and $r$ is an arbitrary non-$\bound$ variable not occurring in $E(\bound,\vec{x})$.   Then $\arfour\vdash \bound\equals \hat{b}\mli H[G(r)]$.}\vspace{5pt}

\begin{subproof} Assume the conditions of the claim. Let $F=F(\bound,\vec{s},r)$ be the    formula such that $\overline{F(\hat{b},\vec{\hat{c}},r)}=H[G(r)]$.  For each constant $m$ whose size does not exceed $b$, let $\beta_m$ be the environment's move that brings $H[\ada yG(y)]$ to $H[G(m)]$, and  let $k_m$ be the (code of the) counterbehavior obtained by appending the timestamped move $(\beta_m,d_1\mminus 1)$ to (the counterbehavior whose code is) $a$. 
 Since $\mathbb{W}^E(\hat{a},\hat{d}_1, \hat{d}_2,\hat{b},\vec{\hat{c}})$ is true,    obviously, for each constant $m$ with $|m|\mleq b$, $\mathbb{W}^{F}(\hat{k}_m,\hat{d}_1, \hat{d}_1,\hat{b},\vec{\hat{c}},\hat{m})$ is also true.  Then, by the induction hypothesis, $\arfour\vdash \bound\equals\hat{b}\mli \overline{F(\hat{b},\vec{\hat{c}},\hat{m}})$, i.e. $\arfour\vdash \bound\equals\hat{b}\mli H[G(\hat{m})]$. But then, by Lemma \ref{pathology2}, $\arfour\vdash \bound\equals\hat{b}\mli \bound\equals\hat{b}\mli  \overline{F(\hat{b},\vec{\hat{c}},r)}$, i.e. $\arfour\vdash \bound\equals\hat{b}\mli \bound\equals\hat{b}\mli   H[G(r)] $. By $\clfour$-Instantiation, we also have \[\arfour\vdash \bigl(\bound\equals\hat{b}\mli \bound\equals\hat{b}\mli   H[G(r)]\bigr)\mli \bigl(\bound\equals\hat{b}\mli   H[G(r)]\bigr) \]
(this is an instance of $(p\mli p\mli Q)\mli (p\mli Q)$). So, by Modus Ponens,  $\arfour\vdash  \bound\equals\hat{b}\mli   H[G(r)]$.  \end{subproof}

From (\ref{jan5b}), Claim 1 and Claim 2, by Wait, we find the desired  $\arfour \vdash \bound\equals\hat{b}\mli \overline{E(\hat{b},\hat{\vec{c}})}$.  
\end{proof}

\begin{lemma}\label{pathology5}
For any positive integer $b$, 
$\arfour\vdash \bound\equals \hat{b}\mli \overline{X(\bound)}$.  \end{lemma} 

\begin{proof} Consider any positive integer $b$.  Let $a$ be (the code of) the empty counterbehavior. Of course, $\mathbb{W}^X(\hat{a},\hat{1},\hat{1},\hat{b})$ is true. Then, by Lemma \ref{pathology4}, $\arfour\vdash \bound\equals \hat{b}\mli \overline{X(\hat{b})}$. But the formula \[\bigl(\bound\equals \hat{b}\mli \overline{X(\hat{b})}\bigr)\mli \bigl(\bound\equals \hat{b}\mli \overline{X(\bound)}\bigr)\] is an instance of the $\clfour$-provable $\bigl(\bound\equals f\mli P(f)\bigr)\mli \bigl(\bound\equals f\mli P(\bound )\bigr)$ and, by $\clfour$-Instantiation, is provable in $\arfour$. Hence, by Modus Ponens, $\arfour\vdash  \bound\equals \hat{b}\mli \overline{X(\bound)}$. \end{proof}

\begin{lemma}\label{pathology6}
$\arfour\vdash \gneg| \xi(\bound)|\mleq\bound\mli   \overline{X(\bound)}$.  \end{lemma} 

\begin{idea} $\arfour$ knows that, if $\gneg| \xi(\bound)|\mleq\bound$, then $\bound\equals\hat{b}$ for one of finitely many particular (``very small'') positive integers $b$. Furthermore, as in Lemma \ref{pathology3}, we can show that such knowledge is constructive, in the sense that $\arfour$ can tell ($\add$) exactly for which $b$ do we have $\bound\equals\hat{b}$. Then the desired conclusion easily follows from Lemma \ref{pathology5}. 
\end{idea}

\begin{proof}  The size of $\xi(\bound)$ can be greater than $\bound$ for only a certain finite number of ``small'' non-$0$ values of $\bound$. Let $N$ be the set of all such values. Obviously 
\[\pa\vdash \bound\notequals\zero\mli  \gneg| \xi(\bound)|\mleq\bound\mli \mld\{\bound\equals \hat{a}\ |\ a\in N\},\]
modus-ponensing which with Lemma \ref{zer} yields 
\begin{equation}\label{jan6e}
\arfour\vdash   \gneg| \xi(\bound)|\mleq\bound\mli \mld\{\bound\equals \hat{a}\ |\ a\in N\}.
\end{equation} 

By Lemma \ref{pathology5}, for each  $a\in N$ we have $\arfour\vdash \bound\equals \hat{a}\mli \overline{X(\bound)}$. Hence, by $\adc$-Introduction, 
\begin{equation}\label{jan6f}
\arfour\vdash \add\{\bound\equals \hat{a}\ |\ a\in N\}\mli \overline{X(\bound)}.
\end{equation}

Next we claim that
\begin{equation}\label{jan6h}
\mbox{\em for each $a\in N$, $\pa\vdash\bound\equals\hat{a}\add\bound\notequals\hat{a}$}.
\end{equation}
Below is a justification of this claim for an arbitrary $a\in N$: \vspace{7pt}

\noindent 1. $\begin{array}{l} 
\gneg |\hat{a}|\mleq \bound\add \ade z (z\equals \hat{a})
\end{array}$ \ \  Lemma \ref{com} \vspace{3pt}

\noindent 2. $\begin{array}{l} 
\gneg |\hat{a}|\mleq \bound\mli \bound\notequals\hat{a}
\end{array}$ \ \  $\pa$\vspace{3pt}

\noindent 3. $\begin{array}{l} 
\gneg |\hat{a}|\mleq \bound\mli \bound\equals\hat{a}\add\bound\notequals\hat{a}
\end{array}$ \ \  $\add$-Choose: 2\vspace{3pt}

\noindent 4. $\begin{array}{l} 
\ada x \ada y(y\equals x\add y\notequals x)
\end{array}$ \ \ Lemma \ref{nov18a}  \vspace{3pt}

\noindent 5. $\begin{array}{l} 
s\equals\bound\add s\notequals \bound
\end{array}$ \ \ $\ada$-Elimination (twice): 4  \vspace{3pt}

\noindent 6. $\begin{array}{l} 
s\equals\bound \mli s\equals \hat{a}  \mli \bound\equals\hat{a} 
\end{array}$ \ \ Logical axiom  \vspace{3pt}

\noindent 7. $\begin{array}{l} 
s\equals\bound \mli s\equals \hat{a}  \mli \bound\equals\hat{a}\add\bound\notequals\hat{a}
\end{array}$ \ \  $\add$-Choose: 6 \vspace{3pt}

\noindent 8. $\begin{array}{l} 
s\notequals \bound\mli s\equals \hat{a}  \mli \bound\notequals\hat{a}
\end{array}$ \ \ Logical axiom  \vspace{3pt}

\noindent 9. $\begin{array}{l} 
s\notequals \bound\mli s\equals \hat{a}  \mli \bound\equals\hat{a}\add\bound\notequals\hat{a}
\end{array}$ \ \  $\add$-Choose: 8 \vspace{3pt}

\noindent 10. $\begin{array}{l} 
s\equals \hat{a}  \mli \bound\equals\hat{a}\add\bound\notequals\hat{a}
\end{array}$ \ \ $\add$-Elimination: 5,7,9  \vspace{3pt}

\noindent 11. $\begin{array}{l} 
\ade z (z\equals \hat{a}) \mli \bound\equals\hat{a}\add\bound\notequals\hat{a}
\end{array}$ \ \ $\ada$-Introduction: 10  \vspace{3pt}

\noindent 12. $\begin{array}{l} 
\bound\equals\hat{a}\add\bound\notequals\hat{a}
\end{array}$ \ \  $\add$-Elimination:  1,3,11 \vspace{7pt}

The following formula can be easily seen to be provable in $\clthree$ and hence in $\arfour$:
\[ \arfour\vdash \mlc\{\bound\equals \hat{a}\add\bound\notequals\hat{a}\ |\ a\in N\}\mli \mld\{\bound\equals \hat{a}\ |\ a\in N\} \mli \add\{\bound\equals \hat{a}\ |\ a\in N\}.\]
Modus-ponensing the above with (\ref{jan6h}) yields
\begin{equation}\label{jan6g}
 \arfour\vdash \mld\{\bound\equals \hat{a}\ |\ a\in N\} \mli \add\{\bound\equals \hat{a}\ |\ a\in N\} .
\end{equation}

Now, from (\ref{jan6e}), (\ref{jan6g}) and (\ref{jan6f}), by Transitivity applied twice, we get   $\arfour\vdash \gneg| \xi(\bound)|\mleq\bound\mli   \overline{X(\bound)}$ as desired.
 \end{proof}

\subsection{Ptarithmetizing  HPM-computations}
In this subsection we prove the earlier-mentioned ``provable traceability'' of the work of $\cal X$, in a certain technically strong form necessary for our further treatment. As we remember, roughly it means the constructive knowledge by $\arfour$ of the configurations of $\cal X$ in its interaction with a given adversary (the latter thought of as a counterbehavior). The present elaboration is the first relatively advanced example of ``{\em ptarithmetization}''\label{iptarz} or, more generally, ``{\em clarithmetization}''\label{iclarz} --- extending G\"{o}del's {\em arithmetization} technique from the classical context to the context of computability logic. 

Let \mbox{\em STATES} be the set of all states of the machine $\cal X$, and \mbox{\em SYMBOLS} be the set of all symbols that may appear on any of its tapes. As we know, both sets are finite. We assume that the cells of each of the three tapes are numbered consecutively starting from $0$ (rather than $1$).

Below we introduce elementary formulas that naturally arithmetize the corresponding metapredicates. 

\begin{itemize}
\item $\mbox{\em Adequate}(z,w,t)$ means ``$z$ is a $(w,t)$-adequate counterbehavior''.
\item For each $a\in\mbox{\em STATES}$, $\mbox{\em State}_a(z,w,t)$ means ``In the $(z,e_w)$-branch, at time $t$, $\cal X$ is  in state $a$''. 
\item  For each $a\in\mbox{\em SYMBOLS}$, $\mbox{\em VSymbol}_{a}(z,w,t,u)$ means ``In the $(z,e_w)$-branch, at time $t$, cell $\#u$ of the valuation tape contains symbol $a$''. Similarly for $\mbox{\em WSymbol}_{a}(z,w,t,u)$ (for the work tape) and $\mbox{\em RSymbol}_{a}(z,w,t,u)$ (for the run tape).  
\item  $\mbox{\em VHead}(z,w,t,u)$ means ``In the $(z,e_w)$-branch, at time $t$, the head of the valuation tape is over cell $\#u$''.  Similarly for $\mbox{\em WHead}(z,w,t,u)$ (for the work tape) and $\mbox{\em RHead}(z,w,t,u)$ (for the run tape).  
\item $\mbox{\em Runsize}(z,w,t,u)$ means ``In the $(z,e_w)$-branch, at time $t$, the leftmost blank cell of the run tape is cell $\#u$''.
\item $\mathbb{E}(z,t)$\label{ieee} abbreviates 
\[
\begin{array}{l}
\mbox{\em Adequate}(z,\bound,t)\ \mlc \\
\add\{\mbox{\em State}_a(z,\bound,t)\ |\ a\in \mbox{\em STATES}\}\ \mlc \\
\Bigl(\cle x \bigl(\mbox{\em Runsize}(z,\bound,t,x)\mlc |x|\mleq\bound\bigr)\adi  \ade x \mbox{\em Runsize}(z,\bound,t,x)\Bigr)\ \mlc \\
\ade x\bigl(\mbox{\em VHead}(z,\bound,t,x)\mlc \add \{\mbox{\em VSymbol}_a(z,\bound,t,x)\ |\ a\in \mbox{\em SYMBOLS}\}\bigr)\ \mlc \\
\ade x\bigl(\mbox{\em WHead}(z,\bound,t,x)\mlc \add \{\mbox{\em WSymbol}_a(z,\bound,t,x)\ |\ a\in \mbox{\em SYMBOLS}\}\bigr)\ \mlc \\
\ade x\bigl(\mbox{\em RHead}(z,\bound,t,x)\mlc \add \{\mbox{\em RSymbol}_a(z,\bound,t,x)\ |\ a\in \mbox{\em SYMBOLS}\}\bigr).
\end{array}\]
\item $\mathbb{F}(z,t)$\label{ifff} abbreviates 
\[
\begin{array}{l}
\ada x\bigl(\add\{\mbox{\em VSymbol}_a(z,\bound,t,x)\ |\ a\in \mbox{\em SYMBOLS}\}\bigr)\ \mlc \\ 
\ada x\bigl(\add\{\mbox{\em WSymbol}_a(z,\bound,t,x)\ |\ a\in \mbox{\em SYMBOLS}\}\bigr)\ \mlc \\ 
\Bigl(\ada x\bigl(\add\{\mbox{\em WSymbol}_a(z,\bound,t,x)\ |\ a\in \mbox{\em SYMBOLS}\}\bigr)\adc \ada x\bigl(\add\{\mbox{\em RSymbol}_a(z,\bound,t,x)\ |\ a\in \mbox{\em SYMBOLS}\}\bigr)\Bigr). 
\end{array}\]
\end{itemize}
Note that both formulas $\mathbb{E}(z,t)$ and $\mathbb{F}(z,t)$, in addition to $z$ and $t$, contain $\bound$ as a free variable, which we however do not explicitly indicate as it will never be replaced by any other term.

We use $\cle !$ as a standard abbreviation, defined by 
\[\cle ! z T(z) \ =\ \cle z\Bigl(T(z)\mlc \cla y\bigl(T(y)\mli y\equals z)\bigr)\Bigr).\]

Let $z$ be any variable and $T$ --- let us (also) write it in the form $T(z)$ --- any elementary formula. We say that $T$ is {\bf functional for $z$}\label{iffz} iff $\arfour\vdash \cle ! zT(z)$.  

For  variables $z$,$t$ and an elementary formula $T=T(z)$ functional for $z$, we will be using $\mathbb{E}( z^T,t)$\label{i876} as an abbreviation defined by 
\[\mathbb{E}( z^T,t) \ =\ \cla z\bigl(T(z)\mli \mathbb{E}(z,t)\bigr).\]
Similarly for $\mathbb{F} ( z^T,t )$. It is our convention that, whenever using these abbreviations, the variables $z$ and $t$ are not the same, so that $t$ does not get bound by the external $\cla z$. Similarly, if we write $\mathbb{E} ( z^T,\theta )$ or $\mathbb{F} ( z^T,\theta )$ where $\theta$ is a term, it will be assumed that $\theta$ does not contain $z$.

\begin{lemma}\label{oct17a}
For any elementary formula $T$ functional for $z$, $\arfour$ proves
\[\mathbb{E} ( z^T, t)\mli \mathbb{E} ( z^T, t)\mlc \mathbb{E} ( z^T,t).\] 
\end{lemma}

\begin{idea} As explained in Subsection \ref{gggg}, $\mathbb{E}$ --- whether in the form  $\mathbb{E}(z, t )$ or $\mathbb{E} ( z^T, t)$ --- is essentially a ``recyclable'' resource because it does not contain $\adc$,$\ada$. \end{idea}
 
\begin{proof} Bottom-up, a proof of the target formula goes like this. Keep applying $\ada$-Introduction and $\adc$-Introduction   until the antecedent (in the given branch of the proof tree) becomes 
\[
\begin{array}{l}
\cla z \Bigl(T\mli \\
\mbox{\em Adequate}(z,\bound,t)\ \mlc \\
\mbox{\em State}_a(z,\bound,t)\ \mlc \\
 \mbox{\em Runsize}(z,\bound,t,u)\ \mlc \\
\bigl(\mbox{\em VHead}(z,\bound,t,v)\mlc \mbox{\em VSymbol}_b(z,\bound,t,v)\bigr)\ \mlc \\
\bigl(\mbox{\em WHead}(z,\bound,t,w)\mlc \mbox{\em WSymbol}_c(z,\bound,t,w)\bigr)\ \mlc \\
\bigl(\mbox{\em RHead}(z,\bound,t,r)\mlc \mbox{\em VSymbol}_d(z,\bound,t,r)\bigr)\Bigr) 
\end{array}\]
--- or, maybe, the same but with ``$\gneg\cle x\bigl( \mbox{\em Runsize}(z,\bound,t,x)\mlc |x|\mleq\bound\bigr)$'' instead of ``$ \mbox{\em Runsize}(z,\bound,t,u) $'' --- for some variables $u,v,w,r$, state $a$ and symbols $b,c,d$. 
Then apply a series of $\add$-Chooses and $\ade$-Chooses and bring the consequent to a conjunction of two copies of the antecedent. Now we are dealing with a classically valid and hence provable elementary formula of the form $F\mli F\mlc F$. 
\end{proof}

\begin{lemma}\label{oct17c}
For any elementary formula $T$ functional for $z$,  
$\arfour$ proves
\begin{equation}\label{jan10aa}
\mathbb{E} ( z^T,  t )\mlc \mathbb{F} ( z^T,  t )\mli \mathbb{F} ( z^T, t\successor  ).
\end{equation} 
\end{lemma}

\begin{idea} For reasons in the spirit of an explanation given in Subsection \ref{gggg}, a single copy of the resource $\mathbb{E} ( z^T,  t) $ and a single copy of the resource $\mathbb{F} ( z^T,  t )$ turn out to be sufficient so solve the problem $\mathbb{F} ( z^T,  t\successor )$. \end{idea}

\begin{proof} 
The following formula is provable in $\clfour$ by Match applied three times: 
\begin{equation}\label{jan10cc}
\begin{array}{l} 
\cla z\bigl(P_1(z)\mlc P_2(z)\mli P_3(z)\bigr)\ \mli\  \Bigl(\cla z\bigl(q(z)\mli P_1(z)\bigr)\mlc \cla z\bigl(q(z)\mli P_2(z)\bigr)\mli \cla z(q(z)\mli P_3(z)\bigr)\Bigr). 
\end{array}  
\end{equation}
 Consider the formula 
\begin{equation}\label{jan10dd} 
\cla z\Bigl(\mathbb{E} ( z,  t )\mlc \mathbb{F} ( z,  t ) \mli  \mathbb{F} ( z, t\successor  )\Bigr).
\end{equation}
The  formula $(\ref{jan10dd})\mli(\ref{jan10aa})$, which --- after disabbreviating $z^T$ in (\ref{jan10aa})  --- is 
\[\cla z\Bigl(\mathbb{E} ( z,  t )\mlc \mathbb{F} ( z,  t ) \mli  \mathbb{F} ( z, t\successor  )\Bigr)\mli \Bigl(\cla z \bigl(T(z)\mli \mathbb{E} ( z ,  t )\bigr)\mlc \cla z \bigl(T(z)\mli \mathbb{F} ( z ,  t )\bigr)\mli \cla z \bigl(T(z)\mli\mathbb{F} ( z , t\successor  )\bigr) \Bigr),\]
can be seen to be an instance of (\ref{jan10cc}) and hence, by $\clfour$-Instantiation, provable in $\arfour$. Therefore, if $\arfour$ proves (\ref{jan10dd}), then, by Modus Ponens, it also proves the target (\ref{jan10aa}). Based on this observation, we now forget about (\ref{jan10aa}) and, in what follows, exclusively devote our efforts to showing that  $\arfour\vdash (\ref{jan10dd})$.

This is one of those cases where giving a full formal proof in the style  practiced earlier is not feasible. But by now we have acquired enough experience in working with $\arfour$ to see that the informal argument provided below can be translated  into a strict $\arfour$-proof if necessary.

Argue in $\arfour$.  Consider an arbitrary ($\cla$) counterstrategy $z$.  
The context of our discourse will be the play of 
$\cal X$ against  $z$ on the standard valuation $e_\bound$ --- the $(z,e_\bound)$-branch, that is. Assume that a single copy of the antecedental resource $ \mathbb{E} ( z,  t )\mlc \mathbb{F} ( z,  t )$ is at our disposal. We need to show how to resolve the consequental problem $ \mathbb{F} ( z, t\successor  )$.  

For resolving the first conjunct of $ \mathbb{F} ( z, t\successor  )$, we need to tell, for an arbitrary ($\ada$) given $x$, the content of cell $\#x$ of the valuation tape at time $t\successor$. This is very easy: the content of the valuation tape never changes. So, the symbol in cell $\#x$ at time $t\successor$ will be the same as at time $t$, and what symbol it is we can learn from the first conjunct of (the antecedental resource) $ \mathbb{F} ( z, t)$. In more detailed  terms, a solution/deduction strategy corresponding to the above outline is to wait (bottom-up $\ada$-introduction) till the environment specifies a value (syntactically, a ``fresh'' variable) $s$ for $x$ in the first $\mlc$-conjunct of 
$ \mathbb{F} ( z, t\successor)$; then, using the same $s$ (bottom-up $\ade$-Choose), specify the value of $x$ in the first $\mlc$-conjunct of $ \mathbb{F} ( z, t)$; after that, wait 
(bottom-up $\adc$-introduction)  till the environment selects one of the $\add$-disjuncts in the first $\mlc$-conjunct of $ \mathbb{F} ( z, t)$ (or rather of what that formula has become), and then select (bottom-up $\add$-Choose) the same $\add$-disjunct in the first $\mlc$-conjunct of $ \mathbb{F} ( z, t\successor)$. Henceforth we will no longer provide such details, and will limit ourselves to just describing strategies, translatable (as we just saw) into bottom-up $\arfour$-deductions.

 For resolving the second conjunct of $ \mathbb{F} ( z, t\successor  )$, we need to tell, for an arbitrary ($\ada$) given $x$, the content of cell $\#x$ of the work  tape at time $t\successor$. This is not hard, either. At first, using  the fifth conjunct of $ \mathbb{E} ( z, t )$, we determine the location $m$ of the work-tape head  and the tape symbol $c_W$ at that location at time $t$. If $m\not=x$ (Lemma \ref{nov18a} can be used to tell whether this is the case or not), then the symbol in cell $\# x$ at time $t\successor$ will remain the same $c_W$. Suppose now $m=x$. Then we further  use the second,  fourth and sixth conjuncts of $ \mathbb{E} ( z, t )$ to learn about  the state $a$ of the machine at time $t$ and the symbols $c_V$ and $c_R$ scanned at that time by the heads of the valuation and the run tapes. Now, knowing $c_V,c_W,c_R$ and $a$, based on the transition function of $\cal X$, we can tell what symbol will be written in cell $\#x$ of the work tape by time $t\successor$.  

The left $\adc$-conjunct of the third $\mlc$-conjunct of $ \mathbb{F} ( z, t\successor  )$ is identical to the second $\mlc$-conjunct of $ \mathbb{F} ( z, t\successor  )$, and it can be resolved as we just saw above. However, to avoid an (unacceptable/unavailable) repeated usage of resources, we will employ the first $\adc$-conjunct of the third $\mlc$-conjunct of $ \mathbb{F} ( z, t )$ instead of the second $\mlc$-conjunct of $ \mathbb{F} ( z, t )$ as was done in the previous case. Of course, we will also need to use some parts of the resource  $ \mathbb{E} ( z, t )$ which were already used by the procedure of the previous case. This, however, does not create any resource conflicts. Because any information extracted from $ \mathbb{E} ( z, t )$ earlier is still there, so the relevant parts of $ \mathbb{E} ( z, t )$ do not really need to be ``queried'' again, as we already know answers.   That (re)using $\mathbb{E}(z,t)$ does not create any competition for resources should be remembered through the remaining part of this proof and the proof of the following lemma as well. This phenomenon of the ``recycleability'' of $ \mathbb{E} ( z, t )$ was, in fact, already established in Lemma \ref{oct17a}. 

Finally, for resolving the right $\adc$-conjunct of the third $\mlc$- conjunct of $ \mathbb{F} ( z, t\successor  )$, we need to tell, for an arbitrary ($\ada$) given $x$, the content of cell $\#x$ of the run  tape at time $t\successor$. This is how it can be done. Let us call $j$ the location of the leftmost blank cell of the run tape at time $t$. At first, we wait till the environment selects one of the $\add$-disjuncts of the third $\mlc$-conjunct of $\mathbb{E}(z,t)$. If the left disjunct is selected, then  $\bound\mless |j|$ (or else the selected disjunct is false and we win). Then we also have ($|x|\mless |j|$ and hence) $x\mless j$, because the size of (the $\add$-bound) $x$ cannot exceed $\bound$. 
If the right disjunct is selected instead, the environment will have to further provide the actual value of $j$. Then, using Lemma \ref{minus}, we can figure out whether $x\mless j$ or not. Thus, in either case, we will  know whether  $x\mless j$ or $x\mgeq j$ and, if $x\mgeq j$, we  will also know the value of $j$.  
First, suppose   $x\mless j$. Then  the content of cell $\#x$ at time $t\successor$ is obviously the same as at time $t$, and  information about this content can be obtained from the right $\adc$-conjunct of the third $\mlc$-conjunct of $\mathbb{F}(z, t)$.\footnote{The third $\wedge$-conjunct of $\mathbb{F}(z, t)$ was already used in the previous paragraph. But there is no resource conflict here, as we have a choice (rather than parallel) conjunction between the problems whose solutions are described in the present and the previous paragraphs, so that only one of them  will actually have to be solved.} Similarly if the state of $\cal X$ was not a move state at time $t$ (and information about whether this was the case is available from the second conjunct of $\mathbb{E}(z,t)$). 
Now assume (we know the value of $j$ and) $x\mgeq j$, and also assume 
 the   state of $\cal X$ at time $t$ was a move state.   If $x\equals j$ (use Lemma \ref{nov18a} to tell if this is so or not), then the content of cell $\#x$ at time $t\successor$   will be the symbol $\pp$. Otherwise, if $x\notequals j$, meaning that $x\mgreater j$, 
then the content of cell $\#x$ at time $t\successor$   will be the content $c$ of cell $\#(x\mminus j\mminus 1)$ of the work tape at time $ t$ (Lemma \ref{minus} can again be used to compute the value of $x\mminus j\mminus 1$). Such a $c$ can be found using the left $\adc$-conjunct of the third $\mlc$-conjunct of $\mathbb{F}(z,t)$. Well, what we just said is true unless $x\mminus j\mminus 1$ is greater than or equal to the location of the work-tape head at time $ t$ (known from $\mathbb{E}(z,t)$), in which case the content of cell $\#x$ of the run tape  at time $t\successor$ will be blank. 
\end{proof}

\begin{lemma}\label{oct17b}
For any elementary formula $T$ functional for $z$,   
$\arfour$ proves
\begin{equation}\label{jan10a}
|t\successor|\mleq \bound \mlc \mathbb{E} ( z^T,  t )\mlc \mathbb{F} ( z^T,  t )\mli \mathbb{E} ( z^T, t\successor  ).
\end{equation} 
\end{lemma}

\begin{idea} 
As in the previous lemma, a single copy of the resource $\mathbb{E} ( z^T,  t) $ and a single copy of the resource $\mathbb{F} ( z^T,  t )$ turn out to be sufficient so solve the problem $\mathbb{E} ( z^T,  t\successor )$. A minor additional technical condition for this in the present case   is that the size of $t\successor$ should not exceed $\bound$.  \end{idea}

\begin{proof} For reasons similar to those given at the beginning of the proof of Lemma \ref{oct17c}, it would be sufficient to show the $\arfour$-provability of the following formula instead of (\ref{jan10a}):  
\begin{equation}\label{jan10c}
\begin{array}{l} 
\cla z\Bigl(|t\successor|\mleq \bound \mlc \mathbb{E} ( z,  t )\mlc \mathbb{F} ( z,  t )\mli \mathbb{E} ( z, t\successor  )\Bigr).
\end{array}  
\end{equation}

Argue in $\arfour$.  Consider an arbitrary ($\cla$) counterstrategy $z$. 
As in the proof of the previous lemma, the context of our discourse will be the play   
according to the scenario of the $(z,e_\bound)$-branch. Assume $|t\successor|\mleq \bound $. And assume that a single copy of the  resource $ \mathbb{E} ( z,  t )\mlc \mathbb{F} ( z,  t )$ is at our disposal. We need to show how to resolve $ \mathbb{E} ( z, t\successor  )$.  

The first conjunct of $\mathbb{E} ( z,  t )$ is $\mbox{\em Adequate}(z,\bound,t)$. It implies that the environment does not move at $t$ or any greater time, so that $z$ will remain adequate for any value greater than $t$ as well. Thus, $\mbox{\em Adequate}(z,\bound,t\successor)$ is true, which takes care of the first conjunct of $\mathbb{E} ( z, t\successor  )$. 

The resource $\mathbb{E} ( z,  t )$  contains full information about the state of the machine  at time $ t$, the locations of the three scanning heads, and the symbols at those three locations. This allows us to determine the next state, and the next locations of the heads (``next'' means ``at time $t\successor$''). Note that we will have no problem naming those locations, as they cannot exceed $t'$ (moving a head farther than cell $\#t\successor$ would require more than $t\successor$ steps)  and hence, in view of the assumption $|t\successor|\mleq\bound$, their sizes cannot exceed $\bound$. What we just said fully takes care of the second conjunct of $\mathbb{E} ( z, t\successor  )$, and partially takes care of the fourth, fifth and sixth conjuncts. To turn this ``partial care'' into a full one, we need to show how to tell the symbols looked at by the three heads at time $t\successor$.   

 The content of the cell scanned by the valuation-tape head at time $t\successor $ will be the same as the content of that cell at time $t$, and this information can be obtained from the first conjunct of $\mathbb{F} ( z ,t )$. 

Since scanning heads (almost) always move left or right, the content of the cell scanned by the work-tape head at time $ t\successor $ will generally also be the same as the content of that cell 
at time $ t$, which can be obtained from the second conjunct of $\mathbb{F} ( z ,  t )$. An exception is when the head is at the beginning of the tape at time $t$, writes a new symbol and tries to move left which, however, results in staying put. In such a case, we can obtain the symbol just written 
(i.e., the content of the cell scanned by the head at time $ t\successor $) directly from our knowledge of the transition function and our knowledge --- already obtained earlier from $\mathbb{E} ( z,  t )$ --- of the state of $\cal X$ and the contents of the three cells scanned at time $ t$.  

Let the cell scanned by the head of the run tape at time $t\successor$ be cell $\# i$ (the value of $i$ has already been established earlier). Let the leftmost blank cell of that tape at time $t$ be cell $\#j$. Since the run-tape head can never move past the leftmost blank cell, we have either $i\equals j$ or ($i\notequals j$ and hence) $i\mless j$.  The third conjunct of $\mathbb{E}(z,t)$ in combination with Lemma \ref{nov18a} can be used to tell which of these two alternatives is the case. If $i\mless j$, then  the content of the run-tape cell $\#i$ at time $t\successor$ will be the same as at time $t$, and this information can be obtained from the right$\adc$-conjunct of the third $\mlc$-conjunct of $\mathbb{F}(z, t)$. Similarly if the state of $\cal X$ was not a move state at time $t$ (and information about whether this was the case is available from the second conjunct of $\mathbb{E}(z,t)$). Assume now $i\equals j$, and the state of $\cal X$ at time $t$ was a move state. 
Then the content of cell $\#i$ at time $t\successor$   will be the symbol $\pp$ (the label of the move made at time $t$).

The above three paragraphs complete taking care of the fourth, fifth and sixth conjuncts of $ \mathbb{E} ( z, t\successor  )$.

Finally, to solve the remaining third conjunct of   $ \mathbb{E} ( z, t\successor  )$, wait till the environment selects one of the two $\add$-disjuncts of the third conjunct of  $ \mathbb{E} ( z, t  )$. If the left disjunct is selected there, do the same in the third conjunct of   $ \mathbb{E} ( z, t\successor  )$. Suppose now the right conjunct is selected. Wait till the environment further specifies a value $j$ for $x$ there. If $\cal X$ is not in a move state at time $t$, do the exact same selections in the third conjunct of   $ \mathbb{E} ( z, t\successor  )$. Suppose now $\cal X$ is in a move state at time $t$. Then the location of the leftmost blank cell at time $t\successor$ will be $j\plus i\plus 1$, where $i$ is the location of the work-tape head at time $t$. Using the results of Section \ref{s17}, try to compute $m$ with $m\equals j\plus i\plus 1$. If $|m|$ turns out to exceed $\bound$, select the left $\add$-disjunct of the third conjunct of   $ \mathbb{E} ( z, t\successor  )$. Otherwise select the right disjunct, and specify $x$ as $m$ there. 
\end{proof}

\begin{lemma}\label{oct17d}
For any elementary formula $T$ functional for $z$,  
$\arfour$ proves
\[|t\successor |\mleq\bound \mlc  \mathbb{E} ( z^T,  t )\mlc \mathbb{F} ( z^T,  t )\mli \mathbb{E} ( z^T,  t )\mlc \bigl(\mathbb{E} ( z^T,  t\successor  )\adc \mathbb{F} ( z^T,  t\successor  )\bigr).\] 
\end{lemma}

\begin{idea} This is a logical consequence of the previous three lemmas (i.e. a consequence exclusively due to logical axioms and rules, without appealing to induction or any nonlogical axioms of $\arfour$). Correspondingly, the proof given below is a purely syntactic exercise. \end{idea}

\begin{proof} The following sequence is a $\clfour$-proof: \vspace{7pt}

\noindent 1. $\begin{array}{l} 
(p_{1} \mli p_{2} \mlc \twg)\mlc (\tlg \mlc \tlg \mli \twg)\mlc (q\mlc \tlg \mlc \tlg \mli \twg)\mli q\mlc p_1\mlc \twg\mli p_2\mlc   \twg 
\end{array}$ \ \  Tautology \vspace{3pt}

\noindent 2. $\begin{array}{l} 
(p_{1} \mli p_{2} \mlc p_{3})\mlc (\tlg \mlc \tlg \mli \twg)\mlc (q\mlc p_3\mlc p_4\mli p_5)\mli q\mlc p_1\mlc p_4\mli p_2\mlc    p_5  
\end{array}$ \ \  Tautology \vspace{3pt}

\noindent 3. $\begin{array}{l} 
(p_{1} \mli p_{2} \mlc p_{3})\mlc (Q_{1} \mlc Q_{2} \mli Q_{4})\mlc (q\mlc p_3\mlc p_4\mli p_5)\mli q\mlc p_1\mlc p_4\mli p_2\mlc    p_5  
\end{array}$ \ \  Wait: 2 \vspace{3pt}

\noindent 4. $\begin{array}{l} 
(p_{1} \mli p_{2} \mlc Q_{1})\mlc (Q_{1} \mlc Q_{2} \mli Q_{4})\mlc (q\mlc Q_1\mlc Q_2\mli Q_3)\mli q\mlc p_1\mlc Q_2\mli p_2\mlc    Q_3  
\end{array}$ \ \  Match (3 times): 3 \vspace{3pt}

\noindent 5. $\begin{array}{l} 
(p_{1} \mli p_{2} \mlc p_{3})\mlc (p_{3} \mlc p_{4} \mli p_{5})\mlc (q\mlc \tlg\mlc \tlg\mli \twg)\mli q\mlc p_1\mlc p_4\mli p_2\mlc    p_5 
\end{array}$ \ \  Tautology \vspace{3pt}

\noindent 6. $\begin{array}{l} 
(p_{1} \mli p_{2} \mlc p_{3})\mlc (p_{3} \mlc p_{4} \mli p_{5})\mlc (q\mlc Q_1\mlc Q_2\mli Q_3)\mli q\mlc p_1\mlc p_4\mli p_2\mlc    p_5 
\end{array}$ \ \ Wait: 5  \vspace{3pt}

\noindent 7. $\begin{array}{l} 
(p_{1} \mli p_{2} \mlc Q_{1})\mlc (Q_{1} \mlc Q_{2} \mli Q_{4})\mlc (q\mlc Q_1\mlc Q_2\mli Q_3)\mli q\mlc p_1\mlc Q_2\mli p_2\mlc    Q_4 
\end{array}$ \ \ Match (3 times): 6 \vspace{3pt}

\noindent 8. $\begin{array}{l} 
(p_{1} \mli p_{2} \mlc Q_{1})\mlc (Q_{1} \mlc Q_{2} \mli Q_{4})\mlc (q\mlc Q_1\mlc Q_2\mli Q_3)\mli q\mlc p_1\mlc Q_2\mli p_2\mlc   (Q_3\adc Q_4) 
\end{array}$ \ \  Wait: 1,4,7 \vspace{3pt}

\noindent 9. $\begin{array}{l} 
(Q_{1} \mli Q_{1} \mlc Q_{1})\mlc (Q_{1} \mlc Q_{2} \mli Q_{4})\mlc (q\mlc Q_1\mlc Q_2\mli Q_3)\mli q\mlc Q_1\mlc Q_2\mli Q_1\mlc   (Q_3\adc Q_4) 
\end{array}$   \mbox{Match (twice): 8}\vspace{7pt}
 
The following formula matches the last formula of the above sequence and hence, by $\clfour$-Instantiation, it is provable  in $\arfour$:
\[\begin{array}{l}
\bigl(\mathbb{E} ( z^T,  t )\mli \mathbb{E} ( z^T, t )\mlc \mathbb{E} ( z^T,  t )\bigr)\mlc  
\bigl(  \mathbb{E} ( z^T,  t)\mlc \mathbb{F} ( z^T,  t  )\mli \mathbb{F} ( z^T, t\successor )\bigr)  \mlc  
\bigl(  |t\successor |\mleq\bound  
\mlc  \mathbb{E} ( z^T,  t )\mlc \mathbb{F} ( z^T, t ) \mli \mathbb{E} ( z^T,  t\successor )\bigr)   
 \\
\mli |t\successor |\mleq\bound\mlc \mathbb{E} ( z^T,  t )\mlc \mathbb{F} ( z^T,  t )\mli \mathbb{E} ( z^T,  t )\mlc\bigl(\mathbb{E} ( z^T, t\successor )\adc \mathbb{F} ( z^T,  t\successor  )\bigr).
\end{array}\]

But, by Lemmas \ref{oct17a}, \ref{oct17c} and \ref{oct17b}, the three conjuncts of the antecedent of the above formula are also provable. Hence, by Modus Ponens, so is (the desired) consequent. 
\end{proof}

\begin{lemma}\label{oct17e}
Assume $R$ is an  elementary formula, $w$ is any variable, $t$ is a  variable other than $\bound$, $z$ is a  variable other than $\bound,w,t$,  $T$ is an  elementary formula functional for $z$,  and  
\begin{equation}\label{nov29a}
\begin{array}{l}
\arfour \vdash R \mli \mathbb{E} ( z^T, w )\mlc \mathbb{F} ( z^T, w ).
\end{array}
\end{equation}
Then   
\begin{equation}\label{jan8b}
\begin{array}{l}
\arfour\vdash R\mlc  w \mleq t\mleq\xi(\bound) \mli   \mathbb{E} ( z^T, t)\mlc \mathbb{F} ( z^T, t) .
\end{array}
\end{equation} 
\end{lemma}

\begin{proof}  Immediately from Lemmas \ref{oct17d} and \ref{noct17e}.
\end{proof}

\subsection{Taking care of the case of large bounds}\label{jjjj}

We will be using 
\[\mathbb{A}(z,r,t)\label{iaaa}\]
for a natural formalization of the predicate saying that $r\mleq t$, \ 
$z$ is a $(\bound,r)$-adequate counterbehavior (so that $\bound$ is a hidden free variable of this formula) and, in the $(z,e_\bound)$-branch,   $\cal X$  is not in a move state at any time  $v$ with $r\mleq v\mless t$. 

Next, we will be using 
\[\mathbb{B}(z,r,t)\label{ibbb}\]
as an abbreviation of \[t\mless\xi(\bound)\mlc \mathbb{A}(z,r,t)\mlc \gneg \mathbb{A}(z,r,t\successor). \]

In the context of the $(z,e_\bound)$-branch, $\mathbb{B}(z,r,t)$ thus asserts that, on the interval $[r,t]$, one single move $\beta$ was made, and it was made exactly at time $t$. Note that, since the condition of the $(\bound,r)$-adequacy of $z$ is implied by ($\mathbb{A}(z,r,t)$ and hence) $\mathbb{B}(z,r,t)$, $\pa$ knows that the above move $\beta$ can only be made by $\cal X$.

For a variable $z$ and an elementary formula $T$ functional for $z$, as we did in the case of $\mathbb{E}$ and $\mathbb{F}$, we will write $\mathbb{A} ( z^T,r,t )$\label{iopo} as an abbreviation of $\cla z ( T\mli \mathbb{A}(z,r,t) )$. Similarly for $\mathbb{B} ( z^T,r,t )$ and $\mathbb{W}^E ( z^T,t_1,t_2,\bound,\vec{s} )$. 
 
\begin{lemma}\label{dec19}
Assume $x,u,z,w,\vec{s}$ are pairwise distinct non-$\bound$ variables, $R$ is an elementary formula,  $T$ is an elementary formula functional for $z$,  $E=E(\bound,\vec{s})$ is a  safe  formula all of whose free variables are among $\bound,\vec{s}$, and the following provabilities hold:
\begin{equation}\label{mmm1}
\arfour\vdash R\mli\xi(\bound)\equals u;
\end{equation}
\begin{equation}\label{mmm2}
\arfour\vdash R\mli \mathbb{W}^E ( z^T,w,w,\bound,\vec{s}  );
\end{equation}
\begin{equation}\label{mmm3}
\arfour\vdash R\mli  \mathbb{E} ( z^T, w  )\mlc \mathbb{F} ( z^T, w  ).
\end{equation}
 Then $\arfour$ proves
\begin{equation}\label{nnn1}
R\mli \mathbb{W}^E ( z^T,w,u,\bound, \vec{s} ) \add \ade x  \mathbb{B} ( z^T, w , x ) .
\end{equation}
\end{lemma}

\begin{idea} According to (\ref{mmm2}), $\arfour$ knows that,  under the  assumptions (of the truth of) $R$ and $T$,   $z$ is a $(\bound,w)$-adequate counterbehavior  and,   in the context of the $(z,e_\bound)$-branch, by time $w$, the play is legal and it has evolved to the position $E(\bound,\vec{s})$. Under the above assumptions, the target (\ref{nnn1}) is the problem of telling whether the same situation persists up to time $u$ (the left $\add$-disjunct of the consequent), or whether a (legal or illegal) move is made at some time $x$ with $w\mleq x\mless\xi(\bound)$ (the right $\add$-disjunct), i.e. --- in view of (\ref{mmm1}) --- at some time $x$ with $w\mleq x\mless u$. 

Solving this problem is not hard. Conditions (\ref{mmm1}) and (\ref{mmm3}), by Lemma \ref{oct17e}, imply full knowledge of the configurations of the machine at any time $t$ with $w\mleq t\mless u$.  Using this knowledge, we can trace the work of the machine step-by-step starting from $w$ and ending with $u\mminus 1$ and see if a move is made or not.  Technically, such ``tracing'' can be implemented relying on the induction rule of Lemma    \ref{noct17ee}. 
\end{idea}

\begin{proof} Assume all  conditions of the lemma. We shall point out that the condition on the safety of  $E$ is not relevant to the present proof, and it is included in the formulation of the lemma merely for the convenience of future references.

By Lemma \ref{oct17e}, condition (\ref{mmm3}) implies
\begin{equation}\label{oct26az}
\arfour\vdash    R\mlc  w \mleq t\mleq \xi(\bound) \mli  \mathbb{E} ( z^T, t )\mlc \mathbb{F} ( z^T, t )
\end{equation} 
which, in turn, in view of condition (\ref{mmm1}), can be easily seen to further imply 
\begin{equation}\label{oct26a}
\arfour\vdash    R\mlc  w \mleq t\mleq u \mli  \mathbb{E} ( z^T, t )\mlc \mathbb{F} ( z^T, t ).  
\end{equation}
 
Obviously $\pa\vdash \mathbb{W}^E ( z^T,w,w,\bound,\vec{s}  )\mli\mathbb{A} ( z^T, w , w )$. This, together with (\ref{mmm2}), by Transitivity, yields  $\arfour\vdash R\mli  \mathbb{A} ( z^T, w , w )$, whence, by $\add$-Choose, 
\begin{equation}\label{dec20a}
 \arfour\vdash R \mli  \mathbb{A} ( z^T, w , w   )\add \ade x  \mathbb{B} ( z^T, w , x )  .  
\end{equation}

{\bf Claim 1:} {\em $\arfour$ proves} 
\begin{equation}\label{jan16a}
\begin{array}{l}
R\mlc  w \mleq  t\mless \xi(\bound) \mlc  \bigl(\mathbb{A} ( z^T, w ,t )\add\ade x   \mathbb{B} ( z^T, w , x ) \bigr) 
\mli 
  \mathbb{A} ( z^T, w ,t\successor )\add \ade x   \mathbb{B} ( z^T, w , x ) .
\end{array}\end{equation}  

\begin{subproof} 
As in the case of Lemmas \ref{oct17c} and \ref{oct17b}, we will have to limit ourselves to an informal reasoning within $\arfour$.  Assume $R\mlc  w \mleq t\mless \xi(\bound)$, and   (a single copy) of the resource 
\begin{equation}\label{j16a}
 \mathbb{A} ( z^T, w , t )\add \ade x  \mathbb{B} ( z^T, w ,x )  
\end{equation}
 from the antecedent of (\ref{jan16a}) is at our disposal. Our task is to solve the consequental problem 
\begin{equation}\label{j16b}
\mathbb{A} ( z^T, w ,t\successor )\add \ade x  \mathbb{B} ( z^T, w , x ) . 
\end{equation}
The environment will have to choose one of the two $\add$-disjuncts of (\ref{j16a}). If the right disjunct is chosen, then we also choose the identical right disjunct in (\ref{j16b}), thus reducing the (relevant part of the) overall play to 
\(\ade x \mathbb{B} ( z^T, w ,x ) \mli \ade x  \mathbb{B} ( z^T, w ,x ) \)
 which, having the form $F\mli F$, is, of course, solvable/provable. 

Suppose now the left disjunct of (\ref{j16a}) is chosen, bringing the latter to $\mathbb{A} ( z^T, w , t )$. If this formula is false, we win. So, assume it is true.  
In view of (\ref{oct26az}), we have access to   the resource $ \mathbb{E} ( z^T, t )$, which contains information about the state of the machine at time $t$ in the play against the  counterbehavior $z$ (``{\em the}'' due to the functionality of $T$ for $z$) for which $T$ is true. If that state is not a move state, then we resolve (\ref{j16b}) by choosing its left component. And if that state is a move state, then we resolve (\ref{j16b}) by choosing its right component and specifying $x$ as $t$ in it. With a little thought, this can be seen to guarantee a win. \end{subproof} 

From (\ref{dec20a}) and Claim 1, by the rule of Lemma \ref{noct17ee}, we find 
\[\arfour\vdash R\mlc  w \mleq  t\mleq \xi(\bound) \mli  \mathbb{A} ( z^T, w ,t )\add \ade x \mathbb{B} ( z^T, w ,x )  \]
which, in view of condition (\ref{mmm1}), obviously implies 
\[\arfour\vdash R\mlc  w \mleq  t\mleq u \mli  \mathbb{A} ( z^T, w ,t )\add \ade x  \mathbb{B} ( z^T, w ,x ) . \]
Applying first $\ada$-Introduction and then $\ada$-Elimination to the above formula, we get
\begin{equation}\label{j16c}
\arfour\vdash R\mlc  w \mleq  u\mleq u \mli  \mathbb{A} ( z^T, w ,u )\add \ade x  \mathbb{B} ( z^T, w ,x ) .
\end{equation}
But the condition $w\mleq\xi(\bound)$ is part of $\mathbb{W}^E ( z^T,w,w,\bound,\vec{s}  )$ and hence, in view of (\ref{mmm2}) and (\ref{mmm1}),  $\arfour$ obviously proves $R\mli  w \mleq u\mleq u$.  This, in conjunction with (\ref{j16c}), can be easily seen to imply the $\arfour$-provability of
\begin{equation}\label{j16d}
 R\mli  \mathbb{A} ( z^T, w ,u )\add \ade x \mathbb{B} ( z^T, w ,x ) .
\end{equation}
Clearly $\pa\vdash \mathbb{W}^E ( z^T,w,w,\bound,\vec{s} )\mli \mathbb{A} ( z^T, w ,u )\mli \mathbb{W}^E ( z^T,w,u,\bound,\vec{s} )$. This, together with 
 (\ref{mmm2}), by Transitivity, implies that $\arfour$ proves 
\begin{equation}\label{j16e}
 R\mli \mathbb{A} ( z^T, w ,u )\mli \mathbb{W}^E ( z^T, w ,u,\bound,\vec{s} ).
\end{equation}

One can easily verify that $\clfour$ proves \[ (p\mli q_1\add Q)\mlc (p\mli q_1\mli q_2)\mli (p\mli q_2\add Q).\] Now,   $(\ref{j16d})\mlc(\ref{j16e})\mli(\ref{nnn1})$ can be seen to be an instance of the above formula and hence provable in $\arfour$. Modus-ponensing it with (\ref{j16d}) and (\ref{j16e}) yields (the $\arfour$-provability of) the desired 
 (\ref{nnn1}).
\end{proof}

Assume $E$ is a safe formula.   We say that a formula $G$ is a {\bf $\add$-deletion}\label{ideletion} of $E$ iff $G$ is the result of replacing  in $E$ some surface occurrence of a subformula $H_1\add\ldots\add H_m$ by $H_i$ (some $i\in\{1,\ldots,m\}$). And we say that a formula $G(y)$ is a {\bf $\ade$-deletion} of $E$ iff $G(y)$ is the result of   replacing  in $E$ some surface occurrence of a subformula $\ade yH(y)$ by $H(y)$ (deleting ``\hspace{2pt}$\ade y$'', that is). Note that  $\add$-deletions and $\ade$-deletions of a safe formula remain safe, and do not create free occurrences of variables that also have bound occurrences, which would otherwise violate  Convention \ref{jan26}. 

\begin{lemma}\label{dec21}
Assume the conditions of Lemma \ref{dec19} are satisfied. 
Let $G_1=G_1(\bound,\vec{s}),\ldots,G_m=G_m(\bound,\vec{s})$ be all of the $\add$-deletions of $E$, and $H_1=H_1(\bound,\vec{s},y_1),\ldots,H_n=H_n(\bound,\vec{s},y_n)$ be all of the $\ade$-deletions of $E$ (each $H_i $ is obtained from $E$ by deleting a surface occurrence of ``\hspace{2pt}$\ade y_i$''). Let $t$ be a fresh variable, and $\mathbb{C}(t)$ and $\mathbb{D}(t)$ be abbreviations defined by 
\[\begin{array}{rcl}
\mathbb{C}(t) & = & \mathbb{W}^{G_1} ( z^T,t\successor,t\successor,\bound,\vec{s}) \add\ldots\add \mathbb{W}^{G_m} ( z^T,t\successor,t\successor,\bound,\vec{s} ) ;\label{iccc}\\
\mathbb{D}(t) & = & \ade y_1\mathbb{W}^{H_1 } ( z^T,t\successor,t\successor,\bound,\vec{s},y_1 )\add\ldots\add \ade y_n\mathbb{W}^{H_n } (z^T,t\successor,t\successor,\bound,\vec{s},y_n ) .\label{iddd}
\end{array}\]
Then $\arfour$ proves
\begin{equation}\label{mmm44}
 R\mlc \mathbb{B} ( z^T, w ,t )\mli   \mathbb{L}\add    
\mathbb{C}(t)\add \mathbb{D}(t).
\end{equation}
\end{lemma}

\begin{idea} By the conditions of the lemma plus the additional condition expressed by the antecedent of (\ref{mmm44}), and in the context of the play according to the scenario  of the $(z,e_\bound)$-branch (for the counterbehavior $z$ satisfying $T$), we --- $\arfour$, that is --- know that, by time $w$, the play has evolved to the position $E$, and that, at time $t$ with $w\mleq t\mless\xi(\bound)$, some new move $\beta$ has been made by the machine. From (\ref{oct26az}), we have all information necessary to determine whether $\beta$ is legal or not and --- if $\beta$ is legal --- what move exactly it is. If $\beta$ is illegal, the machine does not win $X$ after all, so we can choose $\mathbb{L}$  in the consequent of (\ref{mmm44}). And if $\beta$ is legal, then, depending on what it is, we can choose   $\mathbb{C}(t)$ or $\mathbb{D}(t)$ in the consequent of (\ref{mmm44}), and then further choose in it the corresponding subcomponent $\mathbb{W}^{G_i} ( z^T,t,t,\bound,\vec{s})$ or ($\ade y_i\mathbb{W}^{H_i} ( z^T,t\successor,t\successor,\bound,\vec{s},y_i)$ and then)  $\mathbb{W}^{H_i} ( z^T,t\successor,t\successor,\bound,\vec{s},c)$.  
\end{idea}

\begin{proof} Assume the conditions of the lemma. Let us fix the two  sets $\{\alpha_1,\ldots,\alpha_m\}$ and $\{\beta_1,\ldots,\beta_n\}$ of strings such that the move  that brings $E(\bound,\vec{s})$ down to $G_i(\bound,\vec{s})$ is $\alpha_i$,\footnote{Strictly speaking, more than one move can bring $E$ to the same $\sqcup$-deletion (e.g., think of the case $E=Y\add Y$). But this is not a serious problem, and is easily taken care of by assuming that the list $G_1,\ldots,G_m$ has repetitions if necessary so that, for each move that turns $E$ into one of its $\sqcup$-deletions, the list contains a separate copy of the corresponding $\sqcup$-deletion.} and the move that brings $E(\bound,\vec{s})$ down to $H_i(\bound,\vec{s},c)$ (whatever constant $c$) is $\beta_i.c$. 

For each $i\in\{1,\ldots,m\}$, let $\mathbb{G}_i(z)$\label{iggg} be an elementary formula saying that  $ \mathbb{B} ( z, w ,  t )$ is true and the move made by the machine at time $  t$ in the $(z,e_\bound)$-branch is $\alpha_i$. Extending our notational practice to this formula, $\mathbb{G}_i(z^T)$ will be an abbreviation of $\cla z\bigl(T\mli \mathbb{G}_i(z)\bigr)$.\vspace{7pt}

{\bf Claim 1.}  {\em $\arfour$ proves}
\begin{equation}\label{jan19a}
\begin{array}{l}
   R\mlc   \mathbb{B} ( z^T, w ,  t )\mli  
\mathbb{G}_1( z^T)\add\ldots\add\mathbb{G}_m( z^T)\add\gneg \bigl(\mathbb{G}_1( z^T)\mld\ldots\mld\mathbb{G}_m( z^T)\bigr).\vspace{7pt}
\end{array}\end{equation} 

\begin{subproof}  Let $k$ be the greatest of the lengths of the moves $\alpha_1,\ldots,\alpha_m$. Argue in $\arfour$. Assume  $R\mlc   \mathbb{B} ( z^T, w ,  t )$. Consider the counterbehavior $z$ for which $T$ is true, and consider the play according to the scenario  of the $(z,e_\bound)$-branch.   $\mathbb{B} ( z^T, w ,  t )$ implies $w\mleq t\mless \xi(\bound)$. Therefore,  in view of (\ref{oct26az}), full information is available about the situation in the machine at time $t$. Using this information, we first find the location $l$ of the work-tape head and, using the results of Section \ref{s17}, find $a$ with $a=min(l,k)$. Then we construct  a full picture of the contents of cells $\#0$ through $\# (a\mminus 1)$ of the work tape at time $t$. From this picture, we can determine whether it shows making one of the moves $\alpha_i$ (and which one), or none, and accordingly choose the true $\add$-disjunct of the consequent of (\ref{jan19a}).  
\end{subproof}

For each $i\in\{1,\ldots,n\}$, let $\mathbb{H}_i(z)$\label{ihhh} be an elementary formula saying that  $ \mathbb{B} ( z, w ,  t )$ is true and the move made by the machine at time $ t$ in the $(z,e_\bound)$-branch has the prefix ``$\beta_i.$''. $\mathbb{H}_i(z^T)$ will be an abbreviation of $\cla z\bigl(T\mli \mathbb{H}_i(z)\bigr)$.\vspace{7pt}

{\bf Claim 2.}  {\em $\arfour$ proves}
\begin{equation}\label{jan19b}
  R\mlc   \mathbb{B} ( z^T, w ,  t )\mli \mathbb{H}_1(z^T)\add\ldots\add\mathbb{H}_n(z^T)\add\gneg \bigl(\mathbb{H}_1(z^T)\mld\ldots\mld\mathbb{H}_n(z^T)\bigr).\vspace{7pt}\end{equation}

\begin{subproof} Similar to the proof of Claim 1.
\end{subproof}

For each $i\in\{1,\ldots,n\}$, let $\mathbb{H}'_i(z,y)$\label{ihhhp} be an elementary formula saying that  $ \mathbb{H}_i ( z)$ is true and the move made by the machine at time $ t$ in the $(z,e_\bound)$-branch is $\beta_i.y$. $\mathbb{H}'_i(z^T,y)$ will be an abbreviation of $\cla z\bigl(T\mli \mathbb{H}'_i(z,y)\bigr)$.\vspace{7pt}

{\bf Claim 3.} {\em For each $i\in\{1,\ldots,n\}$,  $ \arfour$ proves}   
\begin{equation}\label{jan19c}
R\mlc  \mathbb{H}_i(z^T) \mli \ade y_i\mathbb{H}'_i(z^T,y_i)\add \mathbb{L}.\vspace{7pt}
\end{equation} 

\begin{subproof} Take any $i\in\{1,\ldots,n\}$. Let $k$ be the length of the string ``$\beta_i.$''. Let $\mathbb{J}(z,v,y)$\label{ijjj} be a formula saying  
\begin{quote}{\em `` $\mathbb{H}_i ( z)$ (is true) and, in  the $(z,e_\bound)$-branch, at time $  t$, on the work tape, cells $\# k$ through $\#(k\plus v)$    spell constant $y$, and the location of the head is not any of the cells $\#0,\#1,\ldots,\#(k\plus v\plus 1)$''.}\end{quote}
 $\mathbb{J}(z^T,v,y)$ will be an abbreviation of $\cla z\bigl(T\mli\mathbb{J}(z,v,y)\bigr)$. 

Argue in $\arfour$. We want to prove, by WPTI induction on $v$, that 
\begin{equation}\label{ooo1}
 v\mleq \hat{k}\plus \bound \mli   R\mlc    \mathbb{H}_i(z^T) \mli\mathbb{L}\add \ade y_i\mathbb{H}'_i( z^T,y_i)
\add \ade y \mathbb{J}(z^T,v,y).
\end{equation}
The basis is 
\begin{equation}\label{ooo2}
   R\mlc    \mathbb{H}_i(z^T) \mli \mathbb{L}\add  \ade y_i\mathbb{H}'_i( z^T,y_i) \add \ade y \mathbb{J}(z^T,\zero,y).
\end{equation} 
Assume the (truth of the) antecedent of the above. Consider the  counterbehavior $z$ for which $T$ is true, and consider the play according to the scenario  of the $(z,e_\bound)$-branch. We will implicitly rely on the fact that,  in view of (\ref{oct26az}) (whose antecedent is implied by $   R\mlc    \mathbb{H}_i(z^T)$), full information is available about the situation in the machine at time $t$.   The problem (\ref{ooo2}) is 
solved as follows, where ``head'' and ``cell'' always mean those of the work tape, and ``located'' or ``contains'' mean that this is so at time $t$. 
\begin{enumerate}
\item Using the results of Section \ref{s17}, figure out whether  $|k\plus 1|\mleq\bound$ ($|\hat{k}\successor|\mleq\bound$, that is) and, if so, find the values of $k$ and $k\plus 1$ and then continue according to Steps 2-4. If, however,    $|k\plus 1|\mgreater\bound$, then choose $\mathbb{L}$ in the consequent of (\ref{ooo2}) and you are done as it is guaranteed to be true. This is so because, from Axiom 13, we know that $|t|\mleq\bound$, and thus $k\plus 1\mgreater t$; this, in turn, means that the head would not have enough time to go as far as cell $\#(k\plus 1)$; and, if so, the machine cannot make a legal move at time $t$, so it loses.   
\item If the location of the head is not greater than $k$, then we are dealing with the fact of $\cal X$ having just (at time $t$) made an illegal move which is ``$\beta_i.$'' or some proper initial substring of it, so choose $\mathbb{L}$ in the consequent of (\ref{ooo2}) because $\cal X$ loses.
\item Suppose the head is located at cell $\#(k\plus 1)$. Then:
\begin{itemize}
\item If cell $\#k$ contains $0$, then we are dealing with the fact of $\cal X$ having made the move  $\beta_i.0$, so choose $\ade y_i\mathbb{H}'_i( z^T,y_i)$ 
in the consequent of (\ref{ooo2}) and specify $y_i$ as $0$ in it.
\item If cell $\#k$ contains $1$, then we are dealing with the fact of $\cal X$ having made the move  $\beta_i.1$, so choose $\ade y_i\mathbb{H}'_i( z^T,y_i)$ in the consequent of (\ref{ooo2}) and specify $y_i$ as $1$ in it.
\item If cell $\#k$ contains any other symbol,  then we are dealing with the fact of $\cal X$ having made an illegal move,  so choose $\mathbb{L}$. 
\end{itemize}
\item Suppose the location of the head is greater than  $k\plus 1$. Then:
\begin{itemize}
\item If cell $\#k$ contains $0$, choose $\ade y \mathbb{J}(z^T,\zero,y)$ in the consequent of (\ref{ooo2}) and specify $y$ as $0$ in it.
\item If cell $\#k$ contains $1$, choose $\ade y \mathbb{J}(z^T,\zero,y)$ in the consequent of (\ref{ooo2}) and specify $y$ as $1$ in it.
\item If cell $\#k$ contains any other symbol, choose $\mathbb{L}$.
\end{itemize}
\end{enumerate}

The inductive step is 
\begin{equation}\label{ooo3}
\begin{array}{l}
\bigl( R\mlc    \mathbb{H}_i(z^T) \mli  
\mathbb{L}\add \ade y_i \mathbb{H}'_i(z^T,y_i) \add \ade y \mathbb{J}(z^T,v,y)\bigr)\mli\\  
\bigl( R\mlc     \mathbb{H}_i(z^T) \mli \mathbb{L}\add \ade y_i \mathbb{H}'_i(z^T,y_i) \add \ade y \mathbb{J}(z^T,v\successor,y)\bigr).
\end{array}
\end{equation}
 Assume $ R\mlc  \mathbb{H}_i(z^T)$ is true (otherwise (\ref{ooo3}) is won). Under this assumption, solving (\ref{ooo3}) essentially means solving the following problem:  
\begin{equation}\label{oooo3}
\begin{array}{l}
\mathbb{L}\add \ade y_i \mathbb{H}'_i(z^T,y_i) \add \ade y \mathbb{J}(z^T,v,y)\mli 
\mathbb{L}\add \ade y_i \mathbb{H}'_i(z^T,y_i) \add \ade y \mathbb{J}(z^T,v\successor,y).
\end{array}
\end{equation}
This problem is solved as follows. Wait for the environment to choose a $\add$-disjunct in the antecedent. If that choice is one of the first two disjuncts, choose the identical disjunct in the consequent, and then resolve the resulting problem of the form $F\mli F$. Suppose now the third disjunct $\ade y \mathbb{J}(z^T,v,y) $ is chosen. Wait till it is further brought to $ \mathbb{J}(z^T,v,c) $ for some $c$. Consider the  counterbehavior $z$ for which $T$ is true, and consider the play according to the scenario  of the $(z,e_\bound)$-branch.  As was done when justifying the basis of induction, we will  rely on the fact that,  in view of (\ref{oct26az}), full information is available about the situation in the machine at time $t$. In our subsequent discourse,   ``head'' and ``cell'' always mean those of the work tape,  and ``located'' or ``contains'' mean that this is so at time $t$.  Note that, as implied by $ \mathbb{J}(z^T,v,c) $, the location of the head is greater than $k\plus v\successor$. So, using the results of Section \ref{s17}, we can tell ($\add$) whether that location is $k\plus v\successor\plus 1$ or greater than $k\plus v\successor\plus 1$. We correspondingly consider the following two cases and resolve the consequent of (\ref{oooo3}) accordingly:
\begin{enumerate}  
\item Suppose the head is located at cell $\#(k\plus v\successor\plus 1)$. Then:
\begin{itemize}
\item If cell $\#(k\plus v\successor)$ contains $0$, then we are dealing with the fact of $\cal X$ having made  the move  $\beta_i.c0$, so choose $\ade y_i\mathbb{H}'_i( z^T,y_i)$ in the consequent of (\ref{oooo3}) and specify $y_i$ as $c0$ in it.
\item If cell $\#(k\plus v\successor)$ contains $1$, then we are dealing with the fact of $\cal X$ having made the move  $\beta_i.c1$, so choose $\ade y_i\mathbb{H}'_i( z^T,y_i)$ in the consequent of (\ref{oooo3}) and specify $y_i$ as $c1$ in it.
\item If cell $\#k$ contains any other symbol,  then we are dealing with the fact of $\cal X$ having made an illegal move,  so choose $\mathbb{L}$. 
\end{itemize}
\item Suppose the location of the head is greater than $k\plus v\successor\plus 1$. Then:
\begin{itemize}
\item If cell $\#k$ contains $0$, choose $\ade y \mathbb{J}(z^T,v',y)$ in the consequent of (\ref{oooo3}) and specify $y$ as $c0$ in it.
\item If cell $\#k$ contains $1$, choose $\ade y \mathbb{J}(z^T,v',y)$ in the consequent of (\ref{oooo3}) and specify $y$ as $c1$ in it.
\item If cell $\#k$ contains any other symbol, choose $\mathbb{L}$.
\end{itemize}
\end{enumerate}

Now, (\ref{ooo1}) follows by WPTI from (\ref{ooo2}) and (\ref{ooo3}).

We continue our proof of Claim 3 by arguing in $\arfour$ towards the goal of justifying (\ref{jan19c}). Assume (the truth of) the antecedent of the latter. As before, we let $z$ be the counterbehavior satisfying $T$, and let the context of our discourse be the play according to the scenario of the $(z,e_\bound)$-branch. From (\ref{oct26az}), we   find the location $l$ of  the work-tape head at time $t$. If $l\equals\zero$, we are dealing with the fact of the machine having made an illegal move (the empty move), so choose $\mathbb{L}$ in the consequent of  (\ref{jan19c}). Otherwise, we find the number $a$ with $l\equals a\successor$ (the results of Section \ref{s17} will allow us to do so). From (\ref{ooo1}), we get 
\begin{equation}\label{ooo11}
 a\mleq \hat{k}\plus \bound \mli   R\mlc  \mathbb{H}_i(z^T) \mli\mathbb{L}\add \ade y_i\mathbb{H}'_i( z^T,y_i)
\add \ade y \mathbb{J}(z^T,a,y).
\end{equation}

Next, we figure out (again relying on the results of Section \ref{s17}) whether $a\mleq k\plus \bound$ or not. If not, we are obviously dealing with the case of the machine having made an illegal (``too long a'') move, so  we choose $\mathbb{L}$ in the consequent of  (\ref{jan19c}). Suppose now $a\mleq k\plus \bound$. Then, from  
 (\ref{ooo11}) by Modus Ponens applied twice, we get    
\begin{equation}\label{ooo4}
 \mathbb{L}\add \ade y_i\mathbb{H}'_i( z^T,y_i) \add \ade y \mathbb{J}(z^T,a,y).
\end{equation}
Our remaining task is to show how to solve the consequent 
\begin{equation}\label{ooo444}
 \ade y_i\mathbb{H}'_i( z^T,y_i)\add  \mathbb{L} 
\end{equation}
of (\ref{jan19c}) using the resource (\ref{ooo4}). This is very easy. Wait till the environment selects a $\add$-disjunct of (\ref{ooo4}). If one of the first two disjuncts is selected, select the identical disjunct in (\ref{ooo444}) and, having brought things down to a problem of the form $F\mli F$, solve it. And if the third disjunct $\ade y \mathbb{J}(z^T,a,y)$ of (\ref{ooo4}) is selected, we  win. That is because, no matter what $c$ the environment further selects for $y$ in it, the resulting formula $\mathbb{J}(z^T,a,c)$ will be false as it implies that the work-tape head at time $t$ is not located at cell $\#(a\plus 1) $,  which is a contradiction --- as we remember, $l\equals a\successor$ is exactly the location of the head.  

Our proof of Claim 3 is now complete. 
\end{subproof}

{\bf Claim 4.} $\pa\vdash  R\mlc   \mathbb{B} ( z^T, w ,t )\mli \gneg \bigl(\mathbb{G}_1(z^T)\mld\ldots\mld\mathbb{G}_m(z^T)\bigr)\mlc \gneg \bigl(\mathbb{H}_1(z^T)\mld\ldots\mld\mathbb{H}_n(z^T)\bigr)\mli\mathbb{L}$.\vspace{7pt}

\begin{subproof} This and the following two claims can be proven by a straightforward argument within $\pa$ based on the meanings of the predicates involved in the formula. Assume 
$ R\mlc   \mathbb{B} ( z^T, w ,t )$ and 
\begin{equation}\label{j19b}
\gneg \bigl(\mathbb{G}_1(z^T)\mld\ldots\mld\mathbb{G}_m(z^T)\bigr)\mlc \gneg \bigl(\mathbb{H}_1(z^T)\mld\ldots\mld\mathbb{H}_n(z^T)\bigr). \end{equation}  
Consider the counterbehavior $z$ satisfying $T$, and the play according to the scenario of the $(z,e_\bound)$-branch. According to (\ref{mmm2}), by time $w$   the play has evolved to position $E(\bound,\vec{s})$. And, according to $\mathbb{B} ( z^T, w ,t )$, a (first new) move $\beta$ has been made by the machine at time $t$. Obviously the assumption (\ref{j19b}) precludes the possibility of such a  $\beta$ being a legal  move of   $E(\bound,\vec{s})$. So, $\beta$ is illegal, which makes the machine lose the game, and hence $\mathbb{L}$ is true. 
\end{subproof}

{\bf Claim 5.} {\em For each} $i\in\{1,\ldots,m\}$, $\pa\vdash  R\mli   \mathbb{G}_i(z^T)\mli \mathbb{W}^{G_i} ( z^T,t\successor,t\successor,\bound,\vec{s})$.\vspace{7pt}

\begin{subproof} Argue in $\pa$. Assume $R$ and  $\mathbb{G}_i(z^T)$. Consider the counterbehavior $z$ satisfying $T$, and the play according to the scenario of the $(z,e_\bound)$-branch. According to (\ref{mmm2}), by time $w$   the play has evolved to position $E(\bound,\vec{s})$. And, according to $\mathbb{G}_i(z^T)$, a (first new) move  has been made by the machine at time $t$, and such a move is $\alpha_i$. But this move brings $E(\bound,\vec{s})$ down to $G_i(\bound,\vec{s})$. This, in turn,  implies that $ \mathbb{W}^{G_i} ( z^T,t\successor,t\successor,\bound,\vec{s})$ is true.
\end{subproof} 

{\bf Claim 6.} {\em For each} $i\in\{1,\ldots,n\}$, \(\pa\vdash  R\mli   \mathbb{H}'_i(z^T,y_i)\mli  \mathbb{W}^{H_i} ( z^T,t\successor,t\successor,\bound,\vec{s},y_i).\)\vspace{7pt}

\begin{subproof} Very similar to the proof of Claim 5. Argue in $\pa$. Assume $R$ and  $\mathbb{H}'_i(z^T,y_i)$. Consider the counterbehavior $z$ satisfying $T$, and the play according to the scenario of the $(z,e_\bound)$-branch. According to (\ref{mmm2}), by time $w$   the play has evolved to position $E(\bound,\vec{s})$. And, according to $\mathbb{H}_i(z^T,y_i)$, a (first new) move  has been made by the machine at time $t$, and such a move is $\beta_i.y_i$. But this move brings $E(\bound,\vec{s})$ down to $H_i(\bound,\vec{s},y_i)$. This, in turn, implies that $ \mathbb{W}^{H_i} ( z^T,t\successor,t\successor,\bound,\vec{s},y_i)$ is true.
\end{subproof}

To complete our proof of Lemma \ref{dec21}, it remains to  observe that  (\ref{mmm44}) is a  logical consequence of Claims 1-6.  Since we have played more than enough with  $\clfour$, here we only schematically outline how to do this purely  syntactic exercise.

First of all, Claims 1 and 2 can be easily seen to  imply 
\[\begin{array}{l}
\arfour\vdash R\mlc  \mathbb{B} ( z^T, w ,  t )\mli 
  \mathbb{G}_1(z^T)\add\ldots\add\mathbb{G}_m(z^T)\add \mathbb{H}_1(z^T)\add\ldots\add\mathbb{H}_n(z^T)\add\\
\Bigl(\gneg \bigl(\mathbb{G}_1(z^T)\mld\ldots\mld\mathbb{G}_m(z^T)\bigr)\mlc \gneg \bigl(\mathbb{H}_1(z^T)\mld\ldots\mld\mathbb{H}_n(z^T)\bigr)\Bigr).
\end{array}\]
The above, together with Claim 4, further implies 
\[\begin{array}{l}
\arfour\vdash R\mlc   \mathbb{B} ( z^T, w ,  t )\mli 
  \mathbb{G}_1(z^T)\add\ldots\add\mathbb{G}_m(z^T)\add \mathbb{H}_1(z^T)\add\ldots\add\mathbb{H}_n(z^T)\add\mathbb{L}.
\end{array}\]
This, in turn, together with Claim 5, further implies
\[\begin{array}{l}
\arfour\vdash R\mlc  \mathbb{B} ( z^T, w ,  t )\mli 
  \mathbb{C}(t)\add \mathbb{H}_1(z^T)\add\ldots\add\mathbb{H}_n(z^T)\add\mathbb{L}.
\end{array}\]
The above, together with Claim 3, further implies 
\begin{equation}\label{jan19d}
\begin{array}{l}
\arfour\vdash R\mlc  \mathbb{B} ( z^T, w ,  t )\mli 
  \mathbb{C}(t)\add \ade y_1\mathbb{H}'_1(z^T,y_i)\add\ldots\add\ade y_n\mathbb{H}'_n(z^T,y_n)\add\mathbb{L}.
\end{array}\end{equation}
Claim 6 can be seen to imply 
\[\arfour\vdash  R\mli   \ade y_i\mathbb{H}'_i(z^T,y_i)\mli \ade y_i \mathbb{W}^{H_i} ( z^T,t\successor,t\successor,\bound,\vec{s},y_i)
\begin{array}{l}
\end{array}\]
for each $i\in\{1,\ldots,n\}$. This, together with (\ref{jan19d}), can be seen to imply the desired (\ref{mmm44}).
\end{proof}

\begin{lemma}\label{oct}
Under  the conditions of Lemma  \ref{dec19} and using the abbreviations of Lemma  \ref{dec21}, $\arfour$ proves  
\begin{equation}\label{mmm4}
R\mli \mathbb{L}\add\ade x  \bigl(\mathbb{C}(x)\add \mathbb{D}(x) \bigr) \add \mathbb{W}^{E} ( z^T,w,u,\bound,\vec{s})  .
\end{equation}
\end{lemma}

\begin{idea} This is a logical consequence of the previous two lemmas. 
\end{idea}

\begin{proof}  Assume the conditions of Lemma  \ref{dec19}. Then, according to  Lemmas \ref{dec19} and \ref{dec21}, $\arfour$ proves (\ref{nnn1}) and (\ref{mmm44}) (where, in the latter, $t$ is a fresh variable). The target formula (\ref{mmm4}) is a logical consequence of those two formulas, verifying which is a purely syntactic exercise. As we did in the proof of the previous lemma, here we only provide a scheme for such a verification. It is rather simple. 
First, applying $\ade$-Choose and $\ada$-Introduction to (\ref{mmm44}), we get  
\[\arfour\vdash   R\mlc  \ade x  \mathbb{B} ( z^T, w ,x ) \mli   \mathbb{L}\add    
\ade x\bigl(\mathbb{C}(x)\add \mathbb{D}(x)\bigr). \]
And then we observe that the above, together with $\arfour\vdash (\ref{nnn1})$, implies $\arfour\vdash (\ref{mmm4})$. 
\end{proof}

\begin{lemma}\label{oct26}
Under  the conditions of Lemma \ref{dec19},  
$\arfour\vdash R\mli    \overline{E(\bound,\vec{s})}$. 
\end{lemma}

\begin{idea} Under the assumption of the truth of $R$, one of the three $\add$-disjuncts of the consequent of (\ref{mmm4}) is available as a resource. In each case, we need to show (in $\arfour$) how to solve the target $\overline{E}=\overline{E(\bound,\vec{s})}$.

1. The case of $\mathbb{L}$ is taken care of by Lemma \ref{jan4d}, according to which $\arfour\vdash \mathbb{L}\mli \overline{E}$.

2. The case of $\ade x  \bigl(\mathbb{C}(x)\add \mathbb{D}(x) \bigr)$, depending on which of its $\ade$- and $\add$-components are further chosen, allows us to jump to a formula $F$ (one of the $G_i$, $1\mleq i\mleq m$  or $H_i$, $1\mleq i\mleq n$) from which $E$ follows by $\add$-Choose or $\ade$-Choose. With appropriately readjusted $R$ and certain other parameters, by the induction hypothesis, we know how to solve  $\overline{F}$. Then (by $\add$-Choose or $\ade$-Choose) we also know how to solve $\overline{E}$.

3. Finally, consider the case of    $\mathbb{W}^{E} ( z^T,w,u,\bound,\vec{s})$.  $E$ can be critical or non-critical. The case of $E$ being critical is almost immediately taken care of by Lemma \ref{august20a}, according to which $\arfour\vdash \cle z \mathbb{W}^{E} ( z,w,u,\bound,\vec{s})\mli \overline{E}$. Suppose now $E$ is non-critical. Then, by Lemma \ref{august20b}, according to which  $\arfour\vdash \cle z \mathbb{W}^{E} ( z,w,u,\bound,\vec{s})\mli \elz{\overline{E}}$, the elementarization of $\overline{E}$ is true/provable. Relying on the induction hypothesis as in the previous case, and replacing $T(z)$ by 
a formula $S(z)$ saying that $z$ is a certain one-move extension of the counterbehavior satisfying $T$, we manage to show that any other (other than $\elz{\overline{E}}$) necessary Wait-premise of $\overline{E}$ is also solvable/provable. Then, by Wait, we know how to solve/prove $\overline{E}$. 
\end{idea}

\begin{proof} We prove this lemma by (meta)induction on the complexity of $E(\bound,\vec{s})$. Assume the conditions of  Lemma  \ref{dec19}.  Then, by Lemma \ref{oct}, $\arfour\vdash (\ref{mmm4})$.  

By Lemma \ref{jan4d},  $\arfour\vdash \mathbb{L} \mli \overline{E(\bound,\vec{s})}$, whence, by Weakening, 

\begin{equation}\label{dec15b}
\begin{array}{l}
\arfour\vdash \mathbb{L} \mlc R\mli \overline{E(\bound,\vec{s})}.\vspace{7pt}
\end{array}
\end{equation}

In what follows, we will rely on the additional assumptions and abbreviations of Lemma \ref{dec21}.\vspace{7pt}

{\bf Claim 1}. {\em For each} $i\in\{1,\ldots,m\}$,   $\arfour\vdash  \mathbb{W}^{G_i} ( z^T,   t\successor,  t\successor,\bound,\vec{s} )\mlc R\mli   \overline{E(\bound,\vec{s})}$.\vspace{7pt}

\begin{subproof} Pick any   $i\in\{1,\ldots,m\}$  and a fresh variable $v$. 

From condition   (\ref{mmm1}), by Weakenings, we have  
\begin{equation}\label{dec16d}
\arfour\vdash  v\equals t\successor\mlc \mathbb{W}^{G_i} ( z^T, t\successor,t\successor,\bound,\vec{s}) \mlc R\mli
\xi(\bound)\equals u. 
\end{equation}

And, of course, we also have    
\begin{equation}\label{dec16dd}
\arfour\vdash   v\equals t\successor\mlc\mathbb{W}^{G_i} ( z^T,t\successor,t\successor,  \bound,\vec{s})\mlc R\mli
 \mathbb{W}^{G_i} ( z^T,  v,  v,\bound,\vec{s} ). 
\end{equation}

Condition (\ref{mmm3}) and Lemma \ref{oct17e} imply 
\begin{equation}\label{dec16a}
\arfour\vdash  R\mlc    w \mleq v\mleq\xi(\bound) \mli   \mathbb{E} ( z^T,  v )\mlc \mathbb{F} ( z^T, v ). 
\end{equation}
In view of condition (\ref{mmm2}), we also obviously have 
\begin{equation}\label{dec16b}
\arfour\vdash  v\equals t\successor\mlc \mathbb{W}^{ G_i} ( z^T, t\successor,t\successor,\bound,\vec{s})\mlc R\mli  R\mlc    w \mleq v\mleq\xi(\bound) .
\end{equation}
From (\ref{dec16b}) and (\ref{dec16a}), by Transitivity, we get
\begin{equation}\label{dec16c}
\arfour\vdash  v\equals t\successor\mlc\mathbb{W}^{G_i} ( z^T, t\successor,t\successor,\bound,\vec{s})\mlc R\mli \ \mathbb{E} ( z^T,  v )\mlc \mathbb{F} ( z^T,  v ). 
\end{equation}

By the induction hypothesis of our lemma, with $  v$, $G_i(\bound,\vec{s})$ and $v\equals t\successor\mlc\mathbb{W}^{G_i} ( z^T,t\successor,t\successor, \bound,\vec{s})\mlc R$ in the roles of $ w $, $E(\bound,\vec{s})$ and $R$, (\ref{dec16d}),  (\ref{dec16dd}) and (\ref{dec16c}) --- which correspond to (\ref{mmm1}), (\ref{mmm2}) and (\ref{mmm3}), respectively --- imply 
\[\arfour\vdash  v\equals t\successor\mlc \mathbb{W}^{G_i} (z^T,t\successor,t\successor,\bound,\vec{s}) \mlc R\mli \overline{G_i(\bound,\vec{s})}.\]
The above, by $\ada$-Introduction, yields  
\begin{equation}\label{dec16e}
\arfour\vdash \ade x(x\equals t\successor)\mlc \mathbb{W}^{G_i} (z^T,t\successor,t\successor,\bound,\vec{s}) \mlc R\mli \overline{G_i(\bound,\vec{s})}.
\end{equation}
Remembering the definition of $\mathbb{W}$, the condition $t\successor\mleq\xi(\bound)$ is one of the conjuncts of $\mathbb{W}^{G_i}\bigl(z^T,t\successor,t\successor,\bound,\vec{s})$. Hence $\pa\vdash \mathbb{W}^{G_i}\bigl(z^T,t\successor,t\successor,\bound,\vec{s})\mli t\successor\mleq\xi(\bound)$. Together with condition (\ref{mmm1}), this implies \[\arfour\vdash  \mathbb{W}^{G_i}\bigl(z^T,t\successor,t\successor,\bound,\vec{s})\mlc R \mli t\successor\mleq u.\] But, by Axiom 13, $\arfour\vdash |u|\mleq\bound$. Hence, obviously,  
$\arfour\vdash   \mathbb{W}^{G_i}\bigl(z^T,t\successor,t\successor,\bound,\vec{s})\mlc R \mli |t\successor|\mleq \bound$. This, together with in Axiom 10, by Transitivity, yields $\arfour\vdash   \mathbb{W}^{G_i}\bigl(z^T,t\successor,t\successor,\bound,\vec{s})\mlc R\mli \ade x(x\equals t\successor)$. And the latter, in turn, in conjunction with (\ref{dec16e}), can be seen to imply 
\begin{equation}\label{dec16z}
\arfour\vdash \mathbb{W}^{G_i} (z^T,t\successor,t\successor,\bound,\vec{s}) \mlc R\mli \overline{G_i(\bound,\vec{s})}.
\end{equation}
Now, it remains to notice that the desired $ \mathbb{W}^{G_i} ( z^T,    t\successor, t\successor,\bound,\vec{s} )\mlc R\mli \overline{E(\bound,\vec{s})}$ follows from (\ref{dec16z}) by $\add$-Choose. \end{subproof}

{\bf Claim 2}. {\em For each} $i\in\{1,\ldots,n\}$,   $\arfour\vdash  \ade y_i\mathbb{W}^{ H_i}  (z^T,t\successor,t\successor,\bound,\vec{s},y_i )\mlc R\mli\overline{E(\bound,\vec{s})}$.\vspace{7pt}

\begin{subproof} Pick any   $i\in\{1,\ldots,n\}$. Arguing as we did for (\ref{dec16z}) in the proof of Claim 1, we find
\[\arfour\vdash \mathbb{W}^{H_i} (z^T, t\successor, t\successor,\bound,\vec{s},y_i )\mlc R\mli  \overline{H_i(\bound,\vec{s},y_i)}.
\]
Applying first $\ade$-Choose and then $\ada$-Introduction to the above, we get the desired conclusion $\arfour\vdash  \ade y_i\mathbb{W}^{H_i} ( z^T,     t\successor,  t\successor ,\bound,\vec{s},y_i)\mlc R\mli  \overline{E}$. \end{subproof}

Claims 1 and 2, by $\adc$-Introductions, imply 
\[\begin{array}{ll}\arfour\vdash &     
\Bigl(\bigl(\mathbb{W}^{  G_1} (z^T,t\successor,t\successor,\bound,\vec{s}) \add\ldots\add \mathbb{W}^{ G_m }(z^T,t\successor,t\successor,\bound,\vec{s})\bigr)    \add \\ & \bigl(\ade y_1\mathbb{W}^{H_1}( z^T,t\successor,t\successor,\bound,\vec{s},y_1) \add\ldots\add\ade y_n\mathbb{W}^{H_n}( z^T,t\successor,t\successor,\bound,\vec{s},y_n)\bigr) \Bigr)\mlc R   \mli \overline{E(\bound,\vec{s})}  \end{array} \]
which, using the abbreviations of Lemma \ref{dec21}, is written as 
\(\arfour\vdash      
\bigl(\mathbb{C}(t)    \add \mathbb{D}(t)\bigr) \mlc R   \mli \overline{E(\bound,\vec{s})}.\)
The latter, 
 by $\ada$-Introduction, yields
\begin{equation}\label{dec16g}
\arfour\vdash      
\ade x\bigl(\mathbb{C}(x)    \add \mathbb{D}(x)\bigr) \mlc R   \mli \overline{E(\bound,\vec{s})}.\end{equation}

{\bf Claim 3}. {\em If $E(\bound,\vec{s})$ is not critical, then} $\arfour\vdash \elz{\mathbb{W}^{  E  } ( z^T,w,u,  \bound,\vec{s} ) \mlc R\mli \overline{E(\bound,\vec{s})}}$.\vspace{7pt}

\begin{subproof} Assume $E(\bound,\vec{s})$ is not critical. Then, Lemma \ref{august20b}, together with the fact of $T $ being functional for $z$, can be easily seen to imply  $\arfour\vdash \mathbb{W}^{  E }  ( z^T,w ,\xi(\bound),\bound,\vec{s}  \bigr) \mli \elz{\overline{E(\bound,\vec{s})}}$. Remembering condition (\ref{mmm1}), 
the latter can be seen to further imply  $\arfour\vdash \mathbb{W}^{  E }  ( z^T,w ,u,\bound,\vec{s})\mlc R \mli \elz{\overline{E(\bound,\vec{s})}}$, which 
is the same as to say that  $\arfour\vdash \elz{\mathbb{W}^{  E  } ( z^T,w,u,  \bound,\vec{s} ) \mlc R\mli \overline{E(\bound,\vec{s})}}$, because 
 both $\mathbb{W}^{  E }  ( z^T,w ,u,\bound,\vec{s})$ and $R$ are elementary. \end{subproof}

{\bf Claim 4}. {\em Assume $E(\bound,\vec{s})$ has the form $ F[J_1\adc\ldots \adc J_k]$, and $i\in\{1,\ldots,k\}$. Then}   \[\arfour\vdash \mathbb{W}^{   E } ( z^T, w,u, \bound,\vec{s})\mlc R\mli \overline{F[J_i]}.\vspace{7pt}\]

\begin{subproof} 
From  (\ref{mmm1}), by Weakening, we  have 
\begin{equation}\label{oct26q}
\arfour\vdash    \mathbb{W}^{  E } ( z^T, w,u,  \bound,\vec{s})\mlc R   \mli \xi(\bound)\equals u.
\end{equation}

Assume $E(\bound,\vec{s})=F[J_1\adc\ldots \adc J_k]$ and $1\mleq i\mleq k$.  Let $\alpha$ be the move whose effect is turning $F[J_1\adc\ldots \adc J_k]$ into $F[J_i]$. Let us write our formula $T$ in the form $T(z)$. Let $S=S(z)$ be a formula saying that $z$ is the code of the counterbehavior resulting by adding the timestamped move $(\alpha,w\mminus 1)$ to the  counterbehavior $a$ for which $T(a)$ holds. Of course, $S$ is functional for $z$. It is not hard to see that $\pa$ proves $\mathbb{W}^{  E } ( z^T, w,u,  \bound,\vec{s}) \mli \mathbb{W}^{E} ( z^T, w,w,  \bound,\vec{s}) $ and $\mathbb{W}^{E} ( z^T, w,w,  \bound,\vec{s})  \mli \mathbb{W}^{F[J_i]  }(z^S,w,w,\bound,\vec{s}) $. Therefore it proves $\mathbb{W}^{  E } ( z^T, w,u,  \bound,\vec{s})   \mli \mathbb{W}^{F[J_i]  } (z^S,w,w,\bound,\vec{s}) $, whence, by Weakening,  
\begin{equation}\label{oct26d}
\arfour\vdash \mathbb{W}^{  E } ( z^T, w,u,  \bound,\vec{s})  \mlc R  \mli \mathbb{W}^{F[J_i]  } (z^S,w,w,\bound,\vec{s}) .
\end{equation}

Next, we claim that 
\begin{equation}\label{oct26ee}
\arfour\vdash   R   \mli \mathbb{E}(z^S, w   )\mlc \mathbb{F} (z^S, w   ). 
\end{equation} 
Here is a brief justification of (\ref{oct26ee}) through reasoning in $\arfour$. Let $a$ be the counterbehavior for which $T(a)$ is true, and let $d$ be the counterbehavior for which $S(d)$ is true. Assume $R$. Then, in view of (\ref{mmm3}), the resource  $\mathbb{E}(z^T, w   )\mlc\mathbb{F} (z^T, w)$ is available for us in  unlimited supply. That is, we have full information about the configuration of $\cal X$ at time $w$ in the $(a,e_\bound)$-branch. Solving  (\ref{oct26ee}) means being able to generate full information about the configuration of $\cal X$ at time $w$ in the $(d,e_\bound)$-branch. Since the time $w$ is fixed and is the same in both cases, let us no longer explicitly mention it. 
Note that the two configurations are identical, for the exception of the contents of the run tape. So, from the resource  $\mathbb{E}(z^T, w   )\mlc\mathbb{F} (z^T, w)$ which describes the configuration of the $(a,e_\bound)$-branch, we can directly tell the (identical) state of $\cal X$ in the configuration of the $(d,e_\bound)$-branch, as well as  
the locations of all three scanning heads, and the contents of any cells of the valuation and work tapes.  Next, in order to tell the location of the leftmost blank cell  on the run tape in the configuration of the $(d,e_\bound)$-branch (or tell that the size of this location exceeds $\bound$), all we need is to compute $i\plus j\plus 1$, where $i$ is the location of the leftmost blank cell of the run tape in the configuration of the $(a,e_\bound)$-branch (unless $|i|\mgeq \bound$, in which case the size of the sought value also exceeds $\bound$), and $j$ is the location of the work-tape head in the configuration of the $(a,e_\bound)$-branch.
Finally, consider any cell $\#c$ of the run tape. If $c$ is less than the above $i$, then the content of cell $\#c$ in the configuration of the $(d,e_\bound)$-branch is the same as in the $(a,e_\bound)$-branch. Otherwise, if $c\mgeq i$, then the sought content is  the $(c \mminus i)$th symbol (starting the count of those symbols from $0$ rather than $1$) of the labmove $\oo\alpha$ --- unless   $c \mminus i$ is greater or equal to the length of this labmove, in which case the sought content of cell $\#c$ is blank. 

From (\ref{oct26ee}), by Weakening, we have
\begin{equation}\label{oct26e}
\arfour\vdash    \mathbb{W}^{E} ( z^T, w,u,  \bound,\vec{s})\mlc R   \mli \mathbb{E}(z^S, w   )\mlc \mathbb{F} (z^S, w   ). 
\end{equation}

By the induction hypothesis of our lemma, with  $F[J_i]$ and     $ \mathbb{W}^{E} ( z^T, w,u,  \bound,\vec{s})\mlc R$ in the roles of   $E $ and $R$, (\ref{oct26q}),  (\ref{oct26d}) and (\ref{oct26e}) --- which correspond to (\ref{mmm1}), (\ref{mmm2}) and (\ref{mmm3}), respectively --- imply the desired $\arfour\vdash \mathbb{W}^{E} ( z^T, w,u,\bound,\vec{s})\mlc R\mli \overline{F[J_i]}$.
 \end{subproof}

{\bf Claim 5}. {\em Assume $E(\bound,\vec{s})$ has the form $ F[\ada xJ(x)]$, and $v$ is a  variable not occurring in $E(\bound,\vec{s})$.  Then}   \[\arfour\vdash \mathbb{W}^{   E } ( z^T, w,u, \bound,\vec{s})\mlc R\mli \overline{F[J(v)]}.\vspace{7pt}\]

\begin{subproof} Assume $E=F[\ada xJ(x)]$, and $v$ is a fresh variable.  Let $\alpha$ be the string such that, for whatever constant $c$, $\alpha.c$ is the move which  brings $F[\ada xJ(x)]$ down to $F[J(c)]$. Arguing almost literally as in the proof of Claim 4, only with ``$\alpha.v$'' instead of $\alpha$ and ``$J(v)$'' instead of ``$J_i$'',  we find that $\arfour\vdash \mathbb{W}^{   E } ( z^T, w,u, \bound,\vec{s})\mlc R\mli \overline{F[J(v)]}$. The only difference and minor complication is related to the fact that, while in the proof of Claim 4 the labmove $\pp\alpha$ was constant, the corresponding labmove $\pp\alpha.v$ in the present case is not.
 Hence, its size is not given directly but rather needs to determined (while arguing within $\arfour$). No problem, this (for the ``$v$'' part of the labmove) can be done using Lemma \ref{comlen}. Similarly, various symbols of the labmove that were given directly in the proof of Claim 4 will now have to be determined using some general procedure. Again no problem: this (for the ``$v$'' part of the labmove) can be done using Lemma \ref{combit}. \end{subproof}

Now we claim that 
\begin{equation}\label{oct26c}
\mbox{$\arfour\vdash \mathbb{W}^{   E } ( z^T, w,u, \bound,\vec{s})\mlc R\mli \overline{E(\bound,\vec{s})}$.}
\end{equation}
Indeed, if $E(\bound,\vec{s})$   is not critical, then the above follows from Claims 3, 4 and 5 by the closure of $\arfour$ under Wait. Suppose now $E(\bound,\vec{s})$ is critical. Then, by Lemma \ref{august20a}, $\arfour\vdash \cle z\mathbb{W}^{   E } ( z, w,\xi(\bound), \bound,\vec{s})\mli \overline{E(\bound,\vec{s})}$. This, in view of the functionality of $T$ for $z$ and condition (\ref{mmm1}), can be easily seen to imply $\arfour\vdash \mathbb{W}^{   E } ( z^T, w,u, \bound,\vec{s})\mlc R\mli \overline{E(\bound,\vec{s})}$, as claimed.

From  (\ref{dec15b}), (\ref{dec16g}) and (\ref{oct26c}), by $\adc$-Introduction, we find that $\arfour$ proves
\begin{equation}\label{mmm5}
\Bigl(\mathbb{L}\add\ade x  \bigl(\mathbb{C}(x)\add \mathbb{D}(x) \bigr) \add \mathbb{W}^{E} ( z^T,w,u,\bound,\vec{s})\Bigr)\mlc R\mli \overline{E(\bound,\vec{s})}.
\end{equation}

In turn, the $\arfour$-provability of  (\ref{mmm4}) and (\ref{mmm5}) can be easily seen to imply the desired $\arfour$-provability of 
$ R\mli \overline{E(\bound,\vec{s})}$.  This completes our proof of Lemma \ref{oct26}.
\end{proof}

\begin{lemma}\label{dec14a}
$\arfour \vdash  \ade x\bigl(x\equals \xi(\bound)\bigr)\mli \overline{X(\bound)}$. 
\end{lemma}

\begin{idea} We take $u\equals \xi(\bound)\mlc w\equals\zero\successor$ in the role of $R$, $X$ in the role of $E$ and show that the conditions of Lemma \ref{dec19} are satisfied. Then, by Lemma \ref{oct26}, $\arfour$ proves $u\equals \xi(\bound)\mlc w\equals\zero\successor\mli \overline{X}$. And the target formula $\ade x\bigl(x\equals \xi(\bound)\bigr)\mli \overline{X}$ is an almost immediate logical consequence of the latter and Lemma \ref{zersuc}. 
\end{idea}

\begin{proof} Let $R$ be the formula $u\equals \xi(\bound)\mlc w\equals\zero\successor$. Then, of course, we have
\begin{equation}\label{j20a}
\arfour\vdash R\mli \xi(\bound)\equals u.
\end{equation}

 Let $T(z)$ be an elementary formula saying that $z$ is (the code of) the empty counterbehavior.  Obviously $\pa$ proves $\bound\notequals\zero\mli \mathbb{W}^{X}(z^T,\zero\successor,\zero\successor,\bound)$ and hence, in view of Lemma \ref{zer}, $\arfour$ proves 
$\mathbb{W}^{X}(z^T,\zero\successor,\zero\successor,\bound)$.  Therefore, as $R$ contains the condition $w\equals\zero\successor$,
\begin{equation}\label{j20c}
\arfour\vdash R\mli \mathbb{W}^{X}(z^T,w,w,\bound).
\end{equation}

Next, we observe that $\arfour \vdash \mathbb{E} ( z^T,\zero)\mlc\mathbb{F} ( z^T,\zero)$. Indeed, arguing in $\arfour$, solving both $\mathbb{E} ( z^T,\zero )$ and $\mathbb{F} ( z^T,\zero )$ is very easy as we know exactly and fully the situation in the machine at time $0$, which is nothing but the start configuration of the machine. The observation that we just made, of course, implies 
\begin{equation}\label{feb13}
\arfour \vdash v\equals\zero\mli \mathbb{E} ( z^T,v)\mlc\mathbb{F} ( z^T,v).
\end{equation}
From (\ref{feb13}), by Lemma \ref{oct17e}, we get 
\begin{equation}\label{feb13a}
\arfour \vdash v\equals\zero\mlc v\mleq w\mleq \xi(\bound)\mli \mathbb{E} ( z^T,w)\mlc\mathbb{F} ( z^T,w).
\end{equation}
Since $w\equals\zero\successor$ is a conjunct of $R$, $\pa$ obviously proves $\bound\notequals \zero \mli v\equals\zero\mli R\mli v\equals\zero\mlc v\mleq w\mleq \xi(\bound) $.\footnote{Remember that $\cal X$ runs in time $\xi(\bound)$. By definition, this means that $\pp$'s time in any play is {\em less} than $\xi(\bound)$. Hence, the term $\xi(\bound)$ cannot be $\zero$, or $\zero\mult\bound$, or anything else that always evaluates to $0$. Therefore, of course, $\pa\vdash \xi(\bound)\mgeq\bound$.}  
But, by Lemma \ref{zer}, $\arfour\vdash \bound\notequals \zero$. Hence, by Modus Ponens, $\arfour\vdash   v\equals\zero\mli R\mli v\equals\zero\mlc v\mleq w\mleq \xi(\bound) $. From here and (\ref{feb13a}), by Transitivity, we get 
\[\arfour\vdash   v\equals\zero\mli R\mli \mathbb{E} ( z^T,w)\mlc\mathbb{F} ( z^T,w) \]
whence, by $\ada$-Introduction,
\[\arfour\vdash   \ade x(x\equals\zero)\mli R\mli \mathbb{E} ( z^T,w)\mlc\mathbb{F} ( z^T,w), \]
modus-ponensing which with Axiom 8 yields 
\begin{equation}\label{j20d}
\arfour\vdash R\mli \mathbb{E} ( z^T,w )\mlc \mathbb{F} ( z^T,w ).
\end{equation}

Now, with $X(\bound)$ in the role of $E(\bound,\vec{s})$, the conditions (\ref{j20a}), (\ref{j20c}) and (\ref{j20d}) are identical to the conditions (\ref{mmm1}), (\ref{mmm2}) and (\ref{mmm3}) of Lemma \ref{dec19}. Hence, by Lemma \ref{oct26}, we have $\arfour\vdash R\mli\overline{X(\bound)}$, i.e.  
\[\arfour\vdash u\equals \xi(\bound)\mlc w\equals\zero\successor  \mli\overline{X(\bound)}.\]
From the above, by $\ada$-Introduction, we get 
\[\arfour\vdash u\equals \xi(\bound)\mlc \ade x(x\equals\zero\successor)\mli \overline{X(\bound)}.\]
But the second conjunct of the antecedent of the above formula is provable by Lemma \ref{zersuc}. Hence, we obviously  have
\(\arfour\vdash u\equals \xi(\bound) \mli\overline{X(\bound)}\) which, by $\ada$-Introduction, yields the desired $\arfour \vdash  \ade x\bigl(x\equals \xi(\bound)\bigr)\mli \overline{X(\bound)}$.
\end{proof}

\subsection{Completing the completeness proof}  By Lemma \ref{com}, \(\arfour\vdash \gneg |\xi(\bound)|\mleq\bound\add  \ade x\bigl(x\equals \xi(\bound)\bigr).\) By Lemmas \ref{pathology6} and \ref{dec14a}, we also have 
\(\arfour \vdash   \gneg |\xi(\bound)|\mleq\bound\mli \overline{X}\)
and
\(\arfour \vdash  \ade x\bigl(x\equals \xi(\bound)\bigr)\mli \overline{X}.\)
From these   three facts, by $\add$-Elimination,  $\arfour\vdash\overline{X}$. 

\section{Inherent extensional incompleteness in the general case}\label{sincom}
The extensional completeness of $\arfour$ is not a result that could be taken for granted. In this short section we argue  that, if one considers computability-in-general instead of polynomial time computability, extensional completeness is impossible to achieve for whatever recursively axiomatizable sound extension of $\arfour$. 

First of all, we need to clarify what is meant by considering computability-in-general instead of polynomial time computability. This simply means a minor readjustment of the semantics of ptarithmetic. Namely, such a readjusted semantics would be the same as the semantics we have been considering so far, with the only difference that the time complexity of the machine solving a given problem would no longer be required to be polynomial, but rather it would be allowed to be arbitrary without any restrictions. Alternatively, we can treat $\ada,\ade$ as the ordinary $\ada,\ade$ of computability logic (rather than $\ada^{\bound},\ade^{\bound}$ as done throughout the present paper), and then forget about any complexity altogether. 

In either case, our extensional incompleteness argument goes like this. Consider any system $\bf S$ in the style of $\arfour$ whose proof predicate is decidable\footnote{$\arfour$ can easily be readjusted to satisfy this condition by requiring that each logical axiom in a $\arfour$-proof be supplemented with a proof of that axiom in some known (fixed) sound and complete recursively axiomatized calculus for classical logic.}  and hence the theoremhood predicate recursively enumerable.   Assume $\bf S$ is sound in the same strong sense as $\arfour$ --- that is, there is an effective procedure that extracts an algorithmic solution (HPM) for the problem represented by any formula $F$ from any $\bf S$-proof of $F$. 
 
Let then $A(s)$ be the predicate which is true iff:
\begin{itemize}
\item $s$ is (the code of) an $\bf S$-proof of some formula of the form $\ada x\bigl(\gneg E(x)\add E(x)\bigr)$, where $E$ is elementary,  
\item and $E(s)$  is false.
\end{itemize}   

On our assumption of the soundness of $\bf S$, $A(s)$ is a decidable predicate. Namely, it  is decided by a procedure that first checks if $s$ is the code of  an $\bf S$-proof of some formula of the form $\ada x\bigl(\gneg E(x)\add E(x)\bigr)$, where $E$ is elementary. If not, it rejects. If yes, the procedure extracts from $s$ an HPM $\cal H$ which solves $\ada x\bigl(\gneg E(x)\add E(x)\bigr)$, and then simulates the play of $\cal H$ against the environment which, at the very beginning of the play, makes the move $s$, thus bringing the game down to $\gneg E(s)\add E(s)$. If, in this play, $\cal H$ responds by choosing $\gneg E(s)$, then the procedure accepts $s$; and if $\cal H$ responds by choosing $E(s)$, then the procedure rejects $s$. Obviously this procedure indeed decides the predicate $A$.

 Now, assume that $\bf S$ is extensionally complete. Since $A$ is decidable, the problem $\ada x\bigl(\gneg A(x)\add A(x)\bigr)$ has an algorithmic solution. So, for some formula $F$ with $F^\dagger=\ada x\bigl(\gneg A(x)\add A(x)\bigr)$ and some $c$, we should have that $c$ is an {\bf S}-proof of $F$. Obviously $F$ should have the form $\ada x\bigl(\gneg E(x)\add E(x)\bigr)$, where $E$ is an elementary formula with $E^\dagger(x)=A(x)$. We are now dealing with the absurd of $A(c)$ being true iff it is false.

\section{On the intensional strength of $\arfour$}\label{sculprit}

\begin{theorem}\label{feb15}
Let $X$ and $\mathbb{L}$ be as in Section \ref{s19}.  Then $\arfour\vdash \gneg \mathbb{L}\mli X$.
\end{theorem}

\begin{proof} As established in Section \ref{s19}, $\arfour\vdash \overline{X}$. By  induction on the complexity of $X$, details of which we omit, it can also easily be seen that $\arfour\vdash  \overline{X}\mli \gneg\mathbb{L}\mli X$. So, by Modus Ponens, $\arfour\vdash \gneg\mathbb{L}\mli X$.
\end{proof} 

Remember that, in Section \ref{s19},  $X$ was an arbitrary $\arfour$-formula assumed to have a polynomial time solution  under the standard interpretation $^\dagger$. And $\gneg\mathbb{L}$ was a certain true sentence of the language of classical Peano arithmetic. We showed in that section that $\arfour$ proved a certain formula $\overline{X}$ with $\overline{X}^\dagger=X^\dagger$. That is, we showed that $X$ was ``extensionally provable''. 

According to our present Theorem \ref{feb15}, in order to make $X$ also provable in the intensional sense, all we need is to add to the axioms of $\arfour$ the true elementary sentence $\gneg\mathbb{L}$. 

In philosophical terms, the import of Theorem \ref{feb15} is that the culprit of the intensional incompleteness of $\arfour$ is the (G\"{o}del's) incompleteness of its classical, elementary part. Otherwise, the ``nonelementary rest'' --- the extra-Peano axioms and the PTI rule --- of $\arfour$, as a bridge from classical arithmetic to ptarithmetic, is 
as perfect/strong as it could possibly be: it guarantees not only extensional but also intensional provability of every polynomial time computable problem   as long as all necessary  true elementary formulas are taken care of. This means that if, instead of $\pa$,  we take the truth arithmetic {\bf Th(N)} (the set of all true sentences of the language of $\pa$) as the base arithmetical theory,   the corresponding version of $\arfour$ will be not only extensionally, but also intensionally complete. Unfortunately, however, such a system will no longer be recursively axiomatizable.

So, in order to make $\arfour$ intensionally stronger, it would be sufficient to add to it new true elementary (classical) sentences, without any need for also adding some nonelementary axioms or rules of inference that deal with nonelementary formulas.  Note that this sort of an extension, even if in a language more expressive than that of $\pa$, would automatically remain sound and extensionally complete: virtually nothing in this paper relies on the fact that $\pa$ is not stronger than it really is. Thus, basing applied theories on computability logic allows us to construct ever more expressive and intensionally  strong (as well as extensionally so in the case of properly more expressive languages) theories without worrying about how to preserve soundness and extensional completeness. Among the main goals of this paper was to illustrate the scalability of computability logic rather than the virtues of the particular system $\arfour$ based on it. The latter is in a sense arbitrary, as is $\pa$ itself: in the role of the classical part of $\arfour$, we could have chosen not only any true extension of $\pa$, certain weaker-than-$\pa$ theories as well, for our proof of the extensional completeness of $\arfour$ does not require the full strength of $\pa$. The reason for not having done so is purely ``pedagogical'': $\pa$ is the simplest and best known arithmetical theory, and reasoning in it is much more relaxed, easy and safe than in weaker versions. $\arfour$ is thus the simplest and nicest representative of the wide class of ``ptarithmetics'', all enjoying the same relevant properties  as $\arfour$ does.  

 Among the potential applications of ptarithmetic-style systems is using them as formal tools for finding efficient solutions for problems (after developing reasonable theorem-provers, which, at this point, only belongs to the realm of fantasy, of course). One can think of those systems as ideally declarative programming languages, where human ``programming'' simply means stating the problem/formula whose efficient solution  is sought (for systematic usage in the future), and hence the program verification problem is non-existent. Compiling such a ``program'' means finding a proof, followed by the easy step of translating it into an assembly-language program/solution. This process of compiling may take long  but, once compiled, the program runs fast ever after. The stronger such a system is, the better the chances that a solution for a given problem will be found. Of course, what matters in this context is intensional rather than extensional strength. So, perfect strength is not achievable, but we can keep moving ever closer to it.  

One may ask why not think of simply using $\pa$ (or even, say, {\bf ZFC}) instead of $\arfour$ for the same purposes: after all,   $\pa$ is strong enough to allow us reason about polynomial time computability. This is true,  but $\pa$ is far from being a reasonable alternative to $\arfour$. First of all, 
 as a tool for finding  solutions, $\pa$ is very indirect and hence hopelessly inefficient. Pick any of the basic arithmetical functions of Section \ref{s17} and try to generate, in $\pa$, a full formal proof of the fact that the function is polynomial-time computable (or even just {\em express} this fact) to understand the difference. Such a proof would have to proceed by clumsy reasoning about {\em non-number} objects such as Turing machines and computations, which, only by good luck, happen to be amenable to being understood as numbers through encoding. In contrast, reasoning in $\arfour$ would be directly about numbers and their properties, without having to encode any foreign beasts and then try to reason about them as if they were just kind and innocent natural numbers. Secondly, even if an unimaginably strong theorem-prover succeeded in finding such a proof, there would be no direct use of it because, from a proof of the existence of a solution we cannot directly extract a solution. Furthermore, even knowing that a given HPM $\cal X$ solves the problem in {\em some} polynomial time $\xi$, would have no practical significance without knowing {\em what} particular polynomial $\xi$ is, in order to asses whether it is ``reasonable'' (such as $\bound^2$, $\bound^3$, \ldots) or takes us beyond the number of nanoseconds in the lifespan of the universe (such as $\bound^{9999999999}$). In order to actually obtain a solution and its polynomial degree, one would need a {\bf constructive} proof, that is, not just a proof that a polynomial $\xi$ and a $\xi$-time solution exist, but a proof of the fact that certain particular numbers $a$ and $b$ are (the codes of) a polynomial term $\xi$ and a $\xi$-time solution $\cal X$. This means that a theorem-prover would have to be used not just once for a single target formula, but an indefinite (intractably many) number of times, once per each possible pair of values of $a,b$ until the ``right'' values is encountered.  To summarize, $\pa$ does not provide any reasonable mechanism for handling queries in the style ``{\em find} a polynomial time solution for problem $A$'': in its standard form, $\pa$ is merely a YES/NO kind of a ``device''.   

The above dark picture can be somewhat brightened by switching from $\pa$ to Heyting's arithmetic {\bf HA} --- the version of $\pa$  based on intuitionistic logic instead of classical logic,  which is known to  allow us to directly extract, from a proof of a formula $\cle  xF(x)$, a particular value of $x$ for which $F(x)$ is true. But the question is why intuitionistic logic and not computability logic? Both claim to be ``constructive logics'', but the constructivistic claims of computability logic have a clear semantical meaning and justification, while intuitionistic logic is essentially an ad hoc invention whose constructivistic claims are mainly based on certain syntactic and hence circular considerations,\footnote{What creates circularity is the common-sense fact that syntax is merely to serve a meaningful semantics, rather than vice versa. It is hard   not to remember the following words from \cite{Japfin} here: ``The reason for the failure of $P\add\gneg P$ in computability logic is not that this principle \ldots is not included in its axioms. Rather, the failure of this principle is exactly the reason why this principle, or anything else entailing it, would not be among the axioms of a sound system for computability logic''.}   without being supported by a convincing and complete constructive semantics.   And, while {\bf HA} is immune to the second one of the two problems pointed out in the previous paragraph,  it  still suffers from the first problem. At the same time, as a reasoning tool, {\bf HA} is inferior to $\pa$, for it is intensionally weaker and, from the point of view of the philosophy of computability logic, is so for no good reasons. As a simple example, consider the function $f$ defined by ``$f(x)=x$ if $\pa$ is either consistent or inconsistent, and $f(x)=2x$ otherwise''. This is a legitimately defined function, and we all --- just as $\pa$ --- know that extensionally it is the same as the identity function $f(x)=x$. Yet, {\bf HA} can be seen to fail to prove --- in the intensional sense --- its computability.

A natural question to ask is: {\em Is there a formula $X$ of the language of $\arfour$ whose polynomial time solvability is constructively provable in $\pa$ yet $X$ is not provable in $\arfour$?} Remember that, as we agreed just a while ago, by constructive provability of the polynomial time solvability of $X$ in $\pa$ we mean that, for some particular HPM $\cal X$ and a particular polynomial (term) $\xi$, $\pa$ proves that $\cal X$ is   a $\xi$-time solution of $X$. If the answer to this question was negative, then $\pa$, while indirect and inefficient, would still have at least {\em something} to say   in its defense when competing with $\arfour$ as a problem-solving tool. But, as seen from the following theorem, the answer to the question is negative:

\begin{theorem}\label{jan30}
Let $X$ be any formula of the language of $\arfour$ such that $\pa$ constructively proves (in the above sense) the polynomial time solvability of $X$. Then $\arfour\vdash X$. 
\end{theorem}

\begin{proof} Consider   any formula $X$ of the language of $\arfour$. Assume $\pa$ constructively proves the polynomial time solvability of $X$, meaning that,
for a certain HPM $\cal X$ and a certain term $\xi$ (fix them),  
$\pa$ proves that $\cal X$ solves $X$ in    time $\xi$. But this is exactly what the formula $\mathbb{L}$ of Section  \ref{s19} denies. So, $\pa\vdash\gneg \mathbb{L}$. But, by  
Theorem \ref{feb15}, we also have $\arfour\vdash\gneg \mathbb{L}\mli X$. Consequently, $\arfour\vdash X$.  
\end{proof} 

An import of the above theorem is that, if we tried to add to $\arfour$ some new nonelementary axioms in order to achieve a properly greater intensional strength, the fact that such axioms are computable in time $\xi$ for some particular polynomial $\xi$ would have to be unprovable in $\pa$, and hence would have to be ``very nontrivial''.  The same applies to attempts to extend $\arfour$ through some new rules of inference.  

\section{Give Caesar what belongs to Caesar} The idea of exploring versions  of Peano arithmetic motivated by and related to various complexity-theoretic considerations and concepts is not  new. In this connection one should mention a solid amount of work on studying {\em bounded arithmetics}, with the usage of the usual quantifiers $\cla,\cle$ of $\pa$ restricted to forms such as $\cla x\bigl(x\mleq\tau\mli F(x)\bigr)$  and  $\cle x\bigl(x\mleq\tau\mlc F(x)\bigr)$, where $\tau$ is a term not containing $x$. Parikh \cite{Parikh} was apparently the first to tackle bounded quantifiers in arithmetic. A systematic study of bounded arithmetics and their connections to complexity theory was initiated in the seminal work \cite{Buss} by Buss. Hajek and Pudlak \cite{Hajek} give an extensive survey of this area.  The main relevant results in it can be summarized saying that, by appropriately weakening the induction axiom of $\pa$ and then further restricting it to bounded formulas of certain forms, and correspondingly readjusting the nonlogical vocabulary and axioms of $\pa$, certain  soundness and completeness   for the resulting system(s) $\bf S$ can be achieved. Such soundness results typically read like   ``If $\bf S$ proves a formula of the form $\cla x\cle y F(x,y)$, where $F$ satisfies such and such constraints, then there is function of such and such computational complexity which, for each $a$, returns a $b$ with $F(a,b)$''. And completeness results typically read  like ``For any function $f$ of such and such computational complexity, there is an $\bf S$-provable formula of the form  $\cla x\cle y F(x,y)$ such that, for any $a$ and $b$, \  $F(a,b)$ is true iff $b=f(a)$''.   
  
Among the characteristics that  make our approach very different from the above, one should point out that it {\em extends} rather than {\em restricts} the language and the deductive power of $\pa$. Restricting the language and power of $\pa$ in the style of the approach of bounded arithmetics   throws out the baby with the bath water. Not only does it expel from the system many complexity-theoretically unsound yet otherwise meaningful and useful theorems, but it apparently also reduces --- even if only in the intensional rather than extensional sense --- the class of complexity-theoretically  correct provable principles. This is a necessary sacrifice there, related to the inability of the underlying classical logic  to clearly differentiate between constructive ($\adc,\add,\ada,\ade$) and ``ordinary'', non-constructive versions ($\mlc,\mld,\cla,\cle$) of operators. Classical logic has never been meant to be a constructive logic, let alone a logic of efficient computations. Hence an attempt to still make it work as a logic of computability or efficient computability cannot go without taking a toll, and  results such as the above-mentioned soundness  can only be partial.

The problem of the partiality of the soundness results has been partially overcome in \cite{Bussint} through basing bounded arithmetic on intuitionistic logic instead of classical logic. In this case, soundness extends to all formulas of the form $\cla x\cle y F(x,y)$, without the ``$F$ satisfies such and such constraints'' condition (the reason why we still consider this sort of soundness partial is that it remains to be limited to formulas of the form $\cla x\cle y F(x,y)$, even if for arbitrary $F$s; similarly, completeness is partial because it is limited to functions only which, for us, are only special cases of computational problems). However, for reasons pointed out in the previous section, switching to intuitionistic logic signifies throwing out even more of the ``baby'' from the bath tub, further decreasing the intensional strength of the theory. In any case, whether being based on classical or intuitionistic logic, bounded arithmetics do not offer the flexibility of being amenable to being strengthened without losing soundness, and are hence ``inherently weak'' theories.   

In contrast, computability logic avoids all this trouble and sacrifices by giving Caesar  what belongs to Caesar, and  God  what belongs to God. As we had a chance to see throughout this paper, classical ($\mlc,\mld,\cla,\cle$) and constructive ($\adc,\add,\ada,\ade$) logical constructs can peacefully coexist and complement each other in one natural system that seamlessly extends the classical, constructive, resource- and complexity-conscious visions and concepts, and does so not by mechanically putting things together, but rather on the basis of one natural, all-unifying game semantics. Unlike most other approaches where only few, special-form expressions (if any) have clear computational interpretations, in our case every formula is a meaningful computational problem. Further, we can capture not only computational problems in the traditional sense, but also problems in the more general --- interactive --- sense.   That is, ptarithmetic or computability-logic-based theories in general, are by an order of magnitude more expressive and deductively powerful than the classical-logic-based $\pa$, let alone the far more limited bounded arithmetics. 

Classical logic and classical arithmetic, so close to the heart and mind of all of us,  do not at all need to be rejected or tampered with (as done in Heyting's arithmetic or bounded arithmetic) in order to achieve constructive heights. Just the opposite, they can be put in faithful and useful service to this noble goal. Our heavy reliance on reasoning in $\pa$ throughout this paper is an eloquent illustration of it. 

\section{Thoughts for the future}

The author wishes to hope that the present work is only a beginning of a longer and more in-depth line of research on exploring computability-logic-based theories (arithmetic in particular) with complexity-conscious semantics. There is an ocean of problems to tackle in this direction. 

First of all, it should be remembered that the particular language of ptarithmetic employed in this paper is only a modest fragment of the otherwise inordinately expressive and, in fact, open-ended formalism of computability logic. Attempting to extend the present results to more expressive versions of ptarithmetic is one thing that can be done in the future. Perhaps a good starting point would be considering the language employed in \cite{Japtowards} which, in addition to the present connectives, has the operator   $\fintimpl$, with $A\fintimpl B$ being the problem of reducing $B$ to $A$ where any finite number of reusages of $A$ is allowed. In a more ambitious perspective, a development of this line may yield a discovery of a series of new, complexity-conscious operators that are interesting and useful in the context of interactive computational complexity while not quite so in the ordinary context of computability-in-principle.

Another direction to continue the work started in this paper would be to try to consider complexity concepts other than polynomial time complexity.  Who knows, maybe these studies can eventually lead to a discovery of substantially new, not-yet tried weapons for attacking the famous and notorious open problems in complexity theory. Two most immediate candidates for exploration are logarithmic space and polynomial space computabilities. 
While the precise meaning of logarithmic space computability in our interactive context is yet to be elaborated, a definition of polynomial space computability comes almost for free. It can be defined exactly as we defined polynomial time computability in Section \ref{s7}, only, instead of counting the number of steps  taken by the machine ($\pp$'s time, to be more precise), we should count the number of cells ever visited by the head of the work tape. What, if any, variations  of the PTI rule (and perhaps also the nonlogical axioms) would yield  systems of {\bf psarithmetic} (``polynomial space arithmetic'') or {\bf larithmetic} (``logarithmic space arithmetic''), sound and complete with respect to polynomial space or logarithmic space computability in the same sense as $\arfour$ is sound and complete with respect to polynomial time computability?

\begin{theindex}
\item adequate counterbehavior \pageref{iadequate}
\item arity:  of function letter \pageref{iar3}; of game \pageref{igarity}; of predicate letter \pageref{iar2} 
\item  binary numeral \pageref{ibinnum}
\item binary predecessor \pageref{ibp}
\item binary $0$-successor \pageref{ibzs}
\item binary $1$-successor \pageref{ibos}
\item blind: existential quantifier \pageref{icle}; universal quantifier \pageref{icla}
\item bounded valuation \pageref{ibv}
\item branch: $(C,e)$-$\sim$ \pageref{icebranch};  $e$-computation $\sim$  \pageref{icb}
\item BSI \pageref{ibsi}
\item BSI+ \pageref{irp}
\item BPI \pageref{ibpi}
\item choice: conjunction \pageref{ichoicecon}; disjunction \pageref{ichoicedis}; existential quantifier \pageref{ichoiceeq}; implication \pageref{icimpl}; 
universal quantifier \pageref{ichoiceuq}  
\item Choose: $\add$-$\sim$ \pageref{sep10aa}; $\ade$-$\sim$ \pageref{sep10bb}
\item Choose-premise: $\add$-$\sim$ \pageref{iaddprem}; $\ade$-$\sim$ \pageref{iadeprem}
\item clarithmetic \pageref{iclarithmetic}
\item $\clthree$ \pageref{ss6}
\item $\clthree$-formula \pageref{icl3f} 
\item $\clfour$ \pageref{ss66}
\item $\clfour$-formula \pageref{icl4f}
\item $\clfour^\circ$ \pageref{icl4c}
\item $\clfour^\circ$-formula \pageref{icl4cf} 
\item $\clfour$-Instantiation
\item clarithmetization \pageref{iclarz}
\item clock cycle \pageref{icc}
\item computable (game, problem) \pageref{icomputable}; polynomial time $\sim$ \pageref{iptc}
\item computation branch \pageref{icb}
\item computation step \pageref{icc}
\item configuration \pageref{iconfiguration}
\item constant \pageref{iconstant}
\item constant game \pageref{iconstantgame}
\item counterbehavior \pageref{icnb}
\item critical formula \pageref{icritical}
\item deletion \pageref{ideletion} 
\item depth (of game) \pageref{idepth}, \pageref{fdpth}  
\item distinctive valuation \pageref{idistinctive}
\item elementarization \pageref{ielz},\pageref{ielz2},\pageref{ielz3} 
\item elementary: atom \pageref{ielat}; component \pageref{iec}; formula \pageref{ief}; game \pageref{ielem1},\pageref{ielgame2},\pageref{ielem2}; letter \pageref{iel2}; literal \pageref{iell}  
\item Elimination: $\add$-$\sim$ \pageref{iaddel}; $\ada$-$\sim$ \pageref{iadael}
\item empty run (position) \pageref{iempty}
\item extra-Peano axioms \pageref{iextra}
\item environment \pageref{ienvironment}
\item extensional (in)completeness \pageref{iextcom}  
\item formal numeral \pageref{ifornum}
\item function letter \pageref{ifl}
\item functional for $z$ \pageref{iffz}
\item game \pageref{ngame}
\item general: atom \pageref{igenat}; component \pageref{igc};   letter \pageref{igl2}; literal \pageref{igenl} 
\item HPM \pageref{ihpm}
\item hybrid  letter \pageref{ihl} 
\item hyperformula \pageref{ihf} 
\item illegal: move \pageref{iillegmove}; run \pageref{iillegrun};  $\xx$-$\sim$  \pageref{ipillegal}
\item intensional (in)completeness \pageref{iintcom}
\item initial legal (lab)move \pageref{iilm}
\item interpretation \pageref{iint}
\item Introduction: $\adc$-$\sim$ \pageref{iadcintro}; $\ada$-$\sim$ \pageref{iadaintro}
\item instance: of  game \pageref{iinstance}; of $\clfour$-formula \pageref{iinstance2}
\item label \pageref{ilabel}
\item labeled move (labmove) \pageref{ilabmove}
\item legal: move \pageref{ilegmove}; run \pageref{ilegrun}
\item literal \pageref{iell}
\item $\legal{}{}$ \pageref{ilr}
\item $\legal{}{e}$ \pageref{ilre}  
\item logical axiom \pageref{ilax}
\item logical rule \pageref{ilogr}
\item lost run \pageref{ilost}
\item machine \pageref{imachine}
\item Match \pageref{imatch}
\item Match$^\circ$ \pageref{imatchc}
\item matching literals \pageref{imatlit}
\item Modus Ponens (MP) \pageref{imp}
\item move \pageref{imove}
\item move state \pageref{imovestate}
\item negation \pageref{igneg}
\item nonlogical axiom \pageref{inlax}
\item nonlogical rule \pageref{inlr}
\item WPTI \pageref{inpti}
\item WPTI+ \pageref{irp}
\item parallel: conjunction \pageref{imlc}; disjunction \pageref{imld}
\item Peano arithmetic ({\bf PA}) \pageref{iPA},\pageref{ipa3}
\item Peano axioms \pageref{ipeanax}
\item politeral \pageref{ipoliteral}
\item polynomial time computable \pageref{iptccc}
\item polynomial time induction (PTI) \pageref{ipti}
\item polynomial time machine \pageref{iptm}
\item position \pageref{iposition}
\item predicate letter \pageref{ipl}
\item prefixation \pageref{iprefixation}
\item prompt (run, play) \pageref{iprompt}
\item $\arfour$ \pageref{iPTA},\pageref{ij31}
\item $\arfour$-formula \pageref{iarfff}
\item ptarithmetic \pageref{iPTA},\pageref{iptarithmetic2} 
\item ptarithmetization \pageref{iptarz}
\item PTI \pageref{ipti}
\item PTI+ \pageref{iptiplus}
\item quasiinstance \pageref{iquasi}
\item reduction \pageref{imli}
\item representation of problem \pageref{irepresentation}
\item run \pageref{irun};  $(C,e)$-$\sim$ \pageref{icerun}
\item   run generated by HPM \pageref{irgb}
\item run spelled by computation branch \pageref{irsb}
\item run tape \pageref{iruntape}
\item safe formula \pageref{isafe}
\item sentence \pageref{isentence}
\item solution: algorithmic \pageref{isol}; polynomial time \pageref{ipts}; uniform polynomial time \pageref{iupts}
\item standard interpretation \pageref{isi}
\item standard valuation \pageref{istandard}
\item state (of HPM) \pageref{istate}
\item static game \pageref{istatic}
\item substitution: for $\clfour$-formula \pageref{isubstitution}; for $\clfour^\circ$-formula \pageref{ish}
\item surface occurrence \pageref{isoc}
\item synchronizing \pageref{imatching}
\item term \pageref{ipterm}; $\bound$-$\sim$ \pageref{ibterm}
\item thinking period \pageref{itp}
\item time: $\pp$'s $\sim$, $\oo$'s $\sim$ \pageref{itm}
\item timestamp \pageref{itimestamp}
\item transition function \pageref{itf}
\item Transitivity \pageref{itr}
\item TROW-premise \pageref{itrow}
\item true $\arfour$-formula \pageref{itrue}
\item unary predecessor \pageref{iup}
\item unary successor \pageref{ius}
\item uniform-constructively sound rule \pageref{iucs}
\item uniform-constructive soundness of logic \pageref{iucsl}
\item unistructural game \pageref{iunistructural}
\item valuation \pageref{ivaluation}
\item valuation tape \pageref{ivaluationtape}
\item variable \pageref{ivariable}
\item Wait \pageref{iwait}
\item Wait-premise (special, ordinary) \pageref{iwaitpremise}
\item Weakening \pageref{iwk}
\item $\win{}{}$ \pageref{iwn}
\item $\win{}{e}$ \pageref{ilre} 
\item winning strategy \pageref{isol} 
\item won run\pageref{iwon}
\item work tape \pageref{iworktape}
\item \ 
\item \
\item $^\dagger$ \pageref{isi}
\item $e[A]$ \pageref{iea}
\item $\seq{\Phi}A$ \pageref{ipr}
\item $ \xx $ \pageref{ixx}
\item $\overline{\xx}$ \pageref{ixxneg}
\item $\bound$ \pageref{ipi}
\item $\pp$:  as game \pageref{itwg2}; as  player \pageref{ipp}
\item $\oo$:  as  game \pageref{itwg2}; as player \pageref{ioo}
\item $\models$ \pageref{imodels}
\item $\models^{P}$ \pageref{imodelsp}
\item $\elz{F}$ \pageref{ielz},\pageref{ielz2},\pageref{ielz3} 
\item $\ulcorner O\urcorner$ \pageref{ign}
\item  $|x|$ \pageref{iii}
\item $\lfloor s/2\rfloor$ \pageref{ifloor}
\item $\gneg$ \pageref{igneg},\pageref{igneg2}
\item $\mlc$ \pageref{imlc},\pageref{imlc2}
\item $\mld$ \pageref{imld},\pageref{iadd2}
\item $\mli$ \pageref{imli},\pageref{imli2}
\item $\cla$ \pageref{icla},\pageref{op5}
\item $\cle$ \pageref{icle},\pageref{op5}
\item $\adc$ \pageref{ichoiceop},\pageref{iadc2}
\item $\add$ \pageref{ichoiceop},\pageref{iadd2}
\item $\adi$ \pageref{icimpl}
\item $\ada$ \pageref{ichoiceop},\pageref{con1}
\item $\ade$ \pageref{ichoiceop},\pageref{con1}
\item $\ada^{\bound}$  \pageref{iadab},\pageref{ibcuq2}
\item $\ade^{\bound}$ \ \pageref{iadeb},\pageref{ibceq2}
\item $\pst^n$ \pageref{ipst}
\item $X$ \pageref{ix}
\item $\cal X$ \pageref{ixxx}
\item $\xi$ \pageref{ixi}
\item $\overline{E}$ \pageref{ipver}
\item $\hat{n}$ \pageref{ifornum}
\item $e_b$ \pageref{ieb}
\item $z^T$ \pageref{i876}
\item $\mathbb{A}$ \pageref{iaaa}
\item $\mathbb{B}$ \pageref{ibbb}
\item $\mathbb{C}$ \pageref{iccc}
\item $\mathbb{D}$ \pageref{iddd}
\item $\mathbb{E}$ \pageref{ieee}
\item $\mathbb{F}$ \pageref{ifff}
\item $\mathbb{G}$ \pageref{iggg}
\item $\mathbb{H}$ \pageref{ihhh}
\item $\mathbb{H}'$ \pageref{ihhhp}
\item $\mathbb{J}$ \pageref{ijjj}
\item $\mathbb{L}$ \pageref{illl}
\item $\mathbb{W}$ \pageref{idoubleu}
\end{theindex}


\begin{thebibliography}{99}

\bibitem{Abr94}
S. Abramsky and R. Jagadeesan. {\em Games and full completeness for multiplicative linear logic}. {\bf Journal of Symbolic Logic} 59 (1994), pp. 543-574.

\bibitem{Babai}
L. Babai and M. Shlomo. {\em Arthur-Merlin games: a randomized proof system, and a hierarchy of complexity classes}. {\bf Journal of Computer System Sciences} 36 (1988), pp. 254-276.

\bibitem{Bla72} A. Blass. {\em Degrees of indeterminacy of games}. {\bf Fundamenta Mathematicae}  77 (1972), pp. 151-166. 

\bibitem{Bla92} A. Blass. {\em A game semantics for linear logic}. {\bf Annals of Pure and Applied Logic} 56 (1992), pp. 183-220.

\bibitem{Buss} S. Buss. {\bf Bounded arithmetic} (revised version of Ph. D. thesis). Bibliopolis, 1986. 

\bibitem{Bussint} S. Buss. {\em The polynomial hierarchy and intuitionistic bounded arithmetic}.   {\bf Lecture Notes in Computer Science} 223 (1986),  pp. 77-103.

\bibitem{Chandra} A.Chandra, D. Kozen and L. Stockmeyer. {\em Alternation}.     {\bf Journal of the ACM} 28 (1981), pp. 114--133.

\bibitem{Goldwasser}  S. Goldwasser, S. Micali and C. Rackoff. {\em The knowledge complexity of interactive proof systems}. {\bf SIAM Journal on Computing} 18 (1989), pp. 186-208. 

\bibitem{Hajek} P. Hajek and P. Pudlak. {\bf Metamathematics of First-Order Arithmetic}. Springer, 1993. 

\bibitem{Jap02}  G. Japaridze. {\em The logic of tasks}. {\bf Annals of Pure and Applied Logic} 117 (2002), pp. 263-295.



\bibitem{Jap03} G. Japaridze. {\em Introduction to computability logic}. {\bf Annals of Pure and Applied Logic} 123 (2003), pp. 1-99.

\bibitem{Japtocl1} G. Japaridze. {\em Propositional computability logic I}. {\bf ACM Transactions on Computational Logic} 7 (2006), pp. 302-330.

\bibitem{Japtocl2} G. Japaridze. {\em Propositional computability logic II}. {\bf ACM Transactions on Computational Logic} 7 (2006), 
pp.  331-362.

\bibitem{Cirq} G. Japaridze. {\em Introduction to cirquent calculus and abstract resource semantics}. {\bf Journal of Logic and Computation} 16 (2006), pp. 489-532.

\bibitem{Japic} G. Japaridze. {\em Computability logic: a formal theory of interaction}. In: {\bf Interactive Computation: The New Paradigm}. D. Goldin, S. Smolka and P. Wegner, eds. Springer Verlag,  Berlin, 2006, pp. 183-223. 
   
\bibitem{Japtcs} G. Japaridze. {\em From truth to computability I}. {\bf Theoretical Computer Science} 357 (2006), pp. 100-135.

\bibitem{Japtcs2} G. Japaridze. {\em From truth to computability II}. {\bf Theoretical Computer Science} 379 (2007), pp. 20-52.

\bibitem{Japjsl} G. Japaridze. {\em The logic of interactive Turing reduction}. {\bf Journal of Symbolic Logic} 72 (2007), pp. 243-276. 

\bibitem{int1} G. Japaridze. {\em Intuitionistic computability logic}. {\bf Acta Cybernetica} 18 (2007),  pp. 77-113.  

\bibitem{Propint} G. Japaridze. {\em The intuitionistic fragment of computability logic at the propositional level}. {\bf Annals of Pure and Applied Logic} 147 (2007),  pp.187-227. 

\bibitem{Japdeep} G. Japaridze. {\em Cirquent calculus deepened}. {\bf Journal of Logic and Computation}  18 (2008),  pp. 983-1028.

\bibitem{Japseq} G. Japaridze. {\em Sequential operators in computability logic}. {\bf Information and Computation} 206 (2008),  pp. 1443-1475. 

\bibitem{Japfour} G. Japaridze. {\em Many concepts and two logics of algorithmic reduction}. {\bf Studia Logica} 91 (2009), pp. 1-24.  

\bibitem{Japfin} G. Japaridze. {\em In the beginning was game semantics}. In: {\bf Games: Unifying Logic, Language, and Philosophy}. O. Majer,
A.-V. Pietarinen and T. Tulenheimo, eds. Springer   2009,  pp. 249-350.
 


\bibitem{Japtowards} G. Japaridze. {\em Towards applied theories based on computability logic}. {\bf Journal of Symbolic Logic} (to appear in 2010).  

\bibitem{Ver} I. Mezhirov and N. Vereshchagin. {\em On abstract resource semantics and computability logic}. {\bf Journal of Computer and System Sciences} (to appear in 2010). 

\bibitem{Parikh} R. Parikh. {\em Existence and feasibility in arithmetic}. {\bf Journal of Symbolic Logic} 36 (1971), pp. 494-508.

\bibitem{Xu} W. Xu and S. Liu. {\em Knowledge representation and reasoning based on computability logic}. {\bf Journal of Jilin University}  47 (2009), pp. 1230-1236. 


\end{thebibliography}
\end{document}